\documentclass[11pt,letterpaper]{article}
\usepackage[margin=1in]{geometry}
\usepackage[T1]{fontenc}
\usepackage[utf8]{inputenc}
\usepackage{lmodern}
\usepackage{amsmath,amssymb,amsthm,mathtools}
\usepackage{enumitem,booktabs,array,tabularx,longtable}
\usepackage{xcolor}
\usepackage[hidelinks]{hyperref}
\allowdisplaybreaks
\newif\ifstoc \stocfalse
\hypersetup{bookmarksdepth=3}

\newtheorem{theorem}{Theorem}[section]
\newtheorem{lemma}[theorem]{Lemma}
\newtheorem{proposition}[theorem]{Proposition}
\newtheorem{corollary}[theorem]{Corollary}
\newtheorem{claim}[theorem]{Claim}
\newtheorem{fact}[theorem]{Fact}

\theoremstyle{definition}
\newtheorem{definition}[theorem]{Definition}

\newtheorem{convention}[theorem]{Convention}
\newtheorem{example}[theorem]{Example}

\theoremstyle{remark}
\newtheorem{remark}[theorem]{Remark}

\newcommand{\F}{\mathbb{F}}
\newcommand{\RS}{\mathrm{RS}}
\newcommand{\ev}{\mathrm{ev}}
\newcommand{\spn}{\mathrm{span}}
\newcommand{\poly}{\mathrm{poly}}
\newcommand{\polylog}{\mathrm{polylog}}
\newcommand{\NTIME}{\mathrm{NTIME}}
\newcommand{\Sim}{\mathsf{Sim}}
\newcommand{\View}{\mathrm{View}}
\newcolumntype{Y}{>{\raggedright\arraybackslash}X}
\providecommand{\arxivid}[1]{arXiv:~\texttt{\detokenize{#1}}}

\title{Zero Knowledge for Free up to Linear Query Bounds:\\
Quasilinear-Size Constant-Query Zero-Knowledge PCPs}
\author{Inhyuk Park\\ Seoul National University \and Chenglin Fan\\ Seoul National University}
\date{}

\begin{document}
\maketitle

\begin{abstract}
We show that for all polynomials $T$ and $b$, every relation in $\NTIME(T)$ has a PCP of length
$(T+b)\cdot(\log(T+b))^{O(1)}$ that is perfect zero knowledge against every adaptive, possibly malicious verifier
making at most $b$ queries. The PCP is binary, with $O(1)$ non-adaptive queries, perfect completeness and soundness
error $\frac12$. This improves the length $\tilde O(T+b^2)$ of Zeilberger, O'Connor, Gur and Spooner (2026). For
$b\le T$ the length is $T\cdot\polylog T$, the form of the shortest known constant-query PCPs without zero
knowledge. Unless $\mathrm{NP}\subseteq\mathrm{BPP}$, a PCP for an NP-complete relation cannot be zero knowledge, with a
polynomial-time simulator, against a verifier that reads all of it. So the dependence on $b$ is optimal up to the
polylogarithmic factor: zero knowledge is free up to linear query bounds. Two-round zero-knowledge interactive oracle
proofs with polylogarithmically many queries achieved length $\tilde O(T+b)$ against $b$ queries (Ben-Sasson, Chiesa,
Gabizon and Virza; Ben-Sasson, Chiesa, Forbes, Gabizon, Riabzev and Spooner; 2016--17); we obtain it for PCPs. Our
universal, black-box, straight-line simulator is exact in expected polynomial time, or in strict polynomial time with
a failure symbol (Goldreich's formulation), and has negligible statistical error when truncated to strict polynomial
time.

Gur, O'Connor and Spooner (STOC~2025) proved a zero-knowledge PCP theorem and asked, as their Open Problem~1, for
quasilinear length. This paper answers the question
in its strong reading, for every polynomial query bound, with the simulators above in place of the strict
polynomial-time perfect simulator of their definition.
The proofs were developed with extensive AI assistance and have been checked by the authors.

The construction masks the Reed--Solomon proof of proximity of Ben-Sasson and Sudan at every node of its recursion,
with same-radix bivariate lifts of both the masked word and the mask. The proof rests on a locator lemma: at every
node, whatever $k$ of its proof symbols reveal about its input is determined by $O(k)$ input values, selected
monotonically. We prove it with the coprime-order argument of Bafna and Vyas (2026), which combines the divisibility
argument of Polishchuk and Spielman with the Corvaja--Zannier bound; our treatment of domain points is new. Constant
query complexity comes from Dinur's gap amplification, which preserves zero knowledge, as Zeilberger, O'Connor, Gur
and Spooner observed. The query count and the polylogarithmic exponent are absolute but very large.
\end{abstract}

\thispagestyle{empty}
\newpage
\ifstoc\setcounter{tocdepth}{1}\fi
\tableofcontents
\newpage

\begingroup
\section{Introduction}\label{sec:intro}

\subsection{Short zero-knowledge PCPs}\label{sec:intro-problem}

The PCP theorem \cite{AS98,ALMSS98} says that every NP statement has a proof that a verifier can check by
reading a constant number of its bits. Zero-knowledge PCPs (ZK-PCPs) also ask that the proof reveal
nothing beyond the truth of the statement. A verifier that reads the whole proof learns it, so zero knowledge
can hold only against verifiers that read a bounded number of symbols. Kilian, Petrank and Tardos
\cite{KPT97} identified the two natural regimes: polynomial-length proofs that are zero knowledge against a fixed
polynomial number $b$ of queries (the \emph{query bound}), and exponential-length proofs that are zero knowledge
against every efficient verifier. Their constructions have polylogarithmic or
adaptive honest verifiers, statistical rather than perfect zero knowledge, or a polynomial gap between the honest
query complexity and the query bound.
So do the later ones of Ishai and Weiss \cite{IW14}, Ishai, Weiss and Yang \cite{IWY16}, and Hazay,
Venkitasubramaniam and Weiss \cite{HVW22}.

Gur, O'Connor and Spooner removed these restrictions. Following their perfect zero-knowledge PCPs for $\#\mathrm P$
\cite{GOS24}, they proved a \emph{zero-knowledge PCP theorem} \cite{GOS25}: for every query bound $q^*\le2^{\poly(n)}$ (their
notation for $b$), every language in NP has PCPs of length $\poly(n,q^*)$ that a non-adaptive verifier checks by
reading $O(1)$ bits and that are perfect zero knowledge against every adaptive verifier reading at most
$q^*$ bits. Their open problems begin with
proof length:

\begin{quote}\small
``\emph{Open Problem~1 (Short ZK-PCPs).} Do there exist $O(1)$-query ZK-PCPs for NP with proof
length $\tilde O(n)$? We remark that our algebraic zero-knowledge techniques are based on Reed--Muller
arithmetisation and the sumcheck protocol, whereas (non-ZK) constructions of quasilinear length PCPs are based on
Reed--Solomon arithmetisation and combinatorial gap amplification, hence new ideas are necessary for such
strengthening of our theorem.'' \hfill\cite[\S1.2]{GOS25}
\end{quote}

\noindent O'Connor's thesis \cite[\S6.2]{OCo26} states the problem with the parameters made explicit: ``Do there
exist $O(1)$-query ZK-PCPs for NP over the boolean alphabet with proof length $\tilde O(n)$, with constant
soundness error and query bound $\Omega(n^c)$ for some constant $0<c<1$.'' What we call the \emph{strong reading} of the
question, matching the zero-knowledge PCP theorem it would strengthen, asks for this for \emph{every} polynomial
query bound: binary $O(1)$-query PCPs for $\NTIME(T)$ of length $\tilde O(T+b)$ that are zero knowledge against
$b$ queries for every polynomial $b$. Here $T$ plays the role of $n$: the shortest known constant-query PCPs without zero knowledge
have length $T\cdot\polylog T$ for $\NTIME(T)$ \cite{BS08,Din07}, and the additive $b$ is forced
(Proposition~\ref{prop:length-floor}).

Two recent works approach the question from different directions. Zeilberger, O'Connor, Gur and Spooner
\cite{ZOGS26} construct, for all polynomials $T,b$, $O(1)$-query PCPs for $\NTIME(T)$ of length
$\tilde O(T+b^2)$ that are perfect zero knowledge against $b$ queries. In their words, this gives zero knowledge
``up to $b(n)\approx\sqrt{T(n)}$ essentially `for free', in an asymptotic sense''. In O'Connor's formulation this
answers the question for $c\le\frac12$, and since that formulation asks for \emph{some} constant $c$, \cite{ZOGS26}
already meets it literally. What distinguishes our result is the range $\frac12<c\le1$ and beyond, that is, the
linear dependence on $b$. Bafna and Vyas \cite{BV26} construct \emph{private} PCPs from product expansion of tensor
Reed--Solomon codes on multiplicative grids of coprime orders. For $k=2$ unconditionally, and for larger $k$ under
their product-expansion conjecture, they obtain PCPs of quasilinear length $N$ over an alphabet of polynomial size,
with $O(kN^{1/k})$ queries. These PCPs are private against sets of positions of small measure under the verifier's
query distribution, not against a bounded number of queries \cite[Def.~8.2, Cor.~8.3]{BV26}. Bafna and Vyas
also give a perfect zero-knowledge PCP of length at most $N$ against $\Omega_k(N/\log n)$ queries, with no
efficient simulator \cite[Cor.~8.6]{BV26}; Section~\ref{sec:related} states both guarantees precisely.

\subsection{Our results}\label{sec:intro-results}

Our main theorem gives, for every polynomial query bound $b$, binary $O(1)$-query PCPs for $\NTIME(T)$ of
length $(T+b)\cdot\polylog(T+b)$ that are perfect zero knowledge against $b$ queries.
Section~\ref{sec:defs} defines the zero-knowledge notions, and Section~\ref{sec:main} proves the theorem as
Theorem~\ref{main:thm:main}. The query bound counts \emph{distinct} queries; running times count each call to
the verifier as one step.

\begin{theorem}[main theorem]\label{thm:intro-main}
For all time-constructible polynomials $T(n)\ge n$ and $b(n)\ge1$ (including $b\gg T$), every relation
$R\in\NTIME(T)$ has a PCP with the following properties.
\begin{enumerate}[label=(\roman*),nosep]
\item \emph{Proof and verifier.} The alphabet is $\{0,1\}$. The verifier runs in time $\poly(|x|)$, is non-adaptive,
makes $Q_{\rm fin}$ queries for an
absolute constant $Q_{\rm fin}$, and tosses $O(\log(T+b))$ fair coins.
\item \emph{Completeness and soundness.} Completeness is perfect and the soundness error is at most $\frac12$.
\item \emph{Length.} The proof has length at most $c_{\rm len}(T+b)\cdot(\ln(T+b))^{C}$ for an absolute constant $C$
and a constant $c_{\rm len}$ that may depend on $R$, $T$ and $b$ but not on the input.
\item \emph{Prover.} Given a witness, the honest prover runs in expected polynomial time and its output law is
exact (Las Vegas).
\item \emph{Zero knowledge.} There is one universal, black-box, straight-line simulator $\Sim$ such that for every
adaptive, possibly malicious and computationally unbounded verifier $V^*$ that makes at most $b$ distinct queries,
and every $(x,w)\in R$, the output $\Sim^{V^*}(x)$ is distributed identically to the view of $V^*$ on the honest
proof. $\Sim$ never sees $w$ and runs in expected time $p(|x|)\cdot(1+Q^*)+O(|\rho|)$, where $p$ is a polynomial,
$Q^*$ is the number of queries of $V^*$ counted with multiplicity and $|\rho|$ is the number of coins $V^*$ tosses
(the term $O(|\rho|)$ pays for tossing $V^*$'s coins, which are part of its view, and is added to every time bound in
this item; it is polynomial when $V^*$ tosses polynomially many coins). Run on fair coins, with a pool of a fixed
polynomial number of field elements drawn up front by bounded rejection sampling, $\Sim$ runs in strict polynomial
time and outputs a failure symbol $\bot$ with probability at most $\frac12$, being exact conditioned on not failing
(perfect zero knowledge in the form of \cite[Def.~4.3.1]{Gol01}); and for every
polynomial $k$, truncating $\Sim$ gives a strict polynomial-time simulator whose output is within statistical distance
$2^{-k(|x|)}$ of the view.
\end{enumerate}
\end{theorem}

\paragraph{Zero knowledge for free up to linear query bounds.} For $b\le T$ the length in
Theorem~\ref{thm:intro-main} is $T\cdot\polylog T$, the form of the shortest known constant-query PCPs without zero
knowledge \cite{BS08,Din07}. For $b>T$ it is $b\cdot\polylog b$, and the dependence on $b$ is optimal up to the
polylogarithmic factor. Indeed, by Proposition~\ref{prop:length-floor}, unless $\mathrm{NP}\subseteq\mathrm{BPP}$, a PCP for an
NP-complete relation that is zero knowledge against $b$ queries, for a polynomial $b$ and with a simulator polynomial
in the input length, has length greater than $b$ on infinitely many input lengths. The constants $Q_{\rm fin}$ and
$C$ are absolute and very large: with the constants of Dinur's analysis \cite{Din07}, as repaired in Section~\ref{sec:thmD}, the query count is at least $10^{13}$
and $C$ is galactic (Section~\ref{sec:intro-limits}). We make no claim about how $C$ compares with the exponent of
the non-zero-knowledge constructions \cite{BS08,Din07}, whose $\tilde O$ also hides the exponents of \cite{Din07}.

\begin{corollary}[Open Problem~1 of \cite{GOS25}]\label{cor:intro-op1}
For every polynomial $q^*(n)$, every relation $R$ such that $|w|\le T(|x|)$ and membership of $(x,w)$ is decided in
time $T(|x|)$ with $T(n)=n\cdot\polylog n$ (for instance Circuit-SAT and 3SAT, with $n$ the instance size) has a binary PCP with $O(1)$
non-adaptive queries, perfect completeness, soundness error $\frac12$ and length $(n+q^*)\cdot\polylog(n+q^*)$, which
is perfect zero knowledge against every adaptive verifier making at most $q^*$ queries, with the simulators of
Theorem~\ref{thm:intro-main}(v). In particular, for every constant $0<c\le1$ there are such PCPs of length
$\tilde O(n)$ with query bound $n^c$, which is O'Connor's formulation \cite{OCo26} with the constant $c=1$ allowed.
\end{corollary}

\begin{remark}[quasilinear $T$, by padding]\label{rem:intro-quasilinear}
Theorem~\ref{thm:intro-main} is stated, as proved in Section~\ref{sec:main}, for polynomial $T$; Corollary~\ref{cor:intro-op1}
follows from it by padding. Let $T(n)\le n\lceil\log_2 n\rceil^{c'}$, put
$N(n):=n\lceil\log_2n\rceil^{c'}+1$ and $R':=\{(x\#1^{N(|x|)-|x|-1},w):(x,w)\in R\}$. Since $N$ is computable from $|x|$ in
time $O(N)$, $R'\in\NTIME(T')$ with $T'(N)=c_0N$ for a constant $c_0$, a time-constructible polynomial. Let $\hat q$ be a
nondecreasing time-constructible polynomial with $\hat q\ge q^*$ and $\hat q=O(q^*)$ (the sum of the absolute values of
the monomials of $q^*$) and apply Theorem~\ref{thm:intro-main} to $R'$ with $b(N):=\hat q(N)\ge q^*(n)$. The verifier and
prover for $x$ pad $x$ and run those for $R'$; completeness, soundness and zero knowledge transfer verbatim, because the
proof for $x$ is the proof for the padded instance and $x\in L(R)$ iff its padding lies in $L(R')$. The length is
$(N+\hat q(N))\polylog(N+\hat q(N))=(n+q^*(n))\polylog(n+q^*(n))$, since $N\le n\polylog n$ and
$\hat q(N)\le\hat q(n)\polylog n$.
\end{remark}

\noindent Theorem~\ref{thm:intro-main} answers Open Problem~1 of \cite{GOS25} in its strong reading, for every polynomial query bound, with the
simulators of item~(v). \cite{GOS25} ask for a perfect simulator in strict
polynomial time. Our simulator meets three standard relaxations: perfect in expected
polynomial time \cite[\S4.3]{Gol01}, perfect in strict polynomial time with a failure symbol \cite[Def.~4.3.1]{Gol01},
and statistical with negligible error in strict polynomial time. For our construction, no strict polynomial-time
simulator on fair coins is perfect without a failure symbol: the honest prover samples uniform elements of a field
of odd characteristic, so some view probabilities are not dyadic (Remark~\ref{main:rem:strict}).
Section~\ref{sec:defs-gos} compares the definitions.

\paragraph{The locator lemma.} Our new technical result, the locator lemma (Theorem~\ref{kl:thm:KL}), concerns
the linear algebra of masked Reed--Solomon lifts. Section~\ref{sec:overview-node} gives the setting. A \emph{node} of
our construction (in what we call the \emph{same-radix configuration}) consists of an evaluation domain $D$, a union
of at most eight cosets of a multiplicative group $\mu_m$ with $m=m_1m_2$, $\gcd(m_1,m_2)=1$, $m_1<m_2\le2m_1$, in a
field of characteristic $p>m$; an input $g\in\RS[D,m]$; a uniformly random mask $r\in\RS[D,m]$; and proof symbols.
These are $y=g+r$ and $r$ on $D$, and the radix-$m_1$ bivariate lifts of $y$ and $r$ on grids $A_y\times B$ and
$A_r\times B$, where $A_y,A_r$ are unions of seven cosets of $\mu_{m_1}$ each, disjoint from each other and from $D$,
and $B=\{x^{m_1}:x\in D\}$. The \emph{row}~$\beta$ holds the lift symbols at $(z,\beta)$ and the domain symbols at
the $x$ with $x^{m_1}=\beta$. The \emph{leak} of a set $V$ of proof symbols is the space of linear functionals of $g$
through which the distribution of the symbols in $V$ depends on $g$.

\begin{theorem}[locator lemma, informal; see Theorems~\ref{main:thm:KL} and~\ref{kl:thm:KL}]\label{thm:intro-key}
There are absolute constants $c_*=10^{-13}$ and $a\approx2.0\cdot10^{26}$ such that at every node the
following holds. For a set $V$ of proof symbols, let $W(V)\subseteq D$ consist of the domain points read by $V$, the
fibres $\{x\in D:x^{m_1}=\beta\}$ of the rows $\beta\in B$ that carry at least $m_1+1-\lfloor c_*m_1\rfloor$ points of
$V$, and, if $|V|\ge(\lfloor c_*m_1\rfloor+1)^2$, one whole coset of $\mu_m$ in $D$. Then $W$ is monotone and
computable in polynomial time, $|W(V)|\le a|V|$, and the leak of $V$ is spanned by evaluations of $g$ at points of
$W(V)$; that is, the distribution of the symbols in $V$ depends only on $g|_{W(V)}$.
\end{theorem}

\noindent In words: each domain point read by $V$ costs one input value; a heavy row (at least
$m_1+1-\lfloor c_*m_1\rfloor$ points of $V$) costs its fibre of $m_1$ points, at most $\frac{16}{15}$ per point of
$V$ on it; and the whole coset, charged only when $|V|\ge(\lfloor c_*m_1\rfloor+1)^2$, costs at most $2/c_*^2$ per
symbol of $V$, which is almost all of $a$. No condition on the choice of the cosets is needed beyond the disjointness of $A_y$, $A_r$
and $D$: the lemma holds for every such choice with $m<p$ (Section~\ref{sec:overview-leak} shows why $A_y$ and $A_r$
must be disjoint). Its proof, sketched in Section~\ref{sec:overview-key} and given in Section~\ref{sec:keylemma}, runs the divisibility argument of Polishchuk and
Spielman \cite{PS94}, in the version of Bafna and Vyas \cite[\S4.4]{BV26} for multiplicative groups of coprime
orders, on grids that are unions of cosets. Like theirs, it goes through the Corvaja--Zannier bound, which we
use directly \cite[Thm.~2]{CZ13}. A new normalisation handles the domain points of the view. Over a whole recursion on a root
domain of size $s$, the lemma gives our masked Reed--Solomon proof of proximity a \emph{lazy sampler}, which samples
the proof exactly, query by query, and makes at most $(\ln s)^{\log_2a}$ input reads per query, amortised
(Section~\ref{sec:overview-tree}).

\paragraph{Zero-knowledge-preserving gap amplification.} For $O(1)$ queries we use Dinur's gap
amplification \cite{Din07}. That it is local, and hence preserves zero knowledge by the local-computability lemma
\cite[Lemma~3.2]{GOS25}, was observed and sketched by Zeilberger, O'Connor, Gur and Spooner
\cite[App.~A, Lemma~A.1]{ZOGS26}. Section~\ref{sec:thmD} states and proves it in full, repairing six points in the
printed proofs of \cite{Din07} (listed at its start), and claims no new technique. In
Theorem~\ref{thm:intro-D}, $Q_D$, $E_C$, $E_L$ are constants of Dinur's amplification
(Section~\ref{main:sec:D}) and $a$, $q$, $s$ describe the given PCP.

\begin{theorem}[zero-knowledge-preserving gap amplification, informal; see Theorem~\ref{thmD:thm}]\label{thm:intro-D}
Let $(P_0,V_0)$ be a non-adaptive PCP with alphabet $\{0,1\}^{a}$, $q$ queries, residual decision circuits of size at
most $s$, perfect completeness and rejection probability at least $\gamma$ on NO instances. Then there is a binary PCP
$(P',V')$ with $Q_D$ non-adaptive queries, perfect completeness, soundness error at most $\frac12$ and length
$|U|\cdot(s/\gamma)^{O(E_C)}$, where $|U|$ is the number of tests of $V_0$, such that $P'(x,w)=f_x(P_0(x,w))$ for a
deterministic, efficiently computable map $f_x$, depending only on $x$, each of whose output bits depends on at most
$L_D=q\cdot(s/\gamma)^{O(E_L)}$ symbols of its input. Consequently, if $(P_0,V_0)$ is perfect (or statistical) zero
knowledge against $b$ queries, then $(P',V')$ is perfect (or statistical) zero knowledge against $\lfloor b/L_D\rfloor$
queries, with a simulator that is straight-line if the original one is and, when $b$ is polynomial in $|x|$, inherits its running-time class.
\end{theorem}

\subsection{Comparison with previous work}\label{sec:intro-compare}

Table~\ref{tab:compare} compares PCPs with zero knowledge or privacy, including the recent preprints
\cite{ZOGS26,BV26}. Ours is not the first quasilinear zero-knowledge PCP, and the $\tilde O(T)$ term of its length
matches that of \cite{ZOGS26}.

\begin{table}[ht]
\centering\small
\begin{tabularx}{\linewidth}{>{\raggedright\arraybackslash}p{2.3cm}>{\raggedright\arraybackslash}p{1.6cm}>{\raggedright\arraybackslash}p{1.7cm}>{\raggedright\arraybackslash}X>{\raggedright\arraybackslash}p{4.8cm}}
\toprule
work & alphabet & queries & length & privacy / zero knowledge\\\midrule
\cite{BS08,Din07} & binary & $O(1)$ & $T\polylog T$ & none\\
\cite{GOS25}, Thm.~2 & binary & $O(1)$ & $\poly(n,b)$ & perfect ZK against $b$ queries, efficient simulator\\
\cite{ZOGS26} & binary & $O(1)$ & $\tilde O(T+b^2)$ & perfect ZK against $b$ queries, efficient simulator\\
\cite{BV26}, Cor.~8.6 & $\F_q$, $q=\poly(n)$ & $O(kN^{1/k})$ & $N=n\log n\,e^{O(k\log k)}$ & perfect ZK against $\Omega_k(N/\log n)$ queries; no efficient simulator; $k>2$ conditional\\
\textbf{this paper} & binary & $O(1)$ & $(T+b)\polylog(T+b)$ & perfect ZK against $b$ queries; simulator in expected polynomial time (strict: failure symbol or statistical ZK)\\
\bottomrule
\end{tabularx}
\caption{PCPs with zero knowledge or privacy for $\NTIME(T)$ (for \cite{GOS25} and \cite{BV26}: NP, with $n$ the
instance size); all are constant-query except \cite{BV26}. The query bound is $b$.}\label{tab:compare}
\end{table}

\noindent In the \emph{interactive} setting the linear dependence on $b$ is known: \cite{BCGV16} (for NP) and
\cite[Thms.~1.1 and~8.1]{BCFGRS17} (for $\NTIME(T)$) give two-round interactive oracle proofs of length
$\tilde O(T+b)$ with polylogarithmically many non-adaptive queries that are perfect zero knowledge against $b$
queries. Their verifier's challenge randomises the input of the Reed--Solomon proof of proximity.
\cite{ZOGS26} remove the interaction with length $\tilde O(T+b^2)$. We remove it and keep the linear dependence on~$b$:
the per-node masking and the locator lemma play the challenge's role inside a single proof. Compared with \cite{BV26}, our
PCP is binary with $O(1)$ queries and has an efficient simulator. Their private PCP \cite[Cor.~8.3]{BV26} (not the
one in the table) also hides unsatisfying assignments, which zero knowledge does not require, but only from sets of
positions of small verifier-induced measure.

Compared with \cite{GOS25}, whose verifier is non-adaptive like ours, our length is quasilinear and additive in the
query bound, but we treat polynomial query bounds only, with a Las Vegas honest prover
(Section~\ref{sec:defs-gos}(g)--(h)). Our simulator's running time differs from that of \cite{GOS25,ZOGS26} because
of the sampling model. The obstruction of Section~\ref{sec:defs-gos}(e) to a strict polynomial-time perfect simulator
on fair coins without a failure symbol applies to any construction whose honest proof is a non-constant function of
uniform elements of a field of odd order, such as ours: our proof of the locator lemma needs characteristic $p>m$,
which excludes characteristic two.
Given an oracle for uniform field elements, which the fair-coin model of \cite{GOS25} lacks, our simulator is exact
and runs in strict polynomial time (Definition~\ref{main:def:PPTF}; Step~4 of the proof of Theorem~\ref{main:thm:main}).

\subsection{Why new ideas were needed}\label{sec:intro-why}

The remark of \cite{GOS25} quoted above locates the difficulty: the known quasilinear constant-query PCPs
\cite{BS08,Din07} are built from Reed--Solomon proofs of proximity (RS PCPPs), which recursively test bivariate
lifts of a univariate polynomial along rows and columns. The zero-knowledge techniques of \cite{GOS24,GOS25} are
tied instead to Reed--Muller codes and the sumcheck protocol.

\paragraph{The simulation paradigm and the square root.} The approach of \cite{BCGV16,BCFGRS17}, which
\cite{ZOGS26} carry over to PCPs, masks the witness encoding by a random codeword of a $t$-wise independent code. It
reduces zero knowledge to the existence of a lazy sampler for the RS PCPP that, after $k$ queries, has made
at most $\bar\kappa k$ input reads. Zero knowledge then holds against about $t/\bar\kappa$ queries, so
taking $t\approx\bar\kappa b$ gives length about $T+\bar\kappa b$. For the PCPP of \cite{BS08} on a domain of size
$s$, a single row or column of a lift is a whole Reed--Solomon codeword of length about $\sqrt s$. The analysis of
\cite[Thm.~7.18]{ZOGS26} through locally simulatable sheaf codes gives a $\sqrt s$-local simulator,
that is, $\bar\kappa\approx\sqrt s$. Since the domain has more than $t\approx\sqrt s\,b$ points, $s$ is at least
about $b^2$, hence the~$b^2$. Call a simulation of a view \emph{point-charged} if it learns about the input only by
reading values $g(x)$ at domain points $x$.
For point-charged simulation the square root recurs in our setting under several natural departures from the
hypotheses of our locator lemma (Appendix~\ref{sec:negatives}): a single lift on the standard layout of
\cite{BS08}, shared columns, degrees not divisible by the radix, and unbounded or free column sets for coprime
radices. In each case, simulating some views costs $\Theta(\sqrt s)$ input reads per symbol read on balanced splits,
$m_1\approx m_2$. These are limits of point-charged simulation, not lower bounds for zero-knowledge PCPs.

\paragraph{What we do instead.} We keep the RS PCPP of \cite{BS08} but mask it \emph{at every node of its
recursion}: each node holds a fresh random mask $r$ of its input $g$, the masked word $y=g+r$, and lifts of both $y$
and $r$. The leak of a view is then the intersection $\spn\Phi\cap\spn\Psi$ of the spans of the functionals read on the
two sides (Lemma~\ref{main:lem:DL}; cf.\ \cite[Lemma~4.3, Claim~7.7]{BCFGRS17}), and the locator lemma bounds it by $O(1)$ input reads per symbol read. Composed
over the $O(\log\log s)$ levels of the recursion, this gives $\bar\kappa=\polylog s$ instead of $\sqrt s$, so the
length is $(T+b)\polylog(T+b)$ already for the polylogarithmic-query PCP. Zero-knowledge-preserving gap amplification
then reduces the query count to $O(1)$ at the cost of another polylogarithmic factor in the length. Our proof needs
two lifts of the \emph{same} radix: both then live in one tensor product of Reed--Solomon codes, so
every leak comes from a dual codeword of a product code, to which the divisibility argument of Polishchuk and Spielman \cite{PS94} applies.

\paragraph{Standard composition of \cite{BV26} certifies only about $N^{1-1/k}$ queries.} The zero-knowledge PCP of
\cite[Cor.~8.6]{BV26} is zero knowledge against $\Omega_k(N/\log n)$ queries, nearly linear in its length, but its
verifier reads $O(kN^{1/k})$ symbols over a large alphabet. Reducing this to $O(1)$ binary queries by proof
composition \cite{BGHSV06} attaches to each outer test an inner proof that depends on all $\Theta(N^{1/k})$ symbols
the test reads. The local-computability lemma \cite[Lemma~3.2]{GOS25} then divides the query bound by
$\Theta(N^{1/k})$, certifying zero knowledge only against about $N^{1-1/k}$ queries.

\subsection{Limitations and open questions}\label{sec:intro-limits}

\begin{itemize}[nosep]
\item \emph{Constants.} The query count $Q_{\rm fin}=3Q_D$ and the exponent $C$ are absolute and given by explicit
formulas in the constants of \cite{Din07}, whose values are known only through lower bounds (the multiplicative length
constant and the size thresholds are not computed). Here $Q_D=6\lceil1/\alpha_{\mathcal D}\rceil$, where
$\alpha_{\mathcal D}\in(0,1)$ is the gap ceiling of Dinur's amplification theorem \cite[Thm.~1.5]{Din07}; with the
analysis of \cite{Din07} as repaired in Section~\ref{sec:thmD}, $Q_D$ is at least $10^{13}$ (Section~\ref{main:sec:D}). The exponent $C$ is
dominated by a multiple of $E_C+E_L$, where $E_C$ and $E_L$ are the size and locality exponents of Dinur's
amplification (galactic, since $E_C$ is doubly exponential in Dinur's powering parameter). The constant
$a\approx2.0\cdot10^{26}$ of our locator lemma contributes only the additive term $\log_2a\approx87.4$
(Appendix~\ref{app:params}). The result is a feasibility result. The exponents of \cite{ZOGS26} have the same
origin; our gain is in the dependence on~$b$ (Section~\ref{prm:sec-read}).
\item \emph{Asymptotics.} The locator lemma is non-vacuous only at nodes with $m_1\ge1/c_*=10^{13}$. Below that size
the locator takes a whole coset, which is still within the ratio $a$, so the main theorem holds at every size, with
very large constants. For domain sizes up to about $2^{128}$ the recursion has depth at most~$2$
(Remark~\ref{main:rem:regime}). Every polylogarithmic factor is polylogarithmic only asymptotically.
\item \emph{Simulator.} The perfect simulator runs in expected polynomial time (or in strict polynomial time with a
failure symbol). With a negligible instead of a constant failure probability, our strict polynomial-time simulator has negligible statistical error;
for this construction, no strict polynomial-time simulator on fair coins is perfect without a failure symbol
(Section~\ref{sec:defs-gos}(e)).
\end{itemize}

\paragraph{Open questions.} (i)~Can the constants be made moderate? Most of $C$ comes from Dinur's gap amplification,
which might be replaced by a zero-knowledge-preserving constant-query step with small constants
(Section~\ref{prm:sec-read}). (ii)~Can a PCP with the parameters of Theorem~\ref{thm:intro-main} have a strict
polynomial-time perfect simulator on fair coins without a failure symbol, as the definitions of \cite{GOS25} require? Ours cannot
(Section~\ref{sec:defs-gos}(e)). (iii)~Does Theorem~\ref{thm:intro-main} extend to superpolynomial query bounds,
which the class of \cite[Def.~2.12]{GOS25} includes (Section~\ref{sec:defs-gos}(g))? Section~\ref{neg:place}
lists questions about point-charged simulation.

\subsection{Organisation}\label{sec:intro-org}

Section~\ref{sec:overview} outlines the proof. Section~\ref{sec:defs} gives definitions, compares them with those of
\cite{GOS25} and proves the length floor (Proposition~\ref{prop:length-floor}). Section~\ref{sec:main} describes
the construction and proves Theorem~\ref{thm:intro-main} from the locator lemma
(Theorem~\ref{main:thm:KL}) and zero-knowledge-preserving gap amplification (Theorem~\ref{main:thm:D}).
Sections~\ref{sec:keylemma} and~\ref{sec:thmD} prove them as Theorems~\ref{kl:thm:KL} and~\ref{thmD:thm}. Section~\ref{sec:related} discusses related work and lists every ingredient taken from prior work
(Table~\ref{tab:credits}). Appendix~\ref{app:params} lists the parameters and constants, and
Appendix~\ref{sec:negatives} shows that natural weakenings of the hypotheses of the locator lemma fail for
point-charged simulation. An AI methodology statement follows the appendices.

\section{Technical overview}\label{sec:overview}

The proof has three layers. An outer \emph{shell} reduces zero knowledge for $\NTIME(T)$ to a Reed--Solomon PCP of
proximity with a lazy sampler that makes few amortised input reads per query (proved in
Section~\ref{main:sec:shell}). The RS PCPP, masked at every node of its recursion, has such a sampler by the locator lemma
(Sections~\ref{main:sec:nodes}--\ref{main:sec:zk} and~\ref{sec:keylemma}). Zero-knowledge-preserving gap amplification
then reduces the query count to a constant and divides the zero-knowledge query bound only by its polylogarithmic
locality (Section~\ref{sec:thmD}). Only the
second layer is new in substance; the first is the method of \cite{BCGV16,BCFGRS17} in the PCP form of \cite{ZOGS26},
and the third makes rigorous an observation of \cite{ZOGS26}.

\subsection{The shell: from $\NTIME(T)$ to masked Reed--Solomon proximity}\label{sec:overview-shell}

The quasilinear Cook--Levin theorem \cite{Coo88} turns an instance $x$ of $R\in\NTIME(T)$ into a 3CNF formula $\Phi_x$
of size $O(T\log T)$, normalised so that every variable occurs at most four times \cite{Tov84}. We arithmetise it
over a prime field $\F_q$: $\Phi_x$ is satisfiable if and only if some table $A:\mu_h\to\F$ satisfies
$\mathcal Q\big(z,(A(\chi_\psi z))_{\psi\in\Psi}\big)=0$ for all $z\in\mu_h$. Here $\mathcal Q$ has degree at most $3$
in the table values, $|\Psi|=\eta=O(\log T)$, and the multipliers $\chi_\psi\in\mu_h$ implement the variable copies and
a Bene\v{s} routing network \cite{Ben64,Wak68} that brings every literal to its clause. We use a routine
multiplicative variant of the routing-network arithmetisation of \cite{BCGT13r} (additive setting), which goes
back to \cite{BFLS91,PS94}.
The prover randomises the interpolant $A_0$ of a satisfying table as $w_0=A_0+Z_H\cdot R$, where $Z_H=X^h-1$ vanishes on
$H=\mu_h$ and $R$ is uniformly random of degree $<t$, and writes $w_0$ on a domain $\hat H$ of size $s$ disjoint from
$\mu_h$. Then $w_0$ is $t$-wise independent on $\hat H$ \cite{BCGV16}, lies in $\RS[\hat H,h+t]$, and
$w_1:=\mathcal Q(X,(w_0(\chi_\psi X))_\psi)$ is divisible by $Z_H$. The proof is $w_0$, $w_1$ and RS PCPPs for the
proximity of $w_0$ and $w_1/Z_H$ to Reed--Solomon codes of the right degrees; the verifier also checks the identity
defining $w_1$ at a random point. Soundness is the usual distance argument.

Zero knowledge follows the paradigm of \cite{BCGV16,BCFGRS17}, in the PCP form of \cite[\S8]{ZOGS26} and with the
straight-line simulators of \cite{VW26}. Suppose the RS PCPP has a \emph{lazy sampler}: an algorithm that answers proof
queries one at a time, with exactly the law of the honest proof for its input, and that has made at most
$\bar\kappa k$ input reads after $k$ queries. A malicious verifier making $b_A$ queries then causes at most
$\eta\bar\kappa b_A$ reads of $w_0$; by $t$-wise independence, the simulator answers them with fresh uniform field
elements if $\eta\bar\kappa b_A\le t$. So the shell is perfect zero knowledge against
$b_A=\lfloor t/(\eta\bar\kappa)\rfloor$ queries, and taking $t\approx b\cdot\eta\bar\kappa\cdot\polylog(T+b)$ adds
$t\cdot\polylog(T+b)$ to the length (Section~\ref{sec:overview-params}). Everything therefore hinges on $\bar\kappa$,
the amortised number of input reads per query.

\subsection{The RS PCPP, masked at every node}\label{sec:overview-node}

For a factorisation $m=m_1m_2$, the RS PCPP of \cite{BS08} tests that $f:D\to\F$ is close to $\RS[D,m]$ by
writing on a grid a bivariate polynomial of bidegree $<(m_1,m_2)$ with value $f(x)$ at $(x,x^{m_1})$ and recursively testing that its rows
and columns are low-degree. We use the multiplicative version: $D$ is a union of eight cosets of $\mu_m$;
$m_1$ and $m_2$ are coprime with $m_1<m_2<2m_1$; and rows and columns are again unions of eight cosets of
$\mu_{m_1}$ and $\mu_{m_2}$. Each level takes roughly a square root of the degree, so the recursion has depth
$O(\log\log|D|)$. The deterministic canonical rule of Definition~\ref{main:def:can} fixes all cosets;
Section~\ref{main:sec:family} builds the domain family.
In this unmasked PCPP, charging one off-domain lift symbol takes at least $m_1=\Theta(\sqrt{|D|})$ input reads
(Section~\ref{sec:overview-leak}), so we \emph{mask at every node}. A node with input $g\in\RS[D,m]$ samples a
uniform mask $r\in\RS[D,m]$ and writes $y=g+r$ and $r$ on $D$, together with the radix-$m_1$ lifts
\[
Q_f(X,Y)=\sum_{i<m_1,\ j<m_2}\hat f_{i+m_1j}\,X^iY^j\qquad(f\in\{y,r\}),\qquad Q_f(x,x^{m_1})=f(x),
\]
of \emph{both} $y$ and $r$, on grids $A_y\times B$ and $A_r\times B$. Here $\hat f_k$ is the coefficient of $X^k$ in
$f$; $A_y,A_r$ are unions of seven cosets of $\mu_{m_1}$ each, disjoint from each other and from $D$; and
$B=\{x^{m_1}:x\in D\}$ is a union of eight cosets of $\mu_{m_2}$. This is an instance of the same-radix configuration
of Section~\ref{sec:intro-results}; from now on every \emph{node} has this form. The verifier checks $g=y-r$ at a
random point of $D$ and recurses into random rows and columns of both lifts: rows of degree $<m_1$ on $F_\beta\cup A_y$ or
$F_\beta\cup A_r$, where $F_\beta=\{x\in D:x^{m_1}=\beta\}$, and columns of degree $<m_2$ on $B$. Every
child input is a codeword for every value of the mask, and every child uses fresh, independent masks. Soundness at
each level comes from the bivariate test of Polishchuk and Spielman in the form of \cite[Lemma~6.13]{BS08}.

\subsection{Leaks and locators}\label{sec:overview-leak}

Fix a node and a set $V$ of its own proof symbols. The $y$-side symbols are linear functionals $\Phi$ of the
coefficient vector $\hat y=\hat g+\hat r$, and the $r$-side symbols are linear functionals $\Psi$ of $\hat r$.
Since $\hat r$ is uniform, the joint law of the answers depends on $\hat g$ only through the functionals in
\[
\Lambda(V)=\spn\Phi\cap\spn\Psi,
\]
the \emph{leak} of $V$ (Lemma~\ref{main:lem:DL}; cf.\ \cite[Lemma~4.3, Claim~7.7]{BCFGRS17}). A \emph{locator} is a map $V\mapsto W(V)\subseteq D$ with
$\Lambda(V)\subseteq\spn\{\ev_x:x\in W(V)\}$; then the law of the answers to $V$ depends only on $g|_{W(V)}$. We need
locators that are \emph{monotone} ($V\subseteq V'\Rightarrow W(V)\subseteq W(V')$, so that a lazy sampler can commit
to input reads as queries arrive), computable in polynomial time, and \emph{small}: $|W(V)|\le a|V|$ for an absolute
constant $a$. The locator lemma (Theorem~\ref{thm:intro-key}) provides such locators.

\paragraph{What leaks: small views.} The costs below are for point-charged simulation (Section~\ref{sec:intro-why}), not lower
bounds for zero-knowledge PCPs (Appendix~\ref{sec:negatives}). (1)~\emph{An unmasked lift
of $g$.} A symbol $Q_g(z,\beta)$ with $z\notin D$ combines the $m_1$ values of $g$ on the fibre $F_\beta$ with nonzero
coefficients (Lemma~\ref{neg:lem:CR}). This functional is not spanned by fewer than $m_1=\Theta(\sqrt{|D|})$
evaluations on $D$ (Proposition~\ref{neg:prop:SM}), so an exact sampler makes at least $m_1$ input reads for one
query (Lemma~\ref{neg:lem:samp}). (2)~\emph{A shared column $z\in(A_y\cap A_r)\setminus D$.} The symbols
$Q_y(z,\beta)$ and $Q_r(z,\beta)$ give their difference $Q_g(z,\beta)$, again at least $m_1$ input reads for two
queries (Proposition~\ref{neg:prop:SM}); this is why $A_y$ and $A_r$ are disjoint. (3)~\emph{Rows are cheap.} A leak
carried by one row $\beta$ lies in $\spn\{\ev_x:x\in F_\beta\}$ (Lemma~\ref{kl:lem:cheap}); a view with about $m_1$ points
on the row can pay for $F_\beta$. Leaks that no single row produces do occur
(Proposition~\ref{neg:prop:SATSR}); for views of fewer than $(\lfloor c_*m_1\rfloor+1)^2$ symbols the locator lemma still charges them to the domain points read and
the fibres of heavy rows (Theorem~\ref{thm:intro-key}).

\subsection{From node locators to a lazy sampler}\label{sec:overview-tree}

A child's input is a line of its parent's proof, so the input values a child's locator asks for are parent
symbols, which the parent's locator charges in turn. Write $\mathrm{Sym}(v)$ for the symbols of a node $v$ itself,
$\pi_w$ for the subtree proof of a child $w$, and $\iota_w$ for the map placing $w$'s input in $\mathrm{Sym}(v)$. Charging
bottom-up gives, for a set $S$ of positions in the subtree proof of $v$,
\[
W^T_v(S)=W_v\Big(S\cap\mathrm{Sym}(v)\ \cup\ \bigcup_{w\text{ child of }v}\iota_w\big(W^T_w(S\cap\pi_w)\big)\Big).
\]
Every child's input is a codeword for every value of $v$'s mask, and child masks are independent, so induction on
the height shows that $W^T_v(S)$ is a dependence set for $S$: two codeword inputs that agree on it induce the same law
on the answers to $S$ (Lemma~\ref{main:lem:C}). It is monotone, and each of the at most $\log_2\ln s$ levels multiplies its size by at most $a$, so $|W^T(S)|\le(\ln s)^{\log_2a}|S|$. The lazy sampler answers a
new query $\alpha$ by reading the input on the new points of $W=W^T(S\cup\{\alpha\})$, choosing a codeword $f^*$
matching the reads, sampling all masks uniformly subject to the earlier answers, and answering $\alpha$ from the
resulting proof. Exactness follows from the dependence property; after $k$ queries the sampler has made at most
$\bar\kappa k$ input reads, for \emph{every} oracle, with $\bar\kappa=(\ln s)^{\log_2a}$.

\subsection{Proof idea for the locator lemma}\label{sec:overview-key}

\emph{Plan.} We pass one curve $E=0$ of low bidegree through all points of a small view $V$; an incidence bound
shows that only the row and column factors of $E$ matter, so the leak lives on rows with nearly $m_1$ points of $V$. Write $n=m_1$ (not the input
length) and $n'=m_2$. The reshaping $f\mapsto Q_f$ identifies $\RS[D,m]$ with the tensor code
$\F[X]_{<n}\otimes\F[Y]_{<n'}$. So every proof symbol is an evaluation of a bivariate polynomial of bidegree
$<(n,n')$: a lift symbol at a grid point $(z,\beta)\in A\times B$, where $A:=A_y\sqcup A_r$, and a domain symbol $f(x)$ at the point
$(x,x^n)$, which lies on the row $\beta=x^n$. Let $S=A\sqcup D'$, where $D'$ is the set of domain abscissae read by $V$.
A linear dependency among the functionals of $V$ is a matrix $M\in C_1\otimes\F^B+\F^S\otimes C_2$ supported on $V$,
where $C_1=\RS[S,n]^\perp$ and $C_2=\RS[B,n']^\perp$; the leak is the $y$-side part of its component in
$C_1\otimes\F^B$ (Section~\ref{sec:keylemma}).

Let $\tau=\lfloor c_*n\rfloor$. If $|V|<(\tau+1)^2$, there is a nonzero polynomial $E$ of bidegree at most $(\tau,\tau)$
vanishing at all points of $V$; a single $E$ serves every dependency $M$, since $M$ is supported in $V$. The
divisibility argument of Polishchuk and Spielman \cite[Lemma~8, Thm.~9]{PS94}, in the version of
Bafna and Vyas \cite[\S4.4]{BV26} for multiplicative groups of coprime orders, writes the components of $M$ through
$E$. For every irreducible factor $F$ of $E$ that is neither a row nor a column factor, it then compares the zeros and
poles of a rational function on the curve $F=0$. If such an $F$ survives, the curve must pass through at least $(n+1-2\tau)e'+(n'+1-2\tau)e$ points
of the lift grid or of the view's domain points, where $(e,e')$ is the bidegree of $F$.
Line counts alone give no contradiction: for $e=e'$, the $8n'$ rows and the $14n$ lift columns allow up to $8n'e$ and
$14ne$ points, both more than this. On the
other hand, such a curve carries only $O(m^{1/3}(ee')^{2/3}+ee')$ grid points when $m<p$, by the bound of Corvaja and
Zannier \cite[Thm.~2]{CZ13} on rational points of curves on products of cosets, applied to each of the $14\cdot8$
pairs of a column coset and a row coset. For $c_*$ small enough, this contradicts the lower bound. So the $C_1\otimes\F^B$ component is carried by row factors of $E$.
Each nonzero row of it is a dual Reed--Solomon word of weight at least $n+1$; it differs from $M$ only on the at most
$\tau$ columns on which $E$ vanishes identically; so the row carries at least $n+1-\tau$ points of $V$. These are the heavy rows of
Theorem~\ref{thm:intro-key}. Their fibres $F_\beta$ have size $n$, so they cost at most $\frac{16}{15}|V|$; the domain
points read by $V$ cost at most $|V|$; and when $|V|\ge(\tau+1)^2$ one coset of $\mu_m$ (a basis of the dual of the
code) costs $m<2n^2<(2/c_*^2)|V|$. Hence $a=1+\frac{16}{15}+2/c_*^2\approx2.0\cdot10^{26}$.

One step is, as far as we know, new: the treatment of domain points. Since $D$ is a union of $8n'$ cosets of $\mu_n$,
taking all of $D$ as columns would make the number of column cosets grow with $n'$ and destroy the incidence bound. Only the read
abscissae $D'$ are columns, each with one point $(x,x^n)$. We normalise the decomposition on them and divide
out $\prod_{x\in D'}(X-x)$; this moves the loss in the zero count from all of $D'\times B$ to the view's own domain
points on the curve (at most $\min(e+ne',\,8n'e)$), which a case analysis absorbs. A single polynomial $E$
through $V$ yields one monotone set of rows for every dependency, with explicit thresholds; this is an
alternative to the known route from product expansion to a monotone locator through collective product expansion and closure
\cite[Lemma~5.13, Prop.~5.14]{DLV24}, \cite[Lemmas~17--18]{KP25}, \cite[Def.~4.7, Lemma~4.8]{BV26}.

\subsection{From polylogarithmic queries to a constant}\label{sec:overview-D}

In the notation of Theorem~\ref{thm:intro-D}, the shell with the masked RS PCPP is a non-adaptive PCP with
$q=\polylog(T+b)$ queries over an alphabet of
$O(\log(T+b))$ bits, residual circuits of size $s=\polylog(T+b)$, rejection probability at least
$\gamma=1/\polylog(T+b)$ on NO instances, and perfect zero knowledge against $b_A$ queries. Dinur's gap amplification \cite{Din07} turns it into a binary PCP with $O(1)$ queries
and constant soundness error. The amplification is \emph{local}: each round (preprocessing, powering by walks of length
$t_{\mathcal D}$ on a $d_{\mathcal D}$-regular expander, and composition with the long-code assignment tester)
computes every new label from at most $c_{\rm it}=2d_{\mathcal D}^{\lceil t_{\mathcal D}/2\rceil}$ old labels,
deterministically and without the witness. A Tseitin front end \cite{Tse68} keeps every round on Dinur's alphabet.
After $k=O(\log(s/\gamma))$ rounds every output bit depends on at most $L_D=q\,c_{\rm it}^k=\polylog(T+b)$ symbols of
the original proof, and the local-computability lemma \cite[Lemma~3.2]{GOS25} (in a cache-and-cap form that keeps the
simulator straight-line and exact) gives zero knowledge against $\lfloor b_A/L_D\rfloor$ queries
(Theorem~\ref{thm:intro-D}). Zeilberger, O'Connor, Gur and Spooner observed that the amplification is local and
hence preserves zero knowledge \cite[App.~A, Lemma~A.1]{ZOGS26}; Section~\ref{sec:thmD} proves this in full.

\subsection{Parameters}\label{sec:overview-params}

The division by $L_D$ is paid for in advance, at a polylogarithmic cost in length: choosing the independence
parameter $t=\lceil b\cdot\Xi\rceil$ with an explicit polylogarithmic $\Xi\ge2L_D\eta\bar\kappa$ gives $b_A\ge bL_D$,
and hence zero knowledge against $b$ queries after amplification. The length is
$O((T\log^2T+b\,\Xi)\polylog(T+b))$ before amplification and at most $c_{\rm len}(T+b)(\ln(T+b))^{C}$ after it, with
$C$ dominated by a multiple of $E_C+E_L$ (Section~\ref{sec:main}). The query count is $3Q_D$: drawing a test with fair coins at
most halves the rejection probability, so the final verifier runs three independent tests (Lemma~\ref{main:lem:CC}). The perfect simulator samples uniform field elements by rejection
sampling on fair coins, in expected polynomial time; truncating the rejection loops gives the strict polynomial-time
statistical simulator (Section~\ref{sec:defs-gos}). Appendix~\ref{app:params} lists all constants.

\section{Definitions}\label{sec:defs}

This section fixes the notions used in Theorem~\ref{thm:intro-main} (Sections~\ref{sec:defs-notation}--\ref{sec:defs-zk}),
compares them with those of \cite{GOS25} (Section~\ref{sec:defs-gos}), proves the length floor behind the optimality
claim (Proposition~\ref{prop:length-floor}) and recalls the local-computability lemma used in Section~\ref{sec:thmD}.

\subsection{Notation}\label{sec:defs-notation}

For a prime power $r$, $\F_r$ is the field with $r$ elements; for $m\mid r-1$, $\mu_m\le\F_r^\times$ is the
subgroup of order $m$, and a $\mu_m$-coset is a set $z\mu_m$. For a finite $D\subseteq\F$ and $u\le|D|$,
$\RS[D,u]=\{p|_D:p\in\F[X],\ \deg p<u\}$. For $x\in\F$, $\ev_x$ is the evaluation functional on coefficient vectors.
Logarithms are to base~2 unless written $\ln$. $\NTIME(T)$ refers to nondeterministic multitape Turing machines; a
function $T$ is time-constructible if $T(n)$ can be computed from $1^n$ in time $O(T(n))$. A relation $R$ defines the
language $L(R)=\{x:\exists w,\ (x,w)\in R\}$, and $R\in\NTIME(T)$ means that membership of $(x,w)$ is decided in time
$T(|x|)$ with $|w|\le T(|x|)$.

Some symbols are local to a section. The letter $n$ is the radix $m_1$ in Sections~\ref{sec:overview-key}
and~\ref{sec:keylemma}, and the order of the root group in Section~\ref{sec:main}, where the input length is written
$n_{\rm in}$. In Sections~\ref{main:sec:nodes}--\ref{main:sec:zk} and~\ref{sec:keylemma}, $b$ is a row of a node, not
the query bound. In Section~\ref{sec:thmD}, $a$, $q$ and $s$ are the alphabet width, query count and circuit size of a
PCP. Appendix~\ref{sec:negatives} fixes its own notation.

\subsection{PCPs}\label{sec:defs-pcp}

\begin{definition}[PCP]\label{def:pcp}
A \emph{PCP} for a relation $R$ over an alphabet $\Sigma$ is a pair $(P,V)$ of a randomised prover $P$, which on
$(x,w)\in R$ outputs a proof $\pi\in\Sigma^{\ell(|x|)}$, and a probabilistic verifier $V$ that runs in time
$\poly(|x|)$ with oracle access to a proof. It has
\emph{perfect completeness} if $V^\pi(x)$ accepts with probability~$1$ for every $(x,w)\in R$ and every $\pi$ in the
support of $P(x,w)$, and \emph{soundness error} $\varepsilon$ if $\Pr[V^{\pi^*}(x)\text{ accepts}]\le\varepsilon$ for every
$x\notin L(R)$ and every $\pi^*\in\Sigma^{\ell(|x|)}$. Its \emph{length} is $\ell$; its \emph{query complexity} is the
maximum number of proof symbols $V$ reads; its \emph{randomness} is the number of fair coins $V$ tosses. $V$ is
\emph{non-adaptive} if it chooses all its query positions from $x$ and its coins before reading the proof. We use the
equivalent description of a non-adaptive verifier by a multiset $U$ of \emph{tests}, each a tuple of positions with a
Boolean decision circuit, of which $V$ draws one uniformly; a final step converts the multiset into fair coins
(Section~\ref{sec:main}).
\end{definition}

\noindent All PCPs in this paper have non-adaptive honest verifiers. For $\NTIME(T)$ we go through the quasilinear
Cook--Levin reduction; the resulting formula's satisfying assignment (the tableau) is a deterministic function of
$(x,w)$, so a simulator that works for every tableau works for $R$.

\subsection{Views and zero knowledge}\label{sec:defs-zk}

\begin{definition}[verifier strategies and views]\label{def:view}
A \emph{verifier strategy} $V^*$ is an oracle algorithm with its own coins $\rho$; it may be adaptive, malicious and
computationally unbounded. It is \emph{$b$-bounded} if on every execution, with every oracle, it makes at most $b$
\emph{distinct} queries. For $(x,w)\in R$, $\View_{V^*,P}(x,w)$ consists of $\rho$ and the sequence of query--answer
pairs of $V^*$ on input $x$ with oracle $\pi\leftarrow P(x,w)$ (\cite[Def.~2.9]{GOS25}).
\end{definition}

\begin{definition}[zero knowledge]\label{def:zk}
A PCP $(P,V)$ for $R$ is \emph{black-box perfect zero knowledge with query bound $b$} if there is an algorithm
$\Sim$ with oracle access to $V^*$ such that for every $b$-bounded $V^*$ and every $(x,w)\in R$,
$\Sim^{V^*}(x)$ is distributed identically to $\View_{V^*,P}(x,w)$ (\cite[Def.~2.11]{GOS25}). It is
\emph{$\varepsilon$-statistical} if the statistical distance is at most $\varepsilon(|x|)$. $\Sim$ is \emph{universal}
(one simulator for all $V^*$) and never sees $w$. It is \emph{straight-line} if it runs $V^*$ at most once, answering each
query when it is made, and outputs the view of that run or a failure symbol (a fixed value that is not a view, such
as $\bot$). Running time counts each call to $V^*$ as one step. $\Sim$ also tosses $V^*$'s coins $\rho$,
which are part of its output, at one step per coin, and no other step depends on $|\rho|$. Every simulator running
time in this paper, expected or strict, bounds the other steps and carries these $|\rho|$ steps additively, where
$|\rho|$ is the number of coins $V^*$ tosses on the run; this term is polynomial whenever $V^*$ tosses polynomially many
coins, in particular when $V^*$ runs in polynomial time.
\end{definition}

\begin{convention}[query counting]\label{conv:counting}
Zero knowledge is required against verifiers bounded in \emph{distinct} queries; this is the stronger requirement,
since a verifier that repeats queries learns nothing new. Simulator \emph{running time} is measured as a function of
$|x|$ and of the number $Q^*$ of queries of $V^*$ counted \emph{with multiplicity}: a simulator ``runs in time
$p(|x|)\cdot(1+Q^*)$''. Equivalently, one may give the simulator a memoised interface to $V^*$ in which a call returns
the next query not asked before; then $Q^*\le b$. Against verifiers that make at most $b$ queries counted with
multiplicity, as in \cite{GOS25}, either form means polynomial time.
\end{convention}

\noindent We list the simulator notions that appear in the literature; (Z1) implies each of the others, which are
incomparable in general:
\begin{enumerate}[label=(Z\arabic*),nosep]
\item perfect, with one universal simulator in strict probabilistic polynomial time (\cite[Def.~2.11]{GOS25});
\item perfect, with a simulator for each $V^*$ in strict probabilistic polynomial time (\cite[Def.~2.10]{GOS25});
\item perfect, with one universal simulator in \emph{expected} polynomial time, the standard relaxation for perfect
zero knowledge (see \cite[\S4.3]{Gol01}, which also discusses simulators that may output a failure symbol);
\item statistical with negligible error, with one universal simulator in strict probabilistic polynomial time (the
notion of \cite{KPT97,IW14}).
\end{enumerate}
Theorem~\ref{thm:intro-main} gives (Z3) and (Z4), with a straight-line black-box simulator in both cases. In
Section~\ref{sec:main} these are the classes (EP) and ($\varepsilon$-SP) with negligible $\varepsilon$
(Definition~\ref{main:def:ZK}).

\subsection{Comparison with the definitions of \cite{GOS25}}\label{sec:defs-gos}

We list every point at which Theorem~\ref{thm:intro-main} could differ from the definitions of
\cite[Defs.~2.1 and~2.9--2.12]{GOS25}.
\begin{enumerate}[label=(\alph*),nosep]
\item \emph{Views} (Def.~2.9): identical (coins and answers; we also record the queries, which are determined by them).
\item \emph{Adaptivity and maliciousness} (Defs.~2.10, 2.11): identical; neither definition restricts the verifier's
running time.
\item \emph{Universality and black-box access} (Def.~2.11): identical; our simulator is moreover straight-line.
\item \emph{Query bound}: \cite{GOS25} speak of ``$q^*$-query-bounded'' verifiers; we bound distinct queries, which
covers every verifier bounded with multiplicity (Convention~\ref{conv:counting}).
\item \emph{Running time}: \cite{GOS25} ask for a PPT simulator. Like theirs \cite[Constructions~3.3--3.4]{GOS25},
ours writes out $V^*$'s coins, so its time carries the additive $O(|\rho|)$ of Definition~\ref{def:zk}, which is
polynomial whenever $V^*$ tosses polynomially many coins. Up to that term, ours is (Z3), expected polynomial time with perfect
simulation; for this construction also strict polynomial time with a failure symbol $\bot$ of probability at
most $\frac12$ and exact output conditioned on not failing (\cite[Def.~4.3.1]{Gol01}; the field elements the simulator
needs are drawn as an up-front pool of a fixed polynomial size, each by a bounded number of rejection trials, and the
simulator then runs on the pool and its other fair coins, item (F-b) of Section~\ref{sec:main}; sampling on
demand would
\emph{not} give this, because the number of draws can depend on the adaptive verifier's path and conditioning on
success would reweight paths);
or (Z4), strict polynomial time with statistical error $2^{-k}$ for any polynomial $k$. It is \emph{neither} (Z1)
\emph{nor} (Z2): for our construction no strict-PPT simulator on fair coins is perfect without a failure symbol as
soon as some bit of the final proof is non-constant
under the honest prover, and such a bit exists (Remark~\ref{main:rem:strict}). Indeed, the honest prover's randomness
consists of independent uniform elements of the prime field $\F_q$, $q$ odd (Section~\ref{main:sec:conv}, item (F-e)), and the proof
is a deterministic function of them, so for a deterministic $V^*$ (one that tosses no coins) every view probability
lies in $q^{-K}\mathbb Z$ for some $K$; a non-constant bit
has probability $j/q^{K}$ with $0<j<q^{K}$, which is not of the form $i/2^{r}$ because $q$ is odd; but a strict-PPT
simulator with $r$ fair coins outputs every view with probability in $2^{-r}\mathbb Z$. The deterministic verifier
that reads that one bit separates them. (Our proof of the locator lemma
needs a field whose characteristic exceeds the group orders it uses, so the construction cannot be moved to
characteristic two as it stands.)
\item \emph{Alphabet, queries, randomness} (Def.~2.12): binary alphabet, $O(1)$ non-adaptive queries, randomness
$O(\log(T+b))$, matching the class $\mathrm{PZK\text{-}PCP}_{b,\{0,1\}}[\log(T+b),1]$ of \cite[Def.~2.12]{GOS25} up to
the running-time point~(e).
\item \emph{Range of query bounds}: \cite[Def.~2.12]{GOS25} intersects over all $q^*\le2^{\poly(n)}$; we treat polynomial
$b$ only, as does the summary of their result that precedes Open Problem~1 (``for any polynomial $q^*$'',
\cite[\S1.2]{GOS25}).
\item \emph{Honest prover}: \cite{GOS25} call a PCP efficient if the prover runs in polynomial time given a witness; ours
runs in expected polynomial time with an exact output law (Las Vegas), because it samples uniform field elements by
rejection sampling on fair coins.
\item \emph{Verifier} (Def.~2.1): \cite{GOS25} ask for a PPT verifier; ours runs in strict polynomial time
(Theorem~\ref{main:thm:main}(1), by Theorem~\ref{thmD:thm}\ref{thmD:it:a} and Lemma~\ref{main:lem:CC}).
\end{enumerate}

\subsection{The query bound cannot exceed the length}\label{sec:defs-floor}

The following observation is folklore (compare \cite{GO94} as used in \cite[App.~A]{BCFGRS17}, and \cite{IMSX15}); we
include its short proof because it shows that the dependence on $b$ in Theorem~\ref{thm:intro-main} is optimal up to the
polylogarithmic factor.

\begin{proposition}[length floor]\label{prop:length-floor}
Let $(P,V)$ be a PCP for a relation $R$ with a polynomial-time verifier, perfect completeness and soundness error at most
$\frac12$, of length $\ell(n)$ computable from $1^n$ in time $\poly(n)$, which is black-box perfect zero knowledge with
query bound $b(n)\le\poly(n)$ via a simulator running in expected time $p(|x|)\cdot(1+Q^*)$ for a polynomial $p$
(Convention~\ref{conv:counting}; or $\frac1{10}$-statistical zero knowledge via a simulator running in strict time
$p(|x|)\cdot(1+Q^*)$). If
$\ell(n)\le b(n)$ for all but finitely many $n$, then $L(R)\in\mathrm{BPP}$. In particular, if $R$ is NP-complete and
$\mathrm{NP}\not\subseteq\mathrm{BPP}$, then $\ell(n)>b(n)$ for infinitely many $n$.
\end{proposition}

\begin{proof}
Let $n_0$ be such that $\ell(n)\le b(n)$ for all $n\ge n_0$; inputs shorter than $n_0$ are finitely many and are
decided by table look-up. For $|x|\ge n_0$ let $V^*_{\rm all}$ be the deterministic verifier that queries positions
$1,\dots,\ell(|x|)$ in order; it is $b$-bounded because $\ell\le b$. Its view is the whole proof. It makes
$Q^*=\ell(|x|)$ queries and tosses no coins, so $t(|x|):=p(|x|)\cdot(1+\ell(|x|))$ bounds the expected running time of
$\Sim^{V^*_{\rm all}}(x)$ (in the statistical case its strict running time; skip the cut-off below); $t$ is a
polynomial, as $\ell\le b\le\poly(n)$, and is computable in polynomial time, as $\ell$ is. Consider the
algorithm $\mathcal A$ that on input $x$ runs $\Sim^{V^*_{\rm all}}(x)$ for at most $10\,t(|x|)$ steps, rejects if it has not
halted or if its output is not a string $\tilde\pi\in\Sigma^{\ell(|x|)}$, and otherwise runs $V^{\tilde\pi}(x)$ with fresh
coins and outputs its decision. $\mathcal A$ runs in polynomial time: it computes $\ell(|x|)$ and $t(|x|)$, answers
each call of $\Sim$ to $V^*_{\rm all}$ with the next position, and $V$ runs in polynomial time.

If $x\in L(R)$, fix $w$ with $(x,w)\in R$. The output of the untruncated simulator is distributed as $P(x,w)$, every proof
in whose support $V$ accepts with probability~$1$; by Markov's inequality the cut-off fires with probability at most
$\frac1{10}$; so $\mathcal A$ accepts with probability at least $\frac9{10}$ (at least $\frac9{10}$ in the statistical case
too, since the accepting event has probability~$1$ under $P(x,w)$ and the simulated proof is $\frac1{10}$-close). If
$x\notin L(R)$, then for \emph{every} string $\tilde\pi$ the verifier accepts with probability at most $\frac12$, so
$\mathcal A$ accepts with probability at most $\frac12$, whatever the simulator outputs on $x$. Standard amplification gives
$L(R)\in\mathrm{BPP}$. The last sentence is the contrapositive for NP-complete $R$.
\end{proof}

\noindent Theorem~\ref{thm:intro-main} has $\ell\le(T+b)\polylog(T+b)$; by Proposition~\ref{prop:length-floor} no
construction meeting its hypotheses (in particular with polynomial $b$, as here) can have $\ell\le b$ on all large
inputs for NP-complete $R$ (unless $\mathrm{NP}\subseteq\mathrm{BPP}$),
so for $b\ge T$ the length is optimal up to the polylogarithmic factor.

\subsection{Local computability}\label{sec:defs-local}

The step from polylogarithmic to constant queries uses the local-computability lemma of \cite{GOS25}: if every symbol of
a proof $\pi$ is computed, by a fixed efficient deterministic procedure that depends on $x$ but not on the witness, from
at most $L$ symbols of a proof $\pi_0$ that is zero knowledge with query bound $b$, then $\pi$ is zero knowledge with
query bound $\lfloor b/L\rfloor$ \cite[Def.~3.1, Lemma~3.2]{GOS25}. Section~\ref{sec:thmD} uses it in a cache-and-cap form
that keeps the simulator exact and straight-line and states the two hypotheses that the lemma uses implicitly (the map
does not depend on the witness, and it is defined on every proof string).

\endgroup

\begingroup
\providecommand{\F}{\mathbb{F}}
\providecommand{\RS}{\mathrm{RS}}
\providecommand{\ev}{\mathrm{ev}}
\providecommand{\Sym}{\mathrm{Sym}}
\providecommand{\spn}{\mathrm{span}}
\providecommand{\PPTF}{\mathrm{PPT}^{\F}}
\providecommand{\Hh}{\hat H}
\providecommand{\Wa}{\mathcal W}
\providecommand{\rcan}{\mathfrak r^{\mathrm{can}}}
\providecommand{\NN}{\mathsf N}
\providecommand{\LL}{\mathsf L}
\providecommand{\aKL}{a_*}

\section{The construction and the main theorem}\label{sec:main}\label{main:sec}

This section describes the construction and proves our main theorem. Prior work is credited at each use.

\subsection{Roadmap and credits}\label{main:sec:statement}

Theorem~\ref{main:thm:main} at the end of this section is the precise form of Theorem~\ref{thm:intro-main}. Its proof
rests on two results proved later: the locator lemma at recursion nodes (Theorem~\ref{main:thm:KL}, proved in
Section~\ref{sec:keylemma}) and zero-knowledge-preserving gap amplification (Theorem~\ref{main:thm:D}, proved in
Section~\ref{sec:thmD}). Section~\ref{main:sec:conv} fixes conventions and simulator classes,
Section~\ref{main:sec:family} builds the domain family, Section~\ref{main:sec:nodes} defines nodes and the recursion,
and Section~\ref{main:sec:KL} shows that the locator lemma applies at every recursion node.
Section~\ref{main:sec:pcpp} defines the masked RS PCPP and proves its soundness, and Section~\ref{main:sec:zk} gives
its lazy sampler (Theorem~\ref{main:thm:Z}). Section~\ref{main:sec:shell} builds the shell (Theorem~\ref{main:thm:A}),
Section~\ref{main:sec:D} checks the hypotheses of amplification for it, and Section~\ref{main:sec:main} assembles the
main theorem. Sections~\ref{sec:intro-results} and~\ref{sec:defs-gos} discuss the reading of Open Problem~1 of
\cite{GOS25} that Theorem~\ref{main:thm:main} meets and the zero-knowledge notions it achieves.

\paragraph{Prior work used in the construction.} Every ingredient that is not ours is listed, with the place where we
use it, in Table~\ref{tab:credits}; we credit it again at each use. \emph{What is new here:} the overall assembly, the
domain family of Section~\ref{main:sec:family}, the canonical choice of cosets (Definition~\ref{main:def:can}), the
same-radix configuration in the recursion, Lemmas~\ref{main:lem:N} and~\ref{main:lem:V} (every node and every view met
in the recursion satisfies the hypotheses of the locator lemma), and the explicit parameters. The new ingredient of
the locator lemma is its treatment of domain points (the domain-column refinement of Section~\ref{kl:sec:div}).

\subsection{Conventions}\label{main:sec:conv}

\paragraph{Fields, cosets, codes.} Throughout, $q$ is an odd prime and $\F=\F_q$; put $\lambda_q=\lceil\log_2q\rceil$.
In this section $n$ denotes the order of the root group (Lemma~\ref{main:lem:SM2}); the input length is written
$n_{\rm in}$.
For $m\mid q-1$ let $\mu_m=\{z\in\F^\times:z^m=1\}$. A \emph{$\mu_m$-coset} is a set $z\mu_m$. For finite
$D\subseteq\F^\times$ and $u\le|D|$, $\RS[D,u]=\{p|_D:p\in\F[X],\ \deg p<u\}$; $\delta(\cdot,\cdot)$ is
relative Hamming distance and $\Delta(\cdot,\cdot)$ the same for pairs. For $\hat g\in\F^m$ and $x\in\F$,
$\ev_x(\hat g)=\sum_j\hat g_jx^j$.

\begin{lemma}[subgroups]\label{main:lem:mu}
If $m\mid q-1$ then $|\mu_m|=m$; for $z,z'\in\F^\times$, $z\mu_m=z'\mu_m$ iff $z^m=z'^m$; and if
$\gcd(m_1,m_2)=1$ then $\mu_{m_1}\cap\mu_{m_2}=\{1\}$ and $\mu_{m_1}\mu_{m_2}=\mu_{m_1m_2}$.
\end{lemma}
\begin{proof}
$\F^\times$ is cyclic of order $q-1$; the rest holds in every cyclic group.
\end{proof}

\paragraph{PCPs, coin multisets, binary reading.} A non-adaptive verifier is a finite multiset $U$ of
\emph{tests}; a test is a tuple of proof positions with a \emph{residual predicate}, a Boolean circuit on
the binary encodings of the symbols read; the verifier draws a test uniformly from $U$. \emph{Queries} is
the maximal tuple length; \emph{perfect completeness} means every proof in the support of the honest
prover passes every test; \emph{gap} $\gamma$ means every proof of a NO instance fails at least a
$\gamma$ fraction of $U$. A PCPP for a code $C$ (input $f$, proof $\pi$) has \emph{strong soundness with
function $\rho$} if $\Pr[\text{reject}]\ge\rho(\delta(f,C))$ for every $f,\pi$. A symbol of $\F$ is stored
as $\lambda_q$ bits; $\mathrm{dec}$ maps a $\lambda_q$-bit block to its value if it is $<q$ and to $0$ otherwise. All
soundness statements are made for $\mathrm{dec}(\pi)$, so they hold for every binary string.

\begin{lemma}[coin convention]\label{main:lem:CC}
Let a PCP have coin multiset $U$, perfect completeness and soundness error $\le\frac12$. Let
$R_U=\lceil\log_2|U|\rceil+1$; let $V'$ draw $j\in\{0,1\}^{R_U}$ uniformly and run test $u_{(j\bmod|U|)}$
(a fixed enumeration), and let $V''$ run $V'$ three times independently. Then $V''$ uses $3R_U$ fair
coins and three times the queries, has perfect completeness and soundness error $\le(\frac34)^3=\frac{27}{64}$;
the honest prover and the proof distribution are unchanged, so every zero-knowledge property is unchanged.
\end{lemma}
\begin{proof}
With $N_U=2^{R_U}\ge2|U|$ each test has at least $\lfloor N_U/|U|\rfloor\ge\frac12N_U/|U|$ preimages, so
every event has $V'$-probability at least half its uniform probability. A NO-instance proof is rejected by
$V'$ with probability $\ge\frac14$, so $V''$ accepts it with probability $\le(\frac34)^3$. Zero knowledge
concerns arbitrary verifiers and the proof distribution only.
\end{proof}

\paragraph{Zero knowledge.}
\begin{definition}[black-box perfect zero knowledge; after {\cite[Defs.~2.9--2.11]{GOS25}}]\label{main:def:ZK}
Let $(P,V)$ be a PCP for a relation $R$ and $b=b(|x|)$. A verifier strategy $V^*$ is an oracle algorithm
with its own coins $\rho$, possibly adaptive, malicious and computationally unbounded; it is
\emph{$b$-bounded} if on every execution and against every oracle it makes at most $b$ \emph{distinct}
queries. $\mathrm{View}_{V^*,P}(x,w)$ is $(\rho,\text{the sequence of query--answer pairs})$ when $V^*$ runs
against $\pi\leftarrow P(x,w)$. The PCP is \emph{black-box perfect zero knowledge with query bound $b$} if
there is one algorithm $\mathrm{Sim}$ (universal) with oracle access to $V^*$ such that
$\mathrm{Sim}^{V^*}(x)\equiv\mathrm{View}_{V^*,P}(x,w)$ for every $b$-bounded $V^*$ and every
$(x,w)\in R$; it is \emph{$\varepsilon$-statistical} if the two laws are within statistical distance
$\varepsilon$. $\mathrm{Sim}$ is \emph{straight-line} if it runs $V^*$ at most once, answering each query when it
is made, and outputs the view of that run or a failure symbol (a fixed value that is not a view, such as $\bot$).
\emph{Cost} (memoised interface): a call to $V^*$ returns $V^*$'s next query that has not been asked
before (repeated queries are answered from the record without a call) and costs one step; equivalently,
if repeated queries are passed to $\mathrm{Sim}$, its running time other than the $|\rho|$ coin-tossing steps below
is multiplied by $(1+Q^*)$ with $Q^*$
the number of queries counted with multiplicity, which is how \cite[Defs.~2.10--2.11]{GOS25} count
(``$q^*$-query-bounded''). Tossing $V^*$'s coins $\rho$, which are part of the output, costs one step per coin, and
every time bound below, including those of the classes, bounds the other steps and carries these $|\rho|$ steps
additively (Definition~\ref{def:zk}).
\emph{Simulator classes:} (EP) exact, expected time polynomial in $(|x|,b)$;
(GB) strict polynomial time, outputs $\bot$ with probability at most $\frac12$, exact conditioned on not
$\bot$ (\cite[Def.~4.3.1]{Gol01}, as restated in \cite{GT14}); (SP) strict polynomial time and exact;
($\varepsilon$-SP) strict polynomial time and $\varepsilon$-statistical.
\end{definition}
\noindent Definition~\ref{main:def:ZK} is Definition~\ref{def:zk} with a cost convention and a list of simulator
classes. In the terms of Section~\ref{sec:defs-zk}, (SP) with a universal simulator is notion (Z1), (EP) is (Z3),
($\varepsilon$-SP) with negligible $\varepsilon$ is (Z4), and (GB) is the failure-symbol form of
\cite[Def.~4.3.1]{Gol01}.

\begin{definition}[the uniform-field-oracle model]\label{main:def:PPTF}
An algorithm is $\PPTF$ if it runs in strict probabilistic polynomial time with fair coins and, in addition, has an
oracle that returns independent uniform elements of $\F$, and the number of oracle calls on every run is at most a
fixed polynomial $N_{\rm samp}$. For a simulator, time is counted as in Definition~\ref{main:def:ZK}, and its fair
coins include the coins $\rho$ it tosses for $V^*$.
\end{definition}
\begin{itemize}[nosep]
\item[(F-a)] Implementing each oracle call by rejection sampling on $\lambda_q$ fresh fair bits (accept if the value is
$<q$; success probability $q/2^{\lambda_q}>\frac12$) gives the same output law and expected polynomial time (EP).
\item[(F-b)] Drawing a \emph{pool} of $N_{\rm samp}$ elements up front, each by at most $k$ rejection trials,
outputting $\bot$ if some trial budget is exhausted, and then running on the pool and the other fair coins, gives a
strict polynomial-time algorithm that outputs $\bot$ with probability at most $N_{\rm samp}2^{-k}$ and is
exact conditioned on not $\bot$: conditioned on success the pool is i.i.d.\ uniform and independent of every
other coin. With $k=\lceil\log_2N_{\rm samp}\rceil+1$ this is (GB). (Sampling on demand would be wrong: the
number of draws could depend on the path, and conditioning on success would reweight paths.)
\item[(F-c)] For every polynomially bounded $k=k(|x|)$, the algorithm of (F-b), with $\bot$ kept as its failure symbol, is ($\varepsilon$-SP) with
$\varepsilon=N_{\rm samp}2^{-k}$: its output law is within statistical distance $\Pr[\bot]\le\varepsilon$ of the exact
one.
\item[(F-d)] Exactness statements are distributional and do not depend on running time; sequential
composition of polynomially many $\PPTF$ algorithms is $\PPTF$. This is what lets the composition theorems of
\cite{VW26} (stated for PPT simulators) apply here (Lemma~\ref{main:lem:2}).
\item[(F-e)] The honest provers of this section, the RS PCPP prover of Theorem~\ref{main:thm:M} and the shell prover
$P_A$ of Theorem~\ref{main:thm:A}, toss no fair coins: their random choices (the node masks $\hat r$, and
$u'\leftarrow C'$ drawn through a uniform $R$) are uniform coefficient vectors over $\F$ drawn from the oracle, so
their output is a deterministic function of their oracle samples.
\end{itemize}

\begin{definition}[oracle-robust lazy sampler]\label{main:def:lazy}
Let $P(y;\omega)$ be a randomised proof map on inputs $y\in\mathcal Y\subseteq\F^D$. A
\emph{$(\hat b,\bar\kappa)$-lazy sampler} for $P$ on $\mathcal Y$ is a stateful $\PPTF$ oracle algorithm
$\mathsf{Samp}^{y}$ answering proof-position queries such that: \emph{(exactness)} for every $y\in\mathcal Y$
and every strategy (adaptive, unbounded, with coins independent of the proof and of the sampler) making at
most $\hat b$ distinct queries, the transcript against $\mathsf{Samp}^y$ is distributed identically to the
transcript against $\Pi\leftarrow P(y)$; \emph{(oracle-robust amortised reads)} for every oracle
$y'\in\F^D$ and every query sequence, after $k$ distinct queries it has read at most $\bar\kappa k$
coordinates of $y'$, and it halts within its time bound. $\hat b=\infty$ means exactness for every strategy
making finitely many queries.
\end{definition}

\subsection{The domain family and the field rule}\label{main:sec:family}

\begin{lemma}[smooth orders with ratio-2 coprime factor trees]\label{main:lem:SM2}
Let $\NN$ be an integer with $\LL:=\ln\NN\ge7$. For $m\ge1$ put $\phi_\NN(m):=\LL-3m\ln\LL$ and
$\epsilon_\NN:=1/(16\NN^6)$, let $\tilde\phi_\NN(m)$ be the output of a fixed deterministic evaluation of $\phi_\NN(m)$
to within $\epsilon_\NN$, and say that $m$ \emph{passes} if $\tilde\phi_\NN(m)\ge\epsilon_\NN$ (then $\phi_\NN(m)\ge0$,
i.e.\ $m\le\LL/(3\ln\LL)$). Let $m_\#=2^d$ be the largest power of two that passes,
$x_1=\lfloor\NN^{1/m_\#}\rfloor$, $P=\lfloor x_1/2\rfloor$. For an integer $x$ let
$\varpi(x)$ be the set of the $m_\#$ smallest primes $>x$ and $\Pi(x)=\prod\varpi(x)$. Let $x^*$ be the least
integer in $[P,x_1]$ with $\Pi(x^*)\ge\NN$, and $n=\Pi(x^*)$ with primes $p_1<\dots<p_{m_\#}$. Then:
\begin{enumerate}[nosep,label=(\arabic*)]
\item $x^*$ exists and $\NN\le n<2\NN$;
\item every $p_i\in(x^*,2x^*]$; $\LL^3/2-1\le P<\LL^6/2+1$; every $p_i>P$ and $p_i\le2\LL^6$; $d\le\log_2\LL$;
\item the pairing tree $\mathcal T(n)$ (below) is a complete binary tree of depth $d$ with leaves $p_1,\dots,p_{m_\#}$,
and every internal split $v=v_1v_2$ ($v_1<v_2$) has $\gcd(v_1,v_2)=1$, equal prime counts and $v_2/v_1<2$;
\item every node at depth $j$ has order $<2n^{2^{-j}}$;
\item $m_\#$, $x^*$, $n$ and $\mathcal T(n)$ are computable in time $\mathrm{poly}(\log\NN)$ with a polynomial-time primality
test \cite{AKS04}.
\end{enumerate}
\emph{Pairing tree} (a folklore balanced partition; \cite[\S7]{BS08} uses a smooth order found by search):
to split $2j$ distinct primes, sort them, pair consecutive ones $(p_{2i-1},p_{2i})$, and put the larger
member of pair $i$ into side $\mathcal B$ if $\prod\mathcal B\le\prod\mathcal A$ so far and into side
$\mathcal A$ otherwise (the smaller member goes to the other side); recurse on both sides.
\end{lemma}
\begin{proof}
We use \cite{RS62}: $\pi(x)>x/\ln x$ for $x\ge17$ and $\pi(x)<1.25506\,x/\ln x$ for $x>1$.
(i) For $x\ge144$, $\pi(2x)-\pi(x)>\frac{2x}{\ln 2x}-\frac{1.25506x}{\ln x}\ge\frac{x}{2\ln x}$, since the
last inequality is $2\ln x\ge1.75506(\ln x+\ln2)$, i.e.\ $\ln x\ge4.967$; and $x/(2\ln x)$ increases for
$x\ge e$.
(ii) $m=1$ passes, as $\phi_\NN(1)=\LL-3\ln\LL\ge7-3\ln7>2\epsilon_\NN$, and every passing $m$ has
$m\le\LL/(3\ln\LL)$; so $m_\#$ exists. As $m_\#$ passes, $\phi_\NN(m_\#)\ge0$, i.e.\ $\NN^{1/m_\#}\ge e^{3\ln\LL}=\LL^3$.
As $2m_\#$ does not pass, $\phi_\NN(2m_\#)<2\epsilon_\NN$, i.e.\ $\NN<e^{2\epsilon_\NN}\LL^{6m_\#}$, so
$\NN^{1/m_\#}<e^{2\epsilon_\NN}\LL^6\le\LL^6+4\epsilon_\NN\LL^6<\LL^6+\frac14$ (as $\LL<\NN$). Hence
$x_1\in[\LL^3-1,\LL^6+\frac14)$ and $P\in[\LL^3/2-1,\LL^6/2+\frac18)$; with $\LL\ge7$, $P\ge170>144$.
(iii) For every $x\ge P$: by (i) and monotonicity of $x/(2\ln x)$,
$\pi(2x)-\pi(x)\ge\frac{P}{2\ln P}\ge\frac{P_0}{2\ln P_0}$ with $P_0:=\LL^3/2-1\le P$, and $\ln P_0\le3\ln\LL$;
so $\pi(2x)-\pi(x)\ge\frac{\LL^3/2-1}{6\ln\LL}\ge\frac{\LL}{3\ln\LL}\ge m_\#$, using $\LL^3/2-1\ge2\LL$. So
$\varpi(x)\subseteq(x,2x]$.
(iv) $\Pi$ is nondecreasing, $\Pi(P)\le(2P)^{m_\#}\le x_1^{m_\#}\le\NN$, and every prime $>x_1$ exceeds
$\NN^{1/m_\#}$, so $\Pi(x_1)>\NN$: $x^*$ exists. If $x^*=P$ then $\Pi(P)=\NN$. Otherwise $\Pi(x^*-1)<\NN$ and
$\varpi(x^*-1)\ne\varpi(x^*)$, which forces $x^*$ to be prime and
$\varpi(x^*)=(\varpi(x^*-1)\setminus\{x^*\})\cup\{p_{\rm new}\}$ with $p_{\rm new}\le2x^*$ by (iii); so
$n=\Pi(x^*-1)p_{\rm new}/x^*<2\NN$. This proves (1) and (2) (with $p_i\le2x^*-1\le2x_1-1<2\LL^6$, as the
prime $p_i>P$ is odd and $2x^*$ is even, and $d=\log_2m_\#\le\log_2\LL$).
(v) For (3): all primes lie in $(x^*,2x^*]$, so each pair ratio $r_i\in(1,2)$; with
$\vartheta_i=\ln(\prod\mathcal B/\prod\mathcal A)$ after $i$ pairs, the rule gives
$|\vartheta_i|=\big||\vartheta_{i-1}|-\ln r_i\big|<\ln2$ by induction. Counts are equal, $m_\#=2^d$ keeps them
even down to single primes, and the prime sets are disjoint.
(vi) For (4): $v_2=\sqrt{v\cdot v_2/v_1}<\sqrt{2v}$; iterating, the order at depth $j$ is
$<2^{\sum_{i\le j}2^{-i}}n^{2^{-j}}<2n^{2^{-j}}$.
(vii) For (5): $m_\#$ needs at most $\log_2\LL+2$ tests, each evaluating $\ln\NN$ and $\ln\LL$ to
$O(\log\NN)$ bits; then binary search for $x^*$ ($O(\log\NN)$ evaluations of $\Pi$), each needing $m_\#$ primes in a
window of length $\le\LL^6+1$.
\end{proof}

\begin{remark}[size of the primes]\label{main:rem:2455}
If $m_\#\ge2$, that is, if the tree $\mathcal T(n)$ has an internal node, then $\LL/(3\ln\LL)\ge2$, i.e.\
$\LL\ge6\ln\LL$, which forces $\LL>16.99888$; so $P\ge\LL^3/2-1>2455$ by (2), and every prime factor $p_i$ of $n$
exceeds $2455$. In particular every recursion node (Definition~\ref{main:def:can}) has $m_1>2455$.
\end{remark}

\begin{definition}[field rule]\label{main:def:fieldrule}
For an output $n$ of Lemma~\ref{main:lem:SM2}, $q$ is the least prime $\equiv1\pmod{32n}$, found by testing
$1+32nj$ for $j=1,2,\dots$ with \cite{AKS04}.
\end{definition}

\begin{lemma}[field]\label{main:lem:field}
The rule is well defined, runs in time $\mathrm{poly}(n)$, and gives $q\le c_{\rm Lin}(32n)^{L_{\rm Lin}}$ for
absolute constants $c_{\rm Lin},L_{\rm Lin}$ (Linnik's theorem \cite{Lin44}; any absolute exponent suffices,
e.g.\ \cite{Xyl11}). Moreover $q>32n$, so $m<q$ and $(q-1)/m\ge32$ for every $m\mid n$.
\end{lemma}
\begin{proof}
Linnik's theorem gives a prime $\equiv1\pmod{32n}$ below $c_{\rm Lin}(32n)^{L_{\rm Lin}}$; every such prime is
$\ge32n+1$.
\end{proof}

\begin{definition}[the domain family]\label{main:def:Wa}
$\Wa$ consists of the triples $(n,q,\Hh)$ where $n$ is an output of Lemma~\ref{main:lem:SM2} with its tree
$\mathcal T(n)$, $q$ is given by Definition~\ref{main:def:fieldrule}, and $\Hh=g\mu_{8n}\subseteq\F^\times$ with $g$ the
least integer $\ge2$ such that $g^{8n}\not\equiv1\pmod q$. Put $s:=|\Hh|=8n$.
\end{definition}

\begin{lemma}[root domain]\label{main:lem:root}
For $(n,q,\Hh)\in\Wa$, $\Hh$ is the disjoint union of the eight $\mu_n$-cosets $g\zeta^i\mu_n$ ($i<8$, $\zeta$
of order $8n$), none equal to $\mu_n$; $\mu_h\Hh=\Hh$ for every $h\mid n$; and $\Hh\cap\mu_n=\emptyset$. An
element $\zeta$ of order $8n$ is $z^{(q-1)/8n}$ for the least $z\ge2$ with $(z^{(q-1)/8n})^{8n/p}\ne1$ for every
prime $p\mid8n$ (the primes of $8n$ are known from $\mathcal T(n)$), found in time $\mathrm{poly}(n)$.
\end{lemma}
\begin{proof}
$g\zeta^i\mu_n=g\zeta^j\mu_n$ iff $\zeta^{n(i-j)}=1$ iff $8\mid i-j$. If $g\zeta^i\in\mu_n\le\mu_{8n}$ then
$g\in\mu_{8n}$, contrary to the choice of $g$. The rest is Lemma~\ref{main:lem:mu}.
\end{proof}

\subsection{Nodes, the canonical rule and the recursion}\label{main:sec:nodes}

The recursion uses the \emph{same-radix configuration}: at a node of order $m=m_1m_2$, the input and the mask are
both lifted to bivariate polynomials with the same radix $m_1$.

\noindent In Sections~\ref{main:sec:nodes}--\ref{main:sec:zk} and in Section~\ref{sec:keylemma} the letter $b$
denotes an element of the row set $B$ of a node (in the overview such a row is written $\beta$); in this
section the query bound $b$ appears only in Section~\ref{main:sec:main}. In the same sections $s$ denotes the size of the root domain $\Hh$;
from Section~\ref{main:sec:shell} on it is written $\ell$.

\begin{definition}[node; the same-radix configuration]\label{main:def:node}
A \emph{node} is $\nu=(m_1,m_2;D;\mathcal A^{y},\mathcal A^{r})$ with:
\begin{itemize}[nosep]
\item $m=m_1m_2\mid q-1$, $\gcd(m_1,m_2)=1$ and $m_1<m_2<2m_1$;
\item $D=\bigsqcup_{i<8}d_i\mu_m\subseteq\F^\times$, eight distinct $\mu_m$-cosets; the node's code is
$\RS[D,m]$ (degree bound $=$ group order);
\item two lists of seven representatives $\mathcal A^{y}=(e_0,\dots,e_6)$, $\mathcal A^{r}=(e'_0,\dots,e'_6)$
such that the 22 cosets $d_i\mu_m$, $e_j\mu_m$, $e'_j\mu_m$ are pairwise distinct.
\end{itemize}
Its column sets are $A_y=\bigsqcup_je_j\mu_{m_1}$ and $A_r=\bigsqcup_je'_j\mu_{m_1}$, its row set is
$B:=\{x^{m_1}:x\in D\}$, and for $b\in B$ its fibre is $F_b:=\{x\in D:x^{m_1}=b\}$.
\end{definition}

For an input coefficient vector $\hat g\in\F^m$, a uniform mask $\hat r\in\F^m$ and $\hat y=\hat g+\hat r$
put, for $f\in\{y,r\}$,
\[
Q_f(X,Y)=\sum_{i<m_1,\ j<m_2}\hat f_{i+m_1j}X^iY^j,
\]
so that $Q_f(x,x^{m_1})=f(x)$ and $Q_f(a,b)=\big(f\bmod(X^{m_1}-b)\big)(a)$.
The node's \emph{own symbols} $\Sym(\nu)$ are $y|_D$, $r|_D$, $Q_y|_{A_y\times B}$ and $Q_r|_{A_r\times B}$. A
\emph{view} is a set $V\subseteq\Sym(\nu)$ of symbols (positions); it splits into $y$-side functionals $\Phi$
on $\F^m$ (reads of $y$ and of $Q_y$, as functions of $\hat y$) and $r$-side functionals $\Psi$ (reads of $r$
and of $Q_r$, as functions of $\hat r$). Its \emph{leak} is $\Lambda(V)=\spn\Phi\cap\spn\Psi\subseteq(\F^m)^*$.
Every symbol is the evaluation of a polynomial of bidegree $<(m_1,m_2)$ at a \emph{point}: a lift symbol
$Q_f(a,b)$ at $(a,b)$, a domain symbol $f(x)$ at $(x,x^{m_1})$, which lies in \emph{row} $x^{m_1}$.

\begin{lemma}[coset structure, closure, real symbols]\label{main:lem:CS}
For a node:
\begin{enumerate}[nosep,label=(\alph*)]
\item for $b\in B$, $F_b=x_b\mu_{m_1}$ for any $x_b\in F_b$, $|F_b|=m_1$, and the fibres partition $D$;
\item $B=\bigsqcup_{i<8}d_i^{m_1}\mu_{m_2}$, eight distinct $\mu_{m_2}$-cosets, so $|B|=8m_2$;
\item $A_y$, $A_r$ and $D$ are pairwise disjoint; $A_y\sqcup A_r$ is the union of 14 distinct
$\mu_{m_1}$-cosets; every row domain $F_b\cup A_y$ or $F_b\cup A_r$ consists of eight distinct
$\mu_{m_1}$-cosets, and the column domain $B$ of eight distinct $\mu_{m_2}$-cosets;
\item \emph{closure}: the children of $\nu$ are $8m_2$ $y$-rows (order $m_1$, domain $F_b\cup A_y$), $8m_2$
$r$-rows (order $m_1$, domain $F_b\cup A_r$), $7m_1$ $y$-columns and $7m_1$ $r$-columns (order $m_2$, domain
$B$); every child's code is $\RS[\text{domain},\text{order}]$ on eight cosets of its order, and its split is the
split of its order in $\mathcal T(n)$;
\item \emph{honest lines are codewords}: for every value of the mask, each row line ($Q_f(\cdot,b)$ on $A_f$,
extended by $f|_{F_b}$) is the evaluation of $Q_f(X,b)$, of degree $<m_1$, and each column line
$Q_f(a,\cdot)|_B$ has degree $<m_2$;
\item \emph{real symbols}: every point of a child's domain carries exactly one own symbol of $\nu$ ($x\in F_b$
carries $f(x)$, $a\in A_f$ carries $Q_f(a,b)$, $b\in B$ carries $Q_f(a,b)$), and the map $\iota_w$ from a child
$w$'s domain into $\Sym(\nu)$ is injective.
\end{enumerate}
\end{lemma}
\begin{proof}
(a) $x^{m_1}=x'^{m_1}$ iff $x'/x\in\mu_{m_1}$ (Lemma~\ref{main:lem:mu}), and
$x\mu_{m_1}\subseteq x\mu_m\subseteq D$. (b) On the cyclic group $\mu_m$ the map $z\mapsto z^{m_1}$ has image
of order $m/\gcd(m,m_1)=m_2$, so $(d_i\mu_m)^{m_1}=d_i^{m_1}\mu_{m_2}$; and
$d_i^{m_1}\mu_{m_2}=d_j^{m_1}\mu_{m_2}$ implies $(d_i/d_j)^{m_1m_2}=1$, i.e.\ $d_i\mu_m=d_j\mu_m$. Since the eight
images are disjoint, each $F_b$ lies in one $d_i\mu_m$, on which $x\mapsto x^{m_1}$ is $m_1$-to-one; so
$|F_b|=m_1$. (c) Every $\mu_{m_1}$-coset used lies in one of the 22 pairwise distinct $\mu_m$-cosets, and
$F_b=x_b\mu_{m_1}\subseteq x_b\mu_m\subseteq D$. (d) From (a)--(c); both lifts have radix $m_1$, so all rows
have order $m_1$ and all columns order $m_2$ and domain $B$. (e) $x^{i+m_1j}=x^i(x^{m_1})^j$ gives
$Q_f(x,x^{m_1})=f(x)$, hence $Q_f(x,b)=f(x)$ on $F_b$; $\deg_X<m_1$ and $\deg_Y<m_2$ by definition.
(f) From (e) and the definition of the proof (Section~\ref{main:sec:pcpp}); distinct points of a child domain
are distinct points of $D$ or of a column set, which are disjoint by (c).
\end{proof}

\begin{definition}[canonical rule; recursion nodes]\label{main:def:can}
The \emph{canonical rule} $\rcan$ takes the order $m$ (with its split $m=m_1m_2$ in $\mathcal T(n)$) and eight
representatives $d_0,\dots,d_7$ of distinct $\mu_m$-cosets, and runs the following exact integer loop: set
$\mathcal U:=\{d_i^m\bmod q\}$ and an empty list; for $z=2,3,\dots,q-1$ in increasing order, if $z^m\bmod q\notin\mathcal U$,
append $z$ to the list and add $z^m\bmod q$ to $\mathcal U$; stop when the list has 14 entries. It outputs the
first seven entries as $\mathcal A^y$ and the next seven as $\mathcal A^r$ (FAIL if the loop ends early). The
set $\mathcal N(\rcan)$ of \emph{recursion nodes} is the set of nodes reached from the root
$(n,q,\Hh)\in\Wa$ (order $n$ with the top split of $\mathcal T(n)$, domain $\Hh$) by repeatedly taking the children of
Lemma~\ref{main:lem:CS}(d) of composite order and designing them by $\rcan$. Children of prime order are
\emph{leaves}.
\end{definition}

\begin{proposition}[totality]\label{main:prop:total}
$\rcan$ never outputs FAIL on a recursion node, and runs in time $O(q\log q)$.
\end{proposition}
\begin{proof}
Every recursion node has order $m\mid n$, so $(q-1)/m\ge32$ (Lemma~\ref{main:lem:field}) and $\F^\times$ has at
least 32 $\mu_m$-cosets. Every element of $\F^\times\setminus\{1\}$ is an integer in $[2,q-1]$, and every coset
other than $\mu_m$ contains such an element ($\mu_m\ne\{1\}$ contains one too, as $m>1$), so the loop meets at
least 32 cosets and finds 14 outside the 8 given ones.
\end{proof}

\begin{lemma}[recursion nodes satisfy the hypotheses of the locator lemma]\label{main:lem:N}
Let $(n,q,\Hh)\in\Wa$ and $\nu\in\mathcal N(\rcan)$ with split $m=m_1m_2$. Then, with $n_\nu:=m_1$,
$n'_\nu:=m_2$:
\begin{itemize}[nosep]
\item[(N1)] $\gcd(m_1,m_2)=1$, $m_1<m_2\le Km_1$ with $K=2$, $m\mid q-1$ and $m<q=\mathrm{char}\,\F$; so
$\mu_{m_1},\mu_{m_2},\mu_m\subseteq\F^\times$ and $q\nmid m$;
\item[(N2)] $D$ is the union of $L_2=8$ distinct $\mu_m$-cosets and the node's code is $\F[X]_{<m}$ on $D$;
\item[(N3)] $B=D^{m_1}$ is the union of $L_2=8$ distinct cosets of $\mu_{m_2}$, and $|F_b|=m_1$ for every
$b\in B$;
\item[(N4)] $A=A_y\sqcup A_r\subseteq\F^\times\setminus D$ is the union of $L_1=14$ distinct cosets of
$\mu_{m_1}$, with $A_y$ and $A_r$ unions of seven of them each.
\end{itemize}
\end{lemma}
\begin{proof}
The root has order $n$ and domain $\Hh$, eight distinct $\mu_n$-cosets (Lemma~\ref{main:lem:root}); by
induction and Lemma~\ref{main:lem:CS}(c)--(d) every recursion node has eight distinct cosets of its order and
code $\RS[D,m]$: this is (N2). Its order is an internal node of $\mathcal T(n)$ with split $m_1m_2$, so (N1)'s
divisibility and ratio conditions are Lemma~\ref{main:lem:SM2}(3), $m\mid n\mid q-1$ is
Definition~\ref{main:def:fieldrule}, and $m\le n<32n<q$ is Lemma~\ref{main:lem:field}. (N3) is
Lemma~\ref{main:lem:CS}(a)--(b). (N4): $\rcan$ outputs 14 representatives of $\mu_m$-cosets distinct from
each other and from the $d_i\mu_m$ (Proposition~\ref{main:prop:total}), so Lemma~\ref{main:lem:CS}(c) applies.
\end{proof}
\noindent Items (N1)--(N4) are conditions (N1)--(N4) of Definition~\ref{kl:def:node} for the parameters
$(L_1,L_2,K)=(14,8,2)$, with $n=m_1$, $n'=m_2$ and $p=q$ (for (N3), Lemma~\ref{kl:lem:rows} gives the stated
properties of $B$ and of the fibres). So every recursion node is a node with parameters $(14,8,2)$ in the sense of
Definition~\ref{kl:def:node}.

\subsection{The locator lemma at recursion nodes}\label{main:sec:KL}

Let $L_1=14$, $L_2=8$, $K=2$, and let $c_*:=10^{-13}$. Then $c_*\le1/(16L_2)$ and
$L_1L_2(3\cdot2^{-1/3}K^{1/3}c_*^{1/3}+6c_*)<\frac1{8L_2}$ (for $K=2$ the condition reads
$112(3c^{1/3}+6c)<\frac1{64}$); since $c_*<\frac1{2(L_2+1)}$ and $\frac1{L_2+1}-2c_*\ge\frac1{8L_2}$, $c_*$ is admissible for $(L_1,L_2,K)$ in
the sense of Definition~\ref{kl:def:adm}. Put
\[
\aKL:=1+\tfrac{16}{15}+K/c_*^2\ \le\ 2.0001\cdot10^{26},\qquad \log_2\aKL\le87.39 ;
\]
as $\frac1{1-c_*}\le\frac{16}{15}$, $\aKL$ bounds the constant $a(c_*)$ of Theorem~\ref{kl:thm:KL}(c).

\begin{theorem}[the locator lemma at recursion nodes; proved in Section~\ref{sec:keylemma} as Theorem~\ref{kl:thm:KL}]\label{main:thm:KL}
Let $\nu$ be a node over a field of characteristic $p$ satisfying (N1)--(N4) of Lemma~\ref{main:lem:N}
(with $q$ replaced by $p$, $L_1=14$, $L_2=8$, $K=2$). For a view $V\subseteq\Sym(\nu)$ let $D_V$ be the
set of $x\in D$ with $y(x)$ or $r(x)$ in $V$, let $|V\cap\text{row }b|$ be the number of distinct points of
$V$ in row $b$, let $\tau:=\lfloor c_*m_1\rfloor$ and $H_\tau(V):=\{b\in B:|V\cap\text{row }b|\ge m_1+1-\tau\}$,
and
\[
W_\nu(V):=D_V\ \cup\bigcup_{b\in H_\tau(V)}F_b\ \cup\ \begin{cases}d_0\mu_m,&|V|\ge(\tau+1)^2,\\
\emptyset,&\text{otherwise.}\end{cases}
\]
Then $W_\nu(V)\subseteq D$; $W_\nu$ depends on $V$ as a set of positions only, is monotone
($V\subseteq V'\Rightarrow W_\nu(V)\subseteq W_\nu(V')$) and computable in time $\mathrm{poly}(|\F|)$;
$|W_\nu(V)|\le\aKL|V|$; and $W_\nu$ is a \emph{locator}:
$\Lambda(V)\subseteq\spn\{\ev_x|_{\F^m}:x\in W_\nu(V)\}$.
\end{theorem}
Section~\ref{sec:keylemma} credits the ingredients of the proof. The non-trivial content of the theorem covers views
of size below $(\tau+1)^2\approx10^{-26}m_1^2$, and it is vacuous for $m_1<1/c_*=10^{13}$, where $\tau=0$ and every
non-empty view is located by a whole $\mu_m$-coset.

\begin{corollary}[the locator lemma applies at every recursion node]\label{main:cor:HDL}
For every $(n,q,\Hh)\in\Wa$ and every $\nu\in\mathcal N(\rcan)$, the map $W_\nu$ of Theorem~\ref{main:thm:KL}
is defined, contained in $D$, monotone, position-only, computable in time $\mathrm{poly}(q)$, of size
$\le\aKL|V|$, and a locator for every view $V\subseteq\Sym(\nu)$.
\end{corollary}
\begin{proof}
Lemma~\ref{main:lem:N} gives (N1)--(N4) with $p=q$, and Theorem~\ref{main:thm:KL} applies; the field is $\F_q$, so
$\mathrm{poly}(|\F|)=\mathrm{poly}(q)$.
\end{proof}

\subsection{The RS PCPP in the same-radix configuration}\label{main:sec:pcpp}

Fix $(n,q,\Hh)\in\Wa$ and its tree $\mathcal T(n)$.

\paragraph{The PCPP $\mathcal P(v)$} for a node $v$ (input $g:D_v\to\F$, claim $g\in\RS[D_v,m_v]$):
\begin{itemize}[nosep]
\item \emph{Leaf} (prime order $p$): empty proof; one test reads all of $g$ ($8p$ symbols) and accepts iff
$g\in\RS[D_v,p]$.
\item \emph{Internal node}: the proof is $\pi_v=(\Sym(v),(\pi_w)_w)$ with $w$ over the $16m_2+14m_1$ children of
Lemma~\ref{main:lem:CS}(d), designed by $\rcan$, and $\pi_w$ a proof for the child line $g_w$ (read from $\pi_v$
via $\iota_w$) in $\mathcal P(w)$.
\item \emph{Honest prover}: sample $\hat r$ uniformly, put $\hat y=\hat g+\hat r$, compute the lifts, and
recurse with fresh independent masks at every node.
\item \emph{Tests} (five families): $U_{\rm cons}=\{x\in D\}$ checks $g(x)=y(x)-r(x)$; four families run the
children's verifiers on the $y$-rows, $y$-columns, $r$-rows and $r$-columns.
\item \emph{Coin multiset}: with $N_\phi=|U_\phi|$ for the five families $\phi$, $N_{\max}=\max_\phi N_\phi$,
$U_v=\biguplus_\phi\lceil N_{\max}/N_\phi\rceil\cdot U_\phi$; each family has total weight in
$[N_{\max},2N_{\max})$, hence probability $\ge\frac1{10}$; within a family all children have the same order and
eight cosets, so the child is uniform and its coin is uniform given the child.
\end{itemize}
\paragraph{The root (degree adjustment, the device of \cite{BS08}).} For $f$ claimed in $\RS[\Hh,u]$, $u\le n$: if
$u=n$ run $\mathcal P(\text{root})$; if $u<n$ run two independent copies, on $f$ and on $f'(x)=x^{n-u}f(x)$ (a
read of $f'(x)$ is a read of $f(x)$ times a public constant), with coin multiset $U_A\uplus U_B$.

\paragraph{Soundness.} The multiplicative smooth RS PCPP with ``subgroup $\cup$ coset'' lines is \cite[\S7]{BS08},
which proves soundness for 2-smooth orders (Thm.~3.4) and indicates in Remark~3.5 that the proof extends to
polylog-smooth orders. Every order here has odd coprime factors, so we give a complete soundness proof with the
ingredients of \cite{BS08}.

\begin{lemma}[bivariate step; via {\cite[Lemma~6.13]{BS08}}]\label{main:lem:S}
Let $D$ be eight cosets of $\mu_m$, $m=m_1m_2$ coprime, $A$ seven cosets of $\mu_{m_1}$ in distinct
$\mu_m$-cosets disjoint from $D$, $B=D^{m_1}$. For $z:D\to\F$ and $Q:A\times B\to\F$ let
$R_b=z|_{F_b}\cup Q(\cdot,b)|_A$, $C_a=Q(a,\cdot)$,
$\varepsilon_{\rm row}=\mathbb E_{b\in B}\delta(R_b,\RS[F_b\cup A,m_1])$ and
$\varepsilon_{\rm col}=\mathbb E_{a\in A}\delta(C_a,\RS[B,m_2])$. Then
$\delta(z,\RS[D,m])\le C_{\rm biv}(\varepsilon_{\rm row}+\varepsilon_{\rm col})$ with $C_{\rm biv}=10+2c_0=266$.
\end{lemma}
\begin{proof}
\cite[Lemma~6.13]{BS08} (proved in its \S6.6 with $c_0=128$ from \cite{PS94} via \cite[Thm.~4.2.19]{Spi95},
which holds over any field): for all $\mathcal A,\mathcal B\subseteq\F$, $d_1\le|\mathcal A|/4$,
$d_2\le|\mathcal B|/8$ and $f:\mathcal A\times\mathcal B\to\F$,
$\delta_{(d_1,d_2)}(f)\le c_0(\delta_{(d_1,*)}(f)+\delta_{(*,d_2)}(f))$.
(1) Let $p_b$ be a nearest codeword to $R_b$; its disagreements inside $A$ number $\le8m_1\delta(R_b)$ of
$|A|=7m_1$, so $f':=Q$ has $\delta_{(m_1-1,*)}(f')\le\frac87\varepsilon_{\rm row}$ and
$\delta_{(*,m_2-1)}(f')=\varepsilon_{\rm col}$.
(2) Lemma~6.13 applies with $(\mathcal A,\mathcal B)=(A,B)$ ($m_1-1\le7m_1/4$, $m_2-1\le8m_2/8$) and gives $P$ with
$\deg_XP<m_1$, $\deg_YP<m_2$, $\delta(f',P)\le\eta_0:=c_0(\frac87\varepsilon_{\rm row}+\varepsilon_{\rm col})$.
(3) Call $b$ bad if $p_b\ne P(\cdot,b)$; two distinct polynomials of degree $<m_1$ agree on $\le m_1-1$ points of
$A$, so on a bad row the two disagreement fractions of $Q(\cdot,b)|_A$ sum to $\ge6/7$; by Markov,
$\Pr_b[\text{bad}]\le\frac76(\frac87\varepsilon_{\rm row}+\eta_0)$.
(4) $F(X):=P(X,X^{m_1})$ has degree $\le m-1$ and $F=P(\cdot,b)$ on $F_b$; on a good row $z|_{F_b}$ and $F|_{F_b}$
differ on $\le8m_1\delta(R_b)$ points; fibres have equal size, so $\delta(z,F)\le\Pr[\text{bad}]+8\varepsilon_{\rm row}$.
(5) In total $\delta(z,F)\le(\frac{28}3+\frac43c_0)\varepsilon_{\rm row}+\frac76c_0\varepsilon_{\rm col}\le
C_{\rm biv}(\varepsilon_{\rm row}+\varepsilon_{\rm col})$.
\end{proof}

\begin{theorem}[strong soundness, perfect completeness]\label{main:thm:SND}
For a node of height $h$, every input $g$ and every proof,
$\Pr[\text{reject}]\ge\varrho_h\,\delta(g,\RS[D,m])$ with $\varrho_h=(10C_{\rm biv})^{-h}=2660^{-h}$. At the root
with $u<n$, $\Pr[\text{reject}]\ge\frac38\varrho_d\,\delta(f,\RS[\Hh,u])$. Completeness is perfect.
\end{theorem}
\begin{proof}
Completeness: Lemma~\ref{main:lem:CS}(e) and $y-r=g$. Leaves reject every non-codeword. At an internal node let
$\rho=\Pr_x[g(x)\ne y(x)-r(x)]$ and let the four $\varepsilon$'s be the average distances of the child inputs of
the four families. By induction and the family weights, $\Pr[\text{reject}]\ge\frac{\varrho_{h-1}}{10}(\rho+\sum\varepsilon)$.
Lemma~\ref{main:lem:S} applies to the $y$-lift ($z=y$, $A=A_y$) and to the $r$-lift ($z=r$, $A=A_r$), giving
$\delta(y),\delta(r)\le C_{\rm biv}(\cdots)$; with nearest codewords $Y,R$, $Y-R$ is a codeword and
$\delta(g)\le\rho+\delta(y)+\delta(r)\le C_{\rm biv}(\rho+\sum\varepsilon)$. Root with $u<n$: let $F,G$ be nearest to
$f,f'$ in $\RS[\Hh,n]$ at distances $\delta_1,\delta_2$; $X^{n-u}F$ and $G$ have degree $<2n$ and agree on
$\ge(1-\delta_1-\delta_2)8n$ points. If $\delta_1+\delta_2\le\frac34$ they are equal, so $\deg F<u$ and
$\delta(f,\RS[\Hh,u])\le\delta_1$; otherwise $\delta(f,\RS[\Hh,u])\le1<\frac43(\delta_1+\delta_2)$. Either way
$\Pr[\text{reject}]\ge\frac12\varrho_d(\delta_1+\delta_2)\ge\frac38\varrho_d\,\delta(f,\RS[\Hh,u])$.
\end{proof}

\begin{proposition}[length, coins, queries, residual circuits]\label{main:prop:LEN}
Let $s=8n$ and $d$ the depth of $\mathcal T(n)$, so $2^d\le\ln s$. A node of order $m$ has $128m$ own symbols, and the
proof length $\mathcal L$ (input excluded) is $8m\lambda_h$ with $\lambda_h=16+30\lambda_{h-1}$, $\lambda_0=0$;
so $\mathcal L=\frac{16}{29}(30^d-1)s\le s(\ln s)^{\log_230}$, doubled at the root when $u<n$. The coin multiset
has $|U(m)|\le8m\cdot80^h$. A test reads either three symbols or a whole leaf line of $8p\le16(\ln s)^6$
symbols; the leaf predicate is a parity check with a circuit of at most $C_{\rm leaf}\,p^2\log_2^2q\le
4C_{\rm leaf}(\ln s)^{12}\log_2^2q$ gates for an absolute $C_{\rm leaf}$, computable from the coin in time
$\mathrm{poly}(s)$. The verifier is non-adaptive; the prover is $\PPTF$.
\end{proposition}
\begin{proof}
Own symbols: $2\cdot8m$ domain values and $2\cdot7m_1\cdot8m_2$ lift values. Children contribute
$(16m_2\cdot8m_1+14m_1\cdot8m_2)\lambda_{h-1}=240m\lambda_{h-1}$. Coins: $N_{\rm cons}=8m$, and each child family
has $N_\phi\le64m\cdot80^{h-1}$ by induction, so $|U|\le10N_{\max}\le8m\cdot80^h$. Leaves have $8p$ symbols with
$p\le2\LL^6\le2(\ln s)^6$ (Lemma~\ref{main:lem:SM2}); the parity check multiplies a $7p\times8p$ matrix by the
read vector, and multiplication in $\F$ has circuits of size $O(\log^2q)$. $2^d=m_\#\le\LL\le\ln s$.
\end{proof}

\subsection{Zero knowledge of the PCPP}\label{main:sec:zk}

\begin{lemma}[leak linear algebra; cf.\ {\cite[Lemma~4.3, Claim~7.7]{BCFGRS17}}]\label{main:lem:DL}
At a node, the law of the answers to $V\subseteq\Sym$ under input $\hat g$ equals that under $\hat g'$ iff
$\lambda(\hat g-\hat g')=0$ for all $\lambda\in\Lambda(V)$.
\end{lemma}
\begin{proof}
The answers are $(\Phi(\hat g+\hat r),\Psi\hat r)$ with $\hat r$ uniform: the uniform law on
$(\Phi\hat g,0)+\mathcal M$, $\mathcal M=\{(\Phi\hat r,\Psi\hat r)\}$. Two such laws agree iff
$(\Phi(\hat g-\hat g'),0)\in\mathcal M$ iff $\hat g-\hat g'\in\ker\Phi+\ker\Psi$, and
$(\ker\Phi+\ker\Psi)^\perp=\spn\Phi\cap\spn\Psi$.
\end{proof}
So a locator $W_\nu(V)$ is a \emph{dependence set}: codewords agreeing on $W_\nu(V)$ induce the same law on $V$.

\begin{lemma}[tree dependence]\label{main:lem:C}
For a node $v$ and a set $S$ of positions of its subtree proof, write $S_v=S\cap\Sym(v)$, $S_w=S\cap\pi_w$, and
define recursively ($W^T_{\rm leaf}(\emptyset)=\emptyset$)
\[V^{\rm eff}_v(S)=S_v\cup\textstyle\bigcup_w\iota_w\big(W^T_w(S_w)\big),\qquad W^T_v(S)=W_v\big(V^{\rm eff}_v(S)\big),\]
with $W_v$ the map of Theorem~\ref{main:thm:KL}. Then (i) $W^T_v(S)\subseteq D_v$ is a dependence set for $S$ (for
inputs in $\RS[D_v,m_v]$); (ii) it is monotone in $S$; (iii) $|W^T_v(S)|\le\aKL^{\,d-\mathrm{dep}(v)}|S|$;
(iv) it depends on the positions $S$ only and is computable in time $\mathrm{poly}(q)\cdot|S|$.
\end{lemma}
\begin{proof}
Induction on height; well defined by Lemma~\ref{main:lem:V} below. (ii) and (iv): $W_v$ is monotone and
position-only, and only the touched nodes and their ancestors are visited. (iii)
$|W^T_v|\le\aKL(|S_v|+\sum_w|W^T_w(S_w)|)$ ($\iota_w$ is injective) and $\aKL\ge1$. (i) Let $g,g'$ be codewords
agreeing on $W^T_v(S)$ and condition on $\sigma=\Sym(v)|_{V^{\rm eff}}$. Every child input is a codeword for every
value of $v$'s mask (Lemma~\ref{main:lem:CS}(e)); $g_w|_{W^T_w(S_w)}$ is a sub-vector of $\sigma$
(Lemma~\ref{main:lem:CS}(f)); child masks are independent of each other and of $\Sym(v)$. By induction the
conditional law of $(S_w)_w$ given all of $\Sym(v)$ is $\bigotimes_wK_w(\sigma|_{\iota_wW^T_w})$ with kernels not
depending on $g$. By Lemma~\ref{main:lem:DL} and Corollary~\ref{main:cor:HDL} at $v$, $\sigma$ has the same law
under $g$ and $g'$; hence so do $(\sigma,(S_w)_w)$ and its marginal on $S$.
\end{proof}
At the root with $u<n$, $W^T:=W^T_A(S_A)\cup W^T_B(S_B)$ is a dependence set with respect to $f\in\RS[\Hh,u]$
(independent masks; the law of $S_B$ depends on $(x^{n-u}f(x))_{x\in W^T_B}$).

\begin{lemma}[every invocation is covered by the locator lemma]\label{main:lem:V}
In Lemma~\ref{main:lem:C}, every set to which a node map $W_v$ is applied is a view of $v$ in the sense of
Theorem~\ref{main:thm:KL}, i.e.\ a subset of $\Sym(v)$, and $v\in\mathcal N(\rcan)$. In detail:
\begin{enumerate}[nosep,label=(\roman*)]
\item the verifier's own reads at $v$ ($y(x)$, $r(x)$ in $U_{\rm cons}$, and lift symbols) are elements of
$\Sym(v)$;
\item a read of the input $g_v(x)$ in $U_{\rm cons}$ at a non-root node $v$ is a read of the parent symbol
$\iota_v(x)$, so it belongs to $S_{\rm parent}\subseteq\Sym(\mathrm{parent})$ and never to a view of $v$; at the
root it is a read of the PCPP input, handled by the shell (Proposition~\ref{main:prop:3M}, Lemma~\ref{main:lem:2});
\item child-induced reads $\iota_w(W^T_w(S_w))$ are subsets of $\Sym(v)$, because $W^T_w(S_w)\subseteq D_w$ by
induction and $\iota_w$ maps $D_w$ injectively into $\Sym(v)$ (Lemma~\ref{main:lem:CS}(f)); the same symbol may be
induced by several children (e.g.\ $Q_y(a,b)$ by the $y$-row $b$ and the $y$-column $a$), and $V^{\rm eff}$ is
their union;
\item shared domain reads ($y(x)$ and $r(x)$ for the same $x$) are two distinct symbols of $\Sym(v)$ at the same
point; Theorem~\ref{main:thm:KL} covers them (Lemma~\ref{kl:lem:shared} removes them by the modular law of subspaces
and puts $x\in D_V$);
\item $W_v(V)\subseteq D_v$, depends on positions only, is monotone and computable in time $\mathrm{poly}(q)$
(Corollary~\ref{main:cor:HDL}); hence so is $W^T$ (Lemma~\ref{main:lem:C}(iv)), which is what makes the sampler
of Theorem~\ref{main:thm:Z} oracle-robust.
\end{enumerate}
\end{lemma}
\begin{proof}
(i) is the definition of the proof $\pi_v$. (ii) By Lemma~\ref{main:lem:CS}(f) the input of a child $w$ is the
line $z\mapsto(\text{value of the symbol }\iota_w(z))$, so its positions are positions of $\pi_v$; the PCPP proof is the disjoint union of
the $\Sym(v)$ over all nodes, and $S$ is partitioned accordingly. (iii) and (iv) are restatements of
Lemma~\ref{main:lem:CS}(c),(f) and of the definition of a view as a set of symbols. (v) is
Corollary~\ref{main:cor:HDL}. Finally, every node to which Lemma~\ref{main:lem:C} applies $W_v$ is an internal node
of the recursion, hence in $\mathcal N(\rcan)$; leaves have no own symbols.
\end{proof}

\begin{theorem}[exact lazy sampler]\label{main:thm:Z}
The root PCPP for $\RS[\Hh,u]$ ($u\le n$) has a straight-line $(\infty,\bar\kappa)$-lazy sampler with
$\bar\kappa=\aKL^{\,d}\le(\ln s)^{\log_2\aKL}$; it is $\PPTF$ and oracle-robust.
\end{theorem}
\begin{proof}
State: the answered set $S_t$, the answers $\mathbf a_t$, $W_t=W^T(S_t)$ and the oracle values on $W_t$. On a new
query $\alpha$: (1) $S'=S_t\cup\{\alpha\}$, $W'=W^T(S')\supseteq W_t$; read the oracle on $W'\setminus W_t$.
(2) Find $f^*\in\RS[\Hh,u]$ with $f^*|_{W'}=$ oracle$|_{W'}$ by Gaussian elimination; if none exists answer $0$.
(3) Sample $\omega$ uniformly from the affine space $\{\omega:\Pi(f^*;\omega)|_{S_t}=\mathbf a_t\}$ of masks (at
most the number of mask coordinates many uniform $\F$-samples); if it is empty answer $0$. (4) Answer
$\Pi(f^*;\omega)_\alpha$. Repeated queries are answered from the state.
\emph{Reads, for every oracle}: $W^T$ depends on positions only, so after $k$ distinct queries at most
$|W^T(S_k)|\le\bar\kappa k$ coordinates have been read. \emph{Exactness for codeword oracles}: fix the
strategy's coins and induct on $t$. Given a history of positive probability, the next query $\alpha$ is
determined; in the real run the next answer has law $\mathcal L_f(\Pi_\alpha\mid\Pi|_{S_t}=\mathbf a_t)$. $W'$ is a
dependence set for $S'$ and $f^*|_{W'}=f|_{W'}$, so the joint laws of $\Pi|_{S'}$ under $f$ and $f^*$ coincide;
in particular steps (2)--(3) succeed and the conditionals coincide. As $\Pi$ is affine in $\omega$ and $\omega$
is uniform, steps (3)--(4) sample exactly $\mathcal L_{f^*}(\Pi_\alpha\mid\Pi|_{S_t}=\mathbf a_t)$. The bound on
$\bar\kappa$: $\aKL^d=(2^d)^{\log_2\aKL}$ and $2^d\le\ln s$ (Proposition~\ref{main:prop:LEN}).
\end{proof}

\begin{theorem}[the zero-knowledge RS PCPP on $\mathcal W$]\label{main:thm:M}
For $(n,q,\Hh)\in\Wa$ and every $u\le n$ the PCPP for $\RS[\Hh,u]$ has: a non-adaptive verifier with
$\le16(\ln s)^6$ queries; coin multiset $\le2\cdot80^ds\le2s(\ln s)^{\log_280}$; length
$\le\frac{32}{29}30^ds$; perfect completeness; $\Pr[\text{reject}]\ge\frac38\cdot2660^{-d}\delta(f,\RS[\Hh,u])$
with $2660^d\le(\ln s)^{\log_22660}$; residual circuits of size $\le4C_{\rm leaf}(\ln s)^{12}\log_2^2q+O(\log^2q)$
computable in time $\mathrm{poly}(s)$; a $\PPTF$ prover; and an oracle-robust $\PPTF$ $(\infty,\bar\kappa)$-lazy
sampler with $\bar\kappa\le(\ln s)^{\log_2\aKL}$. Here $\log_280<6.33$, $\log_22660<11.38$, $\log_2\aKL<87.39$.
\end{theorem}
\begin{proof}
Theorems~\ref{main:thm:SND}, \ref{main:thm:Z}, Proposition~\ref{main:prop:LEN}, and $c^d=(2^d)^{\log_2c}\le
(\ln s)^{\log_2c}$. The $O(\log^2q)$ term is the multiplication by the public constant $x^{n-u}$ at the root.
\end{proof}

\subsection{The outer shell}\label{main:sec:shell}

Ingredients. \textbf{(CL)} Quasilinear Cook--Levin \cite{Coo88} (via \cite{PF79}): for
$R\in\mathrm{NTIME}(T)$ (multitape machines), $x\mapsto$ a CNF $\Phi_x$ of size $O(T\log T)$, computable in time
$\mathrm{poly}(T)$, satisfiable iff $x\in L(R)$, with a satisfying assignment computable from a witness. We
normalise to \textbf{3SAT[4]} (clauses of 1--3 literals on distinct variables, each variable in $\le4$ clauses):
first delete repeated literals and tautological clauses and split clauses longer than 3 with fresh variables;
then replace each variable with $k\ge5$ occurrences by copies $y_1,\dots,y_k$ (one per occurrence) and the
2-clauses $(\neg y_i\vee y_{i+1})$ cyclically \cite{Tov84}. Size stays linear; satisfiability and witnesses are
preserved. \textbf{(BW)} The Bene\v{s} network on $2^\beta$ rows with $2\beta-1$ switch stages pairing rows that
differ in bit $b_j$, $(b_j)=(0,1,\dots,\beta-1,\dots,1,0)$, is rearrangeable, and the looping algorithm computes the
settings for a given permutation in time $O(\beta2^\beta)$ \cite{Ben64,Wak68}.

\begin{lemma}[multiplicative routing instance]\label{main:lem:MR}
Let $\Phi$ be a 3SAT[4] formula with variables $X_0,\dots,X_{v-1}$ and clauses $C_0,\dots,C_{M_c-1}$. Put
$\beta=\lceil\log_2\max(4v,3M_c,2)\rceil$, $\mathsf R=2^\beta$, $N_{\rm pts}=(2\beta+1)\mathsf R$, and points
$(r,j)\in[\mathsf R]\times\{0,\dots,2\beta\}$.
\begin{itemize}[nosep]
\item \emph{Routing.} Source row $4i+c$ is the $c$-th occurrence of $X_i$ ($c<4$), target row $3k+t$ the $t$-th
literal of $C_k$. Let $\pi$ be a permutation of $[\mathsf R]$ mapping each used source to its target, realised by
Bene\v{s} settings; column $j+1$ ($0\le j<2\beta-1$) is obtained from column $j$ by stage $j$.
\item \emph{Point sets} (pairwise disjoint): $\mathrm{Var}=\{(4i,0)\}$; $\mathrm{Copy}=\{(4i+c,0):c\in\{1,2,3\}\}$
with shift $c$; $\mathrm{Net}=[\mathsf R]\times\{1,\dots,2\beta-1\}$, where $(r,j+1)$ has predecessor $(r',j)$ with
$r'\in\{r,r\oplus2^{b_j}\}$ and shift $\delta=r-r'\in\{0,\pm2^{b_j}\}$; $\mathrm{Cl}=\{(3k,2\beta)\}$ with clause
$C_k$.
\item \emph{Field data.} Let $h\mid q-1$ with $h\ge N_{\rm pts}$, $\omega\in\F$ of order $h$,
$p(r,j)=\omega^{r+\mathsf Rj}$ (injective), labels
$\Psi_{\rm lab}=\{\mathrm{id}\}\cup\{\sigma_c\}_{c\in\{1,2,3\}}\cup\{\beta_\delta\}_{\delta\in\{0,\pm2^{b}\}}\cup\{\gamma_t\}_{t<3}$
with multipliers $\chi_{\sigma_c}=\omega^{-c}$, $\chi_{\beta_\delta}=\omega^{-(\mathsf R+\delta)}$,
$\chi_{\gamma_t}=\omega^{t-\mathsf R}$ ($|\Psi_{\rm lab}|=2\beta+8$ labels; distinct labels may have equal
multipliers), Lagrange selectors $L_z(X)=\frac zh\frac{X^h-1}{X-z}$ on $\mu_h$, and
\begin{multline*}\mathcal Q=\sum_{z\in p(\mathrm{Var})}L_zY_{\rm id}(Y_{\rm id}-1)+\sum_{z\in p(\mathrm{Copy})}L_z(Y_{\rm id}-Y_{\sigma_{c(z)}})\\
+\sum_{z\in p(\mathrm{Net})}L_z(Y_{\rm id}-Y_{\beta_{\delta(z)}})+\sum_{z\in p(\mathrm{Cl})}L_z\prod_{t<|C_{k(z)}|}(Y_{\gamma_t}-c_t(z)),\end{multline*}
where $c_t(z)\in\{0,1\}$ is the value satisfying the $t$-th literal.
\end{itemize}
Then $\deg_X\mathcal Q<h$, $\mathcal Q$ has total degree $\le3$ in the $Y$'s, and $\Phi$ is satisfiable iff some table
$A:\mu_h\to\F$ satisfies $\mathcal Q(x,(A(\chi_\psi x))_\psi)=0$ for all $x\in\mu_h$; from a satisfying assignment
such an $A$ (with values in $\{0,1\}$) is computed in time $\mathrm{poly}(v+M_c)$, and all data in time
$\mathrm{poly}(v+M_c,\log q)$.
\end{lemma}
\begin{proof}
$L_z(z)=z^h=1$ and $L_z$ vanishes on $\mu_h\setminus\{z\}$, so at $x\in\mu_h$ the sum reduces to the single term of
$x$'s point set, or to $0$. The multipliers send each point to its partner: $p(4i+c,0)\chi_{\sigma_c}=p(4i,0)$;
$p(r,j+1)\chi_{\beta_{r-r'}}=\omega^{r'+\mathsf Rj}=p(r',j)$; $p(3k,2\beta)\chi_{\gamma_t}=p(3k+t,2\beta-1)$.
($\Rightarrow$) For a satisfying $\alpha$ put $A(p(4i+c,0))=\alpha_i$, propagate column by column through the switch
settings, and set $A=0$ elsewhere: Var holds as $\alpha_i\in\{0,1\}$, Copy and Net hold by construction, and the
value at $p(3k+t,2\beta-1)$ is the value entering at the source $\pi^{-1}(3k+t)$, i.e.\ the variable of that
literal, so a satisfied clause has a vanishing factor. ($\Leftarrow$) If $A$ satisfies all constraints,
$\alpha_i:=A(p(4i,0))$ is a root of $Y(Y-1)$, hence in $\{0,1\}$; Copy points equal their home; by induction on the
column, $A(p(r,j))$ is the column-0 value at the source routed to $(r,j)$; so each clause constraint says some
literal of $C_k$ is satisfied by $\alpha$.
\end{proof}
\emph{Credit.} Arithmetising Bene\v{s} routing inside an algebraic CSP, via De Bruijn graphs, is \cite[\S8, App.~A]{BCGT13r} (additive
setting); routing networks in PCP arithmetisation go back to \cite{BFLS91,PS94}; the selector pattern and
neighbour maps as group translations are \cite[\S\S5, 7]{BS08}, and group-preserving instances are
\cite[Def.~7.2, Lemma~7.5]{BCGV16}. Lemma~\ref{main:lem:MR} is a routine multiplicative variant; the only new point is the grid
embedding $p(r,j)=\omega^{r+\mathsf Rj}$, which makes every stage a constant shift.

\begin{proposition}[multiplicative randomisable algebraic CSP for NTIME]\label{main:prop:6M}
Let $R\in\mathrm{NTIME}(T)$, $x$ an instance, $t\ge1$, $\Phi_x$ the 3SAT[4] formula of (CL) and
$N_{\rm pts}=O(T\log^2T)$ its Lemma~\ref{main:lem:MR} size. Put $Y:=\max(N_{\rm pts},t)$,
$\NN:=(\lceil A_Y\rceil+1)Y$, where $A_Y$ is the output of a fixed deterministic evaluation of $2^{13}(\ln Y)^6$
to within $\frac12$ (so $2^{13}(\ln Y)^6Y<\NN<(2^{13}(\ln Y)^6+3)Y$); let $n$ (with $\mathcal T(n)$) be
Lemma~\ref{main:lem:SM2}'s output on $\NN$, $h$ the
least prefix product of $n$'s sorted primes with $h\ge N_{\rm pts}$, $(n,q,\Hh)\in\Wa$, $\omega$ a generator of
$\mu_h$, $\ell:=|\Hh|=8n$ (the $s$ of Definition~\ref{main:def:Wa}). Let $\mathcal Q,\Psi_{\rm lab}$ be
Lemma~\ref{main:lem:MR}'s data, $Z_H:=X^h-1$,
$C_0:=\RS[\Hh,h+t]$, $C_1:=\{(Z_HP)|_{\Hh}:\deg P<3(h+t)\}$, $g(w)[\omega']:=\mathcal Q(\omega',(w(\chi_\psi\omega'))_\psi)$,
$C':=\{(Z_HR)|_{\Hh}:\deg R<t\}$ and $\eta:=|\Psi_{\rm lab}|=2\beta+8$. If $\ln Y\ge30$:
\begin{enumerate}[nosep,label=(\alph*)]
\item \emph{locality}: each $\chi_\psi\in\mu_h\le\mu_n$ permutes $\Hh$; $g$ is $\eta$-local and every input
coordinate is read by at most $\eta$ output coordinates, so $\Delta(g(w),g(w'))\le\eta\Delta(w,w')$;
\item \emph{randomisability}: $Z_H\ne0$ on $\Hh$ and $C'$ is $t$-wise independent on $\Hh$; if $(x,w)\in R$, the
interpolant $A_0$ ($\deg A_0<h$) of Lemma~\ref{main:lem:MR}'s table is computable in time $\mathrm{poly}(T+t)$,
and every $w'_0\in A_0|_{\Hh}+C'$ has $w'_0\in C_0$ and $g(w'_0)\in C_1$;
\item \emph{soundness}: if some $w_0\in C_0$ has $g(w_0)\in C_1$ then $x\in L(R)$;
\item \emph{distance}: $C_1\cup g(C_0)\subseteq\RS[\Hh,4(h+t)]$, of relative distance $\tau_0\ge63/64$;
\item \emph{sizes}: $32(h+t)\le n$; $\ell<2^{18}Y(\ln Y)^6$; $\eta\le2\log_2\ell+8$; $\log q=O(\log\ell)$; the
residual predicate ``$w_1[\omega']=g(w_0)[\omega']$'' is a circuit with $\le2\beta+14$ monomials of degree $\le3$,
of size $O(\beta\log^2q)$, whose constants are computed with $O(h)$ field operations.
\end{enumerate}
\end{proposition}
\begin{proof}
(e) Consecutive prefix products differ by a prime factor $\le2(\ln\NN)^6$ (Lemma~\ref{main:lem:SM2}), so
$h<2(\ln\NN)^6N_{\rm pts}$. For $\ln Y\ge30$, $2^{13}(\ln Y)^6+3\le Y$, so $\ln\NN\le2\ln Y$; then
$h+t\le(128(\ln Y)^6+1)Y$ and $32(h+t)\le4128(\ln Y)^6Y<\NN\le n$. Also $\ell=8n<16\NN\le2^{18}Y(\ln Y)^6$,
$2^\beta\le N_{\rm pts}\le h\le\ell$ and $\log q=O(\log\ell)$ (Lemma~\ref{main:lem:field}). The residual: $\mathcal Q(\omega',\cdot)$
has the monomials $Y_{\rm id}^2$, $Y_{\rm id}$, $Y_{\sigma_c}$, $Y_{\beta_\delta}$ and the $\le8$ products of subsets of
$\{Y_{\gamma_t}\}$.
(a) $\mu_h\le\mu_n$ and $\Hh$ is a union of $\mu_n$-cosets; the number of distinct multipliers is $\le\eta$.
(b) $\Hh\cap\mu_h\subseteq\Hh\cap\mu_n=\emptyset$ (Lemma~\ref{main:lem:root}). For distinct $e_1,\dots,e_t\in\Hh$,
$R\mapsto(Z_H(e_j)R(e_j))_j$ is a bijection $\F[X]_{<t}\to\F^t$ (interpolation times a nonzero diagonal;
\cite[Claim~2.4]{BCGV16}). For $P=A_0+Z_HR$, $\deg P<h+t$, so $P|_{\Hh}\in C_0$. Let
$B_P(X):=\mathcal Q(X,(P(\chi_\psi X))_\psi)$; then $g(P|_{\Hh})=B_P|_{\Hh}$ by (a), and $\deg B_P\le(h-1)+3(h+t-1)$. For
$x\in\mu_h$, $\chi_\psi x\in\mu_h$ and $Z_H(\chi_\psi x)=0$, so $B_P(x)=\mathcal Q(x,(A(\chi_\psi x))_\psi)=0$ by
Lemma~\ref{main:lem:MR}; hence $B_P=Z_HP_1$ with $\deg P_1\le3(h+t)-4$ (the randomisability argument of
\cite[Lemma~7.5]{BCGV16}, with the degree bound computed here).
(c) Let $w_0=P|_{\Hh}$ with $\deg P<h+t$ and $g(w_0)=(Z_HP_1)|_{\Hh}$. $B_P$ and $Z_HP_1$ have degree
$<4(h+t)\le|\Hh|$ and agree on $\Hh$, so they are equal; so $B_P$ vanishes on $\mu_h$, and the table $A:=P|_{\mu_h}$
satisfies every constraint. By Lemma~\ref{main:lem:MR} $\Phi_x$ is satisfiable, and by (CL) $x\in L(R)$.
(d) Two distinct words of degree $<4(h+t)\le\ell/64$ agree on $<\ell/64$ points.
\end{proof}
The factor $2^{13}(\ln Y)^6$ in $\NN$ is a worst-case guard for $h<2(\ln\NN)^6N_{\rm pts}$; it costs a
polylogarithmic factor only.

\begin{proposition}[PCPP for proximity to $C_0\times C_1$]\label{main:prop:3M}
For input $y=(w_0,w_1)\in\F^{2\ell}$ put $f_0:=w_0$, $f_1:=(w_1/Z_H)|_{\Hh}$ pointwise, and let the proof be
$\Pi_0\leftarrow P_{\rm RS}(\Hh,h+t;f_0)$, $\Pi_1\leftarrow P_{\rm RS}(\Hh,3(h+t);f_1)$ with independent coins (the
PCPP of Theorem~\ref{main:thm:M}), with coin multiset
$\hat U=(\{0\}\times U_0)^{\times a_0}\sqcup(\{1\}\times U_1)^{\times a_1}$, weights within a factor~2. It has a
non-adaptive verifier with $\le16(\ln\ell)^6$ queries, $|\hat U|\le8\ell(\ln\ell)^{\log_280}$, perfect completeness,
$\Pr[\text{reject}]\ge\Delta(y,C_0\times C_1)/(8\cdot2660^{d})$, residual circuits of size
$\le4C_{\rm leaf}(\ln\ell)^{12}\log_2^2q+O(\log^2q)$, and an oracle-robust $\PPTF$
$(\infty,\max(1,\bar\kappa))$-lazy sampler on $C_0\times C_1$ that, for every oracle $y$, reads at most
$\#y\text{-queries}+\bar\kappa\cdot\#\text{proof queries}$ coordinates of $y$.
\end{proposition}
\begin{proof}
$\Delta(y,C_0\times C_1)\le\max(\Delta(w_0,C_0),\Delta(w_1,C_1))$; scaling by $1/Z_H$ is a Hamming isometry of $C_1$ onto
$\RS[\Hh,3(h+t)]$ ($Z_H\ne0$ on $\Hh$); both degrees are $\le n$ by Proposition~\ref{main:prop:6M}(e). Each part is drawn
with probability $\ge\frac13$ and rejects with probability $\ge\frac38\cdot2660^{-d}$ times its distance
(Theorem~\ref{main:thm:SND}). The query count, the coin multiset and the residual circuits are those of
Theorem~\ref{main:thm:M} for the two copies. The sampler is the joint sampler of Lemma~\ref{main:lem:2} with $c_i=1$
for the samplers of Theorem~\ref{main:thm:M}, which are lazy on $\mathcal Y_0=\RS[\Hh,h+t]$ and
$\mathcal Y_1=\RS[\Hh,3(h+t)]$ (the inputs of the honest provers; the proof of Theorem~\ref{main:thm:Z} gives
exactness for codeword oracles). For $y\in C_0\times C_1$, $f_0=w_0\in\mathcal Y_0$, and $w_1=(Z_HP)|_{\Hh}$ with
$\deg P<3(h+t)$ gives $f_1=P|_{\Hh}\in\mathcal Y_1$; so it is exact on $C_0\times C_1$. Oracle-robustness is
inherited coordinatewise.
\end{proof}

\begin{lemma}[adaptive reads of a $t$-wise independent vector; after {\cite[Lemma~6.5 of the full version]{BCGV16}} and the proof of {\cite[Thm.~8.9]{BCFGRS17}}, with a cap]\label{main:lem:1}
If $X$ is $t$-wise independent and $M$ is an oracle algorithm with coins independent of $X$ that reads $\le t$
distinct coordinates on every execution with every oracle, then $(\text{coins},M^X)\equiv(\text{coins},M^U)$ for
$U$ uniform.
\end{lemma}
\begin{proof}
Cap $M$ before a $(t+1)$-st read; for fixed coins every leaf event prescribes $\le t$ coordinates, which have the
same probability under $X$ and $U$; the cap never fires.
\end{proof}

\begin{lemma}[lazy samplers sharing an input; after {\cite[Thms.~8.1--8.2]{VW26}}]\label{main:lem:2}
Let $y_i=\phi_i(y)$ ($1\le i\le m$) with each coordinate of $y_i$ depending on $\le c_i$ coordinates of $y$,
$\Pi_i\leftarrow P_i(y_i)$ independent, and $\mathsf{Samp}_i$ a $(\hat b_i,\bar\kappa_i)$-lazy sampler for $P_i$ on
$\mathcal Y_i$. The \emph{joint sampler} answers each query to $y$ by reading $y$, and the queries to $\Pi_i$ by
$\mathsf{Samp}_i^{\phi_i(y)}$, computing each coordinate of $\phi_i(y)$ that $\mathsf{Samp}_i$ reads from $\le c_i$
coordinates of $y$. For every $y$ with $\phi_i(y)\in\mathcal Y_i$ for all $i$, it is exact against every strategy
making $\le\hat b_i$ distinct queries to each $\Pi_i$; for every oracle $y$, after $k_i$ distinct queries to each
$\Pi_i$ it has read $\le\#y\text{-queries}+\sum_ic_i\bar\kappa_ik_i$ coordinates of $y$; and it is $\PPTF$ if the
$\mathsf{Samp}_i$ are.
\end{lemma}
\begin{proof}
Fix $y$ with $\phi_i(y)\in\mathcal Y_i$ for all $i$. Hybrids $H_0,\dots,H_m$: in $H_j$ the first $j$ proofs are
answered by their samplers; the strategy that simulates everything except $\Pi_j$ internally has coins independent
of $\Pi_j$ and makes $\le\hat b_j$ distinct queries to it, and exactness of $\mathsf{Samp}_j$ on
$\phi_j(y)\in\mathcal Y_j$ gives $H_{j-1}\equiv H_j$. The read bound holds for every oracle, since each
$\mathsf{Samp}_i$ reads $\le\bar\kappa_ik_i$ coordinates of $\phi_i(y)$ for every oracle
(Definition~\ref{main:def:lazy}).
\end{proof}

\begin{theorem}[the outer shell; the construction and argument of {\cite[Thm.~8.9]{BCFGRS17}} in the PCP form of {\cite[Thm.~8.8]{ZOGS26}}]\label{main:thm:A}
Let $(\hat P,\hat V)$ be Proposition~\ref{main:prop:3M}'s PCPP, $\bar\kappa'=\max(1,\bar\kappa)$,
$\hat\delta=\tau_0/(4(\eta+1))$. The \emph{shell PCP} $(P_A,V_A)$: the prover $P_A$ samples $u'\leftarrow C'$, puts
$w'_0=A_0|_{\Hh}+u'$, $w'_1=g(w'_0)$, and outputs $(w'_0,w'_1,\hat\pi\leftarrow\hat P(w'_0,w'_1))$; the verifier $V_A$
uses $U_A=\hat U^{\times\lambda_0}\sqcup\Hh^{\times\lambda_1}$ (weights within a factor 2), runs $\hat V$ on a coin of
$\hat U$, and on
$\omega'\in\Hh$ checks $w'_1[\omega']=\mathcal Q(\omega',(w'_0(\chi_\psi\omega'))_\psi)$. Then (i) perfect completeness;
(ii) gap $\gamma_A\ge21/(2048(\eta+1)\cdot2660^d)$; (iii) black-box perfect zero knowledge against adaptive malicious
verifiers with query bound $b_A=\lfloor t/(\eta\bar\kappa')\rfloor$ (distinct symbol queries), via a universal
straight-line $\PPTF$ simulator; (iv) non-adaptive, $q_A:=\max(16(\ln\ell)^6,\eta+1)$ queries,
$|U_A|\le4\max(|\hat U|,\ell)$.
\end{theorem}
\begin{proof}
(i) Proposition~\ref{main:prop:6M}(b). (ii) For a NO instance and any $(\tilde w_0,\tilde w_1,\tilde\pi)$: if
$\Delta((\tilde w_0,\tilde w_1),C_0\times C_1)\ge\hat\delta$, $\hat V$ rejects with probability
$\ge\hat\delta/(8\cdot2660^d)$; otherwise a $\hat\delta$-close $(w_0,w_1)\in C_0\times C_1$ has each half $2\hat\delta$-close,
$g(w_0)\ne w_1$ by Proposition~\ref{main:prop:6M}(c), both lie in $C_1\cup g(C_0)$ so $\Delta(g(w_0),w_1)\ge\tau_0$,
$\Delta(g(w_0),g(\tilde w_0))\le2\hat\delta\eta$, and the consistency test rejects with probability
$\ge\tau_0-2\hat\delta(\eta+1)=\tau_0/2$. Each part is drawn with probability $\ge\frac13$, so the gap is at least
$\frac13\cdot\frac{\hat\delta}{8\cdot2660^d}\ge\frac{21}{2048(\eta+1)2660^d}$ as $\tau_0\ge\frac{63}{64}$ (and $\frac{\hat\delta}{8\cdot2660^d}<\frac{\tau_0}2$).
(iii) $S_A$ runs $V^*$ and the sampler $\hat S$ of Proposition~\ref{main:prop:3M}, answering $\hat S$'s reads of
$\hat w_0[i]$ with fresh uniform elements of $\F$ (stored) and of $\hat w_1[j]$ by computing $g(\hat w_0)[j]$ after
defining $\hat w_0$ on the $\le\eta$ coordinates it reads. Fix $(x,w)$ and a $b_A$-bounded $V^*$. The real view is a
mixture over $w'$ of views against $\hat P(w')$ with $w'=(w'_0,g(w'_0))\in C_0\times C_1$
(Proposition~\ref{main:prop:6M}(b)); for each such $w'$,
Proposition~\ref{main:prop:3M} gives equality with $\hat S^{V^*,w'}$. The oracle algorithm
$\Psi^Z:=\hat S^{V^*,(Z,g(Z))}$ reads $\le\eta\bar\kappa'b_A\le t$ coordinates of $Z$ on every execution and for
every oracle, with coins independent of $Z$; $w'_0$ is $t$-wise independent, so Lemma~\ref{main:lem:1} gives
$\Psi^{w'_0}\equiv\Psi^{U}=S_A$. (iv) By construction.
\end{proof}

\subsection{Zero-knowledge-preserving gap amplification (statement as used)}\label{main:sec:D}

Fix Dinur's instantiation $\mathcal D$ of one amplification round on the alphabet $\Sigma_0=\{0,1\}^3$
(\cite[proof of Thm.~1.5]{Din07}), with the long-code tester's free-variable set $L'$ fixed canonically and an
always-rejecting gadget for unsatisfiable constraints (the canonical folding of Section~\ref{thmD:sec:inst}). Let
$d_{\mathcal D}$ be the degree of the preprocessed graph, $t_{\mathcal D}$ the powering parameter (twice the square of a large integer,
Remark~\ref{thmD:rem:const}), $C_{\mathcal D}$ and
$\alpha_{\mathcal D}\in(0,1)$ the constants of \cite[Thm.~1.5]{Din07} for $|\Sigma|=8$, and
\[c_{\rm it}:=2d_{\mathcal D}^{\lceil t_{\mathcal D}/2\rceil},\qquad E_L:=\log_2c_{\rm it},\qquad E_C:=\log_2C_{\mathcal D}.\]
These are fixed numbers of one explicit algorithm; only lower bounds on them are known from the analysis of
\cite{Din07} with the repairs of Section~\ref{sec:thmD} ($E_L\ge3.958\cdot10^{13}$ and $Q_D\ge10^{13}$ below, and $E_C$
doubly exponential in $t_{\mathcal D}$).

\begin{theorem}[zero-knowledge-preserving gap amplification; Theorem~\ref{thmD:thm}]\label{main:thm:D}
Let $(P_0,V_0)$ be a non-adaptive PCP with coin multiset $U$ ($|U|\le\mathrm{poly}(|x|)$) for a promise relation,
with alphabet $\{0,1\}^{a_0}$, at most $q_0$ queries, residual circuits of at most $s_0$ gates computable with the
query positions in polynomial time, perfect completeness and gap $\gamma\in(0,1]$. Let $m^*\le4s_0+1$ be the number of
clauses per test of the Tseitin front end (Section~\ref{thmD:sec:T}), $\hat\gamma:=\max\{\gamma,1/|U|\}$ and
$k:=\lceil\log_2(3m^*/\hat\gamma)\rceil$. Then there is a PCP $(P',V')$ for the same promise relation with:
\begin{enumerate}[nosep,label=(\alph*)]
\item binary alphabet; non-adaptive $V'$ with $Q_D:=6\lceil1/\alpha_{\mathcal D}\rceil$ queries and a coin multiset of
polynomial size; perfect completeness; soundness error $\le\frac12$;
\item length $\le27(4s_0+1)|U|\,C_{\mathcal D}(3(4s_0+1)/\gamma)^{E_C}$;
\item $P'(x,w)=f_x(P_0(x,w))$ with $f_x$ deterministic, polynomial-time, $x$-only, total, and every output bit
depending on at most $L_D:=q_0c_{\rm it}^{\,k}\le q_0c_{\rm it}(3(4s_0+1)/\gamma)^{E_L}$ symbols of the old proof;
\item if $(P_0,V_0)$ is black-box perfect (resp.\ $\varepsilon$-statistical) zero knowledge with query bound $b_0$
against adaptive malicious verifiers via a universal simulator $\mathrm{Sim}_0$, then $(P',V')$ is black-box perfect
(resp.\ $\varepsilon$-statistical) zero knowledge with query bound $\lfloor b_0/L_D\rfloor$ via a universal simulator
that is straight-line (in the sense of Definition~\ref{thmD:def:zk}) if $\mathrm{Sim}_0$ is;
\item the classes are preserved: a strict-PPT or Las Vegas prover stays so; if $b_0\le\mathrm{poly}(|x|)$, a simulator
of class (SP), (EP) or (GB) stays so, with time $\le|\rho|+p(|x|)+(N_0+1)(2+Q^*+3Q^*L_Dp(|x|))$, where $p$ is the polynomial of
Lemma~\ref{thmD:lem:DP} (construction time and time per dependency-graph node), $N_0$ is
$\mathrm{Sim}_0$'s step count without the $|\rho|$ steps that toss $V^*$'s coins and $Q^*$ bounds the queries of $V^*$ (at most $\lfloor b_0/L_D\rfloor$ under the
memoised interface).
\end{enumerate}
\end{theorem}
\emph{Credit.} Theorem~\ref{main:thm:D} restates, in the form used here, items (a)--(e) of
Theorem~\ref{thmD:thm}, which is proved in Section~\ref{sec:thmD}. It is a complete write-up of the sketch in
\cite[App.~A, Lemma~A.1]{ZOGS26}: Dinur's gap amplification \cite{Din07} with a Tseitin front end \cite{Tse68}, and
zero knowledge through a re-proof of \cite[Lemma~3.2]{GOS25} with a cache and a cap (against adaptive malicious
verifiers; only non-adaptive local post-processing is used). The proof also corrects six points of the printed
presentation of \cite{Din07}: the domination step in the soundness proofs of \cite[Lemmas~1.8 and~7.2]{Din07}
(Remark~\ref{thmD:rem:dom}), the folding of the long-code tester \cite[\S7]{Din07}, the case $\delta>\frac12$ in
the proof of \cite[Prop.~7.3]{Din07}, and three steps in the proof of \cite[Lemma~6.1]{Din07}, whose repair
restricts the powering parameter to $t_{\mathcal D}=2r^2$ with $r$ large (Remark~\ref{thmD:rem:L61}).

The following table checks the hypotheses of Theorem~\ref{main:thm:D} for the shell PCP $(P_A,V_A)$ of
Theorem~\ref{main:thm:A}.

\begin{center}\small
\begin{tabular}{p{5.0cm}p{6.3cm}p{3.6cm}}
\toprule
hypothesis of Theorem~\ref{main:thm:D} & $(P_A,V_A)$ & where\\\midrule
non-adaptive, coin multiset & yes, $|U_A|\le32\ell(\ln\ell)^{\log_280}$ & Thm.~\ref{main:thm:A}(iv), Prop.~\ref{main:prop:3M}\\
alphabet $\{0,1\}^{a_0}$ & $a_0=\lambda_q=O(\log\ell)$, decoded by $\mathrm{dec}$ & \S\ref{main:sec:conv}\\
queries $q_0$ & $q_A=\max(16(\ln\ell)^6,\eta+1)$ & Thm.~\ref{main:thm:A}(iv)\\
residual circuits $s_0$, poly-time & $s_A=O((\ln\ell)^{12}\log^2q+\beta\log^2q+q_A\lambda_q)$ & Props.~\ref{main:prop:6M}(e), \ref{main:prop:3M}\\
perfect completeness, gap & $\gamma_A\ge21/(2048(\eta+1)2660^d)$ & Thm.~\ref{main:thm:A}\\
ZK, universal, adaptive, bound $b_0$ & $b_A=\lfloor t/(\eta\bar\kappa')\rfloor$, straight-line, $\PPTF$ & Thm.~\ref{main:thm:A}(iii)\\
simulator class & (EP) by (F-a); (GB) by (F-b) & \S\ref{main:sec:conv}\\
prover class & $\PPTF$, hence Las Vegas & Props.~\ref{main:prop:LEN}, \ref{main:prop:6M}\\\bottomrule
\end{tabular}
\end{center}

\subsection{The main theorem}\label{main:sec:main}

\begin{theorem}[main theorem]\label{main:thm:main}
For every time-constructible polynomial $T(n)\ge n$, every time-constructible polynomial $b(n)\ge1$ (including
$b\gg T$) and every $R\in\mathrm{NTIME}(T)$, there is a PCP for $R$ with:
\begin{enumerate}[nosep,label=(\arabic*)]
\item alphabet $\{0,1\}$, a verifier running in time $\mathrm{poly}(|x|)$, and
$Q_{\rm fin}=3Q_D=18\lceil1/\alpha_{\mathcal D}\rceil$ non-adaptive queries (an absolute constant);
\item perfect completeness and soundness error $\le\frac{27}{64}$, using $O(\log(T+b))$ fair coins;
\item length $\le c_{\rm len}\,(T+b)\,(\ln(T+b))^{C}$ with a constant $c_{\rm len}$ that may depend on $R$, $T$ and
$b$ but not on the input (the constant in $N_{\rm pts}=O(T\log^2T)$ depends on $R$ and $T$), and
\[C:=26.38\,(E_C+E_L)+122,\]
which dominates the exact exponent $(15+\log_22660)(E_C+E_L)+\log_2\aKL+27+\log_280+o(1)$, where
$15+\log_22660<26.378$, $\log_2\aKL<87.39$ and $\log_280<6.33$;
\item perfect zero knowledge (Definition~\ref{main:def:ZK}) against every adaptive, malicious, computationally
unbounded verifier making at most $b$ distinct bit queries, via one universal straight-line black-box simulator that
never sees the witness and is exact in expected time $\mathrm{poly}(|x|,T,b)$ (class (EP); here and below every time
bound carries the additive $O(|\rho|)$ for $V^*$'s coins $\rho$ of Definition~\ref{main:def:ZK}); the same simulator run
on a pool (F-b) is of class (GB), and truncated it is of class ($\varepsilon$-SP) with
$\varepsilon=N_{\rm samp}2^{-k}$ for every polynomially bounded $k$;
\item an honest prover running in expected polynomial time (Las Vegas) given a witness.
\end{enumerate}
The amortised read factor of the RS PCPP is $\bar\kappa\le(\ln\ell)^{\log_2\aKL}$ with $\log_2\aKL<87.39$, the
function $\Xi$ of Step~1 of the proof satisfies $\Xi(\ell)=(\ln\ell)^{(15+\log_22660)E_L+\log_2\aKL+7+o(1)}$, and the
root domain of the shell PCP has size $\ell=|\Hh|=O\big((T\log^2T+b\,\Xi(Y_0^2))(\log(T+b))^6\big)$, with $Y_0$ as in
Step~2.
\end{theorem}

\begin{proof}
\emph{1. Envelopes.} For $\ell\ge3$ define the explicit functions
$\bar\eta(\ell)=2\log_2\ell+9$, $\bar\kappa(\ell)=(\ln\ell)^{\log_2\aKL}$, $\bar q(\ell)=16(\ln\ell)^6+\bar\eta(\ell)$,
$\bar s(\ell)=C_s(\ln\ell)^{14}(\ln\ln\ell)^2$, $\bar\gamma(\ell)^{-1}=\frac{2048}{21}\bar\eta(\ell)(\ln\ell)^{\log_22660}$,
$\bar L_D(\ell)=\bar q(\ell)c_{\rm it}(3(4\bar s(\ell)+1)\bar\gamma(\ell)^{-1})^{E_L}$ and
$X(\ell):=2\bar L_D(\ell)\bar\eta(\ell)\bar\kappa(\ell)$, all nondecreasing on $[3,\infty)$ (where $\ln\ln\ell>0$), and
let $\Xi(\ell):=\lceil\tilde X(\ell)\rceil+1$, where $\tilde X(\ell)$ is the output of a fixed deterministic evaluation
of $X(\ell)$ to within $\frac12$. As $X(\ell)=(\ln\ell)^{O(1)}$ and $C_s,c_{\rm it},\aKL$ are fixed numbers, $\Xi(\ell)$
is computable in time $\mathrm{poly}(\log\ell)$, and $X(\ell)\le\Xi(\ell)<X(\ell)+3$. Here $C_s$ is an absolute constant
chosen so that $s_A\le\bar s(\ell)$ whenever $\log_2q\le\ln\ell\cdot\ln\ln\ell$ ($C_s$ is explicit: it comes from the
leaf circuits of Proposition~\ref{main:prop:LEN}, the consistency circuit of Proposition~\ref{main:prop:6M}(e) and
$\mathrm{dec}$); by Lemma~\ref{main:lem:field}, $\log_2q=O(\log\ell)$, so this holds for all $\ell\ge\ell_{\rm Lin}$, an
absolute threshold (we take $\ell_{\rm Lin}\ge3$) whose existence is all we use. On $\mathcal W$ with
$\ell\ge\ell_{\rm Lin}$: $\eta+1\le\bar\eta$,
$\max(1,\bar\kappa)\le\bar\kappa(\ell)$ (Theorem~\ref{main:thm:Z}), $q_A\le\bar q$, $s_A\le\bar s$,
$\gamma_A\ge\bar\gamma$ (Theorem~\ref{main:thm:A}, as $2660^d\le(\ln\ell)^{\log_22660}$), hence
$L_D\le\bar L_D(\ell)$ for the shell PCP by Theorem~\ref{main:thm:D}(c), and $X(\ell)\ge2L_D\eta\bar\kappa'$.
\emph{2. Choice of $t$.} Given $x$ with $|x|=n_{\rm in}$, let $Y_0:=\max(N_{\rm pts}(\Phi_x),b(n_{\rm in}),n_{\rm in})$ and
$t:=\lceil b\,\Xi(Y_0^2)\rceil$ ($b$ is time-constructible and $\Xi(Y_0^2)$ is computable in time
$\mathrm{poly}(\log Y_0)$). Run
Proposition~\ref{main:prop:6M} with this $t$; then $Y\le Y_0\Xi(Y_0^2)+1$ and $\ell<2^{18}Y(\ln Y)^6$.
\emph{3. Thresholds.} There is $Y_{\rm thr}$ such that for $Y_0\ge Y_{\rm thr}$: $\ln Y\ge30$; $2^{18}Y(\ln Y)^6\le Y_0^2$,
hence $\ell\le Y_0^2$ and $X(\ell)\le X(Y_0^2)\le\Xi(Y_0^2)$; $\ell\ge\ell_{\rm Lin}$; and $\ln\NN\ge7$. Since
$Y_0\ge n_{\rm in}$
there is $n_{\rm thr}$ with $Y_0\ge Y_{\rm thr}$ for all $n_{\rm in}\ge n_{\rm thr}$, independently of $b$. (No threshold on
node sizes is needed: Corollary~\ref{main:cor:HDL} holds at every recursion node.)
\emph{4. Zero knowledge.} $t/(\eta\bar\kappa')\ge b\,X(\ell)/(\eta\bar\kappa')\ge2bL_D$, so
$b_A=\lfloor t/(\eta\bar\kappa')\rfloor\ge bL_D$, and Theorem~\ref{main:thm:D}(d) applied to the shell PCP
(Theorem~\ref{main:thm:A}) gives perfect zero knowledge against $\lfloor b_A/L_D\rfloor\ge b$ distinct bit queries.
The simulator of Theorem~\ref{main:thm:A}(iii) is universal and straight-line in the sense of
Definition~\ref{thmD:def:zk}, as it outputs the view of its single run; it is $\PPTF$, so it is (EP) by (F-a) and
(GB) by (F-b), with $N_{\rm samp}\le b_A(M_{\rm mask}+1)+t$ samples on every run against every $b_A$-bounded strategy
($M_{\rm mask}$ the number of mask coordinates of the RS PCPP); as $b_A\le t=\mathrm{poly}(|x|)$,
Theorem~\ref{main:thm:D}(e) preserves both classes
and (d) straight-line simulation, and Lemma~\ref{main:lem:CC} below does not touch the proof. Its time is
$\mathrm{poly}(|x|,T,b)$ under the memoised interface plus the $|\rho|$ steps that toss $V^*$'s coins: in the bound of
Theorem~\ref{main:thm:D}(e) these steps are added once, not multiplied, and the number $N_0$ of the other steps of $S_A$
is polynomial (in expectation for (EP)). Given uniform field elements instead (Definition~\ref{main:def:PPTF}), the
final simulator is $\PPTF$ and exact: it calls the oracle only inside $S_A$, which runs on the verifier $A'$ of
Lemma~\ref{thmD:lem:Z}, $b_A$-bounded on every oracle by its cap, so at most $N_{\rm samp}$ times; the step bound of
Lemma~\ref{thmD:lem:DP} (for any $\mathrm{Sim}_0$) holds on every run, with $N_0$ polynomial on every run as $S_A$ is $\PPTF$; and exactness
follows from Theorem~\ref{main:thm:D}(d), which concerns laws only (F-d). Hence (F-a)--(F-c) also apply to the final
simulator directly, and (F-c) gives the class ($\varepsilon$-SP) of item~(4) with $\varepsilon=N_{\rm samp}2^{-k}$.
\emph{5. Length.} Theorem~\ref{main:thm:D}(b) gives length $\le27(4s_A+1)|U_A|C_{\mathcal D}(3(4s_A+1)/\gamma_A)^{E_C}$. With
$|U_A|\le32\ell(\ln\ell)^{\log_280}$, $s_A\le\bar s(\ell)$ and $\gamma_A^{-1}\le\bar\gamma(\ell)^{-1}$ this is
$\le c\,\ell(\ln\ell)^{14+\log_280+(15+\log_22660)E_C+o(1)}$, and
$\ell<2^{18}Y(\ln Y)^6$ with $Y\le\max(N_{\rm pts},b\,\Xi(Y_0^2))+1$, $N_{\rm pts}=O(T\log^2T)$ and
$\Xi(Y_0^2)\le c'(\ln(T+b))^{(15+\log_22660)E_L+\log_2\aKL+7+o(1)}$. Since $N_{\rm pts}=O(T\log^2T)$ and
$\Xi(Y_0^2)\ge(\ln Y_0)^2$, we have $Y\le c''(T+b)(\ln(T+b))^{e_\Xi}$ with
$e_\Xi:=(15+\log_22660)E_L+\log_2\aKL+7+o(1)$, and $\ln\ell=O(\ln(T+b))$. So the length is at most
$c\,(T+b)(\ln(T+b))^{e_\Xi+6+14+\log_280+(15+\log_22660)E_C+o(1)}$, and
$e_\Xi+20+\log_280+(15+\log_22660)E_C=(15+\log_22660)(E_C+E_L)+\log_2\aKL+27+\log_280+o(1)$, the exact exponent
stated in item~(3). The $o(1)$ terms (from $(\ln\ln\ell)^2$ in $\bar s$) are absorbed by the slack in $C$ above an
absolute threshold, below which $c_{\rm len}$ absorbs everything.
\emph{6. Fair coins and queries.} Lemma~\ref{main:lem:CC} applied to Theorem~\ref{main:thm:D}'s output gives
$3(\lceil\log_2|U'|\rceil+1)=O(\log(T+b))$ coins, $3Q_D$ queries and soundness error $\le\frac{27}{64}$. The
verifier runs in time $\mathrm{poly}(|x|)$: $V'$ does by Theorem~\ref{thmD:thm}\ref{thmD:it:a}, whose hypotheses
(a coin multiset of polynomial size and tests computable in polynomial time) the table above checks for the shell
PCP; the conversion of Lemma~\ref{main:lem:CC} adds only the computation of $|U'|$ and of $j\bmod|U'|$ and three
repetitions; and the table of Step~7 is finite.
\emph{7. Small inputs.} For the finitely many $x$ with $|x|<n_{\rm thr}$ the verifier decides membership from a
hard-coded table, the proof is empty and zero knowledge is trivial (the view is $V^*$'s coins).
\end{proof}

\begin{remark}[which reading of Open Problem~1 is met]\label{main:rem:gos25}
Section~\ref{sec:defs-gos} compares Theorem~\ref{main:thm:main} with the definitions of \cite[Defs.~2.1 and~2.9--2.12]{GOS25}
item by item. Open Problem~1 asks for ``$O(1)$-query ZK-PCPs for NP with proof length $\tilde O(n)$''; reading $n$ as
the verification time $T$ of the NP relation (as for the quasilinear PCPs \cite{BS08,Din07} it cites) and ``ZK'' as
zero knowledge against every polynomial query bound $b$ with length $\tilde O(T+b)$, Theorem~\ref{main:thm:main}
answers it with perfect zero knowledge (an exact simulator in
expected polynomial time, or in strict polynomial time
with a failure symbol in the form of \cite[Def.~4.3.1]{Gol01}) and with statistical zero knowledge of negligible error
via a strict polynomial-time simulator. These replace the strict polynomial-time perfect simulator on fair coins
without a failure symbol that \cite[Defs.~2.10--2.11]{GOS25} require, which this construction cannot have
(Remark~\ref{main:rem:strict}). O'Connor's restatement \cite[\S6.2]{OCo26} (boolean alphabet, constant
soundness error, query bound $\Omega(n^c)$ for some $0<c<1$) is met for every $c\le1$ with $b=n$, length
$\tilde O(T)$.
\end{remark}

\begin{remark}[strict-PPT perfect simulation, for this construction]\label{main:rem:strict}
All randomness of the honest prover consists of uniform elements of $\F_q$ (item (F-e) of Section~\ref{main:sec:conv}) and $f_x$ is deterministic, so for every
deterministic $V^*$ (one that tosses no coins) every view probability lies in $q^{-K}\mathbb Z$ for some $K$; a
strict-PPT simulator with fair coins and running time $\le r$
outputs every view with probability in $2^{-r}\mathbb Z$, and a rational in $(0,1)\cap q^{-K}\mathbb Z$ is never dyadic
($q$ odd). So no strict-PPT simulator is exact without $\bot$ for this construction as soon as some bit of the final
proof is non-constant under the honest prover (the deterministic verifier that reads that one bit has a non-dyadic view
probability). Such a bit exists for every $x$ with $|x|\ge n_{\rm thr}$ and every $b\ge1$. (i) The PCPP $\Pi_0$ of
Proposition~\ref{main:prop:3M} for $f_0=w'_0$ has a test reading an input symbol $w'_0(x_0)$, $x_0\in\Hh$: a test
$g(x_0)=y(x_0)-r(x_0)$ of $U_{\rm cons}$ at the root of copy $A$ if the root is internal, and the leaf test of copy $A$,
which reads all of $f_0$, if it is a leaf. As a Boolean function of the bits it reads, the residual predicate of that
test depends on the lowest bit $\beta$ of the block of $w'_0(x_0)$: with the blocks of $y(x_0)$ and $r(x_0)$ encoding
$1$ and $0$, $U_{\rm cons}$ accepts the encoding of $1$ and rejects that of $0$; and the leaf test accepts the zero word
and rejects the word that is $1$ at $x_0$ and $0$ elsewhere, since a nonzero polynomial of degree $<p$ cannot vanish on
the other $8p-1$ points. So $\beta$ is an input wire of that residual circuit, i.e.\ it lies in the set $B(x)$ of
Theorem~\ref{thmD:thm}\ref{thmD:it:f}. (ii) Under the honest prover $w'_0(x_0)$ is uniform on $\F_q$ ($C'$ is
$t$-wise independent with $t\ge1$), so $\Pr[\beta=1]=\frac{q-1}{2q}\in(0,1)$. (iii) By
Theorem~\ref{thmD:thm}\ref{thmD:it:f}, $f_x(P_0(x,w))$ is therefore non-constant, and Lemma~\ref{main:lem:CC} does not
change the proof. The remark is not used by Theorem~\ref{main:thm:main}.
\end{remark}

\begin{remark}[asymptotic regime]\label{main:rem:regime}
The depth $d=\log_2m_\#$ of the domain family depends only on $\NN$ and is nondecreasing in it: if a power of two
$m$ passes the test of Lemma~\ref{main:lem:SM2} at $\NN$, then $3m\le\LL/\ln\LL$, so $\phi_\NN(m)=\LL-3m\ln\LL$ grows
with slope $\ge1-1/\ln7>0.48$ in $\LL$, hence by more than $0.24/\NN>2\epsilon_{\NN'}$ from $\NN$ to any integer
$\NN'>\NN$, and $m$ passes at $\NN'$. Every passing $m$ has $m\le\LL/(3\ln\LL)$, so $d=0$ (the root is a leaf
and the PCPP reads its
whole input) for $\NN<e^{16.99}\approx2^{24.5}$, $d\le1$ for $\NN<e^{45.92}\approx2^{66.2}$, and $d\le2$ for
$\NN<e^{113.58}\approx2^{163.8}$; as the root domain has size $s=8n\ge8\NN$, the depth is at most~$2$ whenever
$s<2^{160}$. The locator lemma is non-vacuous only at nodes with $m_1\ge1/c_*=10^{13}$; the thresholds $Y_{\rm thr}$ of step~3 satisfy
$\log_2Y_{\rm thr}\approx E\log_2E$ with $E\approx26.4E_L$. The query bound $16(\ln s)^6$ and every polylogarithmic factor
are polylogarithmic only asymptotically. The theorem is a feasibility statement; it does not suggest concrete
efficiency.
\end{remark}

\begin{remark}[alternatives not used]\label{main:rem:alt}
(1) A variant in which the second lift has radix $m_2$, so that the two lifts use the coprime radices $m_1$ and $m_2$,
appears harder to analyse: the two sides then lie in different tensor products, and the divisibility argument behind
the locator lemma has no direct analogue. (2) The small-view locator $W_1$ of Section~\ref{kl:sec:practical}
(ratio~$1$ for views of at most $m_1+m_2-1$ reads by Theorem~\ref{kl:thm:small}, and of at most
$m_1+\lfloor3m_2/2\rfloor-1$ reads on nodes with $m\ge m_0$ and $m_2<2m_1$ by the discussion after
Theorem~\ref{kl:thm:mob}) is a
locator only for small views; it is defined, monotone and position-only on all views, so its union with
$W_\nu$ is again a monotone locator (Proposition~\ref{kl:prop:natural}). It is not used, and
Theorem~\ref{main:thm:main} relies on $W_\nu$ alone.
\end{remark}

\endgroup

\begingroup
\providecommand{\F}{\mathbb{F}}
\providecommand{\RS}{\mathrm{RS}}
\providecommand{\ev}{\mathrm{ev}}
\providecommand{\spn}{\mathrm{span}}
\providecommand{\poly}{\mathrm{poly}}
\providecommand{\klK}{\overline{\mathbb F}_p}
\providecommand{\klX}{\mathcal X}
\providecommand{\klord}{\mathrm{ord}}
\providecommand{\klsupp}{\mathrm{supp}}

\section{The locator lemma}\label{sec:keylemma}

\noindent This section proves the locator lemma for a single node of the same-radix configuration
(Definition~\ref{kl:def:node}): every view of a node has a monotone, explicitly computable dependence set (a
\emph{locator}) of size at most an absolute constant times the view (Theorem~\ref{kl:thm:KL}).
Section~\ref{sec:main} uses it through Theorem~\ref{main:thm:KL} and Corollary~\ref{main:cor:HDL}, and
Section~\ref{sec:overview-key} sketches its proof. Section~\ref{kl:sec:setting} states the lemma,
Sections~\ref{kl:sec:tensor}--\ref{kl:sec:div} develop the tools (rows, the tensor form and leaks; the incidence
input; the divisibility lemma with domain columns), and Section~\ref{kl:sec:proof} completes the proof.
Section~\ref{kl:sec:practical} gives a small-view locator that the main theorem does not use.

\paragraph{Credits.} The divisibility argument (with the decomposition) is due to Polishchuk and Spielman
\cite[Lemma~8, Theorem~9]{PS94}. Its version for multiplicative subgroups of coprime orders, with the
Corvaja--Zannier gcd estimate, is \cite[\S4.4: Theorems~4.11, 4.12, Lemma~4.13]{BV26}; we credit it and use none of
its results: every step is proved below. Collective product expansion and the monotone closure locator are
\cite[Lemma~5.13, Prop.~5.14]{DLV24} and \cite[Lemmas~17--18]{KP25} (reproduced as \cite[Def.~4.7, Lemma~4.8]{BV26});
the single-polynomial route of Theorem~\ref{kl:thm:rows} is an alternative to that closure argument. Product expansion
is due to Panteleev--Kalachev \cite{PK22}, and ``product expansion implies privacy'' to \cite[Lemma~4.9]{GG25}. The
incidence input is \cite[Theorem~2]{CZ13}; related bounds for polynomial equations in subgroups are in
\cite{Vyu19,MV19,KMSV17,KSV22}. The tensor reshaping is \cite[\S6]{BS08}; the leak linear algebra is
elementary (Lemma~\ref{main:lem:DL}; cf.\ \cite[Lemma~4.3, Claim~7.7]{BCFGRS17}). The extension to unions of cosets is routine (the incidence count $\Gamma$ of
Section~\ref{kl:sec:cz} is additive over coset pairs). The
\emph{domain-column refinement} (Step~0, the division by $\Pi$, Claim~\ref{kl:clm:delta}) is, as far as we know, new; so
is the sharper margin of Claim~\ref{kl:clm:delta}(c$'$).

\subsection{Setting and statement}\label{kl:sec:setting}

Let $\F$ be a finite field of characteristic $p$ and $K_0=\klK$. In this section $b$ denotes an element of the row set
$B$ of a node, never the query bound.

\begin{definition}[node]\label{kl:def:node}
A \emph{node with parameters $(L_1,L_2,K)$} (of the same-radix double lift: both lifted tables $Q_y$, $Q_r$ below
use the same radix $n$) consists of:
\begin{enumerate}[label=(N\arabic*),nosep]
\item integers $n:=m_1\ge1$, $n':=m_2$ with $\gcd(n,n')=1$, $n<n'\le Kn$, $m:=nn'<p$ and $\mu_m\subseteq\F^\times$;
\item $D=\bigsqcup_{i<L_2}d_i\mu_m\subseteq\F^\times$, a union of $L_2$ \emph{distinct} $\mu_m$-cosets; the code is
$C=\F[X]_{<m}$ (evaluated on $D$);
\item the rows $B:=\{x^n:x\in D\}$ and fibres $F_b:=\{x\in D:x^n=b\}$ (Lemma~\ref{kl:lem:rows});
\item column sets $A=A_y\sqcup A_r\subseteq\F^\times\setminus D$, together a union of $L_1\ge1$ distinct
$\mu_n$-cosets.
\end{enumerate}
For $g\in C$, a uniform mask $r\in C$ and $y=g+r$, the node's symbols are $y|_D$, $r|_D$, $Q_y|_{A_y\times B}$,
$Q_r|_{A_r\times B}$ with $Q_f(a,b)=(f\bmod(X^n-b))(a)$; if a view may also contain reads $g(x)$, $x\in G\subseteq D$,
of the input itself, these are symbols too (they do not occur in the recursion, Lemma~\ref{main:lem:V}(ii), but the
statement allows them). A \emph{view} $V$ is a set of symbols. Its \emph{leak} is
$\Lambda(V)=\spn\Phi\cap\spn\Psi\,(+\spn\{\ev_x:x\in G\})$, where $\Phi,\Psi$ are the functionals on $C$ of the $y$-side
and $r$-side reads (Lemma~\ref{main:lem:DL}; cf.\ \cite[Lemma~4.3, Claim~7.7]{BCFGRS17}): the answers to $V$ have the same law under inputs $g,g'$ iff
$\lambda(g-g')=0$ for all $\lambda\in\Lambda(V)$. A set $W\subseteq D$ is a \emph{locator} for $V$ if
$\Lambda(V)\subseteq\spn\{\ev_x:x\in W\}$.
\end{definition}

For a view $V$ let $D_y$ and $D_r$ be the sets of $x\in D$ with $y(x)\in V$, resp.\ $r(x)\in V$, and let
$D_V:=D_y\cup D_r\cup G\subseteq D$ be the set of all domain abscissae read, input reads included; an input read $g(x)$
has the point $\hat x$ (Lemma~\ref{kl:lem:tensor}). Let $\mathrm{Pts}(V)$ be the set of points of $V$
(Lemma~\ref{kl:lem:tensor}), $|V\cap\text{row }b|$ the number of points of $\mathrm{Pts}(V)$ with second coordinate $b$, and
$H_t(V):=\{b\in B:|V\cap\text{row }b|\ge n+1-t\}$.

\begin{definition}[admissible constants]\label{kl:def:adm}
$c>0$ is \emph{admissible} for $(L_1,L_2,K)$ if $c<\frac1{2(L_2+1)}$ and
\begin{equation}\label{kl:eq:adm}
L_1L_2\bigl(3\cdot2^{-1/3}K^{1/3}c^{1/3}+6c\bigr)<\frac1{L_2+1}-2c .
\end{equation}
\end{definition}

\begin{theorem}[locator lemma]\label{kl:thm:KL}
Let $\nu$ be a node with parameters $(L_1,L_2,K)$, $c$ admissible and $\tau:=\lfloor cn\rfloor$. For every view
$V$ put
\[
W(V):=D_V\ \cup\bigcup_{b\in H_\tau(V)}F_b\ \cup\
\begin{cases} d_0\mu_m,& |V|\ge(\tau+1)^2,\\ \emptyset,&\text{otherwise.}\end{cases}
\]
Then (a) $W$ is monotone, depends only on the positions in $V$, satisfies $W(V)\subseteq D$ and is computable in time
$\poly(|\F|)$; (b) $W(V)$ is a locator for every view $V$; (c) $|W(V)|\le a(c)|V|$ with
$a(c):=1+\frac1{1-c}+\frac K{c^2}$.
\end{theorem}

\begin{remark}[constants]\label{kl:rem:const}
For the parameters $(L_1,L_2,K)=(14,8,2)$ of the recursion (Remark~\ref{kl:rem:route}) we fix, throughout the paper,
\[
c_*:=10^{-13}.
\]
It is admissible: $c_*<\frac1{18}$, and the left side of \eqref{kl:eq:adm} is $112(3c_*^{1/3}+6c_*)<0.0156<\frac19-2c_*$.
(Already the weaker margin $\frac n{8L_2}$ of Claim~\ref{kl:clm:delta}(c) suffices: $\frac1{L_2+1}-2c\ge\frac1{8L_2}$ for
$c\le\frac1{16L_2}$, and $112(3c_*^{1/3}+6c_*)<\frac1{64}$.) It gives $a(c_*)=1+\frac1{1-c_*}+2\cdot10^{26}\le2.0001\cdot10^{26}$
and $\log_2a(c_*)\approx87.37$. Through the tree-dependence lemma (Lemma~\ref{main:lem:C}) the per-node constant enters
the sampler's read bound as $\bar\kappa\le(\ln s)^{\log_2a}$ (Theorem~\ref{main:thm:Z}), so $\log_2a$ is added to the
absolute polylog exponent of the main theorem.

The sharper margin of Claim~\ref{kl:clm:delta}(c$'$), which is what Definition~\ref{kl:def:adm} uses, admits larger
constants: the largest admissible $c$ for $(14,8,2)$ is about $3.6\cdot10^{-11}$, and for instance $c=3.58\cdot10^{-11}$
is admissible and gives $a\le1.561\cdot10^{21}$ and $\log_2a\approx70.40$. This proved sharpening would lower the
exponent of the main theorem by about $17$; we do not use it, and keep $c_*$ throughout.

The non-trivial content of (b) is the regime $|V|<(\tau+1)^2$, i.e.\ views of up to about $c^2n^2\ge(c^2/K)m$ reads
($5\cdot10^{-27}m$ for $c_*$); larger views are located by a whole $\mu_m$-coset. Small views at nodes of every size
are located with ratio about~$2$ by Proposition~\ref{kl:prop:natural}.
\end{remark}

\begin{remark}[recursion nodes]\label{kl:rem:route}
Every node of the recursion over $\F_q$ ($q$ prime, canonical rule; Definition~\ref{main:def:can}) satisfies
(N1)--(N4) with $(L_1,L_2,K)=(14,8,2)$ (Lemma~\ref{main:lem:N}): $m\mid q-1$ gives $m<q=p$ and $\mu_m\le\F_q^\times$; the domain is
eight distinct $\mu_m$-cosets; the two column families are seven $\mu_{m_1}$-cosets each, lying in fourteen further
distinct $\mu_m$-cosets (so $A\cap D=\emptyset$); and $m_1<m_2<2m_1$. Theorem~\ref{kl:thm:KL} is stated for all subsets
of the node's symbols, with no size threshold, and depends on positions only, as the tree-dependence lemma
(Lemma~\ref{main:lem:C}) requires.
\end{remark}

\subsection{Rows, tensor form and leaks}\label{kl:sec:tensor}

\begin{lemma}[rows]\label{kl:lem:rows}
Under (N1)--(N2), $B=\bigsqcup_{i<L_2}d_i^n\mu_{n'}$ with $|B|=L_2n'$, and $F_b=x_b\mu_n$ for any $x_b\in F_b$; in
particular $|F_b|=n$ and the $F_b$ partition $D$.
\end{lemma}
\begin{proof}
$z\mapsto z^n$ maps the cyclic group $\mu_m$ onto $\mu_{n'}$ with kernel $\mu_n$, so $(d_i\mu_m)^n=d_i^n\mu_{n'}$ and
the map is $n$-to-one on $d_i\mu_m$ with fibres $x\mu_n$. If $d_i^n\mu_{n'}=d_j^n\mu_{n'}$ then $(d_i/d_j)^{nn'}=1$,
so $d_i/d_j\in\mu_m$ and $i=j$.
\end{proof}

\begin{lemma}[{tensor form; {\cite[\S6]{BS08}}}]\label{kl:lem:tensor}
Let $U=\F[X]_{<n}$, $U'=\F[Z]_{<n'}$, $v(x)=\ev_x\in U^*$, $w(b)=\ev_b\in U'^*$. The map
$f=\sum_{t<n}X^tf_t(X^n)\mapsto P_f=\sum_tX^t\otimes f_t$ is an isomorphism $C\to U\otimes U'$ with
$Q_f(a,b)=P_f(a,b)$ and $f(x)=P_f(x,x^n)$. Hence a lift read is the point $(a,b)\in A\times B$, a domain read is the
point $\hat x:=(x,x^n)$, and the functional of a point $(x,b)$ is $v(x)\otimes w(b)$.
\end{lemma}
\begin{proof}
$f\bmod(X^n-b)=\sum_tX^tf_t(b)$.
\end{proof}

\begin{lemma}[rows are cheap]\label{kl:lem:cheap}
For $b\in B$, $U^*\otimes w(b)=\spn\{\ev_x:x\in F_b\}$.
\end{lemma}
\begin{proof}
$\ev_x=v(x)\otimes w(b)$ for $x\in F_b$, and the $n$ distinct points of $F_b$ give a basis of $U^*$.
\end{proof}

\begin{lemma}[shared domain points]\label{kl:lem:shared}
Let $\mathrm{Sh}=D_y\cap D_r$ and let $V'$ be $V$ without the $r$-side domain reads at $\mathrm{Sh}$. Then
$\Lambda(V)=\Lambda(V')+\spn\{\ev_x:x\in\mathrm{Sh}\}$.
\end{lemma}
\begin{proof}
$\Sigma=\spn\ev(\mathrm{Sh})\subseteq\spn\Phi$ and $\spn\Psi=\spn\Psi'+\Sigma$; apply the modular law
$X\cap(Y+\Sigma)=(X\cap Y)+\Sigma$ for $\Sigma\subseteq X$.
\end{proof}

So it suffices to treat views whose point sets $P_y,P_r$ are disjoint; input reads only add $\spn\{\ev_x:x\in G\}$ to
the leak, and we also assume $G=\emptyset$ (see the proof of Theorem~\ref{kl:thm:KL}). Let $D'=D'_y\sqcup D'_r\subseteq D$
be the read domain abscissae, $D'_y$ ($D'_r$) those read on the $y$-side ($r$-side), and $S:=A\sqcup D'$. A column $s\in A_y\cup D'_y$ carries only $y$-points, a column $s\in A_r\cup D'_r$ only
$r$-points, and a domain column $x_0$ carries only $\hat x_0$. Let $C_1=\RS[S,n]^\perp$ and $C_2=\RS[B,n']^\perp$
(generalised RS codes, MDS of distances $n+1$, $n'+1$).

\begin{lemma}[dependencies]\label{kl:lem:dual}
$\sum_{s,b}M(s,b)v(s)\otimes w(b)=0$ iff $M=M_1+M_2$ with every column $M_1(\cdot,b)\in C_1$ and every row
$M_2(s,\cdot)\in C_2$.
\end{lemma}
\begin{proof}
The left side vanishes iff $M\perp\RS[S,n]\otimes\RS[B,n']$, a code of dimension $nn'$. The sum
$C_1\otimes\F^B+\F^S\otimes C_2$ is contained in the dual and, as the two summands meet in $C_1\otimes C_2$, has
dimension $(|S|-n)|B|+|S|(|B|-n')-(|S|-n)(|B|-n')=|S||B|-nn'$.
\end{proof}

\begin{lemma}[leaks]\label{kl:lem:leak}
Let $V$ have no input reads ($G=\emptyset$) and $P_y\cap P_r=\emptyset$. Then $\ell\in\Lambda(V)$ iff
$\ell=\sum_{\pi\in P_y}M(\pi)\ev_\pi$ for a dependency $M$ supported on $P_y\cup P_r$; and for any decomposition
$M=M_1+M_2$ as in Lemma~\ref{kl:lem:dual},
$\ell\in U^*\otimes\spn\{w(b):M_1(\cdot,b)\ne0\}$.
\end{lemma}
\begin{proof}
The first claim is the definition of $\spn\Phi\cap\spn\Psi$. Since all $y$-points lie in the columns
$S_y=A_y\cup D'_y$, $\ell=\sum_{s\in S_y,b}M(s,b)v(s)\otimes w(b)$, and the $M_2$-part
$\sum_{s\in S_y}v(s)\otimes\sum_bM_2(s,b)w(b)$ vanishes because $M_2(s,\cdot)\in C_2$.
\end{proof}

\subsection{The incidence input}\label{kl:sec:cz}

Let $F\in K_0[X,Y]$ be irreducible with $a=\deg_XF\ge1$, $b'=\deg_YF\ge1$ (a \emph{mixed} curve); let $\klX_F$ be the
normalisation of the closure of $\{F=0\}$ in $\mathbb P^1\times\mathbb P^1$, with coordinate functions $x,y$; for a
place $\nu$ over a finite point put $o_x(\nu)=\klord_\nu(x-x(\nu))$, $o_y(\nu)=\klord_\nu(y-y(\nu))$, and
$\Gamma_{S_0\times B_0}(F):=\sum_{\nu:(x(\nu),y(\nu))\in S_0\times B_0}\min(o_x,o_y)$. We use the textbook facts
\cite{Har77}: (T1) $\deg x=b'$, $\deg y=a$; (T2) the genus of $\klX_F$ is at most the arithmetic genus $(a-1)(b'-1)$ of
a curve of bidegree $(a,b')$; (T3) a nonzero function has as many zeros as poles; (T4) if $(X-c)\nmid F$ then
$\sum_{x(\nu)=c}o_x(\nu)=b'$.

\begin{lemma}[{\cite[Theorem~2]{CZ13}}, in the form used]\label{kl:lem:CZ}
Let $\kappa$ be a finite field of characteristic $p$, $\klX/\kappa$ a smooth projective geometrically irreducible
curve of genus $g$, and $u,v\in\kappa(\klX)$ multiplicatively independent modulo $\kappa^\times$, with nonzero
differentials, whose zeros and poles form a set $S$ of $\kappa$-rational points; $\chi=|S|+2g-2$. Then
\[
\sum_{\nu\in\klX(\kappa)\setminus S}\min\{\nu(1-u),\nu(1-v)\}\le
\max\Bigl(3\sqrt[3]2\,(\deg u\deg v\,\chi)^{1/3},\ 12\,\frac{\deg u\deg v}{p}\Bigr).
\]
\end{lemma}
This is \cite[Theorem~2]{CZ13} verbatim (rationality of $S$ only fixes the meaning of $|S|$). Two points of the
published proof need a comment; neither affects the statement. (i) \cite[Prop.~2]{CZ13} assumes
$\kappa(\klX)=\kappa(u,v)$. If $\klX\to\klX'$ is the separable map (as $du\ne0$) of degree $\delta$ to the smooth model of
$\kappa(u,v)$, then, after enlarging $\kappa$ so that all relevant places are rational, the left side on $\klX$ is at
most $\delta$ times that on $\klX'$, $\deg_{\klX}=\delta\deg_{\klX'}$, and Riemann--Hurwitz gives
$\chi_{\klX}\ge\delta\chi_{\klX'}$; so the bound on $\klX'$ implies it on $\klX$. (ii) The proof of Theorem~2 splits at
$Q=p^3\chi/(\deg u\deg v)^2$ versus $4$; the hypothesis of \cite[Theorem~4]{CZ13} requires a split at $32$
($t=\sqrt[3]4$ for $Q>32$, $t=(Q/8)^{1/3}$ otherwise), which yields the same two bounds.

\begin{lemma}[coset-pair incidences]\label{kl:lem:pair}
Assume (N1); let $e,f\in K_0^\times$ and $F$ mixed of bidegree $(a,b')$ with $a,b'<p$. Then
$\Gamma_{e\mu_n\times f\mu_{n'}}(F)\le3\cdot2^{2/3}m^{1/3}(ab')^{2/3}+12ab'$, and, if (N2)--(N4) also hold,
$\Gamma_{A\times B}(F)\le L_1L_2\bigl(3\cdot2^{2/3}m^{1/3}(ab')^{2/3}+12ab'\bigr)$.
\end{lemma}
\begin{proof}
\emph{Transfer.} $F$ is defined over a finite field $\kappa_1$ and absolutely irreducible, so its normalised projective
model over $\kappa_1$ is smooth and geometrically irreducible with base change $\klX_F$. Let $\kappa\supseteq\kappa_1(e,f,\mu_m)$
be finite such that all places over $e\mu_n\times f\mu_{n'}$ and all zeros and poles of $x,y$ are $\kappa$-rational. At a
rational point the completed local ring $\kappa[[t]]$ becomes $K_0[[t]]$, so valuations, degrees and the genus are the
same over $\kappa$ and $K_0$, and the left side of Lemma~\ref{kl:lem:CZ} counts every place contributing to $\Gamma$.

\emph{Setup.} Put $u=(x/e)^n$, $v=(y/f)^{n'}$. At $\nu$ with $x(\nu)\in e\mu_n$, $T^n-e^n$ has a simple root at $x(\nu)$
($p\nmid n$), so $\nu(1-u)=o_x(\nu)$; likewise $\nu(1-v)=o_y(\nu)$; such $\nu$ are not zeros or poles of $x,y$. So
$\Gamma$ is at most the left side of Lemma~\ref{kl:lem:CZ}.

\emph{Independent case.} $du=nu\,dx/x\ne0$, since $dx=0$ would make $x$ a $p$-th power and $p\mid b'$; likewise $dv\ne0$.
$\deg u=nb'$, $\deg v=n'a$, $|S|\le2(a+b')$ and, by (T2), $\chi\le2ab'$. Lemma~\ref{kl:lem:CZ} and $m<p$ give the bound.

\emph{Dependent case.} Then $x^{nr}y^{n's}$ is constant for some $r,s\ne0$, so $F$ is a binomial $X^aY^{b'}-\lambda$ or
$X^a-\lambda Y^{b'}$ with $\gcd(a,b')=1$, normalised by $z\mapsto(\kappa z^{b'},\kappa'z^{\mp a})$. Grid points satisfy
$z^{b'n}=c_1$, $z^{an'}=c_2$, a set that is empty or a coset of $\mu_{\gcd(b'n,an')}$, and
$\gcd(b'n,an')\le\gcd(b',n')\gcd(n,a)\le ab'$; each has $o_x=1$. So $\Gamma\le ab'$.

\emph{Union.} By (N4) and Lemma~\ref{kl:lem:rows}, $A\times B$ is the disjoint union of $L_1L_2$ coset pairs
$e\mu_n\times f\mu_{n'}$.
\end{proof}

\subsection{The divisibility lemma with domain columns}\label{kl:sec:div}

Fix a view with disjoint sides, $S=A\sqcup D'$, $k_1=|S|-n$, $k_2=|B|-n'$, and write $\RS(S,k)$ for evaluations of
polynomials of degree $<k$. Rescaling coordinates by nonzero scalars maps $C_1\otimes\F^B+\F^S\otimes C_2$ onto
$\RS(S,k_1)\otimes\F^B+\F^S\otimes\RS(B,k_2)$ and preserves supports, so we treat $M=A_0+B_0$ with columns of $A_0$ in
$\RS(S,k_1)$ and rows of $B_0$ in $\RS(B,k_2)$. Let $0\ne E\in\F[X,Y]$ vanish on $\klsupp M$, with $e_X=\deg_XE$,
$e_Y=\deg_YE$, $t=\max(e_X,e_Y)\le n$; assume that for $x_0\in D'$, $M(x_0,\cdot)$ is supported in row $x_0^n$; put
$\Pi=\prod_{x_0\in D'}(X-x_0)$ and $\Pi_b=\prod_{x_0\in D'\cap F_b}(X-x_0)$.

\paragraph{Step 0 (normalisation on domain columns).} Choose an information set $I\supseteq D'$ of the MDS code
$\RS(S,k_1)$ (possible: $k_1=|A|+|D'|-n\ge|D'|$ as $|A|\ge n$), let $\mathrm{Ext}:\F^{D'}\to\RS(S,k_1)$ be the linear
map giving the codeword equal to $z$ on $D'$ and $0$ on $I\setminus D'$, and $Q_0(\cdot,b):=\mathrm{Ext}(-B_0(D',b))$.
The columns of $Q_0$ lie in $\RS(S,k_1)$; its rows are fixed linear combinations of rows of $B_0$, hence lie in
$\RS(B,k_2)$. Replacing $(A_0,B_0)$ by $(A_0-Q_0,B_0+Q_0)$ we get $B_0(x_0,\cdot)=0$ and $A_0(x_0,\cdot)=M(x_0,\cdot)$
for $x_0\in D'$.

\paragraph{Step 1 (the common polynomial).} Let $f_b$ ($\deg<k_1$) and $g_s$ ($\deg<k_2$) interpolate $A_0(\cdot,b)$ and
$B_0(s,\cdot)$, and let $P$ ($\deg_X<|S|$, $\deg_Y<|B|$) interpolate $EA_0=-EB_0$ on $S\times B$. Since
$e_X+k_1\le|S|$ and $e_Y+k_2\le|B|$, both sides of each identity below have degree $<|S|$ in $X$, resp.\ $<|B|$ in
$Y$, and agree on $S$, resp.\ $B$; so, as polynomials,
\begin{equation}\label{kl:eq:41}
P(X,b)=E(X,b)f_b(X)\ (b\in B),\qquad P(s,Y)=-E(s,Y)g_s(Y)\ (s\in S),
\end{equation}
so $\deg_XP<k_1+e_X$ and $\deg_YP<k_2+e_Y$. As $g_{x_0}=0$ for $x_0\in D'$, $P=\Pi\hat P$ with $\deg_X\hat P<|A|-n+e_X$; as
$f_b=(\Pi/\Pi_b)\hat f_b$ for some $\hat f_b\in\F[X]$ ($f_b$ vanishes on $D'\setminus F_b$),
\begin{equation}\label{kl:eq:42}
\Pi_b(X)\hat P(X,b)=E(X,b)\hat f_b(X)\qquad(b\in B).
\end{equation}

\begin{lemma}[divisibility]\label{kl:lem:div}
Suppose every mixed irreducible factor $F$ of $E$ over $K_0$, of bidegree $(a,b')$, satisfies
\begin{equation}\label{kl:eq:H}
\Gamma_{A\times B}(F)+\Delta_F<(n+1-2t)b'+(n'+1-2t)a,\qquad
\Delta_F:=\sum_{\nu:\,(x(\nu),y(\nu))\in\{\hat x_0:x_0\in D'\}}\min(o_x,o_y).
\end{equation}
Then $E\mid P$. Condition \eqref{kl:eq:H} holds for every mixed $F$ whenever
\begin{equation}\label{kl:eq:43}
L_1L_2\bigl(3\cdot2^{-1/3}m^{1/3}t^{1/3}+6t\bigr)<\frac n{L_2+1}-2t .
\end{equation}
\end{lemma}
\begin{proof}
Work over $K_0$. \emph{(i) Mixed factors.} Let $F\mid E$ be mixed with $F\nmid P$; then $\hat P(x,y)\ne0$ on $\klX_F$.

\emph{Poles.} By (T1), (T3), $\#\text{zeros}=\#\text{poles}\le(|A|-n+e_X-1)b'+(|B|-n'+e_Y-1)a$.

\emph{Column zeros ($s\in A$).} By \eqref{kl:eq:41}, $\Pi(s)\hat P(s,Y)=-E(s,Y)g_s(Y)$ with $\Pi(s)\ne0$, and
$F(s,Y)\not\equiv0$; so $\hat P(s,Y)=F(s,Y)h(Y)$, $\hat P=Fh+(X-s)K$, and $\klord_\nu\hat P\ge o_x(\nu)$ when $x(\nu)=s$,
$y(\nu)$ finite.

\emph{Row zeros ($b\in B$).} If $E(X,b)\equiv0$ then $\hat P(X,b)\equiv0$ by \eqref{kl:eq:42} and
$\klord_\nu\hat P\ge o_y(\nu)$. Otherwise $F(X,b)\mid\Pi_b\hat P(X,b)$. Let $g_b=\gcd(F(X,b),\Pi_b)=\prod_{\xi\in G_b}(X-\xi)$,
$G_b=\{\xi\in D'\cap F_b:F(\xi,b)=0\}$ ($\Pi_b$ is squarefree). Writing $F(X,b)=g_bF'$, $\Pi_b=g_b\Pi'$, the polynomials
$F'$ and $\Pi'$ are coprime, so $F'\mid\hat P(X,b)$, i.e.\ $g_b\hat P(X,b)=F(X,b)h_b$; hence on $\klX_F$,
$g_b(x)\hat P(x,y)=(y-b)K_b(x,y)$, and the polynomials $\hat P$, $K_b$ are regular at every place over a finite point.
So, at $\nu$ with $y(\nu)=b$, $x(\nu)$ finite: $\klord_\nu\hat P\ge o_y(\nu)$ if $x(\nu)\notin G_b$, as $g_b(x)$ is
then a unit at $\nu$; if $x(\nu)=\xi\in G_b$, then $\klord_\nu g_b(x)=o_x(\nu)$ ($\xi$ is a simple root of $g_b$, and
$x-\xi'$ is a unit at $\nu$ for the other roots $\xi'$), so $\klord_\nu\hat P\ge o_y-o_x$, and $\klord_\nu\hat P\ge0$
by regularity; together, $\klord_\nu\hat P\ge\max(0,o_y-o_x)=o_y-\min(o_x,o_y)$. The places with $x(\nu)\in G_b$ lie
over the view's domain points on $F$, and $G_b\cap A=\emptyset$.

\emph{Counting.} Over $A\times B$ both bounds apply: $\klord\ge\max(o_x,o_y)=o_x+o_y-\min(o_x,o_y)$. By (T4),
$\sum_{x(\nu)\in A}o_x=|A|b'$ and $\sum_{y(\nu)\in B}o_y=|B|a$; places with $x(\nu)\in A$, $y(\nu)=\infty$ number at most
$a$, each with $o_x\le b'$, and symmetrically. Hence
$\#\text{zeros}\ge|A|b'+|B|a-2ab'-\Gamma_{A\times B}(F)-\Delta_F$. With the pole bound, $e_X,e_Y\le t$ and
$2ab'\le t(a+b')$:
\begin{equation}\label{kl:eq:44}
\Gamma_{A\times B}(F)+\Delta_F\ge(n+1-2t)b'+(n'+1-2t)a,
\end{equation}
contradicting \eqref{kl:eq:H}. So $F\mid P$. Cancelling $F$ preserves \eqref{kl:eq:41} (with the same $f_b,g_s$: the
mixed irreducible $F$ has no factor $X-s$ or $Y-b$, so $F(s,Y),F(X,b)\not\equiv0$ for $s\in S$, $b\in B$), the
vanishing on domain columns ($F$ has no factor $X-x_0$, $x_0\in D'\subseteq S$, so it is coprime to $\Pi$; hence $F\mid\hat P$, $P/F=\Pi\cdot\hat P/F$, and
\eqref{kl:eq:42} holds for $\hat P/F$, $E/F$ with the same $\hat f_b$) and the degree bounds, and does not increase $t$;
the mixed factors of $E/F$ divide $E$, so satisfy \eqref{kl:eq:H}, and (i) and (ii) use nothing else about $P$ and $E$.
Repeat: each round removes one irreducible factor, so $E$ need not be squarefree.

\emph{(ii) Pure factors.} Now $E=c\,E_1(X)E_2(Y)$, and $P$ denotes the polynomial after the mixed factors have been
cancelled. For $s\in S$, $E_2\mid P(s,Y)$ by \eqref{kl:eq:41}; the remainder of $P$ modulo $E_2$ has $X$-degree
$<k_1+e_X\le|S|$ and vanishes on $S$, so $E_2\mid P$; symmetrically $E_1\mid P$, and $E_1E_2\mid P$. Hence $E\mid P$
over $K_0$, and over $\F$.

\emph{\eqref{kl:eq:43} implies \eqref{kl:eq:H}.} \eqref{kl:eq:43} forces $2t(L_2+1)<n$, so
Claim~\ref{kl:clm:delta}(c$'$) applies: the right side of \eqref{kl:eq:H} minus $\Delta_F$ is at least
$(\frac n{L_2+1}-2t)(a+b')$. By Lemma~\ref{kl:lem:pair} ($a,b'\le t\le n<p$), with
$(ab')^{2/3}\le\frac{a+b'}2t^{1/3}$ and $ab'\le\frac{a+b'}2t$,
$\Gamma_{A\times B}(F)\le L_1L_2(a+b')(3\cdot2^{-1/3}m^{1/3}t^{1/3}+6t)<(\frac n{L_2+1}-2t)(a+b')$.
\end{proof}

\begin{claim}[the domain loss is affordable]\label{kl:clm:delta}
For mixed $F$ of bidegree $(a,b')$ with $1\le a,b'\le t<n<n'$:
(a) $\Delta_F\le a+nb'$; (b) $\Delta_F\le L_2n'a$;
(c) if $t\le n/(16L_2)$ then $(n+1-2t)b'+(n'+1-2t)a-\Delta_F\ge\frac n{8L_2}(a+b')$;
(c$'$) if $2t(L_2+1)<n$ then $(n+1-2t)b'+(n'+1-2t)a-\Delta_F\ge\bigl(\frac n{L_2+1}-2t\bigr)(a+b')$.
\end{claim}
\begin{proof}
(a) At a view domain place over $(\xi,\xi^n)$, $\klord(x^n-\xi^n)=o_x$, so $\klord(y-x^n)\ge\min(o_x,o_y)$; and
$y-x^n\ne0$ on $\klX_F$ (as $Y-X^n$ is irreducible of $X$-degree $n>a$) has at most $a+nb'$ zeros.
(b) Each of the $|B|=L_2n'$ rows contributes $\le\deg y=a$.
(c) Case $b'\ge2L_2n'a/n$: by (b) the left side is $\ge(\frac n2-2t)b'\ge\frac38nb'\ge\frac14n(a+b')$. Case
$b'<2L_2n'a/n$: by (a) it is $\ge(n'-2t)a-2tb'\ge\frac58n'a\ge\frac{5n}{32L_2}(a+b')$.
(c$'$) Let $L=L_2$, $\varepsilon=t/n$. \emph{Case $a+nb'\le Ln'a$:} by (a) the left side is
$\ge(n'-2t)a-2tb'\ge an'(1-2\varepsilon(1+L))$ (using $b'\le Ln'a/n$ and $n<n'$), while $a+b'\le a(n+Ln')/n$; since
$nn'/(n+Ln')\ge n/(1+L)$, the ratio is $\ge\frac n{L+1}-2t$. \emph{Case $a+nb'>Ln'a$:} by (b) the left side is
$\ge(n-2t)b'-\kappa a$ with $\kappa=(L-1)n'+2t-1\ge0$; with $\gamma=\frac n{L+1}-2t$ and $b'>(Ln'-1)a/n$ it suffices
that $(n-2t-\gamma)(Ln'-1)\ge(\kappa+\gamma)n$, which expands to $\gamma\le\frac{nn'}{Ln'+n-1}-2t$, true because
$\frac{nn'}{Ln'+n-1}\ge\frac{nn'}{(L+1)n'}$.
\end{proof}

\begin{remark}
Without Step~0 and the division by $\Pi$ the count loses $\Gamma_{D'\times B}(F)$, which is only bounded by $|D'|b'$
in general; the refinement charges only the view's own domain points on
$F$, which lie on $Y=X^n$. Putting all of $D$ among the columns instead would make $L_1$ grow with $n'$.
\end{remark}

\subsection{Row localisation and the proof of the locator lemma}\label{kl:sec:proof}

\begin{theorem}[row localisation of all dependencies]\label{kl:thm:rows}
Let $V$ have disjoint sides and no input reads, and let $1\le t\le n$ satisfy \eqref{kl:eq:43} and $(t+1)^2>|V|$.
Then every dependency $M$ supported on $V$'s points has a decomposition $M=M_1+M_2$ with
$\{b:M_1(\cdot,b)\ne0\}\subseteq H_t(V)$; consequently $\Lambda(V)\subseteq\spn\{\ev_x:x\in F_b,\ b\in H_t(V)\}$.
\end{theorem}
\begin{proof}
As $(t+1)^2>|\mathrm{Pts}(V)|$ there is $0\ne E$ of bidegree $\le(t,t)$ vanishing on $\mathrm{Pts}(V)$; it serves every dependency. The
domain-column hypothesis holds since the only point of $V$ in a column $x_0\in D'$ is $\hat x_0$. Run Steps~0--1;
Lemma~\ref{kl:lem:div} gives $P=EQ$ with $\deg_XQ<k_1$, $\deg_YQ<k_2$. Put $M_1=A_0-Q$, $M_2=B_0+Q$. If
$E(X,b)\not\equiv0$ then $f_b=Q(X,b)$ by \eqref{kl:eq:41}, so $M_1(\cdot,b)=0$; likewise $M_2(s,\cdot)=0$ unless
$(X-s)\mid E$, so $J=\{s:M_2(s,\cdot)\ne0\}$ has $|J|\le t$. A nonzero column $M_1(\cdot,b)$ of the MDS code
$\RS(S,k_1)$ has weight $\ge n+1$, and $M=M_1$ off $J$, so row $b$ carries $\ge n+1-t$ points of $V$. Undo the scaling and
apply Lemmas~\ref{kl:lem:leak} and~\ref{kl:lem:cheap}.
\end{proof}

\begin{proof}[Proof of Theorem~\ref{kl:thm:KL}]
(a) $D_V$, $H_\tau(V)$ and $[|V|\ge(\tau+1)^2]$ are monotone and depend on positions only; computing $W$ needs row counts,
the powers $x^n$, the fibres $F_b=\{x\in D:x^n=b\}$ and $d_0\mu_m$. (b) Let $V_0\subseteq V$ be $V$ without its input
reads. Then $\Lambda(V)=\Lambda(V_0)+\spn\{\ev_x:x\in G\}$ with $G\subseteq D_V\subseteq W(V)$, and $W(V_0)\subseteq W(V)$
by (a); so it suffices that $W(V_0)$ locates $V_0$, and we may assume $G=\emptyset$. Likewise, by
Lemma~\ref{kl:lem:shared} ($\mathrm{Sh}\subseteq D_V$, $V'\subseteq V$) we may assume disjoint sides. If
$|V|<(\tau+1)^2$ and $\tau\ge1$, then $t=\tau$ satisfies
Theorem~\ref{kl:thm:rows}: $\tau\le cn\le n$, and by $m\le Kn^2$,
$L_1L_2(3\cdot2^{-1/3}m^{1/3}\tau^{1/3}+6\tau)\le nL_1L_2(3\cdot2^{-1/3}K^{1/3}c^{1/3}+6c)<n(\frac1{L_2+1}-2c)
\le\frac n{L_2+1}-2\tau$ (if $\tau=0$ then $V=\emptyset$). If $|V|\ge(\tau+1)^2$, $\{\ev_x:x\in d_0\mu_m\}$ is a basis of
$C^*$. (c) $|D_V|\le|V|$; the rows of $H_\tau(V)$ are disjoint and carry $>n(1-c)$ points each, so they contribute
$\le|V|/(1-c)$ (Lemma~\ref{kl:lem:rows}); and $|V|\ge(\tau+1)^2>c^2n^2$ gives $m\le Kn^2<(K/c^2)|V|$.
\end{proof}

\begin{remark}
(i) The one-polynomial argument of Theorem~\ref{kl:thm:rows} is an alternative to the closure route
\cite{DLV24,KP25,BV26} to a monotone locator; it gives the threshold $n+1-t$ directly. (ii) Nothing is claimed for nodes
with $m_1>m_2$ (Claim~\ref{kl:clm:delta} uses $n<n'$); none occurs in the recursion. (iii) If $(\tau+1)^2>n+n'$, a set of
$n+n'$ points with one point per row and column, both sides and no domain reads, has no heavy row and hence no leak;
this needs $n\gtrsim3/c^2$ and is superseded by Theorem~\ref{kl:thm:mob}.
\end{remark}

\subsection{A small-view locator}\label{kl:sec:practical}

When $m_1<1/c$ (for $c=c_*$: whenever $m_1<10^{13}$) we have $\tau=0$, so the non-trivial part of
Theorem~\ref{kl:thm:KL} is empty; the following ratio-one locator, the \emph{small-view locator}, covers small views at
every node. It is not used in the proof of the main theorem. \emph{Credit.} The case
$e=0$ of Theorem~\ref{kl:thm:CL} is elementary, classical Cayley--Bacharach residuation \cite{EGH96}, not new; the
extension to $e\le E^*$ is elementary (residuation and Lemma~\ref{kl:lem:curve}) and may be classical (see
\cite{GMRag92,GVT15} on finite sets of points in $\mathbb P^1\times\mathbb P^1$). Theorem~\ref{kl:thm:mob} is a direct
application of \cite[Theorem~2]{CZ13}.

Write $\mathcal P_{a,b}=\F[X,Z]_{\le(a,b)}$. A finite $\Gamma\subset\F^2$ is \emph{independent} if its evaluation
functionals are independent on $\mathcal P_{a,b}$; a \emph{circuit} is a minimal dependent set. A \emph{M\"obius graph}
is the zero set of an irreducible form $cXZ+dX+eZ+f$ with $cf\ne de$ (affine maps, $c=0$, included). $\Gamma$ is
\emph{Cayley--Bacharach} (CB) for $(a,b)$ if every $f\in\mathcal P_{a,b}$ vanishing on $\Gamma\setminus\{P\}$ vanishes at $P$; circuits are
CB.

\begin{lemma}[hook separation]\label{kl:lem:H}
If $\Gamma$ is CB for $(a,b)$ and $P\in\Gamma$ has $\rho$ other points of $\Gamma$ on its row and $\sigma$ on its
column, then $\rho\ge a+1$, or $\sigma\ge b+1$, or $|\Gamma|\ge a+b+2$.
\end{lemma}
\begin{proof}
Otherwise split $T=\Gamma\setminus(\text{row}\cup\text{col})$ into parts of sizes $\le a-\rho$ and $\le b-\sigma$; the
product of $\le a$ column factors (other row points and the first part) and $\le b$ row factors (other column points and
the second part) lies in $\mathcal P_{a,b}$, vanishes on $\Gamma\setminus P$ and not at $P$.
\end{proof}

\begin{lemma}[residuation]\label{kl:lem:RES}
If $\Gamma$ is CB for $(a,b)$ and $0\ne G\in\mathcal P_{k,l}$, $k\le a$, $l\le b$, then $\{P\in\Gamma:G(P)\ne0\}$ is
empty or CB for $(a-k,b-l)$.
\end{lemma}
\begin{proof}
If $h$ vanishes on the residual set minus $P$ then $Gh$ vanishes on $\Gamma\setminus P$.
\end{proof}

\begin{lemma}[curve]\label{kl:lem:curve}
If $\Gamma$ is CB for $(a,b)$, rows hold $\le a-k+1$ and columns $\le b-l+1$ of its points, and
$|\Gamma|\le a+b+kl+1$, then some $0\ne G\in\mathcal P_{k,l}$ vanishes on $\Gamma$.
\end{lemma}
\begin{proof}
Take $G$ through all of $\Gamma$ if $|\Gamma|<(k+1)(l+1)$, else through $(k+1)(l+1)-1$ of its points; a nonempty residual
set is CB for $(a-k,b-l)$ and by Lemma~\ref{kl:lem:H} has $\ge a+b-k-l+2$ points, so $|\Gamma|\ge a+b+kl+2$.
\end{proof}

\begin{lemma}[graphs]\label{kl:lem:graph}
Let $A_1,B_1\in\F[Z]$ be coprime, $A_1\ne0$, $l=\max(\deg A_1,\deg B_1)\le b+1$. A set $S$ of points of $\{A_1(z)x+B_1(z)=0\}$ (they have
distinct ordinates) has rank $\min(|S|,al+b+1)$ on $\mathcal P_{a,b}$.
\end{lemma}
\begin{proof}
$\Phi(f)=A_1^af(-B_1/A_1,Z)$ maps $\mathcal P_{a,b}$ onto $\F[Z]_{\le al+b}$: its kernel is $(A_1X+B_1)\mathcal P_{a-1,b-l}$
by Gauss. At points of the curve $A_1(z)\ne0$ and $f(P)=\Phi(f)(z_P)/A_1(z_P)^a$.
\end{proof}

\begin{theorem}[small circuits]\label{kl:thm:CL}
Let $a,b\ge1$, $E^*=\max(\min(a-1,\lfloor\frac{b-1}2\rfloor),\min(b-1,\lfloor\frac{a-1}2\rfloor))$. Every circuit $C$ for
$\mathcal P_{a,b}$ with $|C|\le a+b+2+E^*$ is $a+2$ points of a row, $b+2$ points of a column, a hook ($a+1$ points of a
row and $b+1$ of a column, corner excluded), or $a+b+2$ points of a M\"obius graph; each type is a circuit.
\end{theorem}
\begin{proof}
By symmetry let $e=|C|-a-b-2\le\min(a-1,\lfloor\frac{b-1}2\rfloor)$. Rows (columns) with $a+2$ ($b+2$) points are
dependent, so otherwise rows hold $\le a+1$, columns $\le b+1$, and Lemma~\ref{kl:lem:H} gives $e\ge0$.
\emph{A row $R_0$ with $a+1$ points:} $C\setminus R_0$ is CB for $(a,b-1)$ with $b+1+e$ points; Lemma~\ref{kl:lem:H} puts
each of its points on a column holding $b+1$ of them, so one column ($j(b+1)<2(b+1)$), $e=0$, and $C$ is a hook.
\emph{Rows $\le a$, a column with $b+1$:} its complement is CB for $(a-1,b)$ and has $\ge a+b+1$ points, so $e\ge b$,
impossible. \emph{Rows $\le a$, columns $\le b$:} removing a row (column) with $t$ points shows $t\le e+1$;
Lemma~\ref{kl:lem:curve} with $(k,l)=(1,e+1)$ gives $0\ne G=A(Z)X+B(Z)$ through $C$, $\deg A,\deg B\le e+1$. If $A=0$, $C$
lies on $\le e+1$ rows and residuation leaves one row CB for $(a,b-j+1)$ with $\ge a+b-j+3>e+1$ points, a
contradiction. If $A\ne0$, write $G=g\cdot(A_1X+B_1)$ with $\gcd(A_1,B_1)=1$; points off $A_1X+B_1=0$ lie on rows of $g$
and residuation forces $>e+1$ points on such a row, so $C$ lies on the graph; $l_1\ge1$ (else a column), and by
Lemma~\ref{kl:lem:graph} $|C|=al_1+b+2\le2a+b+1$, so $l_1=1$. Conversely, rows, columns and M\"obius sets are circuits by
Lemma~\ref{kl:lem:graph}; a hook is dependent (its row-plus-column restriction has dimension $a+b+1$) and minimal (after
deleting a point, each remaining point is separated by $\le a$ column and $\le b$ row factors).
\end{proof}

For a view $V$ let $R^*(V)$ be the rows holding $\ge n$ points of $\mathrm{Pts}(V)$, $r(V)=|R^*(V)|$, $\mathrm{Pts}''(V)$ the points off them,
and
\[
W_1(V):=D_V\cup\bigcup_{b\in R^*(V)}F_b .
\]
$W_1$ is monotone, depends on positions only, and $|W_1(V)|\le|\mathrm{Pts}(V)|\le|V|$ (each counted row has $\ge n=|F_b|$ points).

\begin{lemma}[locator criterion]\label{kl:lem:LC}
Let $R^*\subseteq B$, $r=|R^*|\le n'-1$, $W=D_V\cup\bigcup_{b\in R^*}F_b$, and split the points of $\mathrm{Pts}(V)$ off $R^*$ into
$y$-lift points $Y''$, $r$-lift points $R''$ and domain points $\mathrm{Dm}''$. If every circuit of these points for
$\mathcal P_{n-1,n'-1-r}$ lies in $Y''\cup\mathrm{Dm}''$ or in $R''\cup\mathrm{Dm}''$, then $W$ is a locator for $V$.
\end{lemma}
\begin{proof}
Let $f\in\mathcal P_{n-1,n'-1}$ vanish at $\hat x$ for all $x\in W$; we show $\ell(f)=0$ for $\ell\in\Lambda(V)$. For
$b\in R^*$, $f(X,b)$ has $n$ roots and degree $<n$, so $f=\Pi g$ with $\Pi=\prod_{R^*}(Z-b)$ and
$g\in\mathcal P':=\mathcal P_{n-1,n'-1-r}$. Some $h\in\mathcal P'$ equals $0$ on $Y''\cup\mathrm{Dm}''$ and $g$ on $R''$:
the data are orthogonal to every circuit (trivially for circuits in $Y''\cup\mathrm{Dm}''$; for circuits in
$R''\cup\mathrm{Dm}''$ because $g=f/\Pi$ vanishes on $\mathrm{Dm}''\subseteq\hat D_V$), and circuits span dependencies.
Then $f_y=\Pi h$ vanishes at every $y$-point, $f_r=f-f_y$ at every $r$-point, and $\ell(f)=\ell(f_y)+\ell(f_r)=0$.
\end{proof}

\begin{theorem}[small views]\label{kl:thm:small}
On every node: (i) if $V$ touches $\le n'$ rows, or $|V|\le n+n'-1$, then $W_1(V)$ is a locator; (ii) if $n\ge2$,
$r=r(V)\le n'-2$,
\[
|\mathrm{Pts}''(V)|\le n+n'-r+E^*(n-1,n'-1-r)\tag{LP}
\]
and no M\"obius graph contains $n+n'-r$ points of $\mathrm{Pts}''(V)$, then $W_1(V)$ is a locator.
\end{theorem}
\begin{proof}
(i) Let $V$ touch the rows $b_1,\dots,b_k$, $k\le n'$. The $w(b_j)$ are linearly independent on $U'$ (distinct points,
$k\le n'=\dim U'$), so $U^*\otimes\spn\{w(b_j)\}_j=\bigoplus_jU^*\otimes w(b_j)$, and the functional of each point of
row $b_j$ lies in the summand $U^*\otimes w(b_j)$ (Lemma~\ref{kl:lem:tensor}). Hence
$\spn\Phi=\bigoplus_j\Phi_j\otimes w(b_j)$ and $\spn\Psi=\bigoplus_j\Psi_j\otimes w(b_j)$, where $\Phi_j$ ($\Psi_j$) is
spanned by the $v(s)$ of the $y$-side ($r$-side) points of row $b_j$, and the leak splits row by row:
$\Lambda(V)=\bigoplus_j(\Phi_j\cap\Psi_j)\otimes w(b_j)$ (plus the $\ev_x$, $x\in G\subseteq D_V$, of input reads). A row
with $\ge n$ points lies in $R^*(V)$, and its part lies in $U^*\otimes w(b_j)=\spn\{\ev_x:x\in F_{b_j}\}$
(Lemma~\ref{kl:lem:cheap}). A row with $\le n-1$ points has distinct abscissae, whose $v(s)$ are linearly independent in
$U^*$, so $\Phi_j\cap\Psi_j$ is spanned by the $v(s)$ of the points read on both sides; as $A_y\cap A_r=\emptyset$,
these are domain points $\hat x$ with $y(x)$ and $r(x)$ both read, so $x\in D_V$ and $v(x)\otimes w(b_j)=\ev_x$. Thus
$W_1(V)$ is a locator. Now let $|V|\le n+n'-1$. If $n=1$, $V$ touches at most $|\mathrm{Pts}(V)|\le|V|\le n'$ rows and
the first case applies. If $n\ge2$, then $n'\ge3$, so $rn\le|V|\le n+n'-1$ gives $r\le n'-1$, and
$|\mathrm{Pts}''(V)|\le(n-1)+(n'-1-r)+1$, so by Lemma~\ref{kl:lem:H} every circuit is a column
circuit, which is one-sided (a domain column carries one point); apply Lemma~\ref{kl:lem:LC}. (ii) By
Theorem~\ref{kl:thm:CL} the circuits of $\mathrm{Pts}''(V)$ are column circuits (one-sided) or M\"obius circuits (excluded); light
rows hold $\le n-1$ points, so there are no row circuits or hooks; apply Lemma~\ref{kl:lem:LC}.
\end{proof}

\begin{theorem}[M\"obius incidences]\label{kl:thm:mob}
On a node with $n\ge2$, every M\"obius graph contains at most $n+1+L_1L_2M(m)$ points of
$(A\times B)\cup\{\hat x:x\in D\}$, where $M(m)=\max(3\cdot2^{2/3}m^{1/3},12m/p)$.
\end{theorem}
\begin{proof}
\emph{Domain points:} $x^n(cx+e)+dx+f=0$ is a nonzero equation of degree $\le n+1$. \emph{Monomial maps
$\lambda x^{\pm1}$:} for each row coset the admissible $a$ form one $\mu_{n'}$-coset, which meets each of the $L_1$
column cosets (cosets of $\mu_n$) in $\le1$ point, as $\gcd(n,n')=1$: $\le L_1L_2\le L_1L_2M(m)$ in all.
\emph{Other maps:} for a column coset $e\mu_n$ and a row coset $\beta\mu_{n'}$, apply Lemma~\ref{kl:lem:CZ} on
$\mathbb P^1_s$ to $u=s^n$, $v=\psi'(s)^{n'}$, $\psi'(s)=\beta^{-1}\psi(es)$, where $\psi$ is the (non-monomial) M\"obius
map of the graph: $\deg u=n$, $\deg v=n'$, $\chi\le2$, $du,dv\ne0$, independence because $\psi'$ is not monomial, and $[\kappa(s):\kappa(u,v)]$ divides
$\gcd(n,n')=1$; this gives $\le M(m)$ points per pair.
\end{proof}

Put $C_L:=L_1L_2\cdot3\cdot2^{2/3}$ and $m_0:=(C_L+1)^6$. On a node with $n\ge2$ and $m\ge m_0$ we have
$M(m)=3\cdot2^{2/3}m^{1/3}$ (as $12m/p<12\le3\cdot2^{2/3}m^{1/3}$), so Theorem~\ref{kl:thm:mob} excludes M\"obius graphs
with $n+n'-r$ available points for every $r\le r_{\max}:=n'-2-\lceil C_L\,m^{1/3}\rceil$; and as
$n'>\sqrt m\ge(C_L+1)m^{1/3}$, $r_{\max}>m^{1/3}-3>3$. For $(L_1,L_2)=(14,8)$, $C_L=533.366\ldots$ and
$m_0\approx2.33\cdot10^{16}$, so $m\ge m_0$ holds once $m_1\ge1.53\cdot10^8$ (as $m>m_1^2$). Hence on nodes with $m\ge m_0$ and
$n'<2n$ (so $n>\sqrt{m/2}\ge2$), $W_1$ locates every view with at most $N_0:=n+\lfloor3n'/2\rfloor-1$ reads, and every
view satisfying (LP) with $r(V)\le r_{\max}$.

\begin{proposition}[combined locator]\label{kl:prop:natural}
Let $N_\nu=n+n'-1$, or $N_\nu=N_0$ if $m\ge m_0$ and $n'<2n$; let $c$ be admissible and $\tau=\lfloor cn\rfloor$; put
\[
W^\natural(V):=D_V\cup\bigcup_{b\in R^*(V)\cup H_\tau(V)}F_b\ \cup\
\bigl(d_0\mu_m\text{ if }|V|\ge\max((\tau+1)^2,N_\nu+1)\bigr).
\]
Then $W^\natural$ is monotone, position-only, $\poly(|\F|)$-time, a locator for every view, of size
$\le(1+\frac1{1-c})|V|$ when $|V|\le\max((\tau+1)^2-1,N_\nu)$ and $\le a(c)|V|$ always. Moreover
$W^\natural\subseteq W\cup W_1$, and $W\cup W_1$ is a monotone locator of size $\le(a(c)+1)|V|$.
\end{proposition}
\begin{proof}
Monotonicity is that of each piece. If $|V|<(\tau+1)^2$, $W^\natural\supseteq W(V)$ (Theorem~\ref{kl:thm:KL}); if
$|V|\le N_\nu$, $W^\natural\supseteq W_1(V)$ (Theorem~\ref{kl:thm:small} and the paragraph above); otherwise
$W^\natural\supseteq d_0\mu_m$. Each counted row carries $\ge\min(n,n+1-\tau)>n(1-c)$ points; the coset appears only
when $|V|\ge(\tau+1)^2$, where $m\le(K/c^2)|V|$. Finally, the fibres $F_b$ with $b\in R^*(V)$ lie in $W_1(V)$, those
with $b\in H_\tau(V)$ lie in $W(V)$, and the coset term of $W^\natural$ appears only when $|V|\ge(\tau+1)^2$, when
$d_0\mu_m\subseteq W(V)$; so
$W^\natural\subseteq W\cup W_1$. The map $W_1$ is defined on all views and is monotone and position-only, as is $W$; so
$W\cup W_1$ is monotone and position-only, it is a locator because it contains $W$ (Theorem~\ref{kl:thm:KL}), and
$|W\cup W_1|\le|W|+|W_1|\le(a(c)+1)|V|$.
\end{proof}

\endgroup

\begingroup
\providecommand{\Sim}{\mathsf{Sim}}
\providecommand{\View}{\mathrm{View}}
\providecommand{\UNSAT}{\mathrm{UNSAT}}
\providecommand{\SAT}{\mathrm{SAT}}
\providecommand{\size}{\mathrm{size}}
\providecommand{\poly}{\mathrm{poly}}
\providecommand{\rdist}{\mathrm{rdist}}
\providecommand{\cit}{c_{\mathrm{it}}}
\providecommand{\Mem}{\mathsf{Mem}}
\providecommand{\Pcan}{\mathcal P_{\mathrm{can}}}
\providecommand{\Lcan}{L'_{\mathrm{can}}}

\section{Zero-knowledge-preserving gap amplification}\label{sec:thmD}\label{thmD:sec}

This section proves Theorem~\ref{thmD:thm}, the zero-knowledge-preserving gap amplification that
Section~\ref{sec:main} uses (as Theorem~\ref{main:thm:D}) to reduce the polylogarithmic query count of the shell PCP to
a constant. Section~\ref{thmD:sec:defs} gives the definitions it needs, Section~\ref{thmD:sec:inst} fixes Dinur's
instantiation with the canonical folding of the long-code tester, and Sections~\ref{thmD:sec:stmt}
and~\ref{thmD:sec:proof} state and prove the theorem. Section~\ref{thmD:sec:cor} packages it as a black-box
corollary, which Section~\ref{sec:main} does not use.

\paragraph{Credit.} Gap amplification and all of its components (preprocessing, powering, composition with an
assignment tester, the long-code assignment tester and its test) are as in Dinur \cite{Din07}, whose composition
uses the modular notion of composition of \cite{BGHSV06,DR06} and whose long-code test is a reworking of
H{\aa}stad's \cite{Has01}; expander replacement
is due to Papadimitriou and Yannakakis \cite{PY91}. The long code and folding go back to Bellare, Goldreich and
Sudan \cite{BGS98} and H{\aa}stad \cite{Has01}. The local-computability lemma is \cite[Def.~3.1,
Lemma~3.2]{GOS25}. The observation that Dinur's amplification is locally computable and therefore preserves
zero knowledge is sketched in \cite[App.~A, Lemma~A.1]{ZOGS26}. The Tseitin encoding is
\cite{Tse68}. The straight-line convention adapts \cite[Def.~3.5]{VW26}. Theorem~\ref{thmD:thm} is a rigorous
write-up of Dinur's construction in which we correct six points of the printed presentation of \cite{Din07}:
the domination step in the soundness proofs of \cite[Lemmas~1.8 and~7.2]{Din07} (Remark~\ref{thmD:rem:dom}), the
folding of the long-code tester \cite[\S7]{Din07} (Definition~\ref{thmD:def:tester}, via the canonical folding of
Lemma~\ref{thmD:lem:F}), the case $\delta>\frac12$ in the proof of \cite[Prop.~7.3]{Din07} (in the proof of
Lemma~\ref{thmD:lem:P}), and three steps in the proof of \cite[Lemma~6.1]{Din07}, from which \cite[Lemma~1.6]{Din07}
follows: comparisons with opinions outside $\Gamma(w)$, an empty edge set $F$ when $|E|<t$, and the hypotheses of
\cite[Lemma~6.4]{Din07} at the middle indices, whose repair restricts the powering parameter to $t=2r^2$ with $r$ large
(Remark~\ref{thmD:rem:L61}). What this section adds
is the complete statement with explicit hypotheses, a front end that keeps every round on Dinur's alphabet,
the explicit honest labels with their locality (Lemma~\ref{thmD:lem:H}), the information property of
Theorem~\ref{thmD:thm}\ref{thmD:it:f}, and the cache-and-cap form of the
local-computability lemma (Lemma~\ref{thmD:lem:Z}). No new technique is claimed. Numbering of \cite{Din07} follows
its revision of 13 February 2007.

\subsection{Definitions}\label{thmD:sec:defs}

\begin{definition}[PCPs with coin multisets]\label{thmD:def:pcp}
A \emph{promise relation} is a pair $(R^{\rm YES},L^{\rm NO})$ with $L^{\rm NO}\cap L(R^{\rm YES})=\emptyset$.
A \emph{non-adaptive PCP with coin multisets} $(P_0,V_0)$ consists of a randomised prover $P_0$ that on
$(x,w)\in R^{\rm YES}$ outputs $\pi_0\in(\{0,1\}^a)^{\ell_0}$ ($a,\ell_0$ depending on $x$ only) and, for each
$x$, a finite multiset $U=U(x)$ of \emph{tests} $u=(Q_u,D_u)$ with $Q_u\in[\ell_0]^{q_u}$, $q_u\le q$, and
$D_u$ a fan-in-two Boolean circuit (gates $\wedge,\vee,\oplus$ with two inputs, $\neg$ with one) on the
$aq_u$ bits of $\pi_0|_{Q_u}$, with at most $s$ gates and whose output wire is a gate. The list of tests is
computable in time $\poly(|x|)$. $V_0$ draws $u\in U$ uniformly (with multiplicity) and accepts iff
$D_u(\pi_0|_{Q_u})=1$. The PCP has \emph{perfect completeness} if every $\pi_0$ in the support of
$P_0(x,w)$ passes every test, and \emph{gap} $\gamma\in(0,1]$ if for every $x\in L^{\rm NO}$ and every
$\pi_0\in(\{0,1\}^a)^{\ell_0}$ at least a $\gamma$ fraction of $U$ rejects.
\end{definition}

\noindent A circuit whose output is an input wire is normalised by one gate $z=y\wedge y$ ($s$ grows by at
most one). A gap may be smaller than $1/|U|$: a PCP with $|U|=1$ and gap $1$ also has gap $2^{-|x|}$. But for
$x\in L^{\rm NO}$ every proof is rejected by a positive fraction of $U$, hence by at least one test, so
$\hat\gamma:=\max\{\gamma,1/|U|\}\in[\gamma,1]$ is also a gap. Let $B(x)\subseteq[\ell_0]\times[a]$ be the set of proof
bits that are input wires of some gate of some $D_u$; by normalisation, $D_u(\pi_0|_{Q_u})$ depends only on
$\pi_0|_{B(x)}$.

\begin{definition}[views, zero knowledge, query counting]\label{thmD:def:zk}
A verifier strategy $V^*$ is an oracle algorithm with coins $\rho$, possibly adaptive and computationally
unbounded. It is \emph{$b$-bounded} if on every execution, for every oracle, it makes at most $b$
\emph{distinct} queries. $\View_{V^*,P}(x,w)$ consists of $\rho$ and the answers to $V^*$'s queries against
$\pi\leftarrow P(x,w)$ \cite[Def.~2.9]{GOS25}; since $V^*$ is a deterministic function of $\rho$ and the answers,
the view is a deterministic function of $\rho$ and the answers to its distinct queries. $(P,V)$ is
\emph{black-box perfect zero knowledge with query bound $b$} \cite[Def.~2.11]{GOS25} if one algorithm $\Sim$
(\emph{universal}) with oracle access to $V^*$ satisfies $\Sim^{V^*}(x)\equiv\View_{V^*,P}(x,w)$ for every
$b$-bounded $V^*$ and every $(x,w)\in R^{\rm YES}$; \emph{$\varepsilon$-statistical} if the statistical
distance is at most $\varepsilon$. $\Sim$ is \emph{straight-line} if it runs $V^*$ at most once, answering each
query when it is made, and outputs the view of that run or a \emph{failure symbol} (a fixed value that is not a
view, such as $\bot$).

\emph{Query-counting convention for running time.} Simulators access $V^*$ through a memoising interface
$\Mem(V^*)$: $\Mem$ runs $V^*$, answers every repeated query from its own record, and passes only first
occurrences of queries to the simulator. One interaction with $\Mem$ (receiving a new query or the output of
$V^*$, and returning an answer) is one step. Under this convention a $b$-bounded $V^*$ costs at most $b+1$
interactions per run, and $Q^*$, the number of distinct queries of $V^*$ on a run, is at most the proof
length. Without the interface, every running-time bound below holds with its terms other than $|\rho|$ multiplied
by $1+Q^*_{\rm mult}$ (tossing $V^*$'s coins does not depend on repeated queries),
where $Q^*_{\rm mult}$ counts the queries of $V^*$ with multiplicity, as \cite[Defs.~2.10--2.11]{GOS25} do. The
zero-knowledge identities do not depend on the convention.

\emph{Efficiency classes.} Simulators are strict PPT (SP), expected polynomial time (EP), or of Goldreich's
$\bot$ type (GB: strict polynomial time, output $\bot$ with probability at most $\frac12$, exact conditioned on
not $\bot$; \cite[Def.~4.3.1]{Gol01}). Provers are strict PPT (SP) or Las Vegas (LV: expected polynomial
time, exact output law).
\end{definition}

\noindent Definition~\ref{thmD:def:zk} restates Definitions~\ref{def:view} and~\ref{def:zk} for promise relations,
with the memoised interface of Definition~\ref{main:def:ZK} as its cost convention. Its labels (SP) and (EP)
name efficiency classes only; in Definition~\ref{main:def:ZK} they also require exactness.

\begin{definition}[local maps]\label{thmD:def:local}
A \emph{local map} for an instance $x$ is a family $f_x=(f_x(\alpha))_{\alpha\in[N]}$ of functions on
$(\{0,1\}^a)^{\ell_0}$ with sets $S_x(\alpha)\subseteq[\ell_0]$ such that $f_x(\alpha)(\pi)$ depends only on
$\pi|_{S_x(\alpha)}$. It is \emph{$L$-local} if $|S_x(\alpha)|\le L$ for all $\alpha$; \emph{$x$-only} if $f_x$
and $S_x$ are determined by $x$ (never by a witness); \emph{total}, as it is defined on every $\pi$; and
\emph{efficient} if $S_x(\alpha)$, and $f_x(\alpha)(\pi)$ given $\pi|_{S_x(\alpha)}$, are computable in time
$\poly(|x|)$.
\end{definition}

\noindent This is \cite[Def.~3.1]{GOS25} restricted to non-adaptive reads, with the two hypotheses that its
Lemma~3.2 uses made explicit ($x$-only and total).

\begin{fact}\label{thmD:fact:comp}
If $f$ is $L_1$-local in labels $\lambda$ and every label $\lambda(\beta)$ is an $L_2$-local function of $\pi$,
then $f\circ\lambda$ is $L_1L_2$-local with $S(\alpha)=\bigcup_{\beta\in S^{(1)}(\alpha)}S^{(2)}(\beta)$;
determinism, the $x$-only property and totality are preserved.
\end{fact}

\noindent \emph{Constraint graphs} are as in \cite[Def.~1.2]{Din07}, with the edge set a finite multiset
(parallel edges and self-loops allowed, as \cite{Din07} produces them). $\UNSAT_\sigma(G)$ is the fraction of
edges violated by $\sigma$, $\UNSAT(G)=\min_\sigma\UNSAT_\sigma(G)$ and $\size(G)=|V|+|E|$. Let
$\Sigma_0:=\{0,1\}^3$. For $Y\subseteq\{0,1\}^n$ and $a\in\{0,1\}^n$, $\rdist(a,Y)$ is the minimum relative
Hamming distance from $a$ to $Y$, with $\rdist(a,\emptyset):=1$.

\subsection{The instantiation: Dinur's tester with canonical folding}\label{thmD:sec:inst}

Let $\psi:\{0,1\}^n\to\{0,1\}$ be a Boolean function on the variables $X=\{x_1,\dots,x_n\}$, $L$ the set of
all Boolean functions on $\{0,1\}^n$ and $A_a(f):=f(a)$ the long code of $a$. A string $A:L\to\{0,1\}$ is
\emph{folded} if $A(\neg f)=\neg A(f)$ and $A(f)=A(f\wedge\psi)$ for all $f$ \cite[\S7]{Din07}. Dinur takes
$L'$ to contain one of each pair $\{f,\neg f\}$, puts $L'_\psi=\{f\in L':f=f\wedge\psi\}$, and asserts that every
assignment to $L'_\psi$ expands to a folded string. As printed this depends on the choice of $L'$: for
$\psi(x)=x_1$ and $L'\ni 1,\neg x_1$ one gets $L'_\psi=\emptyset$ although folded strings are not unique, and
for $\psi\equiv0$ no folded string exists. We fix the choice as follows.

\begin{lemma}[canonical folding]\label{thmD:lem:F}
Let $\psi\not\equiv0$, let $a^*$ be the lexicographically least element of $\SAT(\psi)$, and put
$\Lcan(\psi):=\{g\in L: g\le\psi,\ g(a^*)=0\}$. For $b:\Lcan(\psi)\to\{0,1\}$ define $E(b):L\to\{0,1\}$ by
$E(b)(f):=b(g)$ if $g(a^*)=0$ and $E(b)(f):=1\oplus b(\psi\oplus g)$ if $g(a^*)=1$, where $g:=f\wedge\psi$.
Then $E$ is a bijection from $\{0,1\}^{\Lcan(\psi)}$ onto the set of folded strings, $E(b)$ restricts to $b$
on $\Lcan(\psi)$, and $|\Lcan(\psi)|=2^{|\SAT(\psi)|-1}$. Moreover $A_a$ is folded iff $\psi(a)=1$.
\end{lemma}
\begin{proof}
On $L_\psi:=\{g:g\le\psi\}$ the map $\iota(g)=\psi\oplus g=\psi\wedge\neg g$ is an involution without fixed
points (as $\psi\ne0$), and $(\psi\oplus g)(a^*)=1\oplus g(a^*)$, so $\Lcan(\psi)$ contains exactly one element
of each orbit $\{g,\iota(g)\}$; this gives the size. If $A$ is folded then for $g\le\psi$,
$A(\iota(g))=A(\neg g\wedge\psi)=A(\neg g)=\neg A(g)$, and $A(f)=A(f\wedge\psi)$; hence $A=E(A|_{\Lcan})$, so
$E$ is onto the folded strings once we show that every $E(b)$ is folded. Folding over $\psi$ holds because
$E(b)(f)$ depends only on $f\wedge\psi$. For folding over true, let $g=f\wedge\psi$; then
$\neg f\wedge\psi=\iota(g)$ and $\iota(g)(a^*)=1\oplus g(a^*)$. If $g(a^*)=0$, then
$E(b)(\neg f)=1\oplus b(\psi\oplus\iota(g))=1\oplus b(g)=\neg E(b)(f)$; if $g(a^*)=1$, then
$E(b)(\neg f)=b(\iota(g))=\neg E(b)(f)$. The restriction property is immediate, so $E$ is injective. Finally,
$A_a(\neg f)=\neg f(a)$, and $A_a(f\wedge\psi)=f(a)\psi(a)$ equals $f(a)$ for all $f$ iff $\psi(a)=1$.
\end{proof}

\begin{definition}[the tester $\Pcan$]\label{thmD:def:tester}
Let $T$ be the long-code test of \cite[Thm.~7.1, App.~B]{Din07} and $T'$ the modified test of
\cite[\S7, item~1]{Din07} (with probability $\frac12$ run $T$; otherwise pick $i\in[n]$ and $f\in L$ uniformly and
check $x_i=A(f)\oplus A(f\oplus x_i)$, where $x_i$ also denotes the $i$-th dictator function). On input a
circuit $\psi$ on $X$:
\begin{enumerate}[label=(\roman*)]
\item if $\psi$ is satisfiable, $\Pcan(\psi)$ is Dinur's constraint graph \cite[\S7, item~2]{Din07} with
vertices $X\cup\Lcan(\psi)\cup Z$, where a location $f\in L$ read by $T'$ is answered by $E(b)(f)$
(Lemma~\ref{thmD:lem:F}), that is, by the variable $g$ or by the negation of the variable $\psi\oplus g$; there
is one variable $z_r\in\Sigma_0$ per outcome $r$ of $T'$'s coins and, for each of the three locations read on
$r$, an edge between $z_r$ and the variable answering that location, whose constraint checks that $z_r$
satisfies $T'$'s predicate on $r$ and agrees with that variable (negated if so indicated); edge multiplicities
are proportional to the probabilities of the outcomes; Boolean variables take values in $\{000,111\}$ and any
constraint touching another value rejects;
\item if $\psi$ is unsatisfiable, $\Pcan(\psi)$ has vertex set $X$ and the edges $\{x_i,x_{i+1}\}$ (indices
mod $n$), each with the empty constraint (every constraint rejects);
\item every output graph is replicated edge-wise to a common edge count $N_E$, the least common multiple
of the edge counts of the graphs of (i) and (ii) over the finitely many $\psi$ on $n$ inputs.
\end{enumerate}
\end{definition}

\begin{lemma}\label{thmD:lem:P}
Let $c_T>0$ be a constant for which \cite[Thm.~7.1]{Din07} holds (e.g.\ $c$ of \cite[Prop.~B.2]{Din07}): for every
$\psi$ and $\delta\in[0,1]$, $T$ rejects every folded $A$ that is $\delta$-far from all $A_a$ with $a\in\SAT(\psi)$
with probability at least $c_T\delta$. Then $\Pcan$ is an assignment tester in the sense of \cite[Def.~2.2]{Din07}
with alphabet $\Sigma_0$ and rejection probability $\varepsilon:=c_T/24\le\frac1{24}$, with the convention
$\rdist(a,\emptyset)=1$; this is half the constant $\varepsilon_{\rm pr}:=c_T/12$ that the printed proof of
\cite[Lemma~7.2]{Din07} gives. Consequently \cite[Lemma~1.8, Thm.~1.5]{Din07} hold with the tester $\Pcan$, with
$\beta_3=\varepsilon\rho/4$ \cite[proof of Lemma~1.8]{Din07} for this $\varepsilon$, where $\rho$ is the relative
distance of the code used in composition.
\end{lemma}
\begin{proof}
For satisfiable $\psi$, $\Pcan(\psi)$ is Dinur's construction for the choice
$L'=\{f:f(a^*)=0\}$, for which $L'_\psi=\Lcan(\psi)$. The proof of \cite[Lemma~7.2]{Din07} uses the representation
only through the claim that every assignment to $L'_\psi$ expands to a folded string answering all locations
consistently (\cite[Prop.~7.3]{Din07}: ``$A$ is folded''), which Lemma~\ref{thmD:lem:F} proves for this choice.
One step of the printed proof of \cite[Prop.~7.3]{Din07} needs a correction: its second case uses
$(1-\delta)\delta\ge\delta/2$, which is false for $\delta>\frac12$, and $\delta=\rdist(\sigma,\SAT(\psi))$ can exceed
$\frac12$ (take $\SAT(\psi)=\{0^n\}$ and $\sigma=1^n$). Run Dinur's case split with $\delta':=\min(\delta,\frac12)$ in
place of $\delta$. If the folded $A$ is $\delta'/2$-far from $A_a$ for every $a\in\SAT(\psi)$, then $T$ rejects with
probability at least $c_T\delta'/2\ge c_T\delta/4$ by \cite[Thm.~7.1]{Din07} (as $\delta'\ge\delta/2$), and $T'$ with at
least $c_T\delta/8$. Otherwise $A$ is $\delta'/2$-close to $A_{a'}$ for some $a'\in\SAT(\psi)$; then
$\Pr_f[A(f)\oplus A(f\oplus x_i)=a'_i]\ge1-\delta'\ge\frac12$ for every $i$, $\sigma$ and $a'$ differ on at least a
$\delta$ fraction of the indices, and $T'$ rejects with probability at least $\frac12\cdot\delta\cdot\frac12=\delta/4$.
As $c_T\le1$ (for $\SAT(\psi)=\{a\}$ the string $\neg A_a$ is folded and $1$-far from $A_a$), $T'$ rejects with
probability at least $c_T\delta/8=3\varepsilon\delta$ for all $\delta\in(0,1]$, while the printed case split gives
$c_T\delta/4=3\varepsilon_{\rm pr}\delta$ only for $\delta\le\frac12$. As in \cite[proof of Lemma~7.2]{Din07}, for each
outcome $r$ on which $T'$ rejects, at least one of the three constraints at $z_r$ is violated, and each of them
carries a $\frac13\Pr[r]$ fraction of the edges; so at least a $\frac13\cdot3\varepsilon\delta=\varepsilon\delta$
fraction of the constraints is violated, which is soundness. Completeness is as in Dinur's proof, as
$E(A_a|_{\Lcan(\psi)})=A_a$ by Lemma~\ref{thmD:lem:F}. For unsatisfiable $\psi$ completeness is
vacuous, and every assignment violates every constraint, so $\UNSAT=1\ge\varepsilon\cdot\rdist(a,\emptyset)$.
Replication does not change any $\UNSAT$ value. \cite[Lemma~1.8]{Din07} assumes only an assignment tester
with constant rejection probability over a constant alphabet; its soundness step uses
$\rdist(\sigma'|_{[u]\cup[v]},\SAT(\tilde c(e)))\ge\rho/4$, which holds with $\rdist(\cdot,\emptyset)=1$; and
\cite[Thm.~1.5]{Din07} is proved from Lemma~1.8 with the tester of \cite[Thm.~5.1]{Din07}, for which $\Pcan$ may be
substituted, and from \cite[Lemma~1.6]{Din07} at the $t$ of Remark~\ref{thmD:rem:const} (Remark~\ref{thmD:rem:L61}).
\end{proof}

\begin{remark}[domination]\label{thmD:rem:dom}
In the soundness proofs of \cite[Lemma~1.8, Lemma~7.2]{Din07} the block variables $[v]$ and the long-code
variables are read as bits, while in $G'$ they carry $\Sigma_0$-labels that may lie outside $\{000,111\}$.
Rounding each such label to a bit (say its first coordinate) can only lower the number of violated
constraints: every constraint touching a non-Boolean label rejects before rounding, and every other constraint
evaluates identically. So $\UNSAT_{\sigma'}\ge\UNSAT_{\rm rounded}$ and Dinur's bounds, proved for Boolean
values, apply.
\end{remark}

\begin{remark}[three steps in the proof of {\cite[Lemma~6.1]{Din07}}]\label{thmD:rem:L61}
\cite[Thm.~1.5]{Din07} rests on \cite[Lemma~1.6]{Din07}, which follows from \cite[Lemma~6.1]{Din07}. Lemma~6.1 is
stated for every $t$; repairing three steps of its printed proof, we prove it for $t=2r^2$ with $r\ge r_0$ ($r_0$ as
in (c)), which suffices for the proof of \cite[Thm.~1.5]{Din07} (Remark~\ref{thmD:rem:const}). We use the notation
of \cite[\S6]{Din07} ($t$ even, $\vec\sigma$ a labelling of $G^t$, $\sigma$ its popular-opinion labelling, $I$, $N$,
$Z$, $X_{u,\ell}$, $X'_{u,k}$, $\tau$; equation numbers are Dinur's) and write $\theta:=|F|/|E|$.
(a) \emph{Opinions outside $\Gamma(w)$.} A label $\vec\sigma(w)$ gives values only on $\Gamma(w)$, the endpoints
of $t/2$-step walks from $w$, but \cite[Def.~6.1]{Din07} and $X_{u,\ell}$ also compare $\vec\sigma(w)_u$ with a
value when $u$ is reached from $w$ by a walk of length $\ell$ with $t/2<\ell\le t/2+\sqrt{t/2}$, where $u\notin\Gamma(w)$
is possible (for $\ell\le t/2$ the self-loops of $G$ put $u$ in $\Gamma(w)$); a hit edge then need not lie in $\Gamma(v_0)\times\Gamma(v_t)$, and ``$N(e)>0$ implies that $e$
rejects'' is not justified. Read every comparison $\vec\sigma(w)_u=a$ with $u\notin\Gamma(w)$ as false. A hit
edge $(v_{i-1},v_i)$ then lies in $F$ with $v_{i-1}\in\Gamma(v_0)$ and $v_i\in\Gamma(v_t)$, so every assignment
to $\Gamma(v_0)\cup\Gamma(v_t)$ that agrees with the labels of $v_0$ and $v_t$ takes the values
$\sigma(v_{i-1}),\sigma(v_i)$ there and violates this edge; hence $e$ rejects. $X_{u,\ell}$ becomes a sub-distribution; (8) and the decomposition (10) hold verbatim, as the event
depends only on the endpoint of the walk; $X_{u,t/2}$ is still a distribution (a $t/2$-step walk from $u$ to $w$,
reversed, is one from $w$ to $u$), so (5) holds; the derivation of (9) goes through, as of $X'_{u,k}$ it needs
only $\Pr[X'_{u,k}=a]\le1$; and \cite[Lemma~6.3]{Din07} is unaffected, as $N\le Z$.
(b) \emph{An empty $F$.} For $\UNSAT_\sigma(G)\ge1/t$ the proof takes $|F|=\lfloor|E|/t\rfloor$, which is $0$
when $|E|<t$, and then uses (6) as an equality. Take instead $F$ to be the set of all edges violated by
$\sigma$ in both cases, so $\theta=\UNSAT_\sigma(G)$, and let $\theta>0$; the proof of \cite[Lemma~6.2]{Din07} does not
use $\theta\le1/t$; (c) completes the argument, with $\xi_0=|I_r|\min(\theta,1/t)$ in place of (6).
(c) \emph{The middle indices.} That proof uses (9) for $\ell\in\{i-1,t-i\}$, $i\in I$, derived from
\cite[Lemma~6.4]{Din07} with $\ell_1=t/2$ and $\ell_2=\ell$, which needs
$\ell_{\rm thr}<\ell_1-\sqrt{\ell_1}\le\ell_2<\ell_1+\sqrt{\ell_1}$ for a threshold $\ell_{\rm thr}$ (Dinur's $\ell_0$)
that depends only on $d$ and $|\Sigma|$. This fails for small $t$, and for $\ell=i-1$ with $i=\min I$ whenever $t/2$ is
not a square; $\mathbb E[N]\ge c|I|\theta$, $c:=(\tau/(2|\Sigma|))^2$ (not the $c$ of \cite[Lemma~6.4]{Din07}), is then unproved, and for $t=2$ ($I=\{1,2\}$)
false: let $G$ be the complete graph on five vertices with a null self-loop at every vertex and every other constraint
rejecting $(a,a)$, and $\vec\sigma(w)_u:=b\ne a$ for $u=w$ and $a$ otherwise; then $\sigma\equiv a$ (a one-step walk from $u$ returns to $u$ with probability only $\frac15$) and $\theta>0$, but
no walk is hit, as $\vec\sigma(v_0)_{v_0}=\vec\sigma(v_2)_{v_2}=b$. So let $t=2r^2$ with $r\ge r_0$, where $r_0$ is
least with $\ell_{\rm thr}<r_0^2-r_0$ (so $r_0\ge2$), and replace $I$ by $I_r:=\{r^2-r+2,\dots,r^2+r-1\}$ in $N$ and
$Z$. For $i\in I_r$, $i-1$ and $t-i$ lie in $[r^2-r+1,r^2+r-2]$, strictly inside the window for $\ell_1=r^2$ and, above
$t/2$, in the range of (a); so (9) holds for them, and the proof of \cite[Lemma~6.2]{Din07} gives
$\mathbb E[N]\ge c|I_r|\theta$. By (11) and \cite[Prop.~6.5]{Din07}, which hold for any set of indices,
$\mathbb E[N^2]\le\mathbb E[Z^2]\le C_Z|I_r|\theta+|I_r|^2\theta^2$ with $C_Z:=1+2/(1-\lambda/d)$, so
\cite[Fact~2.6]{Din07} gives $\Pr[N>0]\ge c^2\xi/(C_Z+\xi)$ for $\xi:=|I_r|\theta$. This is increasing in $\xi$, and
$\xi\ge\xi_0:=|I_r|\min(\theta,1/t)$ with $\xi_0\le1$, since $\sqrt{t/2}=r\le|I_r|=2r-2<t$. Hence
$\Pr[N>0]\ge c^2\xi_0/(C_Z+1)\ge\beta_2\sqrt t\,\min(\UNSAT_\sigma(G),1/t)$ with $\beta_2:=c^2/(\sqrt2(C_Z+1))$, which
by (a) is the lemma. As $\beta_2<c^2<2^{-16}$ ($\tau<1$, $|\Sigma|=8$), and the lower bounds of
Remark~\ref{thmD:rem:const} need of $\beta_2$ only $\beta_2\le2^{-12}$ (a smaller $\beta_2$ forces a larger $t$), they
remain valid.
\end{remark}

Fix once and for all the instantiation $\mathcal D$ of \cite[proof of Thm.~1.5]{Din07} on input alphabet
$\Sigma_0$: explicit expander families for the preprocessing \cite[Lemma~2.1, Cor.~2.4]{Din07}, the even
powering parameter $t=2r^2$ chosen in Remark~\ref{thmD:rem:const}, an explicit code
$e:\Sigma_0^{d^{\lceil t/2\rceil}}\to\{0,1\}^{\ell}$ of linear dimension and relative distance at least $\rho>0$
\cite[Def.~5.1]{Din07}, taken from an explicit code family with this distance at every length (so
$\beta_3=\varepsilon\rho/4$ does not depend on $t$), and the tester $\Pcan$. Let $d$ be the degree
of the preprocessed graph, $C_{\mathcal D}\ge1$ and $\alpha_{\mathcal D}\in(0,1)$ the constants of
\cite[Thm.~1.5]{Din07} for $\mathcal D$ and $|\Sigma|=8$ (computed with the $\beta_3$ of Lemma~\ref{thmD:lem:P}), and
\[
\cit:=2d^{\lceil t/2\rceil},\qquad E_L:=\log_2\cit,\qquad E_C:=\log_2C_{\mathcal D}.
\]
These are fixed numbers of one explicit algorithm, so all algorithms below are uniform; only lower bounds on
their values are known (Remark~\ref{thmD:rem:const}). In Section~\ref{sec:main} the degree and the powering parameter
are written $d_{\mathcal D}$ and $t_{\mathcal D}$; in this section we write $d$ and $t$.

\subsection{The statement}\label{thmD:sec:stmt}

\begin{theorem}[zero-knowledge-preserving gap amplification]\label{thmD:thm}
Let $(P_0,V_0)$ be a non-adaptive PCP with coin multisets (Definition~\ref{thmD:def:pcp}) for
$(R^{\rm YES},L^{\rm NO})$ with alphabet $\{0,1\}^a$, at most $q$ queries, residual circuits of at most
$s\ge1$ gates, $|U|\le\poly(|x|)$, perfect completeness and gap $\gamma$. Let $m^*\le4s+1$ be as in
Section~\ref{thmD:sec:T}, $\hat\gamma:=\max\{\gamma,1/|U|\}$ and $k:=\lceil\log_2(3m^*/\hat\gamma)\rceil$. Then there is
a PCP $(P',V')$ for $(R^{\rm YES},L^{\rm NO})$ such that:
\begin{enumerate}[label=(\alph*)]
\item\label{thmD:it:a} \emph{format}: the alphabet is binary; $V'$ is non-adaptive, runs in time
$\poly(|x|)$, draws its test from the coin multiset $E(G_k)^{m}$ with $m:=\lceil1/\alpha_{\mathcal D}\rceil$, makes
$Q_D:=6m$ queries, and has perfect completeness and soundness error at most $\frac12$;
\item\label{thmD:it:b} \emph{length}: at most $27m^*|U|C_{\mathcal D}^k\le27(4s+1)|U|\,C_{\mathcal D}\,(3(4s+1)/\gamma)^{E_C}$;
\item\label{thmD:it:c} \emph{honest prover}: $P'(x,w)=f_x(P_0(x,w))$ for a deterministic, efficient,
$x$-only, total, $L_D$-local map $f_x$, where $L_D:=q\,\cit^{\,k}\le q\,\cit\,(3(4s+1)/\gamma)^{E_L}$;
\item\label{thmD:it:d} \emph{zero knowledge}: if $(P_0,V_0)$ is black-box perfect (resp.\
$\varepsilon$-statistical) zero knowledge with query bound $b$ via a universal simulator $\Sim_0$, then
$(P',V')$ is black-box perfect (resp.\ $\varepsilon$-statistical) zero knowledge with query bound
$\lfloor b/L_D\rfloor$ via a universal simulator $\Sim$, straight-line if $\Sim_0$ is; more strongly, $\Sim$
is exact (resp.\ $\varepsilon$-close) against every $V^*$ such that on every execution against every
$\pi_0$ in the support of $P_0(x,w)$ the union of the sets $S_x(\alpha)$ over its queries $\alpha$ has at most
$b$ elements;
\item\label{thmD:it:e} \emph{efficiency}: if $P_0$ is (SP) resp.\ (LV), so is $P'$; if $b\le\poly(|x|)$ and
$\Sim_0$ is (SP), (EP) or (GB), so is $\Sim$, under the convention of Definition~\ref{thmD:def:zk} and with the
time bound of Lemma~\ref{thmD:lem:DP};
\item\label{thmD:it:f} \emph{information}: $f_x(\pi_0)$ depends on $\pi_0$ only through $\pi_0|_{B(x)}$, and
$\pi_0|_{B(x)}$ is a deterministic function of $f_x(\pi_0)$; so $f_x(P_0(x,w))$ is non-constant iff
$\pi_0|_{B(x)}$ is non-constant on the support of $P_0(x,w)$.
\end{enumerate}
\end{theorem}

\begin{remark}\label{thmD:rem:variants}
(i) \emph{Fair coins.} If instead each of $m':=\lceil2/\alpha_{\mathcal D}\rceil$ samples draws
$j\in\{0,1\}^{\lceil\log_2|E(G_k)|\rceil}$ and uses edge $j\bmod|E(G_k)|$, every edge has one or two
preimages, a violated set of $\ge\alpha_{\mathcal D}|E(G_k)|$ edges is hit with probability $\ge\alpha_{\mathcal D}/2$,
and the soundness error is $(1-\alpha_{\mathcal D}/2)^{m'}<\frac12$; this uses $6m'$ queries and $m'\lceil\log_2|E(G_k)|\rceil=O(\log|x|)$
coins. Item~\ref{thmD:it:a} refers to the coin-multiset form; where the main theorem needs fair coins it converts
afterwards. Neither form touches the proof, so \ref{thmD:it:c}--\ref{thmD:it:f} are unaffected. (Error reduction
via expander neighbourhoods \cite[Lemma~2.11]{BGHSV06}, applied through expected robustness, would give
$\log_2|E(G_k)|+O(1)$ coins; it is not needed and not used.)
(ii) Theorem~\ref{thmD:thm} needs neither $s$ nor $1/\gamma$ small: the construction runs in polynomial time since
$3m^*/\hat\gamma\le3m^*|U|\le\poly(|x|)$, the list of tests being computable in time $\poly(|x|)$. Small $s$ and
$1/\gamma$ are what make $L_D$ and the length factor polylogarithmic.
(iii) If $V^*$ is non-adaptive, so is the induced reader of $\pi_0$ in Lemma~\ref{thmD:lem:Z}. Hence if
$(P_0,V_0)$ is zero knowledge only against non-adaptive verifiers, the theorem certifies zero knowledge of
$(P',V')$ against non-adaptive verifiers with bound $\lfloor b/L_D\rfloor$; it says nothing about adaptive ones.
\end{remark}

\subsection{Proof}\label{thmD:sec:proof}

\subsubsection{The Tseitin front end}\label{thmD:sec:T}

For a gate $z=g(y,y')$ the Tseitin clauses \cite{Tse68} are: for $\oplus$, $\bar z\vee y\vee y'$,
$\bar z\vee\bar y\vee\bar y'$, $z\vee\bar y\vee y'$, $z\vee y\vee\bar y'$; for $\wedge$, $\bar z\vee y$,
$\bar z\vee y'$, $z\vee\bar y\vee\bar y'$; for $\vee$, $z\vee\bar y$, $z\vee\bar y'$, $\bar z\vee y\vee y'$; for
$z=\neg y$, $z\vee y$, $\bar z\vee\bar y$. For every gate and all values, all of its clauses hold iff
$z=g(y,y')$.

For a test $u$ let $\mathcal K_u$ consist of the clauses of all gates of $D_u$ and the unit clause on its output
gate, so $|\mathcal K_u|\le4s+1$. Pad clauses to width three by repeating their last literal and every list to
length $m^*:=\max_u|\mathcal K_u|$ with tautologies $y\vee\bar y\vee y$ on the output gate $y$ of that test. The
constraint graph $G_T$ over $\Sigma_0$ has a \emph{proof-bit vertex} for every bit in $B(x)$, a \emph{gate
vertex} $g_{u,j}$ for every element $u$ of the multiset $U$ (duplicates get separate copies) and every gate
$j$ of $D_u$, and a \emph{clause vertex} $y_{u,c}$ for every $u$ and $c\in[m^*]$. For each clause vertex and
each occurrence $\kappa\in\{1,2,3\}$ of its clause there is an edge to the occurrence's variable vertex
$w_\kappa$, with the constraint: \emph{the label $z$ of $y_{u,c}$, read as values of the three occurrences,
satisfies the clause, and $z_\kappa$ equals the first coordinate of the label of $w_\kappa$}. So
$|E_T|=3m^*|U|$, $|V_T|\le2|E_T|$, and every vertex lies on an edge. The honest map $f_T$ labels a proof bit
$\beta$ by $(\beta,0,0)$, a gate vertex by (its value on $\pi_0|_{Q_u}$, $0$, $0$), and a clause vertex by the
values of its three occurrences; it is $q$-local, deterministic, efficient, $x$-only and total.

\begin{lemma}\label{thmD:lem:T}
(a) If $\pi_0$ passes every test then $f_T(\pi_0)$ satisfies every edge of $G_T$. (b) If $x\in L^{\rm NO}$
then $\UNSAT(G_T)\ge\hat\gamma/(3m^*)$.
\end{lemma}
\begin{proof}
(a) Gate values satisfy the gate clauses, the output clause holds as $D_u(\pi_0|_{Q_u})=1$, tautologies hold,
and a clause vertex carries the actual values. (b) Fix a labelling $\sigma$. Let $\hat\pi$ take on $B(x)$ the
first coordinates of the proof-bit labels (arbitrary elsewhere) and let $\hat g$ be the first coordinates of
the gate labels. As $\hat\gamma$ is a gap, at least $\hat\gamma|U|$ elements $u$ of $U$ have $D_u(\hat\pi|_{Q_u})=0$.
For such $u$, if
$(\hat\pi,\hat g)$ satisfied all of $\mathcal K_u$, then by the gate property along a topological order $\hat g$
would be the true gate values and the output clause would fail; so some clause $c$ of $u$ is violated by
$(\hat\pi,\hat g)$. If $\sigma(y_{u,c})$ does not satisfy $c$, its three edges reject; otherwise it differs from
$(\hat\pi,\hat g)$ on some occurrence, whose edge rejects. So each such $u$ owns a rejecting edge among its
$3m^*$, and $\UNSAT_\sigma(G_T)\ge\hat\gamma|U|/(3m^*|U|)$.
\end{proof}

\subsubsection{One round and its honest map}\label{thmD:sec:round}

One round is $G\mapsto G':=(\mathrm{prep}(G))^t\circ\Pcan$ \cite[proof of Thm.~1.5]{Din07}. Here
$\mathrm{prep}=\mathrm{prep}_2\circ\mathrm{prep}_1$ \cite[Defs.~4.1, 4.2]{Din07}: $\mathrm{prep}_1$ replaces
each vertex $v$ by a cloud of vertices $(v,e)$, one per incident edge $e$, joined by expander edges with
equality constraints, and keeps each constraint $c(e)$ on the edge $\{(v,e),(v',e)\}$; $\mathrm{prep}_2$ adds
the edges of an expander and a self-loop at every vertex, with null constraints; the result $H_1$ is
$d$-regular. The power $H_2:=H_1^t$ has the same vertices and one edge per $t$-step walk; its labels assign
values to $\Gamma(u)$, the set of endpoints of $\lceil t/2\rceil$-step walks from $u$ in $H_1$, so
$|\Gamma(u)|\le d^{\lceil t/2\rceil}$ and, by the self-loops, $u\in\Gamma(u)$; an edge $(u,v)$ accepts $(a,b)$ iff
some assignment to $\Gamma(u)\cup\Gamma(v)$ agrees with $a$ and $b$ and satisfies every constraint of $H_1$
between them. Composition $\circ\,\Pcan$ \cite[Def.~5.1]{Din07} creates a block $[u]$ of $\ell$ Boolean
variables for each vertex $u$ of $H_2$ and, for each edge $\varepsilon=(u,v)$ of $H_2$, the graph
$G_\varepsilon:=\Pcan(\psi_\varepsilon)$ with $\psi_\varepsilon:=\tilde c(\varepsilon)$, the circuit on
$[u]\cup[v]$ that accepts iff its input is $(e(a),e(b))$ with $(a,b)$ satisfying the constraint of
$\varepsilon$; the blocks are shared and all other vertices of $G_\varepsilon$ are fresh for each
$\varepsilon$; $G'$ is the union. By Lemma~\ref{thmD:lem:P}, Remark~\ref{thmD:rem:L61} and \cite[Thm.~1.5]{Din07}:
\begin{enumerate}[label=(Q\arabic*)]
\item\label{thmD:Q1} $G'$ has alphabet $\Sigma_0$, is constructed in time $\poly(\size(G))$, satisfies
$\size(G')\le C_{\mathcal D}\,\size(G)$, and $\UNSAT(G')\ge\min(2\,\UNSAT(G),\alpha_{\mathcal D})$.
\end{enumerate}

\noindent For a labelling $\sigma$ of $G$ define the labelling $\phi(\sigma)$ of $G'$ by:
\begin{enumerate}[label=(\arabic*)]
\item $\sigma_1(v,e):=\sigma(v)$ for every vertex of $H_1$;
\item $\sigma_2(u):=(\sigma_1(u'))_{u'\in\Gamma(u)}$ in a canonical order, padded with $000$ to length
$d^{\lceil t/2\rceil}$;
\item the $j$-th variable of $[u]$ gets $\beta\beta\beta$ with $\beta:=e(\sigma_2(u))_j$;
\item for an edge $\varepsilon=(u,v)$ put $a:=(e(\sigma_2(u)),e(\sigma_2(v)))\in\{0,1\}^{2\ell}$. If
$\psi_\varepsilon$ is satisfiable, the variable $g\in\Lcan(\psi_\varepsilon)$ gets $g(a)g(a)g(a)$, and $z_r$ gets
$(v_1,v_2,v_3)$, where $v_j:=a_i$ if the $j$-th location read by $T'$ on $r$ is $x_i$ and $v_j:=f(a)$ if it is
$f\in L$. If $\psi_\varepsilon$ is unsatisfiable, $G_\varepsilon$ has no further vertices.
\end{enumerate}

\begin{lemma}[honest round map]\label{thmD:lem:H}
(i) If $\sigma$ satisfies $G$, then $\phi(\sigma)$ satisfies $G'$. (ii) Every label of $\phi(\sigma)$ depends only
on $\sigma|_S$ for a set $S\subseteq V(G)$ determined by $G$ and the label, with $|S|\le\cit$. (iii) $\phi$ is
total, deterministic, determined by $G$, and computable in time $\poly(\size(G))$; one label is computable from
$\sigma|_S$ in time $O(\cit)$ plus a constant. (iv) For every vertex $v$ of $G$ that lies on an edge,
$\sigma(v)$ is a deterministic function, determined by $G$, of the first coordinates of the labels of
$\phi(\sigma)$ on vertices of $G'$ that lie on edges.
\end{lemma}
\begin{proof}
(ii), (iii). Steps (1) and (3) copy or encode one earlier label; step (2) reads $\Gamma(u)$, hence $\sigma$ on
$S_2(u):=\{v:(v,e)\in\Gamma(u)\}$, and $|S_2(u)|\le|\Gamma(u)|\le d^{\lceil t/2\rceil}$ because there are
$d^{\lceil t/2\rceil}$ walks of length $\lceil t/2\rceil$ from $u$; step (4) reads $\sigma$ on
$S_2(u)\cup S_2(v)$, at most $2d^{\lceil t/2\rceil}=\cit$ vertices. Every formula is defined for every input,
whether $\psi_\varepsilon$ is satisfiable is decided by brute force over $2^{2\ell}$ inputs (a constant of
$\mathcal D$), and the sizes $|L|$, $|Z|$ are constants of $\mathcal D$.

(i). Copies satisfy the equality constraints and the inherited constraints, and the constraints added by
$\mathrm{prep}_2$ are null, so $\sigma_1$ satisfies $H_1$. For an edge $(u,v)$ of $H_2$ the assignment
$\sigma_1|_{\Gamma(u)\cup\Gamma(v)}$ witnesses its constraint. For an edge $\varepsilon$ of $H_2$, $a$ encodes
a pair satisfying the constraint of $\varepsilon$, so $\psi_\varepsilon(a)=1$; in particular
$\psi_\varepsilon$ is satisfiable. By Lemma~\ref{thmD:lem:F}, $A_a$ is folded and equals $E$ of its restriction,
which is the labelling of step (4) on $\Lcan(\psi_\varepsilon)$; so every location read by $T'$ is answered by
the value $A_a(f)=f(a)$ or $a_i$, which is $v_j$. By perfect completeness of $T$ \cite[Thm.~7.1]{Din07} and
$A_a(f)\oplus A_a(f\oplus x_i)=a_i$, the triple $z_r$ satisfies $T'$'s predicate on every $r$, and it agrees with
every variable it is joined to. All Boolean labels lie in $\{000,111\}$. So every constraint of $G'$ holds.

(iv). Let $v$ lie on an edge $e$ of $G$. Then $(v,e)$ is a vertex of $H_1$ and of $H_2$; $H_2$ is regular of
positive degree, so some edge $\varepsilon$ of $H_2$ contains $(v,e)$ and $[(v,e)]\subseteq V(G_\varepsilon)$.
Every block variable lies on an edge of $G_\varepsilon$: in case (i) of
Definition~\ref{thmD:def:tester} because $T'$ reads every $x_i$ with positive probability, in case (ii) by
construction. The first coordinates of the labels on $[(v,e)]$ form $e(\sigma_2((v,e)))$; $e$ is injective (its
relative distance is positive), so they determine $\sigma_2((v,e))$, whose coordinate at
$(v,e)\in\Gamma((v,e))$ is $\sigma_1(v,e)=\sigma(v)$.
\end{proof}

\noindent Lemma~\ref{thmD:lem:H}(i) is needed because \cite[Thm.~1.5]{Din07} only asserts that a satisfying
labelling exists, and the proof of \cite[Lemma~1.8]{Din07} chooses the tester's labels by the completeness
property of the tester, an existential step; Lemma~\ref{thmD:lem:H} makes the witness a local deterministic
function of the old labelling, with $z_r$ defined directly.

\subsubsection{Iteration and parameters}\label{thmD:sec:iter}

Let $G_0:=G_T$ and $G_{j+1}:=(\mathrm{prep}(G_j))^t\circ\Pcan$ for $j<k$ (here $k\ge2$, as $m^*\ge1\ge\hat\gamma$);
every $G_j$ has alphabet $\Sigma_0$, so one instantiation serves all rounds. Let
$f_x:=\mathrm{enc}\circ\phi_{G_{k-1}}\circ\cdots\circ\phi_{G_0}\circ f_T$, where $\mathrm{enc}$ writes each label
of $G_k$ as its three bits in a fixed vertex order, and $P'(x,w):=f_x(P_0(x,w))$. The verifier computes $G_k$,
draws $m=\lceil1/\alpha_{\mathcal D}\rceil$ edges independently and uniformly, reads the six bits of each and accepts iff
all $m$ constraints hold; every bit string of the right length is a labelling, as $|\Sigma_0|=2^3$.

\emph{Completeness} follows from Lemma~\ref{thmD:lem:T}(a) and Lemma~\ref{thmD:lem:H}(i). \emph{Soundness}:
by Lemma~\ref{thmD:lem:T}(b) and \ref{thmD:Q1}, $\UNSAT(G_k)\ge\min(2^k\hat\gamma/(3m^*),\alpha_{\mathcal D})
=\alpha_{\mathcal D}$, since $2^k\ge3m^*/\hat\gamma$; so the acceptance probability is at most
$(1-\alpha_{\mathcal D})^{\lceil1/\alpha_{\mathcal D}\rceil}\le e^{-1}<\frac12$. \emph{Length}:
$3|V(G_k)|\le3C_{\mathcal D}^k\size(G_0)\le3C_{\mathcal D}^k\cdot3|E_T|=27m^*|U|C_{\mathcal D}^k$, and
$2^{k-1}<3m^*/\hat\gamma\le3m^*/\gamma$ gives $C_{\mathcal D}^k=C_{\mathcal D}\cdot(2^{k-1})^{E_C}\le
C_{\mathcal D}(3m^*/\hat\gamma)^{E_C}\le C_{\mathcal D}(3m^*/\gamma)^{E_C}$.
\emph{Locality}: by Fact~\ref{thmD:fact:comp}, $f_x$ is $q\cit^{\,k}$-local, and likewise
$\cit^{\,k}\le\cit(3m^*/\hat\gamma)^{E_L}\le\cit(3m^*/\gamma)^{E_L}$. \emph{Time}: the list of tests is computed in
time $\poly(|x|)$, so $m^*,|U|\le\poly(|x|)$; as $1/\hat\gamma\le|U|$,
$C_{\mathcal D}^k\le C_{\mathcal D}(3m^*|U|)^{E_C}$ and $\cit^{\,k}\le\cit(3m^*|U|)^{E_L}$ are $\poly(|x|)$, so
$G_0,\dots,G_k$ are built in polynomial time. One output bit is evaluated through its dependency graph: one
label of $G_k$, at most $\cit$ labels of $G_{k-1}$, \dots, at most $\cit^{\,k}$ labels of $G_0$, and at most
$q\cit^{\,k}$ symbols of $\pi_0$; since $\cit\ge2$, the number of nodes is at most
$\sum_{j\le k}\cit^{\,j}+q\cit^{\,k}\le2\cit^{\,k}+q\cit^{\,k}\le3L_D$, each computed in time $\poly(|x|)$.
\emph{Information} (item~\ref{thmD:it:f}): $f_T(\pi_0)$ depends only on $\pi_0|_{B(x)}$, since proof-bit
vertices exist only for $B(x)$ and gate values depend only on gate inputs; conversely the proof-bit vertices
of $G_T$ lie on edges, so by Lemma~\ref{thmD:lem:H}(iv) applied $k$ times, $\pi_0|_{B(x)}$ is a function of the
first coordinates of the labels of $G_k$, which $\mathrm{enc}$ writes out. This proves
\ref{thmD:it:a}--\ref{thmD:it:c} and \ref{thmD:it:f}.

\subsubsection{Zero knowledge}\label{thmD:sec:zk}

\begin{lemma}[{local computability; {\cite[Lemma~3.2]{GOS25}} with a cache and a cap}]\label{thmD:lem:Z}
Let $(P_0,V_0)$ be black-box perfect (resp.\ $\varepsilon$-statistical) zero knowledge with query bound $b$ via a
universal $\Sim_0$, and let $P'=f_x\circ P_0$ with $f_x$ deterministic, efficient, $x$-only, total and $L$-local.
For a verifier $V^*$ of $P'$ let $A$ run $V^*$ and answer each query $\alpha$ by querying $\pi_0$ on the positions
of $S_x(\alpha)$ not yet cached and evaluating $f_x(\alpha)$ on the cache, and let $A'$ be $A$ halted just
before its $(b+1)$-st distinct query. Let $g$ map a record $(\rho,(\text{query},\text{answer})_i)$ of $A'$ to:
the view of $V^*$ obtained by replaying $V^*$ on $\rho$ and answering each $\alpha$ by $f_x(\alpha)$ on the
recorded answers, if the replay ends with every needed position recorded; and a fixed symbol $\bot_{\rm tr}$
otherwise. Let $\Sim^{V^*}(x)$ run $\Sim_0^{A'}(x)$ and output $\bot$ if it does, and $g$ of its output otherwise.
If $\Sim_0$ is straight-line, $\Sim$ applies no $g$: it outputs $\Sim_0$'s output if this is a failure symbol,
$\bot_{\rm tr}$ if the cap fires in the single run of $A'$ inside $\Sim_0$, and otherwise $\rho$ with the
query--answer pairs of $V^*$ recorded during that run.
Then $\Sim$ is universal and black-box, and $\Sim^{V^*}(x)\equiv\View_{V^*,P'}(x,w)$ (resp.\ within
$\varepsilon$) for every $(x,w)\in R^{\rm YES}$ and every $V^*$ for which $A$ makes at most $b$ distinct queries
on every execution against the support of $P_0(x,w)$, in particular for every $\lfloor b/L\rfloor$-bounded
$V^*$. If $\Sim_0$ is straight-line, so is $\Sim$.
\end{lemma}
\begin{proof}
Against $\pi_0\leftarrow P_0(x,w)$, $A$ answers each $\alpha$ by $f_x(\alpha)(\pi_0)$, because the cache holds
true values and $f_x(\alpha)$ depends only on $S_x(\alpha)$; since $P'(x,w)$ is $f_x(P_0(x,w))$ and $f_x$ does not
depend on $w$, the coins and answers of $V^*$ inside $A$ are distributed as $\View_{V^*,P'}(x,w)$. Because of the
cache $A'$ never repeats a query, so it is $b$-bounded on every oracle and is a legitimate verifier for
$(P_0,V_0)$ (no computational bound is required); hence $\Sim_0^{A'}(x)\equiv\View_{A',P_0}(x,w)$ (resp.\ within
$\varepsilon$). On the support of $P_0(x,w)$ the cap never fires, by hypothesis, so
$\View_{A',P_0}=\View_{A,P_0}$, and on these records the replay finds every needed position: $g$ reproduces
$V^*$'s queries and, as $f_x$ is total and local, its answers. So $g(\View_{A,P_0}(x,w))=\View_{V^*,P'}(x,w)$,
and applying the deterministic total map $g$ to both sides gives the claim, since a deterministic map does not
increase statistical distance. In the perfect case the record has the law of $\View_{A,P_0}$, so $\bot_{\rm tr}$
has probability $0$; in the $\varepsilon$-statistical case it has probability at most $\varepsilon$. A
$\lfloor b/L\rfloor$-bounded $V^*$ makes $A$ read at most $L\lfloor b/L\rfloor\le b$ positions on every execution
for every oracle. If $\Sim_0$ is straight-line, each of its outputs other than a failure symbol is the view of
its single run of $A'$: $\rho$ and the answers given to $A'$. As $V^*$ is a deterministic function of $\rho$ and
its answers, and inside $A'$ each answer is $f_x(\alpha)$ evaluated on these answers, the replay in $g$ repeats
that run; it returns the pairs recorded there if the run ends normally, and $\bot_{\rm tr}$ if the cap fires,
since the query at which $A'$ halts needs a position not recorded. So the straight-line $\Sim$ outputs
$h(\Sim_0^{A'}(x))$, where $h$ fixes failure symbols and agrees with $g$ elsewhere; as $\View_{A',P_0}(x,w)$
never is a failure symbol, $h$ and $g$ agree on it, the argument above applies to $h$, and $V^*$ runs at most once (exactly once unless $\Sim_0$ outputs a failure
symbol without running $A'$).
\end{proof}

\begin{lemma}[efficiency]\label{thmD:lem:DP}
If $P_0$ is (SP) or (LV), so is $P'$. Let $p(|x|)$ bound the time to build $G_0,\dots,G_k$ and the time per
node of a dependency graph, let $Q^*$ be the number of distinct queries of $V^*$ on a run (at most the
length of $P'$, and at most $\lfloor b/L_D\rfloor$ if $V^*$ is $\lfloor b/L_D\rfloor$-bounded), and let $N_0$ be
the number of steps of $\Sim_0^{A'}$ with each interaction with $A'$ counted as one step, not counting the $|\rho|$
steps in which it tosses the coins $\rho$ of $A'$, which are those of $V^*$ (Definition~\ref{def:zk}). Under the
convention of Definition~\ref{thmD:def:zk}, $\Sim^{V^*}$ takes at most $|\rho|+p(|x|)+(N_0+1)(2+Q^*+3Q^*L_D\,p(|x|))$
steps, and at most $|\rho|+p(|x|)+N_0+Q^*+1+3Q^*L_D\,p(|x|)$ if $\Sim_0$ is straight-line. If $\Sim_0$ is of class
(SP), (EP) or (GB), then $N_0$ is bounded by a polynomial in $(|x|,b)$ (in expectation for (EP)) by
Definition~\ref{def:zk}; hence (SP), (EP) and (GB) are preserved when $b\le\poly(|x|)$.
\end{lemma}
\begin{proof}
$f_x$ is a deterministic polynomial-time map applied after sampling, so the prover keeps its exact output law
and its running-time class. An interaction of $\Sim_0$ with $A'$ (a next-message call on a given history)
re-runs $V^*$ through $\Mem$ along that history, at a cost of at most $Q^*+1$ interactions with $\Mem$; it answers
at most $Q^*$ distinct queries of $V^*$, each by evaluating a dependency graph with at most $3L_D$ nodes; the replay in
$g$ is one more such run. In the straight-line case there is at most one run and no replay. The
multiplier is deterministic, so expectations of $N_0$ stay polynomial, and $\Sim$ outputs $\bot$ exactly when
$\Sim_0$ does.
\end{proof}

\noindent Item~\ref{thmD:it:d} of Theorem~\ref{thmD:thm} is Lemma~\ref{thmD:lem:Z} with $L=L_D$, and
item~\ref{thmD:it:e} is Lemma~\ref{thmD:lem:DP}. \qed

\subsection{When an input PCP suffices}\label{thmD:sec:cor}

The following corollary packages Theorem~\ref{thmD:thm} as a black box: any polylogarithmic-query zero-knowledge PCP
with the parameters below yields the conclusion of our main theorem. Section~\ref{sec:main} does not use it; it applies
Theorem~\ref{thmD:thm} directly.

\begin{corollary}\label{thmD:cor:in}
Let $R\in\mathrm{NTIME}(T)$ with $T$ a time-constructible polynomial and $T(n)\ge n$. Suppose there are
absolute constants $e_1,e_2\ge1$, with $e_2$ rational,
and a family $(\Pi_B)_{B\ge1}$ of PCPs for $R$, uniform in $B$ (one prover, one verifier and one simulator
algorithm, each given $1^B$), such that with $N:=T(|x|)+B$, logarithms to base two and $N\ge4$:
\begin{enumerate}[label=(\arabic*)]
\item the honest verifier is non-adaptive with a coin multiset of size $|U|\le N(\log N)^{e_1}$;
\item the residual circuits have at most $s\le(\log N)^{e_2}$ gates and are computable in polynomial time,
and $q\le(\log N)^{e_2}$;
\item perfect completeness and gap $\gamma\ge(\log N)^{-e_2}$;
\item black-box perfect (resp.\ $\varepsilon$-statistical) zero knowledge with query bound $B$ against adaptive
verifiers, via a universal simulator of class (SP), (EP) or (GB) in time $\poly(|x|,B)$, and a prover of class
(SP) or (LV).
\end{enumerate}
Then for every time-constructible polynomial $b\ge1$, $R$ has a binary PCP with at most $Q_D$ non-adaptive
queries, perfect
completeness, soundness error at most $\frac12$, length $O\big((T+b)(\log(T+b))^{e_1+2e_2(1+E_C+E_L)}\big)$,
and perfect (resp.\ $\varepsilon$-statistical) zero knowledge against every adaptive $b$-bounded verifier,
with a simulator of the same class.
\end{corollary}
\begin{proof}
Put $X(N):=\cit\,15^{E_L}(\log N)^{e_2(1+2E_L)}$, a nondecreasing function, and $\Xi(N):=\lceil\tilde X(N)\rceil+1$,
where $\tilde X(N)$ is the output of a fixed deterministic evaluation of $X(N)$ to within $\frac12$ (possible as
$e_2$ is rational and $\cit$, $E_L$ are fixed constants of one explicit algorithm, so the $O(\log\log N)$ bits needed
are obtained in time $\mathrm{poly}(\log N)$), so that $X\le\Xi<X+3$ and $\Xi(N)$ is computable in time $\mathrm{poly}(\log N)$. Since
$4s+1\le5s$ and $1/\gamma\le(\log N)^{e_2}$, Theorem~\ref{thmD:thm}\ref{thmD:it:c} gives
$L_D(\Pi_B)\le q\cit(15(\log N)^{2e_2})^{E_L}\le X(N)\le\Xi(N')$ for every $N'\ge N$. Given $x$, let $Y:=T(|x|)+b(|x|)$ and
$B:=b(|x|)\cdot\Xi(Y^2)$, computable in polynomial time. There is an absolute $Y_{\rm thr}$ such that for
$Y\ge Y_{\rm thr}$ we have $\Xi(Y^2)\le Y$, hence $N=T+B\le Y\,\Xi(Y^2)\le Y^2$, so
$L_D(\Pi_B)\le\Xi(Y^2)$ and $\lfloor B/L_D\rfloor\ge b$; Theorem~\ref{thmD:thm}\ref{thmD:it:d} then gives zero
knowledge against $b$-bounded verifiers, and \ref{thmD:it:e} the simulator class (as $B\le\poly(|x|)$). By
\ref{thmD:it:b} the length is at most
$27\cdot5(\log N)^{e_2}\cdot N(\log N)^{e_1}\cdot C_{\mathcal D}(15(\log N)^{2e_2})^{E_C}$, and with
$N\le Y\,\Xi(Y^2)$ and $\log N\le2\log Y$ this is $O(Y(\log Y)^{e_1+2e_2(1+E_C+E_L)})$. For $Y<Y_{\rm thr}$
we have $|x|\le T(|x|)<Y_{\rm thr}$: the verifier decides $x$ by table look-up and the proof is empty, which is trivially
zero knowledge.
\end{proof}

\begin{remark}[why each hypothesis is needed for this route]\label{thmD:rem:hyp}
\emph{Coins, not length:} the front end creates one gadget per element of $U$ and no round shrinks a graph, so
the output length is at least $|U|$; deterministic reduction of an arbitrary coin multiset to
$\log n+O(\log\log n)$ coins is not known in general. \emph{Residual size, not queries:} with alphabet
$\{0,1\}^a$ a polynomial-time residual predicate can need circuits of polynomial size, and $L_D$ grows as
$s^{E_L}$. \emph{Gap:} $L_D$ grows as $\gamma^{-E_L}$. \emph{Perfect completeness:}
\cite[Thm.~1.5]{Din07} is a statement about $\UNSAT=0$. \emph{Non-adaptive honest verifier:}
the front end needs fixed query tuples. \emph{Adaptive zero knowledge of the input:} with only non-adaptive
zero knowledge of the input, Theorem~\ref{thmD:thm} certifies only non-adaptive zero knowledge of the output
(Remark~\ref{thmD:rem:variants}(iii)). \emph{``$b\cdot\mathrm{polylog}$'':} the input must support query bounds
larger than $b$ by the factor $\Xi=(\log)^{\Theta(e_2E_L)}$, with exponents independent of $B$.
\end{remark}

\subsection{Remarks}\label{thmD:sec:rem}

\begin{remark}[constants]\label{thmD:rem:const}
Correctness does not depend on the values of $C_{\mathcal D}$, $\alpha_{\mathcal D}$ or $\cit$. The proof of
\cite[Thm.~1.5]{Din07} sets $\alpha_{\mathcal D}=\beta_2\beta_3/\sqrt t$ and needs of $t$ only
$\beta_1\beta_2\beta_3\sqrt t\ge2$ (it takes $t=\lceil(2/(\beta_1\beta_2\beta_3))^2\rceil$). As Remark~\ref{thmD:rem:L61}
proves \cite[Lemma~6.1]{Din07} only for $t=2r^2$ with $r\ge r_0$, $\mathcal D$ takes $t:=2r^2$ with
$r:=\max(r_0,\lceil\sqrt2/(\beta_1\beta_2\beta_3)\rceil)$, so $\sqrt t\ge2/(\beta_1\beta_2\beta_3)$. Lower bounds on what
Dinur's analysis, with the repairs of this section, yields are $t\ge9\cdot2^{42}$, $E_L\ge3.958\cdot10^{13}$,
$Q_D\ge10^{13}$ and $E_C\ge2^{6\cdot2^t}$. They only grow as $\beta_1,\beta_2,\beta_3$ decrease (with $\beta_2$ as
repaired in Remark~\ref{thmD:rem:L61}(c)) and as $r_0$ grows: $t$ and $Q_D=6\lceil\sqrt t/(\beta_2\beta_3)\rceil$
directly, $E_L=1+\lceil t/2\rceil\log_2d_{\mathcal D}\ge1+t$ (as $d_{\mathcal D}\ge4$) and $E_C$ through $t$, the degree
$d_{\mathcal D}$ being fixed and the composition blowup nondecreasing in the alphabet size
$8^{d_{\mathcal D}^{\lceil t/2\rceil}}$. The stated values use $\beta_1\le\frac18$, $\beta_2\le2^{-12}$ and
$\beta_3\le\frac1{96}$; these
bound the analysis, not gap amplification itself. The canonical tester $\Pcan$ changes only the constant size
factor of composition (through $N_E$), not the rejection probability. The repair of the case $\delta>\frac12$ in the
proof of Lemma~\ref{thmD:lem:P} changes the analysis, not the tester, and halves the guaranteed rejection probability,
to $\varepsilon=\frac12\varepsilon_{\rm pr}$; the constant $\beta_3=\varepsilon\rho/4$ of Lemma~\ref{thmD:lem:P}, and
through it $t$, $\alpha_{\mathcal D}$ and $Q_D$, are computed from this $\varepsilon$. This only lowers $\beta_3$, which
stays at most $\frac1{96}$ as $\varepsilon\le\frac1{24}$ and $\rho\le1$, so the stated lower bounds remain valid.
\end{remark}

\begin{remark}[a known gap elsewhere in \cite{Din07}]\label{thmD:rem:GM}
Goldreich and Meir \cite{GM20} identify and bridge a gap in the extension of gap amplification to assignment
testers (\cite[\S9]{Din07}); by their account it does not affect PCPs. Theorem~\ref{thmD:thm} uses
\cite[\S\S1--7 and Apps.~A--B]{Din07}, with the corrections listed at the start of this section, not \S9.
\end{remark}

\endgroup

\begingroup
\providecommand{\arxivid}[1]{arXiv:~\texttt{\detokenize{#1}}}

\section{Related work and credits}\label{sec:related}

Sections~\ref{sec:intro-problem} and~\ref{sec:intro-compare} compare our results with the closest work
\cite{GOS25,ZOGS26,BV26,BCGV16,BCFGRS17}. This section surveys the wider literature, states the two guarantees of
Bafna and Vyas \cite{BV26} precisely, and ends with a table of every ingredient taken from prior work
(Table~\ref{tab:credits}).

\paragraph{Zero-knowledge PCPs.} Zero-knowledge PCPs go back to Dwork, Feige, Kilian, Naor and Safra \cite{DFKNS92}, whose soundness error is
constant; Kilian, Petrank and Tardos \cite{KPT97} gave the first with polynomially small soundness error.
Ishai and Weiss \cite{IW14} gave zero-knowledge PCPs of proximity. Ishai, Weiss and Yang \cite{IWY16} and Hazay,
Venkitasubramaniam and Weiss \cite{HVW22} gave non-adaptive constructions from leakage-resilient circuits and secret
sharing; Weiss \cite{Wei22} surveys this line of work. Goyal, Ishai, Mahmoody and Sahai \cite{GIMS10} related
zero-knowledge PCPs to interactive locking and unconditional cryptography, and Ishai, Mahmoody, Sahai and Xiao
\cite{IMSX15} studied their limitations. Mahmoody and Xiao \cite{MX13} showed that languages with efficient
zero-knowledge PCPs are in SZK. Perfect zero knowledge with an efficient simulator and constant queries is due to Gur,
O'Connor and Spooner \cite{GOS25} (see also \cite{OCo26}). This construction builds on earlier work
\cite{BCFGRS17,CFS17,CFGS22} and on the perfect zero-knowledge PCPs for $\#\mathrm P$ of the same authors
\cite{GOS24}, which make polynomially many queries. Zeilberger, O'Connor, Gur and Spooner \cite{ZOGS26} obtained
quasilinear zero-knowledge PCPs, of length $\tilde O(T+b^2)$, the first whose length does not grow multiplicatively
with the query bound. Ron-Zewi and Weiss \cite{RZW26} relate zero-knowledge codes to quantum CSS codes.

Bafna and Vyas \cite{BV26} construct \emph{private} PCPs from product expansion of tensor Reed--Solomon codes on
multiplicative grids of coprime orders: for $k=2$ unconditionally (and for larger $k$ under their product-expansion
conjecture) they obtain PCPs of quasilinear length $N$ over an alphabet of polynomial size, with $O(kN^{1/k})$
queries. In these PCPs, for every assignment $w$, satisfying or not, the honest proof restricted to any set of
positions of measure less than $\nu\ge\exp(-O(k^3+k^2\log(1/\rho)))/\log n$ has a law independent of $w$ ($\rho$ is
the product-expansion constant), where the measure is the one induced by the verifier's queries
\cite[Def.~8.2, Cor.~8.3]{BV26}. This measure is far from uniform (the output wire alone has measure at least
$\frac13$ \cite[Remark~8.4]{BV26}), so this guarantee is not a bound on the number of symbols a view may contain: a
view of the output wire alone is not covered. They also give a separate PCP of length at most $N$ that is perfect
zero knowledge, a guarantee for satisfying assignments only, against fewer than
$N/(\log n\cdot\exp(O(k^3+k^2\log(1/\rho))))$ queries, with no efficient simulator \cite[Cor.~8.6]{BV26}. In their
words, the construction ``also implies perfect-zero-knowledge PCPs with the same asymptotic parameters although we
do not obtain an efficient simulator'' \cite[\S1]{BV26}.

\paragraph{Zero knowledge in interactive models.} In the interactive-oracle-proof model, quasilinear-size zero knowledge
is due to Ben-Sasson, Chiesa, Gabizon and Virza \cite{BCGV16}. Ben-Sasson, Chiesa, Forbes, Gabizon, Riabzev and
Spooner \cite{BCFGRS17} developed the succinct-constraint-detection method, including the simulation of random and masked codewords
\cite[Lemma~4.3, Claim~7.7]{BCFGRS17} that our leak linear algebra (Lemma~\ref{main:lem:DL}) adapts. Ron-Zewi and Weiss \cite{RW24} construct zero-knowledge IOPs approaching witness
length, and Vashishth and Weiss \cite{VW26} prove composition theorems for black-box straight-line simulators, whose paradigm we follow.
Zero-knowledge sumcheck and the sumcheck commitment are due to Chiesa, Forbes and Spooner \cite{CFS17} and to Chiesa,
Forbes, Gur and Spooner \cite{CFGS22}.

\paragraph{Short PCPs.} Short PCPs were developed by Polishchuk and Spielman \cite{PS94} (nearly-linear size) and by
Ben-Sasson, Goldreich, Harsha, Sudan and Vadhan \cite{BGHSV05,BGHSV06}, who also introduced PCPs of proximity and robust
composition. The Reed--Solomon PCPP of Ben-Sasson and Sudan \cite{BS08} with Dinur's gap amplification \cite{Din07}
gives constant queries and length $T\polylog T$. Anshu, Breuckmann, Golowich, Nguyen and Vazirani \cite{ABGNV26} give a
different construction of PCPs for NP with polylogarithmic query complexity, from classical fault tolerance; their PCPs
are not zero knowledge.

\paragraph{Product expansion.} Product expansion was introduced by Panteleev and Kalachev \cite{PK22} (see also
\cite{KP22}) in the construction of good quantum LDPC codes. Golowich and Guruswami \cite[Lemma~4.9]{GG25} and Bafna
and Vyas \cite[\S2.1]{BV26} connect it to privacy. Collective product expansion and monotone closure arguments are in
\cite{DLV24,KP25} and \cite[\S4.3]{BV26}. Product expansion for codes on multiplicative grids of coprime orders,
through the Corvaja--Zannier bound, is the subject of \cite[\S4.4]{BV26}. Our proof of the locator lemma
(Theorem~\ref{kl:thm:KL}) follows the structure of the divisibility argument there, and we re-prove every step we
need in Section~\ref{sec:keylemma}, from \cite{PS94} and \cite{CZ13}. Random codes over large fields are treated in
\cite{KP25,CRS26}, and recent constructions of good quantum locally testable codes use product expansion in other
forms \cite{GGJ26,CHLT26,BLN26} (the last conditionally on a variant of a product-expansion conjecture of
\cite{BV26}). Apart from \cite{BV26}, none of these works constructs zero-knowledge PCPs.

\paragraph{Polynomial equations in subgroups.} We use the bound of Corvaja and Zannier \cite[Thm.~2]{CZ13} on points of
curves in products of cosets of multiplicative groups directly. Konyagin, Makarychev, Shparlinski and Vyugin
\cite[Thm.~1.2]{KMSV17} gave a bound for the number of solutions of polynomial equations in cosets of multiplicative
subgroups that generalises \cite{CZ13}; a version with a weaker constant appears as \cite[Thm.~1.2]{KSV22}. Related bounds are in \cite{MV19,Vyu19}.
The bounds of \cite{KMSV17,KSV22,MV19}, obtained by Stepanov's method, assume that the subgroup is small compared with
the characteristic (order $\ll p^{3/4}$) and use one subgroup on both axes. Our nodes only guarantee $m<p$ and use two
subgroups of different orders; this is why we use \cite{CZ13}.

\paragraph{Credits.} Table~\ref{tab:credits} lists every ingredient that is not ours, with the place where we use it.
Results of other authors that we re-prove are marked ``re-proved''; we do not present any of them as ours.

\begin{longtable}{>{\raggedright\arraybackslash}p{5.4cm}>{\raggedright\arraybackslash}p{5.6cm}>{\raggedright\arraybackslash}p{4.2cm}}
\caption{Credits.}\label{tab:credits}\\
\toprule
ingredient & source & use\\\midrule\endfirsthead
\toprule ingredient & source & use\\\midrule\endhead
\bottomrule\endfoot
masking by a random codeword; $t$-wise independence of the masked witness & \cite{BCGV16,BCFGRS17} (Claim~2.4 of
\cite{BCGV16}); two-point form: random self-reducibility \cite{GMR85,TW87} & shell; every node\\
leak linear algebra $\Lambda(V)=\spn\Phi\cap\spn\Psi$ & elementary; cf.\ \cite[Lemma~4.3, Claim~7.7]{BCFGRS17} & re-proved, \S\ref{sec:main}\\
the outer shell (randomisable constraint system, RS PCPP, lazy sampler) & \cite[Thms.~8.1, 8.9]{BCFGRS17};
\cite[\S8, Thm.~8.8]{ZOGS26} & re-proved in PCP form, \S\ref{sec:main}\\
lazy-sampler composition & \cite[Thms.~8.1, 8.2]{VW26} & re-proved, \S\ref{sec:main}\\
group-preserving constraint systems & \cite[Def.~7.2, Lemma~7.5]{BCGV16} & \S\ref{sec:main}\\
Bene\v{s}-network arithmetisation; routing in PCPs & \cite[\S8, App.~A]{BCGT13r}; \cite{BFLS91,PS94};
\cite{Ben64,Wak68} & routine multiplicative variant, \S\ref{sec:main}\\
quasilinear Cook--Levin; bounded occurrences & \cite{Coo88} (via \cite{PF79}); \cite{Tov84} & \S\ref{sec:main}\\
bivariate lifts, recursive RS PCPP, degree adjustment & \cite[\S\S6--7]{BS08} (multiplicative smooth case:
\cite[Thm.~3.4, Remark~3.5]{BS08}) & \S\ref{sec:main}\\
bivariate low-degree test & \cite[Lemma~6.13]{BS08} with $c_0=128$, from \cite{PS94}, \cite[Thm.~4.2.19]{Spi95} &
soundness, \S\ref{sec:main}\\
least prime in an arithmetic progression; bounds on $\pi(x)$; deterministic primality test & Linnik \cite{Lin44}
(any absolute exponent, e.g.\ \cite{Xyl11}); \cite{RS62}; \cite{AKS04} & node orders and field rule, \S\ref{sec:main}\\
CRT grid $\mu_m\cong\mu_{m_1}\times\mu_{m_2}$ & Good \cite{Goo58}; \cite{BW65,Imai77} & \S\ref{sec:main}\\
divisibility argument for product expansion & \cite[Lemma~8, Thm.~9]{PS94}; for multiplicative groups of coprime
orders \cite[\S4.4: Thms.~4.11, 4.12, Lemma~4.13]{BV26} (we follow this argument and re-prove every step we use) & \S\ref{sec:keylemma}\\
collective product expansion; monotone closure locator & \cite[Lemma~5.13, Prop.~5.14]{DLV24}; \cite[Lemmas~17--18]{KP25};
\cite[Def.~4.7, Lemma~4.8]{BV26} & our single-curve argument is an alternative, \S\ref{sec:keylemma}\\
product expansion; product expansion implies privacy & \cite{PK22}; \cite[Lemma~4.9]{GG25} & context\\
points of curves on products of cosets & \cite[Thm.~2]{CZ13} (cited directly); for unions of cosets see
\cite[Thm.~1.2]{KMSV17} and \cite[Thm.~1.2]{KSV22}; see also \cite{MV19,Vyu19} & \S\ref{sec:keylemma}\\
small circuits on $\mathbb P^1\times\mathbb P^1$ & Cayley--Bacharach residuation \cite{EGH96}; for equal bidegrees
\cite[Thm.~3.8, \S4.2]{Cou12} via the Segre embedding & small-view locator, \S\ref{sec:keylemma}\\
Riemann--Hurwitz & classical & Lemma~\ref{kl:lem:CZ}, Lemma~\ref{prm:lem-DC}\\
zero-knowledge-preserving gap amplification & idea and sketch \cite[App.~A, Lemma~A.1]{ZOGS26}; gap amplification \cite{Din07};
long-code test \cite{Has01} (as reworked in \cite[App.~B]{Din07}); folding \cite{BGS98,Has01}; expander replacement
\cite{PY91}; Tseitin encoding \cite{Tse68}; a gap in \cite[\S9]{Din07}, which we do not use, is bridged in \cite{GM20} &
rigorous statement and proof, \S\ref{sec:thmD}\\
local computability & \cite[Def.~3.1, Lemma~3.2]{GOS25} & \S\ref{sec:thmD}\\
virtual-column and diagonal attacks & \cite[Lemma~4.15]{RW24}; \cite[\S2]{BV26} & \S\ref{sec:negatives}\\
Berlekamp--Massey, BCH bound, interpolation & \cite{Mas69}; textbook
\cite{MS77} & \S\ref{sec:negatives}\\
\end{longtable}

\endgroup

\appendix

\begingroup
\providecommand{\F}{\mathbb{F}}
\providecommand{\RS}{\mathrm{RS}}
\providecommand{\prmlg}{\log_2}
\section{Parameters and constants}\label{app:params}

This appendix collects the parameters of the construction. Section~\ref{prm:sec-table} lists them;
Section~\ref{prm:sec-exp} gives the length bookkeeping of Theorem~\ref{main:thm:main} in detail
(Proposition~\ref{prm:prop-exp}); Section~\ref{prm:sec-KL} records a further reduction of the locator
constant, which is not used elsewhere in the paper; Section~\ref{prm:sec-read} comments on the size of the
constants. Neither Proposition~\ref{prm:prop-exp} nor Section~\ref{prm:sec-KL} changes the form of the main theorem.

\subsection{The parameter table}\label{prm:sec-table}

Notation: $N:=T+b$; $\ell:=|\hat H|=8n$ is the size of the root domain of the shell PCP $(P_A,V_A)$
(Theorem~\ref{main:thm:A}) to which gap amplification is applied; $d$ is the depth of the factor tree; $\eta$ the
number of labels of the routing instance; $s_A,q_A,\gamma_A,U_A$ the residual-circuit size, query count, gap and coin
multiset of $(P_A,V_A)$; $\bar\kappa$ the amortised read factor; $a$ the constant of the locator lemma
(Theorem~\ref{kl:thm:KL}). Dinur's constants $d_{\mathcal D}$ (degree after preprocessing), $t_{\mathcal D}$ (powering
parameter), $C_{\mathcal D}$ (size factor per round) and $\alpha_{\mathcal D}$ (gap ceiling) are fixed numbers of one explicit
instantiation of \cite{Din07}; their values are known only by lower bounds on what its analysis, with the repairs
of Section~\ref{sec:thmD}, gives.

\begin{center}
\small
\begin{tabular}{@{}p{3.6cm}p{6.6cm}p{4.9cm}@{}}
\toprule
quantity & exact formula & value / source\\\midrule
\multicolumn{3}{@{}l}{\emph{Final PCP (Theorem~\ref{main:thm:main})}}\\
alphabet, adaptivity & binary; non-adaptive verifier & \\
queries $Q_{\rm fin}$ & $3Q_D=18\lceil1/\alpha_{\mathcal D}\rceil$ & $\ge3\cdot10^{13}$ (Remark~\ref{thmD:rem:const})\\
completeness, soundness error & perfect; $\le(3/4)^3=27/64$ & Lemma~\ref{main:lem:CC} (three runs)\\
randomness & $3(\lceil\prmlg|U'|\rceil+1)$, $|U'|=|E(G_k)|^{\lceil1/\alpha_{\mathcal D}\rceil}$ & $\le\frac{Q_{\rm fin}}6\prmlg N\,(1+o(1))$\\
zero knowledge & perfect in expected polynomial time, adaptive malicious, $\le b$ distinct bit queries & universal straight-line simulator\\
length & $\le c'_{\rm len}\,N(\ln N)^{C_{\rm exact}}(\ln\ln N)^{2E_L}$, and $\le N(\prmlg N)^{C_{\rm exact}+1}$ for $N\ge N_{\rm thr}$ & Prop.~\ref{prm:prop-exp}\\
length exponent $C_{\rm exact}$ & $27+\prmlg80+\prmlg a+(15+\prmlg2660)(E_C+E_L)$ & $=120.692+26.3772(E_C+E_L)$\\
\midrule
\multicolumn{3}{@{}l}{\emph{Zero-knowledge-preserving gap amplification (Theorem~\ref{thmD:thm})}}\\
$c_{\rm it}$, $E_L$, $E_C$ & $2d_{\mathcal D}^{\lceil t_{\mathcal D}/2\rceil}$, $\prmlg c_{\rm it}$, $\prmlg C_{\mathcal D}$ & $E_L\ge t_{\mathcal D}+1$ ($d_{\mathcal D}\ge4$)\\
$t_{\mathcal D}$, $E_L$, $E_C$ (lower bounds) & --- & $t_{\mathcal D}\ge9\cdot2^{42}$; $E_L\ge3.958\cdot10^{13}$; $E_C\ge2^{6\cdot2^{t_{\mathcal D}}}$ (Remark~\ref{thmD:rem:const})\\
rounds $k$; $m^*$ & $\lceil\prmlg(3m^*/\hat\gamma)\rceil$; $m^*\le4s+1$ & $\hat\gamma=\max\{\gamma,1/|U|\}$\\
locality $L_D$ & $q\,c_{\rm it}^{\,k}\le q\,c_{\rm it}(3(4s+1)/\gamma)^{E_L}$ & ZK bound $\lfloor b_0/L_D\rfloor$\\
length & $\le27\,m^*|U|\,C_{\mathcal D}^{\,k}\le27(4s+1)|U|C_{\mathcal D}(3(4s+1)/\gamma)^{E_C}$ & factor 27\\
$Q_D$; DAG size & $6\lceil1/\alpha_{\mathcal D}\rceil$; $\le3L_D$ nodes per output bit & simulator time only\\
\midrule
\multicolumn{3}{@{}l}{\emph{The shell PCP $(P_A,V_A)$ before gap amplification (Theorem~\ref{main:thm:A})}}\\
BS08 constants & $c_0=128$, $C_{\rm biv}=10+2c_0=266$ & \cite[Lemma~6.13, \S6.6]{BS08}\\
soundness per level & $\varrho_h=(10C_{\rm biv})^{-h}=2660^{-h}$; root $\frac38\varrho_d$ & $\prmlg2660=11.37721\ldots$\\
depth $d$ & $\le\prmlg\ln\mathsf N\le\prmlg\ln\ell$ & Lemma~\ref{main:lem:SM2}\\
node proof length & $\frac{16}{29}(30^d-1)s$, own symbols $128m$ & $\prmlg30=4.90689\ldots$\\
coins of the RS PCPP & $|U(m)|\le8m\cdot80^h$; $|U_A|\le32\cdot80^d\ell$ & $\prmlg80=6.32192\ldots$\\
queries $q_A$ & $\le\max(16(\ln\ell)^6,\eta+1)$ (leaf $8p$, $p\le2\mathsf L^6$) & exponent $6$\\
residual size $s_A$ & $O(p^2\log^2q)=O((\ln\ell)^{12}\log^2q)$ & exponent $14$ (proof of Prop.~\ref{prm:prop-exp})\\
gap $\gamma_A$ & $\ge\frac{21}{2048(\eta+1)2660^d}$ (Thm.~\ref{main:thm:A}(ii)); Prop.~\ref{prm:prop-exp} uses the weaker $\frac{7}{288(\eta+1)\lambda_s}$, $\lambda_s=\frac83\,2660^d$ & exponent of $1/\gamma_A$: $1+\prmlg2660$; $\tau_0\ge\frac{63}{64}$\\
labels $\eta$ & $2\beta+8\le2\prmlg\ell+8$ & Lemma~\ref{main:lem:MR}\\
ZK bound of $(P_A,V_A)$ & $b_A=\lfloor t/(\eta\bar\kappa')\rfloor$, $\bar\kappa'=\max(1,\bar\kappa)$ & \\
$\bar\kappa$ & product of per-node constants, $\le a^d\le(\ln\ell)^{\prmlg a}$ & Lemma~\ref{main:lem:C}\\
size chain & $2^{13}(\ln Y)^6Y<\mathsf N<(2^{13}(\ln Y)^6+3)Y$, $\ell<2^{18}Y(\ln Y)^6$, $Y=\max(N_{\rm pts},t)$ & $N_{\rm pts}=O(T\log^2T)$\\
fixed point & $t=\lceil b\,\Xi(Y_0^2)\rceil$, $\Xi\ge2L_D\eta\bar\kappa'$ & exponent of $\Xi$: Prop.~\ref{prm:prop-exp}\\
field & $q$ least prime $\equiv1\pmod{32n}$ & $q\le c_{\rm Lin}(32n)^{L_{\rm Lin}}$ (Lemma~\ref{main:lem:field})\\
\midrule
\multicolumn{3}{@{}l}{\emph{Locator lemma (Theorem~\ref{kl:thm:KL}), same-radix configuration}}\\
node shape & $L_1=14$ column cosets, $L_2\le8$ row cosets, $m_1<m_2<2m_1$ ($K=2$) & recursion nodes: $m_1>2455$ (Remark~\ref{main:rem:2455})\\
constant $c_*$ & $c_*:=10^{-13}$, admissible (Definition~\ref{kl:def:adm}) & used throughout the paper\\
$\tau$, $a$ & $\tau=\lfloor c_*m_1\rfloor$; $a=a_*\ge a(c_*)=1+\frac1{1-c_*}+K/c_*^2$ & $a\le2\cdot10^{26}+3$, $\prmlg a=87.3701$\\
sharpening & $c_*'=3.58\cdot10^{-11}$, admissible; $a(c_*')$ & $\prmlg a(c_*')=70.4025$ (Remark~\ref{kl:rem:const})\\
further reduction & $\tilde c=1/23{,}400{,}000$; $\tilde a=1+\frac1{1-\tilde c}+K/\tilde c^{\,2}$ & $\tilde a<1.0952\cdot10^{15}$, $\prmlg\tilde a<49.961$ (Remark~\ref{prm:prop-KLsharp})\\
non-trivial regime & views with $|V|<(\tau+1)^2$; $\tau\ge1$ iff $m_1\ge1/c$ & $m_1\ge10^{13}$ (resp.\ $2.79\cdot10^{10}$, $2.34\cdot10^7$)\\
\bottomrule
\end{tabular}
\end{center}

\noindent The locator constant $c_*:=10^{-13}$ is the single value used throughout the paper. It is admissible: for
$(L_1,L_2,K)=(14,8,2)$ condition~\eqref{kl:eq:adm} reads $112(3c^{1/3}+6c)<\frac19-2c$, and at $c=10^{-13}$ the left
side is below $0.0156$. The value $c_*'=3.58\cdot10^{-11}$ of Remark~\ref{kl:rem:const} is the sharpening that the
locator lemma proves; it changes only the additive constant of $C_{\rm exact}$. The value
$C_{\rm exact}=120.692\ldots+26.3772\ldots(E_C+E_L)$ uses $a=a_*$ (the construction's value; $\prmlg a_*$ and
$\prmlg a(c_*)$ agree to $27$ digits); with $a(c_*')$ it becomes
$103.724\ldots+26.3772\ldots(E_C+E_L)$, and with Remark~\ref{prm:prop-KLsharp} it becomes
$83.281\ldots+26.3772\ldots(E_C+E_L)$. The brackets $\prmlg2660\in[11.37721,11.37722]$,
$\prmlg80\in[6.32192,6.32193]$, $\prmlg a\in[87.3701,87.3702]$, $\prmlg a(c_*')\in[70.4024,70.4025]$ and
$\prmlg\tilde a\in[49.9600,49.9601]$ are certified by integer comparisons such as
$2^{1137721}\le2660^{10^5}<2^{1137722}$.

\subsection{The length exponent}\label{prm:sec-exp}

Theorem~\ref{main:thm:main} states the exact length exponent, and Step~5 of its proof collects it briefly. This
subsection carries out the bookkeeping quantity by quantity. It also shows that the locator constant enters the
exponent only additively.

\begin{proposition}[explicit length exponent]\label{prm:prop-exp}
Under the hypotheses of Theorem~\ref{main:thm:main} (locator lemma with an absolute constant $a\ge1$), the final PCP
has length at most $c'_{\rm len}\,(T+b)(\ln(T+b))^{C_{\rm exact}}(\ln\ln(T+b))^{2E_L}=(T+b)(\ln(T+b))^{C_{\rm exact}+o(1)}$
for all $T+b\ge N_1$, with constants $c'_{\rm len},N_1$
that may depend on $R$, $T$ and $b$ (like the constant $c_{\rm len}$ of Theorem~\ref{main:thm:main}) but not on the
input, and with the absolute exponent
\[
C_{\rm exact}=27+\prmlg80+\prmlg a+(15+\prmlg2660)\,(E_C+E_L),
\]
where $E_C=\prmlg C_{\mathcal D}$ and $E_L=\prmlg c_{\rm it}$ are the constants of Theorem~\ref{thmD:thm}. Consequently the
length is at most $(T+b)(\prmlg(T+b))^{C_{\rm exact}+1}$ for all $T+b\ge N_{\rm thr}$. The $T$-term alone has exponent
$22+\prmlg80+(15+\prmlg2660)E_C<C_{\rm exact}$. The locator constant enters only additively, through $\prmlg a$.
\end{proposition}

\begin{proof}
All logarithms of $\ell$ are natural; changing the base, or replacing $\ln\ell$ by $\ln(T+b)$ (below), changes
only constant factors, which go into $c'_{\rm len}$. Write $f\preceq g$ if $f\le c\,g$ for an absolute $c$ and all
large $\ell$.

\emph{Step 1: quantities of the shell PCP that do not involve $a$.}
(i) \emph{Coins.} The node PCPP at the root has coin multiset $\le2\cdot80^d\ell$ (two copies), the
product-code PCPP $\le8\cdot80^d\ell$ (weights within a factor 2), and the shell $|U_A|\le4\max(8\cdot80^d\ell,\ell)
=32\cdot80^d\ell$. Since $d\le\prmlg\ln\ell$, $80^d\le(\ln\ell)^{\prmlg80}$.
(ii) \emph{Queries.} A test reads three symbols, a leaf line of $8p\le16(\ln\ell)^6$ symbols, or the $\eta+1$
symbols of the consistency test; so $q_A\preceq(\ln\ell)^6$.
(iii) \emph{Residual circuits.} The leaf test checks membership of $8p$ symbols in $\RS[D_v,p]$ by a parity-check
matrix, $O(p^2)$ multiplications in $\F$ of $O(\log^2q)$ gates each; the other tests have circuits of size
$O(\eta\log^2q)$. With $p\le2(\ln\ell)^6$ and $\log q\preceq\log\ell$ (Lemma~\ref{main:lem:field}),
$s_A\preceq(\ln\ell)^{14}$.
(iv) \emph{Gap.} The node PCPP rejects with probability $\ge\frac38\,2660^{-d}\delta$; the product-code PCPP
draws one of its two PCPPs with probability $\ge\frac13$, so it rejects with probability $\ge\Delta/(3\lambda_s)$,
$\lambda_s=\frac83\,2660^d$; the shell gives $\gamma_A\ge7/(288(\eta+1)\lambda_s)$ (this is weaker than the bound
$21/(2048(\eta+1)2660^d)$ of Theorem~\ref{main:thm:A}(ii) and suffices here). With
$2660^d\le(\ln\ell)^{\prmlg2660}$ and $\eta\le2\prmlg\ell+8$, $1/\gamma_A\preceq(\ln\ell)^{1+\prmlg2660}$.
None of (i)--(iv) depends on $a$: the locator lemma enters the shell PCP only through the sampler's read factor
$\bar\kappa$, which bounds the zero-knowledge parameter $b_A=\lfloor t/(\eta\bar\kappa')\rfloor$ and nothing
else.

\emph{Step 2: gap amplification.} By (iii) and (iv),
$3(4s_A+1)/\gamma_A\preceq(\ln\ell)^{e}$ with $e:=14+1+\prmlg2660=15+\prmlg2660$. Theorem~\ref{thmD:thm} gives
\[
L_D\le q_Ac_{\rm it}\big(3(4s_A+1)/\gamma_A\big)^{E_L}\preceq(\ln\ell)^{6+eE_L}
\]
up to the constant factor $c_{\rm it}\cdot(\text{const})^{E_L}$, and length
\[
\ell'\le27(4s_A+1)|U_A|\,C_{\mathcal D}\big(3(4s_A+1)/\gamma_A\big)^{E_C}\preceq\ell\,(\ln\ell)^{14+\prmlg80+eE_C}.
\]
Lemma~\ref{main:lem:CC} (three runs of a fair-coin verifier) does not change the proof.

\emph{Step 3: the fixed point.} Since $\bar\kappa\le a^d\le(\ln\ell)^{\prmlg a}$, the function
\[
\Xi_\sigma(\ell):=2\cdot16(\ln\ell)^6\,c_{\rm it}\Big(3\big(4\sigma(\ln\ell)^{14}+1\big)\tfrac{768}{7}(2\prmlg\ell+9)
(\ln\ell)^{\prmlg2660}\Big)^{E_L}(2\prmlg\ell+8)(\ln\ell)^{\prmlg a}
\]
(with $\sigma$ the absolute constant of (iii)) is nondecreasing and $\ge2L_D\eta\bar\kappa'$, and
$\Xi_\sigma(\ell)\preceq(\ln\ell)^{e_\Xi}$ with $e_\Xi=7+\prmlg a+eE_L$. The function $\Xi$ in the proof of
Theorem~\ref{main:thm:main}, which fixes $t$, must be computable, so it uses the explicit envelope
$\bar s=C_s(\ln\ell)^{14}(\ln\ln\ell)^2$ in place of $\sigma(\ln\ell)^{14}$ ($\sigma$ depends on the constants of
Lemma~\ref{main:lem:field}) and the constant $\frac{2048}{21}$ of Theorem~\ref{main:thm:A}(ii) in place of $\frac{768}7$. For a locator constant $a$ it uses $\bar\kappa(\ell)=(\ln\ell)^{\prmlg a}$; the proof
of Theorem~\ref{main:thm:main} takes $a=a_*$, which is valid as $a(c_*)\le a_*$, and
$\prmlg a_*-\prmlg a(c_*)<10^{-27}$.
Comparing factor by factor, $\Xi(\ell)\preceq\Xi_\sigma(\ell)\,(c\,(\ln\ln\ell)^2)^{E_L}$ for an absolute $c$, so
$\Xi(\ell)\preceq(\ln\ell)^{e_\Xi}(\ln\ln\ell)^{2E_L}$. Theorem~\ref{main:thm:main} takes
$t=\lceil b\,\Xi(Y_0^2)\rceil$ with $Y_0=\max(N_{\rm pts},b,|x|)$, $Y=\max(N_{\rm pts},t)$ and
$\ell<2^{18}Y(\ln Y)^6$; above its threshold $\ell\le Y_0^2$, so $\ln Y\le\ln\ell\le2\ln Y_0$. Since
$N_{\rm pts}=O(T\log^2T)$ (with a constant depending on $R$ and $T$) and $Y_0\le N_{\rm pts}+b$, also
$\ln Y_0\le2\ln(T+b)$ for large inputs.

\emph{Step 4: assembly.} Combining, with $N=T+b$ and constant factors absorbed,
\[
\ell'\preceq\big(T(\ln N)^{2}+b(\ln N)^{e_\Xi}(\ln\ln N)^{2E_L}\big)(\ln N)^{6}\,(\ln N)^{14+\prmlg80+eE_C}.
\]
The exponent of $\ln N$ in the $b$-term is $6+e_\Xi+14+\prmlg80+eE_C=27+\prmlg80+\prmlg a+e(E_C+E_L)=C_{\rm exact}$;
that of the $T$-term is $22+\prmlg80+eE_C$, smaller by $5+\prmlg a+eE_L>0$. For the second form,
$c'_{\rm len}(\ln\ln N)^{2E_L}$ is at most $\prmlg N$ once $N\ge N_{\rm thr}:=2^{k_0}$ for a suitable constant
$k_0$; finitely many
smaller inputs are handled by brute force as in Theorem~\ref{main:thm:main}.
\end{proof}

\begin{remark}[size of the constants]\label{prm:rem-size}
$e=15+\prmlg2660=26.37721\ldots$. With $E_L\ge3.958\cdot10^{13}$ (Remark~\ref{thmD:rem:const}) alone,
$C_{\rm exact}\ge1.04\cdot10^{15}$; with $E_C\ge2^{6\cdot2^{t_{\mathcal D}}}$, $C_{\rm exact}$ is doubly exponential in $t_{\mathcal D}$. The
locator lemma's contribution $\prmlg a\approx87.4$ (or $70.4$ with $c_*'$, or $50.0$ after
Remark~\ref{prm:prop-KLsharp}) is negligible beside $e(E_C+E_L)$. The threshold $N_{\rm thr}$ is likewise galactic:
the fixed point needs $\prmlg Y_0\gtrsim e_\Xi\prmlg e_\Xi$. The lower bounds on $t_{\mathcal D}$, $E_L$, $E_C$ and
$Q_D$ are our readings of the analysis of \cite{Din07} with the repairs of Section~\ref{sec:thmD}
(Remark~\ref{thmD:rem:const}); they are not limits of
gap amplification, and they are not re-derived here. The only implication used here is $E_L=1+\lceil t_{\mathcal D}/2\rceil\prmlg d_{\mathcal D}\ge1+t_{\mathcal D}$
for $d_{\mathcal D}\ge4$, which holds because in the $d_{\mathcal D}$-regular graph $H_1$ every vertex has its constraint
edge and its self-loop, and the expanders added by $\mathrm{prep}_1$ and $\mathrm{prep}_2$ have positive degree
(Section~\ref{thmD:sec:round}; \cite[Defs.~4.1, 4.2]{Din07}).
\end{remark}

\subsection{A further reduction of the locator constant}\label{prm:sec-KL}

This subsection is not used elsewhere in the paper; we record it as a remark on how far the locator constant can be
lowered by a second application of \cite[Theorem~2]{CZ13}. Notation is that of Section~\ref{sec:keylemma}: a node
with parameters $(L_1,L_2,K)$ (Definition~\ref{kl:def:node}, the same-radix configuration) over $\F$ of
characteristic $p>m$, with $n:=m_1<n':=m_2$, $m=nn'<Kn^2$, column set $A$ a union of $L_1$ cosets of $\mu_n$, rows
$B=\bigsqcup_{i<L_2}f_i\mu_{n'}$ (distinct cosets, so $|F_b|=n$); $K_0=\overline{\mathbb F}_p$; for an irreducible
$F\in K_0[X,Y]$ of bidegree $(a,b')$ with $a,b'\ge1$, $\mathcal X_F$ is the normalisation of its closure in
$\mathbb P^1\times\mathbb P^1$, $o_x(\nu)=\mathrm{ord}_\nu(x-x(\nu))$, $o_y(\nu)=\mathrm{ord}_\nu(y-y(\nu))$, and
$\Gamma_{A\times B}(F)$ and $\Delta_F$ are as in Section~\ref{kl:sec:cz} and Lemma~\ref{kl:lem:div} ($\Delta_F$ sums
$\min(o_x,o_y)$ over the places above the view's domain points $(\xi,\xi^n)$ on $F$). Call $F$ \emph{binomial} if
it is a unit multiple of $X^aY^{b'}-\lambda$ or of $X^a-\lambda Y^{b'}$ with $\lambda\in K_0^\times$.

\begin{lemma}[incidences with the domain curve]\label{prm:lem-DC}
Let $0<c\le\frac18$, let $F$ be irreducible, mixed and not binomial, with $a,b'\le cn$, and let $f\in K_0^\times$.
Let $\Delta_f(F)$ be the sum of $\min(o_x(\nu),o_y(\nu))$ over the places $\nu$ of $\mathcal X_F$ with $x(\nu)\in
K_0^\times$, $y(\nu)=x(\nu)^n$ and $y(\nu)\in f\mu_{n'}$. Then
\[
\Delta_f(F)\le3\cdot2^{2/3}(1+c)^{1/3}\,m^{1/3}(ab')^{2/3}.
\]
In particular $\Delta_F\le L_2\cdot3\cdot2^{2/3}(1+c)^{1/3}m^{1/3}(ab')^{2/3}$.
\end{lemma}

\begin{proof}
Put $u:=(y/f)^{n'}$ and $v:=y\,x^{-n}$ in $K_0(\mathcal X_F)$.

\emph{Local orders.} Let $\nu$ be a place in the sum, $\xi=x(\nu)$, $\beta=y(\nu)=\xi^n$. As $p\nmid n'$, the
polynomial $T^{n'}-f^{n'}$ has simple roots, one of them $\beta$, so $\nu(1-u)=\nu(y-\beta)=o_y(\nu)$. As
$\xi\ne0,\infty$, $\nu(1-v)=\nu(x^n-y)$, and $x^n-y=(x^n-\beta)-(y-\beta)$ with $\nu(x^n-\beta)=\nu(x-\xi)=o_x(\nu)$
($p\nmid n$); so $\nu(1-v)\ge\min(o_x(\nu),o_y(\nu))$. Hence $\min(o_x,o_y)\le\min(\nu(1-u),\nu(1-v))$, and $\nu$
is neither a zero nor a pole of $x$ or $y$, hence not of $u$ or $v$.

\emph{Hypotheses of \cite[Theorem~2]{CZ13}.} (a) \emph{Independence.} If $u^rv^s\in K_0^\times$ with
$(r,s)\ne(0,0)$, then $y^{n'r+s}x^{-ns}\in K_0^\times$. If $s=0$ then $y$ is constant; if $n'r+s=0$ then $x$ is
constant; both are impossible for a mixed curve. Otherwise $x^\alpha y^\beta\in K_0^\times$ with
$\alpha,\beta\ne0$, and the argument of the dependent case in the proof of Lemma~\ref{kl:lem:pair} shows that $F$ is
binomial, which is excluded. (b) \emph{Differentials.} $du=n'u\,dy/y\ne0$, since $dy=0$ would make $y$ a
$p$-th power and $p\mid\deg y=a$, while $a\le cn<p$. The function $v$ is not constant (otherwise $F$ divides
$Y-\lambda X^n$, which is irreducible of $X$-degree $n>a$), and $\deg v\le\deg y+n\deg x=a+nb'\le cn(n+1)<nn'<p$;
so $p\nmid\deg v$, $v$ is not a $p$-th power, and $dv\ne0$. (c) \emph{Degrees and Euler characteristic.}
$\deg u=n'a$; the zeros and poles of $u$ and $v$ are among those of $x$ and $y$, at most $2a+2b'$ places; the
genus is at most the arithmetic genus $(a-1)(b'-1)$; so $\chi\le2ab'$.

\emph{Applying the theorem.} Over a finite field of definition, with the same transfer as in the proof of
Lemma~\ref{kl:lem:pair} (descent to the curve with function field $\kappa(u,v)$ if needed, as in the discussion after
Lemma~\ref{kl:lem:CZ}: the left side grows at most by the index, the right side at least by it, by
Riemann--Hurwitz), \cite[Theorem~2]{CZ13} gives
\[
\Delta_f(F)\le\max\Big(3\sqrt[3]2\,\big(n'a\,(a+nb')\,2ab'\big)^{1/3},\ \frac{12\,n'a\,(a+nb')}{p}\Big).
\]
Since $n'a^2b'(a+nb')=m(ab')^2(1+a/(nb'))$ and $a/(nb')\le c$, the first term is at most
$3\cdot2^{2/3}(1+c)^{1/3}m^{1/3}(ab')^{2/3}$. As $p>m$, the second term is less than $12ab'(1+c)$, which is at most
the same bound because $ab'(1+c)^2\le c^2(1+c)^2n^2\le m/16$. The last claim sums over the $L_2$ row cosets.
\end{proof}

\begin{remark}[a smaller locator constant]\label{prm:prop-KLsharp}
Let $c\in(0,1)$ satisfy
\[
\text{(R1) }c<\frac1{4L_1L_2^2},\qquad\text{(R2) }c\le\frac18,\qquad
\text{(R3) }3\cdot2^{2/3}\big(L_1L_2+L_2(1+c)^{1/3}\big)c^{1/3}\le2(1-2c).
\]
Then the conclusion of Theorem~\ref{kl:thm:KL} holds for this $c$ (in place of an admissible one), with
$\tau:=\lfloor cn\rfloor$ in the definition of $W$ and with $a=1+\frac1{1-c}+K/c^2$. For $L_1=14$, $L_2=8$ the
supremum of the $c$ satisfying (R1)--(R3) lies in $[4.28669\cdot10^{-8},4.28670\cdot10^{-8}]$;
$\tilde c=1/23{,}400{,}000$ satisfies them, and for $K=2$ it gives $\tilde a<1.0952\cdot10^{15}$,
$49.9600\le\prmlg\tilde a<49.9601$ (against $\prmlg a=87.370$ for $c_*$ and $\prmlg a(c_*')=70.402$).
\end{remark}

\begin{proof}
The proof of Theorem~\ref{kl:thm:KL} uses the constant only through the divisibility lemma
(Lemma~\ref{kl:lem:div}: ``$E\mid P$ whenever $E\ne0$ has $\deg_XE,\deg_YE\le t$''), applied with $t=\tau$. We show
that divisibility holds for every $t\le cn$. The argument for pure factors is unchanged; for a mixed irreducible
factor $F\nmid P$ of bidegree $(a,b')$, $1\le a,b'\le t$, the zero count in the proof of Lemma~\ref{kl:lem:div}
gives, before any symmetrisation,
\[
\Gamma_{A\times B}(F)+\Delta_F\ \ge\ (n+1-t)b'+(n'+1-t)a-2ab'\ \ge\ (n-2t)b'+(n'-2t)a,\tag{$\ast$}
\]
using $e_X,e_Y\le t$ and $2ab'\le t(a+b')$. We show that the left side is smaller, a contradiction.

\emph{Case B ($F$ binomial).} Each coset pair has the dependent-case bound $\le ab'$, so
$\Gamma_{A\times B}(F)\le L_1L_2ab'\le\frac12L_1L_2t(a+b')$. Claim~\ref{kl:clm:delta}(c) (valid as
$t\le cn\le n/(16L_2)$) gives $(n+1-2t)b'+(n'+1-2t)a-\Delta_F\ge\frac n{8L_2}(a+b')$. By (R1),
$\frac12L_1L_2t<\frac n{8L_2}$, so $\Gamma_{A\times B}(F)+\Delta_F<(n+1-2t)b'+(n'+1-2t)a$, which is at most the middle
expression of $(\ast)$.

\emph{Case N ($F$ not binomial).} For each of the $L_1L_2$ coset pairs, \cite[Theorem~2]{CZ13} (as in the independent
case of the proof of Lemma~\ref{kl:lem:pair}) gives the bound
$\max(3\cdot2^{2/3}m^{1/3}(ab')^{2/3},12mab'/p)=3\cdot2^{2/3}m^{1/3}(ab')^{2/3}$, because
$12ab'\le3\cdot2^{2/3}m^{1/3}(ab')^{2/3}$ when $ab'\le m/16$. With Lemma~\ref{prm:lem-DC},
\[
\Gamma_{A\times B}(F)+\Delta_F\le\Omega\,m^{1/3}(ab')^{2/3},\qquad\Omega:=3\cdot2^{2/3}\big(L_1L_2+L_2(1+c)^{1/3}\big).
\]
Now $(ab')^{2/3}\le(ab')^{1/2}t^{1/3}\le(ab')^{1/2}(cn)^{1/3}$, and, by the arithmetic--geometric mean
inequality and $t\le cn<cn'$, the right side of $(\ast)$ is at least
$2\sqrt{(n-2t)(n'-2t)ab'}\ge2(1-2c)(nn')^{1/2}(ab')^{1/2}$. So it suffices that
$\Omega\,c^{1/3}n^{1/3}(nn')^{1/3}<2(1-2c)(nn')^{1/2}$, i.e.\ $\Omega c^{1/3}n^{1/3}<2(1-2c)(nn')^{1/6}$; this
follows from (R3) because $(nn')^{1/6}>n^{1/3}$.

So $F\mid P$ in both cases, and the induction in the proof of Lemma~\ref{kl:lem:div} goes through ($t$ does not
increase when $F$ is cancelled). The size bound: $|D_V|\le|V|$; $n|H_\tau(V)|\le n|V|/(n+1-\tau)\le|V|/(1-c)$; and if
$|V|\ge(\tau+1)^2>(cn)^2$ the coset branch costs $m<Kn^2<(K/c^2)|V|$. Monotonicity and polynomial time are as in the
proof of Theorem~\ref{kl:thm:KL}. The numerical claims are exact computations: (R3) is certified through
$(1+c)^{1/3}\le1+c/3$ and cubing, $108(112+8(1+c/3))^3c\le8(1-2c)^3$ ($7.9754\le7.99999$ at $\tilde c$); (R1), (R2)
are immediate.
\end{proof}

\begin{remark}[what the reduction does and does not do]\label{prm:rem-KL}
(i) The gain is $37.41$ bits in $\prmlg a$ against $c_*$ ($20.44$ against $c_*'$), hence the same amount in $C_{\rm exact}$,
against $e(E_C+E_L)\ge10^{15}$. (ii) The chain of inequalities in Case~N is tight at $a=b'=t$, $n'\to n$. (iii) It
cannot bring $\prmlg a$ near $\prmlg2660=11.378$: that would need $c\ge\sqrt{2/2660}\approx0.027$, while any argument
of this shape needs $\Omega c^{1/3}<2$, i.e.\ $c<4.3\cdot10^{-8}$. (iv) Balancing $t$ against the large-view branch
gains nothing: $a$ decreases in $c$, so the largest $c$ satisfying (R1)--(R3) is optimal. (v) The small-view locator
(ratio $1$ up to $m_1+\lfloor3m_2/2\rfloor-1$ reads at large nodes; Proposition~\ref{kl:prop:natural}) and the bound
$m_1>2455$ on recursion nodes (Remark~\ref{main:rem:2455}) change only the per-node factors of small nodes, which are
$O(1)$ levels of the tree; the exponent $\prmlg a$ in $\bar\kappa\le(\ln\ell)^{\prmlg a}$ is set by the large nodes,
where the coset branch costs $K/c^2$. (vi) Choosing the bidegree bound of $E$ as $(\lfloor cn\rfloor,\lfloor cn'\rfloor)$
should remove the factor $K$ (one bit); we do not pursue this. (vii) The non-trivial regime still starts at $m_1\ge2.34\cdot10^7$,
$m\ge5\cdot10^{14}$, far beyond the sizes at which the leak can be computed exactly.
\end{remark}

\subsection{How to read the constants}\label{prm:sec-read}

\emph{Where the size comes from.} In $C_{\rm exact}$ the dominant terms are $eE_C$ and $eE_L$ with $e=15+\prmlg2660$. Both
come from the constant-query step: Theorem~\ref{thmD:thm} follows the gap amplification of \cite{Din07}, and the
factor $e$ is the polylog exponent of $3(4s_A+1)/\gamma_A$ (residual size times inverse gap of the
\cite{BS08}-style PCPP). The zero-knowledge-specific contributions are $\prmlg a$ and the additive $E_L$ of the
fixed point.

\emph{Comparison with \cite{ZOGS26}.} \cite{ZOGS26} also obtains constant query complexity through Dinur's
transformation (its Lemma~A.1, a sketch: length $\ell q^{c_1}$, query bound $b/q^{c_2}$ for absolute constants
$c_1,c_2$), and its $\tilde O(T+b^2)$ hides those exponents; it does not state them. So both results inherit
galactic exponents of the same origin, and neither is a practical construction. The difference claimed by our main
theorem is the dependence on $b$ (linear rather than quadratic), not the size of the constants.

\emph{Why a feasibility result is meaningful.} Open Problem~1 of \cite{GOS25} asks whether $O(1)$-query
zero-knowledge PCPs of quasilinear length exist at all; it is a question about asymptotic exponents, like the
existence of quasilinear constant-query PCPs answered by combining \cite{BS08} with \cite{Din07}, whose printed
analysis also gives galactic constants. Our main theorem answers the existence question for every polynomial query bound, with the simulators of
Theorem~\ref{thm:intro-main}(v), and with absolute exponents given by explicit formulas in
Dinur's constants. Lowering them, e.g.\ by a constant-query composition with small constants in
place of gap amplification, is a separate question that this paper does not address.
\endgroup

\begingroup
\providecommand{\F}{\mathbb F}
\providecommand{\Fb}{\overline{\mathbb F}}
\providecommand{\RS}{\mathrm{RS}}
\providecommand{\ev}{\mathrm{ev}}
\providecommand{\spn}{\mathrm{span}}
\providecommand{\LC}{\mathrm{LC}}
\providecommand{\poly}{\mathrm{poly}}

\section{Necessity of the hypotheses}\label{sec:negatives}\label{neg:sec}

This appendix explains why the strong reading of \cite[Open Problem~1]{GOS25} needed new ideas, and which
hypotheses of the locator lemma (Theorem~\ref{kl:thm:KL}) are necessary. The locator lemma concerns two masked lifts
of the same radix on a union of cosets (the \emph{same-radix configuration}; a \emph{node} is one instance of it,
Definition~\ref{kl:def:node}). Every result below is a statement about \emph{one simulation paradigm}: simulating a
fixed view of a masked Reed--Solomon proof of proximity by reading few input values (\emph{point charging}). Every
bound is \emph{per fixed view}; since views are fixed non-adaptively, each lower bound holds a fortiori against
adaptive verifiers. \textbf{None of the results is a lower bound for zero-knowledge PCPs as such}: none shows that a
composed PCP fails to be zero knowledge, and none contradicts the locator lemma: every family below either violates
one of its hypotheses or is located by its locator (Proposition~\ref{neg:prop:consistent}). By a \emph{design law} we
mean an additional hypothesis that restricts the incidences between the grid of lift symbols and low-degree curves
(for instance, one forbidding many grid points on the graph of a M{\"o}bius map); the locator lemma uses none. Some
statements are illustrated by examples whose values were obtained by exact computer computation over the stated
fields; the examples are not used elsewhere, and no proof in this paper relies on them.

\subsection{Charges, samplers and basic lemmas}\label{neg:basic}

\paragraph{Notation.} $\RS(u)=\F[X]_{<u}$ with coefficient vectors $(c_j)_{j<u}$; a functional is a vector
$h\in\F^u$ with $h(C)=\sum_jh_jc_j$. We write $\ev_x=(x^j)_{j<u}$ and $\partial_x=(jx^{j-1})_{j<u}$
($C\mapsto C'(x)$). For $\mathrm{char}\,\F\nmid k$, $\mu_k\le\F^\times$ is the subgroup of order $k$. The radix-$k$
lift of $f\in\F[X]$ is $Q^{(k)}_f(X,Y)=\sum_{t<k}X^tf_t(Y)$, where $f=\sum_{t<k}X^tf_t(X^k)$. Thus
$Q^{(k)}_f(X,X^k)=f$, $\deg_XQ^{(k)}_f<k$, $\deg_YQ^{(k)}_f<\lceil u/k\rceil$ for $f\in\RS(u)$, and
$Q^{(k)}_f(a,b)=(f\bmod(X^k-b))(a)=\sum_jc_ja^{j\bmod k}b^{\lfloor j/k\rfloor}$ \cite[\S6]{BS08}. A \emph{view} is a
fixed set of proof symbols. We write $s=|D|$ for the size of the evaluation domain.

\paragraph{Charges, locators and samplers.} For a randomised linear prover and a view $V$, let $\Lambda(V)$ be the space of input
functionals on which the law of the view depends (Lemma~\ref{neg:lem:DL}). A set $W\subseteq D$ is a \emph{valid
charge} for $V$ if $\Lambda(V)\subseteq\spn\{\ev_x:x\in W\}$. A \emph{locator} is a map $V\mapsto W(V)$ that is
monotone, computable in polynomial time, a valid charge for every view, and of size $|W(V)|\le a|V|$. An \emph{exact
point-reading sampler} is a randomised oracle algorithm that queries values $C(x)$, $x\in D$, adaptively, halts with
probability~$1$, and outputs exactly the law of the view for every input; it has no running-time bound. A
straight-line simulator restricted to a fixed view is such a sampler. \emph{Statistical} samplers have output error
$\varepsilon$. Samplers that query other functionals of $C$ (for example Hermite jets) are not covered by the
point-charging bounds below.

\begin{lemma}[{Leak of a masked linear proof; cf.\ the masking argument of {\cite[Lemma~4.3, Claim~7.7]{BCFGRS17}}}]\label{neg:lem:DL}
Let the proof be $(\Pi_1(C+r),\Pi_2(r))$ with $\Pi_1,\Pi_2$ linear and $r$ uniform in $\RS(u)$, and let a view read the
functionals $\Phi$ of $C+r$ and $\Psi$ of $r$. Then the law of the view is uniform on an affine space whose
direction does not depend on $C$ and whose position depends on $C$ only through $\Lambda(V):=\spn\Phi\cap\spn\Psi$.
For every $h\in\Lambda(V)$ there is a fixed linear form $\lambda$ with $\lambda\cdot\pi|_V=h(C)$ for every $C$ and
$r$. Given $(h(C))_{h\in\Lambda(V)}$, the view is sampled exactly by Gaussian elimination.
\end{lemma}
\begin{proof}
For fixed $C$ the view is uniform on $(\Phi(C),0)+U$ with $U=\{(\Phi(r),\Psi(r)):r\in\RS(u)\}$. A form
$(\alpha,\beta)$ annihilates $U$ iff $\sum_i\alpha_i\phi_i=-\sum_j\beta_j\psi_j=:h$, that is, iff $h\in\Lambda(V)$;
then $(\alpha,\beta)\cdot(\Phi(C),0)=h(C)$. Since $U$ is the intersection of the kernels of its annihilators, the
position of the coset is determined by these values. Finally $(\alpha,\beta)\cdot\pi|_V=h(C+r)-h(r)=h(C)$.
\end{proof}

\begin{lemma}[Determined functionals force queries]\label{neg:lem:samp}
Let $\lambda\cdot\pi|_V=h(C)$ for every input $C\in\RS(u)$ and every outcome of the prover's coins, and let
$\mathcal F\subseteq\RS(u)^*$ be the family of functionals a sampler may query ($\mathcal F=\{\ev_x:x\in D\}$ for
point-reading samplers).
\begin{enumerate}
\item[(a)] For every exact sampler querying $\mathcal F$, every input $C$ and every halting run of positive
probability with queried set $J$, we have $h\in\spn J$. The same holds for every halting run with non-$\bot$ output
of a sampler that may output $\bot$ and is exact conditioned on not outputting $\bot$ (the convention of
\cite[Def.~4.3.1]{Gol01}).
\item[(b)] Over a finite field, for every $\varepsilon$-statistical sampler and $C$ uniform in $\RS(u)$, the
probability that the queried set $J$ fails $h\in\spn J$ is at most $\varepsilon|\F|/(|\F|-1)\le2\varepsilon$.
\end{enumerate}
\end{lemma}
\begin{proof}
(a) Exactness and $\lambda\cdot\pi|_V=h(C)$ give $\lambda\cdot o=h(C)$ for every (non-$\bot$) output $o$ in the
support. If $h\notin\spn J$, there is $\delta\in\RS(u)$ with $\theta(\delta)=0$ for all $\theta\in J$ and
$h(\delta)\ne0$. On input $C+\delta$ with the same coins, every query lies in $J$ and receives the same answer, so the
run, its positive probability and its output $o$ are unchanged; but $\lambda\cdot o=h(C)\ne h(C+\delta)$,
contradicting exactness on $C+\delta$.
(b) A transcript $\tau$ fixes $J_\tau$ and the output $o_\tau$. Given $\tau$, the input $C$ is uniform on an affine
space on which $h$ is non-constant when $h\notin\spn J_\tau$, so $\Pr[\lambda\cdot o_\tau=h(C)\mid\tau]=1/|\F|$. The
real pair (input, view) always satisfies $\lambda\cdot\pi|_V=h(C)$, and the joint laws are $\varepsilon$-close.
\end{proof}

By Lemma~\ref{neg:lem:samp}, a lower bound on the size of valid charges for some $h\in\Lambda(V)$ is a lower bound
on the number of input reads of every exact point-reading sampler for $V$, on every input and every halting run.

\begin{lemma}[Trivial and large views]\label{neg:lem:triv}
(a) Any $u$ distinct points of $D$ form a valid charge for every view.
(b) Let $W$ be monotone, polynomial-time, of size $\le a|V|$ and a valid charge on every view with $|V|<u/a'$. Put
$W'(V)=W(V)$ if $|V|<u/a'$ and $W'(V)=D$ otherwise. Then $W'$ is a locator with constant $\max(a,a'|D|/u)$.
\end{lemma}
\begin{proof}
(a) Their evaluations form a basis of $\RS(u)^*$ (Vandermonde). (b) Monotone because $W(V)\subseteq D$ and the
condition $|V|\ge u/a'$ is upward closed; valid by (a); and $|D|=(|D|/u)\,u\le(a'|D|/u)|V|$ when $|V|\ge u/a'$.
\end{proof}

\begin{lemma}[Lift reads are Cauchy vectors]\label{neg:lem:CR}
Let $b=x_0^{k}$, $F_b=x_0\mu_{k}$ and $a\notin F_b$. Then
$(C\bmod(X^{k}-b))(a)=\sum_{y\in F_b}\frac{(a^{k}-b)\,y}{k\,b\,(a-y)}\,C(y)$; with $y=x_0\omega$ and $a=x_0\nu$ the
weight is $\frac{\nu^{k}-1}{k}\cdot\frac{\omega}{\nu-\omega}\ne0$.
\end{lemma}
\begin{proof}
The remainder is the Lagrange interpolant of $C$ on $F_b$; its weight at $y$, evaluated at $a$, is
$\frac{a^{k}-b}{a-y}\cdot\frac{1}{k\,y^{k-1}}$, and $y^{k-1}=b/y$.
\end{proof}

\begin{lemma}[Linear-complexity certificates; \cite{Mas69}]\label{neg:lem:LC}
Let $\LC(h)$ be the linear complexity of the sequence $(h_0,\dots,h_{u-1})$ and $P_h$ the characteristic polynomial
of its shortest linear recurrence (as computed by the Berlekamp--Massey algorithm).
\begin{enumerate}
\item[(a)] If $h\in\spn\{\ev_w:w\in W\}$ for a finite $W\subseteq\Fb$, then $|W|\ge\LC(h)$.
\item[(b)] If $h\in\spn\{\ev_w:w\in W\}$ for some $W\subseteq\F^\times$ with $|W|\le u-\LC(h)$, then $P_h$ has
$\LC(h)$ distinct roots, all in $W$. Hence, if $P_h$ has a repeated root or a root outside $D$, every valid charge
$W\subseteq D$ has $|W|\ge u+1-\LC(h)$.
\end{enumerate}
\end{lemma}
\begin{proof}
(a) $h_j=\sum_wc_ww^j$ satisfies the recurrence with characteristic polynomial $\prod_w(Z-w)$.
(b) Work with generating functions. Let $L=\LC(h)$ and let $h(Z)\equiv A(Z)/K(Z)\pmod{Z^u}$ be the shortest
recurrence, with $K(0)=1$, $\deg K\le L$ and $\deg A<L$; by minimality $L=\max(\deg K,\deg A+1)$, since otherwise
$(A,K)$ would define a recurrence of length $L-1$. The hypothesis gives
$h(Z)\equiv B(Z)/\tilde Q(Z)\pmod{Z^u}$ with $\tilde Q=\prod_{w\in W}(1-wZ)$ and $\deg B<|W|$. Then
$A\tilde Q-BK\equiv0\pmod{Z^u}$ and $\deg(A\tilde Q-BK)<L+|W|\le u$, so $A\tilde Q=BK$. If $\deg K<L$ then
$\deg A=L-1\ge\deg K$, and $A\tilde Q=BK$ gives $\deg A+|W|=\deg B+\deg K<|W|+\deg K$, a contradiction; so
$\deg K=L$. If $G=\gcd(A,K)$ were non-constant, $A/G$ over $K/G$ would be a shorter recurrence; so $K$ divides
$\tilde Q$, and $P_h(X)=X^LK(1/X)$ is a product of $L$ distinct factors $X-w$, $w\in W$. The last sentence is the
contrapositive with $W\subseteq D\subseteq\F^\times$.
\end{proof}

\subsection{The square-root threshold}\label{neg:sqrt}

The quasilinear ZK-PCPs of \cite{ZOGS26} are built from locally simulatable sheaf codes: their Reed--Solomon PZK-PCPP
is $\sqrt n$-local with zero-knowledge query bound $\sqrt n/8-2$ \cite[Thm.~7.18]{ZOGS26}, which leads to length
$\tilde O(T+b^2)$ for $b$-query zero knowledge \cite[Thm.~1.1/8.2]{ZOGS26}. The remaining subsections show that the
same square-root threshold appears, on balanced splits, when the natural masked-lift configurations leave out a
hypothesis of the locator lemma marked P (proved necessary) in Section~\ref{neg:table}; for its other hypotheses,
nothing stronger than per-instance certificates is known. They also show that in the same-radix configuration itself some views have global leaks (leaks not confined to
single rows), which the locator of Theorem~\ref{kl:thm:KL} nevertheless pays for (Section~\ref{neg:SR}).

\subsection{A single lift is not enough}\label{neg:single}

\begin{definition}[Single lift]
Let $\varphi\in\F[X]$ have degree $n_1$, $D_X\ge n_1$, $D_Y\ge2$, $\mathcal B=\{Q:\deg_XQ<D_X,\ \deg_YQ<D_Y\}$ and
$n_1<u\le D_X+n_1(D_Y-1)$. A lift of $C\in\RS(u)$ is any $Q\in\mathcal B$ with $Q(X,\varphi(X))=C$; the prover writes a
(possibly random) lift on a grid $G_X\times G_Y$. For each row $b$ used, $\varphi-b$ has $n_1$ distinct roots
$F_b\subseteq D$. Examples: $\varphi=X^{n_1}$ on $D=\mu_n$, and the subspace polynomial of an $\F_2$-subspace of
$\mathrm{GF}(2^k)$ \cite{BS08}.
\end{definition}

\begin{theorem}[Tangent leak]\label{neg:thm:T}
Let $x_0\in D$, $b=\varphi(x_0)$, and let $V$ consist of $D_X$ points of row $b$ with distinct abscissae and $D_Y$
points of column $x_0$ with distinct ordinates (or the same with $n_1$ of the row points replaced by the input reads
$F_b$). Put $R_b=C\bmod(\varphi-b)$ and $\tau(C)=(C'(x_0)-R_b'(x_0))/\varphi'(x_0)$. Then the functionals $\ev_z$
($z\in F_b$) and $\tau$ are fixed linear functions of the view; every finite $W\subseteq\Fb$ whose evaluations span
them has $|W|\ge u$; and every exact point-reading sampler reads at least $u$ input points. The ratio is at least
$u/(D_X+D_Y)$, which is $\Omega(\sqrt s)$ on balanced splits.
\end{theorem}
\begin{proof}
$Y-\varphi(X)$ divides $Q(X,Y)-Q(X,\varphi(X))$, so $Q(X,b)\equiv C\pmod{\varphi-b}$; the row determines
$Q(X,b)=R_b+(\varphi-b)S$, hence $R_b$ (so every $C(z)=R_b(z)$, $z\in F_b$) and $S$. The column determines
$Q_Y(x_0,b)$. Differentiating $C(X)=Q(X,\varphi(X))$ gives $C'(x_0)=Q_X(x_0,b)+\varphi'(x_0)Q_Y(x_0,b)$, and
$Q_X(x_0,b)=R_b'(x_0)+\varphi'(x_0)S(x_0)$; so $S(x_0)+Q_Y(x_0,b)=\tau(C)$, for every lift. Let $|W|\le u-1$. As
$|W\cup\{z\}|\le u$, Vandermonde forces $F_b\subseteq W$. Then $C^*=\prod_{w\in W}(X-w)\in\RS(u)$ kills $\ev_W$, is
divisible by $\varphi-b$ (so $R_b(C^*)=0$), and $\tau(C^*)=\prod_{w\ne x_0}(x_0-w)/\varphi'(x_0)\ne0$. The sampler bound
is Lemma~\ref{neg:lem:samp}.
\end{proof}

The argument is characteristic-free: $\varphi'(x_0)\ne0$ because the roots of $\varphi-b$ are simple. It needs a
readable grid row through the curve point of an in-domain column, which the standard layout $G_Y=\varphi(D)$ of
\cite{BS08} has. In the complementary layout ($G_X$ a union of fibres, $G_Y=\varphi(D)\setminus H$) the same conclusion
holds only if the proof determines $n_1$ values of the row off the fibre; we do not know whether it does
(Section~\ref{neg:place}).

\begin{remark}[What Theorem~\ref{neg:thm:T} does not show]\label{neg:rem:BV}
Every functional exhibited in Theorem~\ref{neg:thm:T} lies in $\spn\{\ev_z:z\in F_b\}+\F\,\partial_{x_0}$, which is
the space of $n_1+1$ Hermite conditions (order two at $x_0\in F_b$, order one elsewhere on $F_b$). If $0\notin F_b$ and
$n_1+2\le u$, then $\ev_0$ is not in this space: $(X-x_0)\prod_{z\in F_b}(X-z)$ has degree $n_1+1<u$, satisfies all
the conditions and does not vanish at $0$. So Theorem~\ref{neg:thm:T} exhibits no leak of a Shamir-type secret
$C(0)$. It shows that \emph{point-charged simulation} of a grid-plus-curve view costs $u$ reads; it does \emph{not}
decide whether such joint views are private, which is the question of \cite[\S2.5]{BV26} (``We do not know how to
prove privacy for their joint local views'').
\end{remark}

\begin{proposition}[{Virtual column; the mechanism of {\cite[Lemma~4.15]{RW24}} and of the diagonal attack of \cite{BV26}}]\label{neg:prop:VC}
If $x_0\notin D$ and $V$ consists of $D_Y$ points of column $x_0$, then $\ev_{x_0}$ is a fixed linear function of the
view, and every valid $W\subseteq D$ has $|W|\ge u$.
\end{proposition}
\begin{proof}
$C(x_0)=Q(x_0,\varphi(x_0))$, and the column determines $Q(x_0,\cdot)$. If $|W|\le u-1$, the $|W|+1\le u$
evaluations at $W\cup\{x_0\}$ are independent (Vandermonde, $x_0\notin W$).
\end{proof}

\subsection{Two masked lifts must not share a column}\label{neg:shared}

From here on the proof is that of Lemma~\ref{neg:lem:DL} with $\Pi_1$ the radix-$k_1$ lift of $y=C+r$ on
$A_1\times\{x^{k_1}:x\in D\}$ and $\Pi_2$ the radix-$k_2$ lift of $r$ on $A_2\times\{x^{k_2}:x\in D\}$, together with
the domain values $y|_D$ and $r|_D$. The same-radix configuration of the main theorem has $k_1=k_2$; we also consider
the variant with coprime radices $k_1\ne k_2$, which we call \emph{transversal}. The statements of this subsection
allow both.

\begin{proposition}[Shared off-domain column]\label{neg:prop:SC}
If $c\in(A_1\cap A_2)\setminus D$, reading column $c$ on $\lceil u/k_1\rceil$ distinct lift-1 rows and on
$\lceil u/k_2\rceil$ distinct lift-2 rows puts $\ev_c$ into $\Lambda(V)$, and every valid $W\subseteq D$ has
$|W|\ge u$. The ratio is $u/(\lceil u/k_1\rceil+\lceil u/k_2\rceil)$, about $k_1k_2/(k_1+k_2)$.
\end{proposition}
\begin{proof}
The column polynomial $Q^{(k_1)}_y(c,\cdot)$ has degree $<\lceil u/k_1\rceil$, so it is determined by the reads, and
$y(c)=Q^{(k_1)}_y(c,c^{k_1})$; so $\ev_c\in\spn\Phi$. Symmetrically $\ev_c\in\spn\Psi$. Conclude as in
Proposition~\ref{neg:prop:VC}.
\end{proof}

\begin{corollary}[Shared in-domain column]\label{neg:cor:T2}
Let $x_0\in D$ be a column of both lifts, and let the view contain, for $i=1,2$, $\lceil u/k_i\rceil$ reads of lift
$i$ at abscissa $x_0$ on distinct rows, together with values determining $Q^{(k_i)}(\cdot,x_0^{k_i})$ ($k_i$ reads on
that row at distinct abscissae, or the $k_i$ domain reads of the same masked word on $x_0\mu_{k_i}$). Then $\ev_{x_0}$
and $\partial_{x_0}$ lie in $\Lambda(V)$, and every finite $W\subseteq\Fb$ with $\Lambda(V)\subseteq\spn\ev_W$ has
$|W|\ge u$. At $u=n=k_1k_2$ the ratio is at least $n/(2(k_1+k_2))$.
\end{corollary}
\begin{proof}
For a masked word $f$, the column determines $Q^{(k)}_f(x_0,Y)$, hence $f(x_0)$ and $\partial_YQ^{(k)}_f(x_0,x_0^k)$;
the row data determine $R=f\bmod(X^k-x_0^k)$. Differentiating $f(X)=Q^{(k)}_f(X,X^k)$ at $x_0$ gives
$f'(x_0)=R'(x_0)+kx_0^{k-1}\partial_YQ^{(k)}_f(x_0,x_0^k)$, a linear function of the view. Applied to $y$ with $k_1$
and to $r$ with $k_2$, this puts $\ev_{x_0}$ and $\partial_{x_0}$ into $\spn\Phi\cap\spn\Psi$. If $|W|\le u-1$,
Vandermonde forces $x_0\in W$, and $C^*=\prod_{w\in W}(X-w)$ kills $\ev_W$ while
$C^{*\prime}(x_0)=\prod_{w\ne x_0}(x_0-w)\ne0$.
\end{proof}

No hypothesis on the characteristic, coprimality or divisibility is used.

\begin{proposition}[Same position on both proofs]\label{neg:prop:SM}
Let both lifts have radix $k_1=k_2=n_1\le u$, let $a\in(A_1\cap A_2)\setminus D$ be a shared column, and let $V$ be
the two reads at one position $(a,b)$, $b=x_0^{n_1}$, $x_0\mu_{n_1}\subseteq D$. Then
$\phi(C)=(C\bmod(X^{n_1}-b))(a)$ lies in $\Lambda(V)$, $\phi\ne0$, and every $W\subseteq D$ with $\phi\in\spn\ev_W$
has $|W|\ge n_1$; the fibre attains this.
\end{proposition}
\begin{proof}
The two reads are the same functional $\phi$ of $y$ and of $r$, so $\phi\in\Lambda(V)$; the fibre attains the bound by
Lemma~\ref{neg:lem:CR}. If $|W|\le n_1-1$, $Z_W=\prod_W(X-x)$ has degree $<n_1\le u$, kills $\ev_W$, equals its own
remainder, and $\phi(Z_W)=Z_W(a)\ne0$ as $a\notin D$.
\end{proof}

Proposition~\ref{neg:prop:SM} gives the cheapest (two-read) view of a shared off-domain column, with ratio $n_1/2$.
For $u>n_1$ it is a related obstruction, not a special case of Proposition~\ref{neg:prop:SC}: its view is one symbol
of each lift at $(a,b)$, where Proposition~\ref{neg:prop:SC} reads $\lceil u/n_1\rceil\ge2$ symbols of column $a$ on
each, and its leak is the remainder evaluation $\phi\ne\ev_a$ ($\phi_{n_1}=b\ne a^{n_1}$), with charge $n_1<u$. For
$u=n_1$, $\phi=\ev_a$, and it is the case of Proposition~\ref{neg:prop:SC} with both reads on row $b$.

\subsection{Degrees not divisible by the radix: the top-block family}\label{neg:TB}

Put $e_i=\lceil u/k_i\rceil$ and $\rho_i=u-k_i(e_i-1)\in[1,k_i]$. A \emph{full column} of lift $i$ at $a$ is $e_i$
reads at abscissa $a$ on distinct rows.

\begin{proposition}[Top-block leak]\label{neg:prop:TB}
Say $\rho_1\ge\rho_2$. Let $V$ be $K_1$ full lift-1 and $K_2$ full lift-2 columns, at abscissae in $\F^\times$ that are
pairwise distinct (also across the two lifts), with $K_1\le\rho_1$, $K_2\le\rho_2$, $K_1\ge\rho_1-\rho_2+1$,
$K_2\ge1$ and $K_1+K_2\ge\rho_1+1$. Then $\Lambda(V)$ contains some $h\ne0$ with $h_j=0$ for $j<u-\rho_2$, and every
finite $W\subseteq\Fb$ with $h\in\spn\ev_W$ has $|W|\ge u-\rho_2+1$. This holds for $k_1\ne k_2$ and for $k_1=k_2$.
\end{proposition}
\begin{proof}
A full lift-$i$ column at $a$ determines its column polynomial and in particular its top coefficient
$\sum_{t<\rho_i}a^tc_{u-\rho_i+t}$. The vectors $(a^t)_{t<\rho_1}$ over the $K_1$ lift-1 abscissae span a space of
dimension $K_1$; its subspace vanishing on $t<\rho_1-\rho_2$ has dimension $\ge K_1-\rho_1+\rho_2$ and, after dropping
those coordinates, lies in $\F^{\rho_2}$, where the lift-2 vectors $(a^t)_{t<\rho_2}$ span dimension $K_2$. The
dimensions sum to more than $\rho_2$, so a common nonzero $v$ exists, and $h=\sum_tv_tc_{u-\rho_2+t}\in\Lambda(V)$. If
$|W|\le u-\rho_2$, the equations $h=\sum_W\lambda_w\ev_w$ at the indices $j<u-\rho_2$ form a Vandermonde system of full
column rank, forcing $\lambda=0$.
\end{proof}

\begin{proposition}[Scope of the family]\label{neg:prop:TBscope}
Every member has ratio at most $k_{\max}/(\rho_{\max}+1)$, where $\rho_{\max}=\max(\rho_1,\rho_2)$. For coprime
$k_1\ne k_2$ and $u=k_1k_2-1$ it is at most $1$.
\end{proposition}
\begin{proof}
$|V|\ge(K_1+K_2)\min(e_1,e_2)\ge(\rho_{\max}+1)u/k_{\max}$, and the charge is at most $u$. At $u=k_1k_2-1$,
$\rho_i=k_i-1$.
\end{proof}

\begin{corollary}[$u\equiv1$ modulo the radices]\label{neg:cor:TB1}
(a) Transversal radices: if $u\equiv1\pmod{k_1k_2}$, one full column per lift leaks $c_{u-1}$ and every $W$ has
$|W|\ge u$; at $u=k_1k_2+1$ the ratio is $(k_1k_2+1)/(k_1+k_2+2)$.
(b) Equal radices $k_1=k_2=k$: if $u\equiv1\pmod k$, one full column per lift (at distinct abscissae) leaks $c_{u-1}$
and every $W$ has $|W|\ge u$; the ratio is $u/(2\lceil u/k\rceil)$, about $k/2$.
\end{corollary}
\begin{proof}
Proposition~\ref{neg:prop:TB} with $\rho_1=\rho_2=1$ and $K_1=K_2=1$.
\end{proof}

\begin{corollary}[A bad degree below $n$ for every transversal pair]\label{neg:cor:TB1p}
Let $k_1,k_2\ge3$ be coprime, $n=k_1k_2$, $\epsilon_1,\epsilon_2\in[0,n)$ the CRT idempotents
($\epsilon_i\equiv1\bmod k_i$, $\equiv0$ modulo the other radix) and $u^*=1+\max(\epsilon_1,\epsilon_2)$. Then
$(n+3)/2\le u^*\le n-1$, neither radix divides $u^*$, $\{\rho_1,\rho_2\}=\{1,2\}$, and two full columns of the lift
with $\rho=2$ plus one full column of the other leak $c_{u^*-1}$, which forces $|W|\ge u^*$ for every finite
$W\subseteq\Fb$. The ratio is at least $k_{\min}/(3+6/k_{\max})$.
\end{corollary}
\begin{proof}
$\epsilon_1+\epsilon_2\equiv1$ modulo both radices and $\epsilon_i\in[2,n-1]$, so $\epsilon_1+\epsilon_2=n+1$ and
$\epsilon_{\rm hi}\ge(n+1)/2$. If $\epsilon_{\rm hi}=\epsilon_i$ then $u^*\equiv2\pmod{k_i}$ and $u^*\equiv1$ modulo the
other radix, so $\rho_i=2$ and $\rho_j=1$, and neither radix divides $u^*$; $u^*=n$ would need $k_i\mid2$.
Proposition~\ref{neg:prop:TB} with $K_i=2$, $K_j=1$ gives $h\propto c_{u^*-1}$, and for $|W|\le u^*-1$ the polynomial
$X^{u^*-1-|W|}\prod_W(X-w)$ kills $\ev_W$ and has top coefficient $1$. The cost is at most $3u^*/k_{\min}+3$, and
$u^*\ge n/2$.
\end{proof}

\noindent\emph{Examples.} For the coprime pairs with $3\le k_1<120$, $k_1<k_2<2k_1$, the ratio of
Corollary~\ref{neg:cor:TB1p} is at least $0.311\,k_1$, attained at $(7,9)$; it is $33.25$ at $(101,103)$ and $68.37$ at
$(199,211)$.

So a locator lemma stated for every degree needs a divisibility hypothesis on $u$ that excludes these degrees. One
such hypothesis is $k_{\min}\mid u$ (hypothesis (DIV)), which caps the family (Proposition~\ref{neg:prop:TBDIV}); for
coprime radices $k_{\max}\mid u$ is another, since then $\rho_{\max}=k_{\max}$ and every member has
$|V|\ge(k_{\max}+1)u/k_{\max}>u$ (proof of Proposition~\ref{neg:prop:TBscope}). The family is not harmless at
most other degrees either: for each constant $a$, the cheapest member beats $a$ on a fraction of the non-(DIV) degrees
$u\in[n/2,n]$ that is close to $1/a^2$; at $(61,67)$ the fractions for $a=1,2,4,8$ are $0.965$, $0.224$, $0.049$, $0.009$,
and $(31,37)$, $(101,103)$, $(64,97)$ are similar (exact computation).

\begin{proposition}[Under (DIV)]\label{neg:prop:TBDIV}
If $1<k_1<k_2$ are coprime and $k_1\mid u$, every member has $|V|>(k_1+1)u/k_2$, so its ratio is $<k_2/(k_1+1)$.
\end{proposition}
\begin{proof}
$\rho_1=k_1$, $e_1=u/k_1$ and $e_2\ge u/k_2$. If $\rho_2>k_1$ then $K_1+K_2\ge\rho_2+1>k_1+1$. Otherwise $k_2\nmid u$,
so $e_1,e_2>u/k_2$, and $K_1+K_2\ge k_1+1$.
\end{proof}

\begin{proposition}[Sharpness along integer split ratios]\label{neg:prop:TBsharp}
Let $k_1\ge2$, $\kappa\ge2$, $k_2=\kappa k_1+1$ and $u=(k_1-1)k_2+k_1+1$. Then $k_1\mid u$, and the member with one
full lift-1 column and $k_1+1$ full lift-2 columns has $|V|=k_1^2+k_1+(k_1-1)\kappa+2$ and forces
$|W|\ge(k_1-1)k_2+2$; its ratio tends to $\kappa$.
\end{proposition}
\begin{proof}
$u=k_1((k_1-1)\kappa+2)$ gives $\rho_1=k_1$; $u\bmod k_2=k_1+1$ gives $\rho_2=k_1+1$ and $e_2=k_1$.
\end{proof}

The split ratio need not be an integer. For coprime $t>\ell\ge1$ and $k_1\ge\max(2,\ell)$ with $\ell\mid tk_1+1$, put
$k_2=(tk_1+1)/\ell$ and $u=k_1(k_2+1-t)$. Then $\gcd(k_1,k_2)=1$, $k_1\mid u$ and $u=(k_1-\ell)k_2+k_1+1$ with $k_1+1<k_2$, so
$\rho_1=k_1$, $\rho_2=k_1+1$ and $e_2=k_1-\ell+1$; the member with one full lift-1 column and $k_1+1$ full lift-2 columns
has $|V|=(k_1+1)(k_1-\ell+1)+k_2+1-t$ and forces $|W|\ge k_1(k_2-t)+1$, and its ratio tends to $t/\ell$ as $k_1\to\infty$
(Proposition~\ref{neg:prop:TBsharp} is the case $\ell=1$, $t=\kappa$). So (DIV) alone forces $a\ge K-o(1)$ as the split
ratio $K=k_2/k_1$ tends to any rational $t/\ell>1$, and the cap of Proposition~\ref{neg:prop:TBDIV} is attained
asymptotically. For $K<2$
the (DIV)-only ratio is $<2$. An exact scan of all coprime $3\le k_1<90$, $k_1<k_2<3k_1$, over
the degrees $u\in\{k_1,2k_1,\dots,k_1k_2\}$ ($220{,}707$ cases with $K<2$ and $367{,}845$ with $2\le K<3$)
finds best ratios $1.903$ at $(89,177,15664)$ and $2.800$ at $(89,266,23496)$, never at the cap of
Proposition~\ref{neg:prop:TBDIV}. If $k_2\mid u$ as well, every member has $|V|>u$, so the family is absent; this is the
case $u=m$ of the main theorem, for equal radices and for coprime ones.

\subsection{Transversal lifts with unbounded column sets}\label{neg:coset}

In this subsection $k_1=n_1$ and $k_2=n_2$ are coprime, $n=n_1n_2$, $u=n$, and $d\mu_n\subseteq D$. Each
$y\in d\mu_n$ is $d\omega\eta$ with $(\omega,\eta)\in\mu_{n_1}\times\mu_{n_2}$; the lift-1 row $\eta$ has fibre
$d\eta\mu_{n_1}$ and the lift-2 row $\omega$ has fibre $d\omega\mu_{n_2}$ (the CRT grid of Good \cite{Goo58}, Burton--Weldon \cite{BW65} and Imai \cite{Imai77}).
The CRT idempotents satisfy $\epsilon_1+\epsilon_2=n+1$, $z^{\epsilon_1}=\omega$ and $z^{\epsilon_2}=\eta$ for
$z=\omega\eta$; $\epsilon_{\rm lo}=\min(\epsilon_1,\epsilon_2)$. For $B(\omega,\eta)$ without zeros on the grid put
\[f_B=\sum_{\omega,\eta}\frac{\omega\eta}{B(\omega,\eta)}\,\ev_{d\omega\eta},\qquad
f_{B,R}=\sum_{\omega,\eta}\frac{\omega\eta\,R(\omega,\eta)}{B(\omega,\eta)}\,\ev_{d\omega\eta}.\]
By Lemma~\ref{neg:lem:CR}, a lift-1 read at $d\eta\nu$ is the array supported on row $\eta$ with entries
$\gamma_\nu\,\omega/(\nu-\omega)$, $\gamma_\nu\ne0$, and symmetrically for lift~2.

\begin{proposition}[Aligned product leak]\label{neg:prop:AL}
Let $\theta_1=\lambda/d$ and $\theta_2=\lambda'/d$ with $\theta_1^{n_1}\ne1\ne\theta_2^{n_2}$, and read lift-1 row $\eta$ at
$\lambda\eta$ and lift-2 row $\omega$ at $\lambda'\omega$, for all rows ($|V|=n_1+n_2$). Then
$\Lambda(V)=\langle h\rangle$ with $h_j=d^j\theta_1^{j\bmod n_1}\theta_2^{j\bmod n_2}$, and $h\propto f_B$ for
$B=(\omega-\theta_1)(\eta-\theta_2)$.
\end{proposition}
\begin{proof}
The read on row $\eta$ is $j\mapsto d^j\theta_1^{j\bmod n_1}\eta^j$, and the characters $\eta^j$ span all $n_2$-periodic
functions of $j$. Under CRT the two spans are $(\theta_1^t)\otimes\F^{\mathbb Z_{n_2}}$ and $\F^{\mathbb Z_{n_1}}\otimes(\theta_2^s)$,
which meet in a line. Lagrange interpolation on $\mu_{n_1}$ gives
$\theta_1^t=\sum_\omega\frac{(\theta_1^{n_1}-1)\omega}{n_1(\theta_1-\omega)}\omega^t$, and likewise for $\theta_2$.
\end{proof}

\begin{proposition}[One read per line: M{\"o}bius classification]\label{neg:prop:BL}
With at most one read per lift-1 row (at $\nu_\eta$) and per lift-2 row (at $\nu'_\omega$), all off their fibres, and
no domain reads, $\Lambda(V)\ne0$ iff every row is read and some bilinear $B=\alpha\omega\eta+\beta\omega+\gamma\eta+\delta$
without grid zeros vanishes at all $(\nu_\eta,\eta)$ and $(\omega,\nu'_\omega)$; then $\Lambda(V)=\langle f_B\rangle$.
\end{proposition}
\begin{proof}
A leak $\sum_\eta g_\eta v_\eta=\sum_\omega g'_\omega v'_\omega$ of Cauchy arrays with nonzero entries needs all
coefficients nonzero (a zero $g_\eta$ empties column $\eta$ of the left side and forces everything to vanish). Then
$(\nu'_\omega-\eta)/(\nu_\eta-\omega)=P(\omega)Q(\eta)$, so $Q(\eta)(\nu_\eta-\omega)$ is affine in each variable on
the grid, hence bilinear; it vanishes at the read points and nowhere on the grid. The converse is
Lemma~\ref{neg:lem:RK} with $i=j=1$.
\end{proof}

\begin{lemma}[Rational kernels]\label{neg:lem:RK}
Let $B$ have no grid zero, with $B(\cdot,\eta)$ of degree $i$ with $i$ distinct roots $\nu_{\eta,s}\in\F^\times$ for
every $\eta$, and $B(\omega,\cdot)$ of degree $j$ with $j$ distinct roots $\nu'_{\omega,t}\in\F^\times$ for every
$\omega$. The view $V_B$ reading lift-1 row $\eta$ at all $d\eta\nu_{\eta,s}$ and lift-2 row $\omega$ at all
$d\omega\nu'_{\omega,t}$ ($|V_B|=in_2+jn_1$) has $0\ne f_B\in\Lambda(V_B)$, and $\dim\Lambda(V_B)\ge ij$ when $i\le n_1$
and $j\le n_2$.
\end{lemma}
\begin{proof}
By partial fractions the reads on row $\eta$ span $\{\omega P(\omega)/\prod_s(\nu_{\eta,s}-\omega):\deg P<i\}$ on its
fibre, and $\prod_s(\nu_{\eta,s}-\omega)$ is a constant multiple of $B(\omega,\eta)$; row $\eta$ of $f_{B,R}$ has this
form when $\deg_\omega R<i$. Symmetrically for lift~2. The map $R\mapsto f_{B,R}$ is injective for $\deg_\omega R<i\le
n_1$, $\deg_\eta R<j\le n_2$.
\end{proof}

\begin{lemma}[Spread idempotents; the BCH bound]\label{neg:lem:BLW}
Let $\tilde B(z)=B(z^{\epsilon_1},z^{\epsilon_2})=\sum_{j\in\mathbb Z_n}\tilde b_jz^j$ on $\mu_n$,
$E^*=\{j\notin\{0,1\}:\tilde b_j\ne0\}\ne\emptyset$, and $\sigma(B)$ the least cyclic gap of $E^*\cup\{0\}$. Every
finite $W\subseteq\Fb$ with $f_B\in\spn\ev_W$ has $|W|\ge\sigma(B)$.
\end{lemma}
\begin{proof}
Let $\sigma_{\max}\ge\sigma$ be the largest cyclic gap of $E^*\cup\{0\}$, $W'=W\setminus d\mu_n$ and $k=|W'|$. If
$k\ge\sigma$, then $|W|\ge\sigma$; so let $k\le\sigma-1$. Testing $f_B=\sum_Ww_x\ev_x$ on
$C_y=(X^n-d^n)/(X-y)$, $y=d\omega\eta\notin W$, gives $nd^{n-1}/B(\omega,\eta)=\sum_{x\in W'}\hat w_x/(x-y)$. Clearing
denominators with $P(Y)=\prod_{W'}(x-Y)$, the function $\Phi(z)=N(dz)\tilde B(z)-nd^{n-1}P(dz)$ ($\deg N\le k-1$)
vanishes on $\{z:dz\notin W\}$, is nonzero, and has its spectrum in the pairwise disjoint intervals $[0,k]$ and
$e+[0,k-1]$ ($e\in E^*$); so it avoids a run of $g\ge\sigma_{\max}-k-1$ consecutive residues, and by the BCH bound
(see \cite{MS77}) it has at most $n-g-1$ zeros on $\mu_n$. Hence $|W\cap d\mu_n|\ge g+1$ and
$|W|=k+|W\cap d\mu_n|\ge\sigma_{\max}\ge\sigma$ in this case.
\end{proof}

\begin{corollary}[Cost of a M{\"o}bius incidence]\label{neg:cor:MOB}
For a M{\"o}bius kernel with $(\beta,\gamma)\ne(0,0)$, every $W$ with $f_B\in\spn\ev_W$ has
$|W|\ge\sigma_M=\min(\epsilon_{\rm lo}-1,n+1-2\epsilon_{\rm lo})$. If $\beta=\gamma=0$, both sides read the single
column $-d\delta/\alpha$, which is a shared column (Section~\ref{neg:shared}).
\end{corollary}
\begin{proof}
$\tilde B=\alpha z+\beta z^{\epsilon_1}+\gamma z^{\epsilon_2}+\delta$; the gaps of $\{0,\epsilon_{\rm lo},\epsilon_{\rm hi}\}$
are $\epsilon_{\rm lo}$, $n+1-2\epsilon_{\rm lo}$ and $\epsilon_{\rm lo}-1$, and those of $\{0,e\}$ are at least
$\epsilon_{\rm lo}-1$.
\end{proof}

\begin{theorem}[Unbounded column sets]\label{neg:thm:F}
Let $n_1\equiv1\pmod3$, $n_2=n_1+3$, $D$ eight cosets of $\mu_n$ in $\F_p$, $u=n$, and column sets containing
$\lambda\mu_n$ and $\lambda'\mu_n$ in two further cosets. The view of Proposition~\ref{neg:prop:AL}
($|V|=2n_1+3$) forces $|W|\ge n_1(n_1+2)/3-1\ge(n_1/6-1/2)|V|$ for every $W\subseteq\Fb$. Hence, for two masked lifts
with coprime radices whose column sets may contain unboundedly many cosets, and with \emph{no} design law, no locator
has a constant ratio, nor a ratio $n^{o(1)}$.
\end{theorem}
\begin{proof}
$3^{-1}\equiv(2n_1+1)/3\pmod{n_1}$ gives $\epsilon_1=n_2(2n_1+1)/3>n/2$ and $\epsilon_{\rm lo}=\epsilon_2=n_1(n_1+2)/3$, so
$\sigma_M=(n_1^2+2n_1-3)/3$ (Corollary~\ref{neg:cor:MOB}), and $2(n_1^2+2n_1-3)-(n_1-3)(2n_1+3)=7n_1+3>0$.
\end{proof}

The witnesses are M{\"o}bius incidences on every row, which a design law excludes; so Theorem~\ref{neg:thm:F} shows
that bounded column-coset counts \emph{and} a design law cannot both be dropped, not that bounded counts are necessary
when a design law is kept.

\begin{corollary}[Free column sets]\label{neg:cor:Ffree}
If the column sets could be any $c_An_i$ points off $D$ ($c_A\ge2$), the view of Theorem~\ref{neg:thm:F} would be
admissible, so the conclusion of Theorem~\ref{neg:thm:F} would hold.
\end{corollary}
\begin{proof}
The view reads only $\lambda\mu_{n_2}$ ($n_2<2n_1$ points) on side~1 and $\lambda'\mu_{n_1}$ on side~2; take
$\lambda,\lambda'$ in distinct cosets of $\mu_n$ off $D$.
\end{proof}

\begin{theorem}[Three restrictions on the column sets do not suffice]\label{neg:thm:R}
For every constant $a$, the conclusion of Theorem~\ref{neg:thm:F} (no locator of ratio $a$ for coprime radices with
unbounded column-coset counts and no design law) persists when the column sets are further restricted by (ALN) (no
$\mu_n$-coset in $A_1$), (ALN$_2$) (none in either set) or (HALF) (fewer than $n_{\min}/2$ cosets of $A_i$ in each
$\mu_n$-coset). The witnesses are the views of Lemma~\ref{neg:lem:RK} for $B=\omega-c\eta^b$, $c=\gamma^b\notin\mu_n$,
$d=1$, over $\F_p$ for primes $p\equiv1\pmod{bn}$ with $(p-1)/n>13b+45$ (they exist by Dirichlet's theorem), where
$\gamma\in\F_p^\times$ is chosen so that $A_1$ and $A_2$ are disjoint and do not meet $D$ (this is possible, see the proof):
(F2) $b=1$, $n_1\equiv1\pmod6$, $n_2=n_1+3$, $A_1=c\mu_{n/2}$, $A_2=c^{-1}\mu_n$, $\sigma=n_1(n_1+2)/3-1$, $|V|=2n_1+3$;
(F3) $b=5$, $n_1\equiv38\pmod{42}$, $5\nmid n_1$, $n_2=n_1+7$, $A_1=c\mu_{n/3}$, $A_2=\gamma^{-1}\mu_{n/2}$ or
$\gamma^{-1}\mu_{5n/2}$, $\sigma=(2n-6n_2+35)/7$, $|V|=6n_1+7$;
(F4) $b=11$, $n_1\equiv39\pmod{660}$, $n_2=n_1+5$, $A_1=c\mu_{n/4}$, $A_2=\gamma^{-1}\mu_{n/3}$, $\sigma=(n+n_2)/5$,
$|V|=12n_1+5$.
\end{theorem}
\begin{proof}
$B(\cdot,\eta)$ has the root $c\eta^b$ and $B(\omega,\cdot)$ the $b$ roots $\gamma^{-1}\omega^s\zeta$, $s=b^{-1}\bmod
n_1$, $\zeta\in\mu_b$. Lemma~\ref{neg:lem:RK} gives the leak, and $\tilde B=z^{\epsilon_1}-cz^{b\epsilon_2}$ with
Lemma~\ref{neg:lem:BLW} gives the bound. (F3): $n_1$ is even, $\equiv2\pmod3$ and $\equiv3\pmod7$; $7^{-1}\equiv(2n_1+1)/7$;
$\epsilon_1=(2n+n_2)/7$ and $5\epsilon_2\equiv(4n-5n_2+35)/7\pmod n$; the gaps are $(2n+n_2)/7$, $(2n-6n_2+35)/7$ and
$(3n+5n_2-35)/7$. (F4): $n_1\equiv3\pmod4$, $3\mid n_1$, $n_1\equiv4\pmod5$, $n_1\equiv6\pmod{11}$; $5^{-1}\equiv(n_1+1)/5$;
$\epsilon_1=(n+n_2)/5$ and $11\epsilon_2\equiv(4n-11n_2+55)/5$; the gaps are $(n+n_2)/5$, $(3n-12n_2+55)/5$ and
$(n+11n_2-55)/5$. The side-1 abscissae $c\eta^{b+1}$ fill $c\mu_{n_2/\gcd(b+1,n_2)}$ and the side-2 abscissae
$\gamma^{-1}\omega^{1+s}\zeta$ fill the stated sets; a union $\gamma'\mu_{n/g}$ contains a $\mu_n$-coset iff $g=1$.
A suitable $\gamma$ exists once $(p-1)/n>13b+45$, and primes $p\equiv1\pmod{bn}$ with this property exist by Dirichlet's
theorem. Indeed, in the cyclic group $\mathcal G=\F_p^\times/\mu_n$ let $x$ be the class of $\gamma$; then $A_1$ lies in the class
$x^b$ and $A_2$ in the classes $x^{-1}\mathcal Z$, where $\mathcal Z$ is the image of $\mu_{5n}$ in $\mathcal G$ if $A_2=\gamma^{-1}\mu_{5n/2}$
and $\mathcal Z=\{1\}$ otherwise, so $|\mathcal Z|\le5$. It suffices that $x^b$ and all of $x^{-1}\mathcal Z$ avoid the eight classes of $D$
and that $x^{b+1}\notin\mathcal Z$, which excludes at most $8b+40+5(b+1)$ classes $x$. (Since $\mu_n$
is a class of $D$, this also gives $c\notin\mu_n$, and $p\equiv1\pmod b$ puts the $b$ roots $\zeta\in\mu_b$ in $\F_p$.)
\end{proof}

\paragraph{Examples.} Lemma~\ref{neg:lem:LC}(a) gives per-instance bounds above $\sigma$. The values below were
computed exactly over $\F_p$.

\smallskip
{\small\centering
\begin{tabular}{lrrrl}
\toprule
instance & $|V|$ & $\sigma$ & $\LC(f_B)$ & ratio ($\sigma$ / $\LC$) \\
\midrule
Thm.~\ref{neg:thm:F}, $(7,10)$ & 17 & 20 & 30 & 1.18 / 1.76 \\
Thm.~\ref{neg:thm:F}, $(13,16)$ & 29 & 64 & 80 & 2.21 / 2.76 \\
Thm.~\ref{neg:thm:F}, $(31,34)$ & 65 & 340 & 374 & 5.23 / 5.75 \\
Thm.~\ref{neg:thm:R} (F2), $(7,10)$ & 17 & 20 & 41 & 1.18 / 2.41 \\
Thm.~\ref{neg:thm:R} (F2), $(13,16)$ & 29 & 64 & 129 & 2.21 / 4.45 \\
Thm.~\ref{neg:thm:R} (F3), $(38,45)$ & 235 & 455 & 800 & 1.94 / 3.40 \\
Thm.~\ref{neg:thm:R} (F4), $(39,44)$ & 473 & 352 & 781 & 0.74 / 1.65 \\
\bottomrule
\end{tabular}\par}

\begin{proposition}[Monomial kernels are harmless on bounded designs]\label{neg:prop:MON}
Assume $n_1\le n_2$. For $B=\omega^\alpha-c\eta^\beta$ or $B=\omega^\alpha\eta^\beta-c$ put $\ell=\alpha+\beta$ resp.\
$\ell=\alpha-\beta$. Side~1's abscissae meet at least $n_2/\gcd(\ell,n_2)$ cosets of $\mu_{n_1}$ and side~2's at least
$n_1/\gcd(\ell,n_1)$ cosets of $\mu_{n_2}$; if $\ell=0$ the sides share abscissae. Hence if each column set is a union
of at most $c_A$ cosets and $n_1>c_A^2$, the view is inadmissible or its minimal charge is smaller than $|V_B|$.
\end{proposition}
\begin{proof}
Two side-1 abscissae of rows $\eta_1,\eta_2$ in one $\mu_{n_1}$-coset force $\eta_1^\ell=\eta_2^\ell$. If both counts
are at most $c_A$ and $\ell\ne0$, then $|\ell|\ge n/c_A^2$, so $|V_B|\ge(\alpha+\beta)n_1\ge nn_1/c_A^2>n$, while $n$
points always suffice.
\end{proof}

\subsection{The same-radix configuration: global leaks exist, and the locator pays for them}\label{neg:SR}

We recall the setting of the locator lemma (conditions (N1)--(N4) of Definition~\ref{kl:def:node}). $n<n'$ are
coprime, $m=nn'<p$, $\mu_m\le\F^\times$, $D=\bigsqcup_{i<L_2}d_i\mu_m$, $u=m$; the rows are $B=\{x^n:x\in D\}$, with
fibres $F_b=\{x\in D:x^n=b\}$; the column set $A=A_y\sqcup A_r\subseteq\F^\times\setminus D$ is a union of $L_1$
cosets of $\mu_n$. Side~1 reads $y=g+r$ at $A_y\times B$ and at the domain, side~2 reads $r$ at $A_r\times B$ and at
the domain, both with radix $n$. Put $U=\F[X]_{<n}$, $U'=\F[Z]_{<n'}$, $v(a)=\ev_a\in U^*$, $w(b)=\ev_b\in U'^*$.

\begin{lemma}[{Tensor form; {\cite[\S6]{BS08}}}]\label{neg:lem:T}
The map $f=\sum_{t<n}X^tf_t(X^n)\mapsto\sum_tX^t\otimes f_t$ is an isomorphism $\F[X]_{<m}\to U\otimes U'$ under which
the lift read at $(a,b)$ is $v(a)\otimes w(b)$ and the domain read at $x$ is $v(x)\otimes w(x^n)$. For $b\in B$,
$U^*\otimes w(b)=\spn\{\ev_x:x\in F_b\}$ when $|F_b|\ge n$.
\end{lemma}
\begin{proof}
$f\bmod(X^n-b)=\sum_tX^tf_t(b)$ and $f(x)=\sum_tx^tf_t(x^n)$; $n$ distinct points of $F_b$ give a basis of $U^*$.
\end{proof}

\begin{proposition}[One read per row and per column; Cayley--Bacharach]\label{neg:prop:SRMOB}
Let $V$ consist of $N$ lift reads at points $(a_i,b_i)$ with distinct $a_i\in A$, distinct $b_i\in B$, at least one on
each side, and no domain reads; let $n,n'\ge2$.
(a) If $N\le n+n'-1$, the functionals $v(a_i)\otimes w(b_i)$ are linearly independent, so $\Lambda(V)=0$.
(b) If $N=n+n'$, then $\Lambda(V)\ne0$ iff there is $\psi\in\mathrm{PGL}_2(\F)$ with $b_i=\psi(a_i)$ for all $i$; then
$\dim\Lambda(V)=1$.
\end{proposition}

That $n+n'$ points with distinct coordinates fail to impose independent conditions on forms of bidegree $(n-1,n'-1)$ iff
they lie on a curve of bidegree $(1,1)$ is a special case of classical Cayley--Bacharach residuation (\cite{EGH96} treats it in $\mathbb P^n$; for finite
sets of points in $\mathbb P^1\times\mathbb P^1$ see \cite{GMRag92,GVT15}); we include a short proof for completeness and claim
no novelty for it (cf.\ also the M{\"o}bius circuits of Theorem~\ref{kl:thm:CL}). A related classical fact is the
action of $\mathrm{PGL}_2$ on generalised Reed--Solomon codes (cf.\ \cite{Dur87}).

\begin{proof}
(a) Suppose $\sum_{i\in T}c_iv(a_i)\otimes w(b_i)=0$ with all $c_i\ne0$ and $|T|\le n+n'-1$. Fix $i_0\in T$ and split
$T\setminus\{i_0\}=S_P\sqcup S_Q$ with $|S_P|\le n-1$ and $|S_Q|\le n'-1$. Evaluating on $P\otimes Q$ with
$P=\prod_{S_P}(X-a_j)$ and $Q=\prod_{S_Q}(Z-b_j)$ leaves $c_{i_0}P(a_{i_0})Q(b_{i_0})\ne0$. A nonzero leak would give
such a relation.
(b) If $\psi(x)=(\alpha x+\beta)/(\gamma x+\delta)$ with $b_i=\psi(a_i)$, then $\gamma a_i+\delta\ne0$, and for every
$P\in U\otimes U'$ the function $(\gamma X+\delta)^{n'-1}P(X,\psi(X))$ is a polynomial of degree at most $n+n'-2$. So the
$N$ evaluations have rank at most $n+n'-1$, and a relation exists; by (a) it has full support and is unique up to
scale. Its side-1 part is nonzero (side~1 is nonempty) and spans $\Lambda(V)$, so $\dim\Lambda(V)=1$. Conversely, let a
relation $(c_i)$ have full support (it must, by (a), if $\Lambda(V)\ne0$). For distinct indices $i_0,i_1,j,l$ choose
$S_P\ni j$ and $S_Q\ni l$ partitioning $[N]\setminus\{i_0,i_1\}$ with $|S_P|=n-1$, $|S_Q|=n'-1$, and let $S'_P$, $S'_Q$
be the same sets with $j$ and $l$ exchanged. Testing the relation on the two corresponding products leaves two
terms each:
\begin{align*}
c_{i_0}(a_{i_0}-a_j)(b_{i_0}-b_l)K_{i_0}&=-c_{i_1}(a_{i_1}-a_j)(b_{i_1}-b_l)K_{i_1},\\
c_{i_0}(a_{i_0}-a_l)(b_{i_0}-b_j)K_{i_0}&=-c_{i_1}(a_{i_1}-a_l)(b_{i_1}-b_j)K_{i_1},
\end{align*}
with the same nonzero $K_{i_0},K_{i_1}$ (products over $S_P\setminus\{j\}$ and $S_Q\setminus\{l\}$). Dividing gives
\[\mathrm{CR}(a_{i_0},a_{i_1};a_j,a_l)=\mathrm{CR}(b_{i_0},b_{i_1};b_j,b_l),\]
where
$\mathrm{CR}(x_0,x_1;x_j,x_l)=\frac{(x_0-x_j)(x_1-x_l)}{(x_0-x_l)(x_1-x_j)}$. So the M{\"o}bius map determined by
$a_{i_0}\mapsto b_{i_0}$, $a_j\mapsto b_j$, $a_l\mapsto b_l$ sends every $a_{i_1}$ to $b_{i_1}$.
\end{proof}

\begin{proposition}[Global leaks from saturated abscissae]\label{neg:prop:SATSR}
Let $|A_y|,|A_r|\ge n$ and $|B|\ge n'+2$; fix a row $b_1\in B$ and nonempty $T_y\subseteq A_y$, $T_r\subseteq A_r$.
\begin{enumerate}
\item[(a)] \emph{One saturated abscissa.} Let $|T_y|+|T_r|=n$, read every $a\in T_y$ (side~1) and every $a\in T_r$
(side~2) on $b_1$, and read one more abscissa $a_0\in A_y\setminus T_y$ on $n'$ other rows; no domain reads, so
$|V|=n+n'$. Then $\dim\Lambda(V)=1$ and $\Lambda(V)\subseteq U^*\otimes w(b_1)=\spn\ev(F_{b_1})$, while every row
carries a set of side-1 abscissae and a set of side-2 abscissae whose spans meet in $0$ (no row has a ``local'' leak).
\item[(b)] \emph{Near-saturation.} Read the same $n$ abscissae on two rows $b_1,b_2$ and $a_0$ on $n'-1$ other rows
($|V|=2n+n'-1$). Then $0\ne\Lambda(V)\subseteq\spn\ev(F_{b_1}\cup F_{b_2})$, and again no row has a local leak.
\item[(c)] \emph{Light saturation.} Let $D$ have eight cosets, $\theta\in[\tfrac14,1]$, $n\ge9$ and $r=\lceil\theta n\rceil-1$; read $r$
abscissae on $b_1$ (split between the sides, both sides nonempty after adding the next item) and $s=n+1-r$ further
side-1 abscissae, each on $n'$ rows other than $b_1$, spread round-robin over the other $7$ row cosets. Then
$\Lambda(V)\ne0$, no row carries $\theta n$ reads, and $(1-\theta)u<|V|\le(1-\theta)u+3n'$.
\end{enumerate}
\end{proposition}
\begin{proof}
By Lemma~\ref{neg:lem:T} the side spans are $\{\sum_av(a)\otimes\Gamma_a:\Gamma_a\in E_a\}$ with
$E_a=\spn\{w(b):(a,b)\text{ read}\}$; any $n$ distinct $v(a)$ are independent and any $n'$ distinct $w(b)$ form a basis.
(a) The $n+1$ vectors $v(a)$, $a\in T_y\cup T_r\cup\{a_0\}$, satisfy one relation, which has full support; write it as
$\varphi:=\sum_{T_y\cup\{a_0\}}x_av(a)=\sum_{T_r}x_av(a)\ne0$. Since $E_{a_0}=U'^*\ni w(b_1)$, $\varphi\otimes w(b_1)$
lies in both spans. Confinement: applying $\mathrm{id}\otimes\vartheta$ for every $\vartheta\in U'$ with
$\vartheta(w(b_1))=0$ to a common element forces $\Gamma_{a_0}\in\spn w(b_1)$, so $\Lambda(V)\subseteq U^*\otimes w(b_1)$,
and the one-dimensional relation space gives $\dim\Lambda(V)=1$. Row $b_1$ carries $n$ distinct, hence independent,
abscissae split between the sides, so their spans meet in $0$; every other row carries one read.
(b) $E_{a_0}$ is spanned by $n'-1$ of the $n'+1$ vectors $w(b)$ of its rows and $b_1,b_2$, so it meets
$\spn\{w(b_1),w(b_2)\}$ in a line $\gamma=\alpha w(b_1)+\beta w(b_2)$ with $\alpha\beta\ne0$; then $\varphi\otimes\gamma$
is in both spans. Confinement and the local statement follow as in (a).
(c) The leak follows as in (a) (the $s$ saturated abscissae carry $v(a)\otimes w(b_1)$ virtually). Loads:
$\lceil s/7\rceil\le((1-\theta)n+8)/7<\theta n$ because $(8\theta-1)n>8$, and $r<\theta n$. Size: $|V|=sn'+r$ with
$(1-\theta)n+1\le s\le(1-\theta)n+2$.
\end{proof}

So the stronger statement ``every view of size at most $cu$ of a node has only local leaks'' is false at every node,
with no design law involved. It is \emph{not} what the locator lemma asserts: the locator lemma asks for a locator,
and these leaks are located.

\begin{proposition}[Consistency with the locator lemma]\label{neg:prop:consistent}
Let $\tau=\lfloor c_*n\rfloor$ and let
\[W(V)=D_V\cup\textstyle\bigcup_{b\in H_\tau(V)}F_b\cup\big(d_0\mu_m\text{ if }|V|\ge(\tau+1)^2\big)\]
be the locator of Theorem~\ref{kl:thm:KL}, where $H_\tau(V)$ is the set of rows carrying at least $n+1-\tau$ view points, $D_V$ the set of
domain points read, and $c_*=10^{-13}$. Then $W(V)$ is a valid charge for the views of
Proposition~\ref{neg:prop:SATSR}(a), (b) for every $n$, and of Proposition~\ref{neg:prop:SATSR}(c) whenever
$\theta\le\tfrac34$. Every other family of this appendix violates a hypothesis of the locator lemma: a single lift
(Section~\ref{neg:single}); a shared column, $A\cap D\ne\emptyset$ or $A_y\cap A_r\ne\emptyset$
(Section~\ref{neg:shared}); $u\ne m$ (Section~\ref{neg:TB}); transversal radices or non-coset column sets
(Section~\ref{neg:coset}). Free column sets (Examples~\ref{neg:prop:free} and~\ref{neg:prop:deg} below) violate a
hypothesis only asymptotically, through an unbounded number $L_1$ of column cosets: a view whose two sides read
nonzero abscissae off $D$, with no $\mu_n$-coset met by both sides, is a view of the node whose column set is the union
of the $\mu_n$-cosets of these abscissae. So each computed instance of this kind is a view of a node (the $12$-read
instances, $(5,7)$ in Example~\ref{neg:prop:free} and $k=2$ at $(4,5)$ in Example~\ref{neg:prop:deg}, span at most $12$
column cosets) and is located because $\tau=0$ at its size; an instance whose two sides meet a common coset, or that reads an abscissa in $D$, is a
shared-column view (Section~\ref{neg:shared}), which violates a hypothesis. The
scale-model views of Example~\ref{neg:prop:scale} below satisfy the hypotheses and are located because $\tau=0$ at
their sizes.
\end{proposition}
\begin{proof}
(a), (b): if $\tau\ge1$, rows $b_1$ (and $b_2$) carry $n\ge n+1-\tau$ points, so $F_{b_1}$ ($\cup F_{b_2}$) $\subseteq
W(V)$, which spans $\Lambda(V)$. If $\tau=0$ then $(\tau+1)^2=1\le|V|$ and $d_0\mu_m\subseteq W(V)$ spans $\RS(u)^*$.
(c): $|V|>(1-\theta)nn'\ge n^2/4$. Since $c_*<1/10$, $(\tau+1)^2\le(n/10+1)^2\le n^2/4$ for $n\ge3$; so the coset branch
applies. The list of violated hypotheses is read off the statements. Every nonempty view with $\tau=0$ is located, since
then $(\tau+1)^2=1\le|V|$ and $d_0\mu_m\subseteq W(V)$; the empty view has $\Lambda(V)=0$.
\end{proof}

At every computationally reachable size $\tau=0$, so the locator of Theorem~\ref{kl:thm:KL} contains a whole
$D$-coset for every nonempty view and no exact computation can test it there. The computations below therefore test
the families, not the locator lemma.

\noindent\emph{Examples.} On nodes $(5,6)$, $(7,8)$, $(8,9)$, $(9,10)$, $(11,12)$ of the same-radix configuration
(eight $D$-cosets, $7+7$ column cosets, the least prime $p\equiv1\pmod m$ with $(p-1)/m\ge22$): in (a) and (b)
$\dim\Lambda(V)=1$, $\Lambda(V)\subseteq\spn\ev(F_{b_1})$ (resp.\ $F_{b_1}\cup F_{b_2}$), the row part of $W$ with
$\tau=0$ is empty and does not locate, and with $\tau=1$ it locates; in (c) at $\theta=\tfrac12$, $\dim\Lambda(V)=1$,
the maximal row load is below $n/2$, and $|V|/u$ ranges from $0.87$ down to $0.67$.

\begin{proposition}[Product expansion forbids M{\"o}bius circuits]\label{neg:prop:PENEC}
Let $S_1,S_2\subseteq\F$ be finite, $C_1=\RS[S_1,n]^\perp$ and $C_2=\RS[S_2,n']^\perp$. If $S_1\times S_2$ contains $n+n'$
points $(a_i,\psi(a_i))$ of the graph of some $\psi\in\mathrm{PGL}_2(\F)$, with distinct coordinates, then the pair
$(C_1,C_2)$ is not $\rho$-product-expanding in the sense of \cite[Def.~1.3]{BV26} for any $\rho>(n+n')/(n|S_2|)$.
\end{proposition}
\begin{proof}
By Proposition~\ref{neg:prop:SRMOB}(b) and its proof the points carry a relation $M$ of weight $n+n'$ with one point per
line. Take any decomposition $M=M_1+M_2$ with the columns of $M_1$ in $C_1$ and the rows of $M_2$ in $C_2$. If
$M_1\ne0$, a nonzero column of $M_1$ is a $C_1$-word of weight $\ge n+1$ on a line where $M$ has at most one point, so
$M_2$ is nonzero on at least $n$ rows, and $|\mathrm{supp}\,M|=n+n'\ge\rho|S_2|\,n$. If $M_1=0$, then $M=M_2$ has nonzero
rows of weight $\le1<n'+1$, which is impossible.
\end{proof}

The notion of product expansion is that of Panteleev and Kalachev \cite{PK22} (\cite[Def.~1.3]{BV26} is the
two-dimensional form); \cite{PS94} proved it for pairs of Reed--Solomon codes whose rates sum to less than one, and
its use for privacy is \cite[Lemma~4.9]{GG25} (see \cite[\S\S1, 2.1]{BV26}). The proposition is stated on any grid on which
a M{\"o}bius graph has $n+n'$ points. So ordinary product expansion with a constant $\rho$ contains the incidence law
``no $n+n'$ points of a M{\"o}bius graph''. On coset designs this law holds above a threshold depending only on $L_1L_2$ by
\cite[Thm.~2]{CZ13} (Theorem~\ref{kl:thm:mob} and the discussion after it, with threshold $m_0\approx2.3\cdot10^{16}$ for $(L_1,L_2)=(14,8)$),
and also as a corollary of the locator lemma for $n>3/c^2$ with $c$ admissible (e.g.\ $n>3/c_*^2$ for
$(L_1,L_2,K)=(14,8,2)$); below the threshold it can fail (the scale models
below), which the locator lemma absorbs through its coset branch.

\begin{example}[Free column sets: instance certificates]\label{neg:prop:free}
If the column set is not required to be a union of cosets, a M{\"o}bius graph through $n+n'$ grid points with both sides
present gives, by Proposition~\ref{neg:prop:SRMOB}(b), a one-dimensional leak at $n+n'$ reads. In our exact
computation (two random M{\"o}bius maps each, eight $D$-cosets, $u=m$) the certified lower bounds of
Lemma~\ref{neg:lem:LC} are $18$, $32$ and $45$ at $(n,n')=(5,7)$, $(7,9)$, $(8,11)$, against $|V|=12$, $16$, $19$ (ratios
$1.50$, $2.00$, $2.37$, about $u/(2(n+n'))$); moving one point off the graph kills the leak in every case.
\end{example}

These are per-instance certificates; no general lower bound is claimed. They indicate that the coset structure of the
column set (condition (N4)) is what the incidence input of the locator lemma uses.

\begin{example}[Degree-$k$ graphs on free designs]\label{neg:prop:deg}
On free column sets, the points $(\sigma(b),b)$ of the graph of a polynomial map $\sigma$ of degree $k$, one read per
row and per abscissa, both sides present, give in our exact computations ($k=2,3$ at $(4,5)$, $(5,6)$, $(5,7)$) a
one-dimensional leak at $k(n-1)+n'+1$ reads and none at one read fewer; every row carries a single read, so every
such leak is global.
\end{example}

These computations cover only $k\le3$ at three small nodes; we do not prove that such graphs leak for general $k$.
If they do, any ``no global leak'' form of a locator lemma for the same-radix configuration would have to exclude such
curves up to degree about $cn'$, i.e.\ impose design laws of \emph{linear} degree; this extrapolation is a conjecture,
not a result. The locator form needs none: the locator
lemma uses no design law (Theorem~\ref{kl:thm:KL}), and its incidence input (\cite[Thm.~2]{CZ13} on each pair of
cosets) is a property of every coset design with $m<p$.

\begin{example}[Small-field threshold leaks are scale-model phenomena]\label{neg:prop:scale}
On same-radix coset designs at orders and fields that the recursion of the main theorem never produces, M{\"o}bius
maps can have $n+n'$ incidences with both sides present, and the resulting $(n+n')$-read views leak. For five
instances, an exact computation gives
\[\begin{array}{lcccc}
\text{node, field} & |V| & \dim\Lambda & \LC & \text{certified }|W|\ (\text{ratio})\\
(5,7),\ \F_{3011} & 12 & 1 & 18 & 18\ (1.50)\\
(7,8),\ \F_{2017} & 15 & 1 & 28 & 29\ (1.93)\\
(9,10),\ \F_{2161} & 19 & 1 & 45 & 46\ (2.42)\\
(9,10),\ \F_{2791},\ L_2=16 & 19 & 1 & 45 & 46\ (2.42)\\
(11,12),\ \F_{4093},\ L_2=16 & 23 & 1 & 66 & 67\ (2.91)
\end{array}\]
(all points are symbols of the stated designs). These views are minimal, by
Proposition~\ref{neg:prop:SRMOB}(a).
\end{example}

\noindent\emph{Scope.} (i) These are scale models, not nodes of our construction: every split node of the
construction has all prime factors above $2455$ (Remark~\ref{main:rem:2455}), and root fields satisfy $32n\mid q-1$; at the actual root fields of
$(7,8)$, $(8,9)$, $(9,10)$, $(12,13)$ the maximal incidence counts are $9$, $9$, $11$, $17$, far below the threshold
$n+n'$. (ii) The bounds are per view: some views of the same family at $(7,8)$ over $\F_{1289}$ are located by exactly
$m/2=28$ points, so no universal bound ``$\ge m/2+1$'' holds. (iii) They are located by the locator $W$ of
Theorem~\ref{kl:thm:KL}, since $\tau=0$ at these sizes. (iv) They bound from below the constant of any locator lemma
stated for all nodes (Definition~\ref{kl:def:node}) at these sizes ($a\ge2.42$ at $L_2=8$, $a\ge2.91$ at $L_2=16$); the
constant of Theorem~\ref{kl:thm:KL} ($a\approx2\cdot10^{26}$) is consistent with this.

\begin{proposition}[No subgroup-aligned child layout]\label{neg:prop:DC}
In the child layout of our construction (Lemma~\ref{main:lem:CS}(d)) the children of a node include the row codes on
$D_b:=F_b\cup A_y$ (and $F_b\cup A_r$), $b\in B$, with the same $A_y$ for every $b$. Let $k\mid n$ with $k<n$, put
$k'=n/k$, and let $H$ be a subgroup of $\F^\times$. Then for two rows $b\ne b'$:
(i) $D_b$ and $D_{b'}$ are not both cosets of $H$;
(ii) if $A_y$ is a union of $7$ cosets of $\mu_n$ lying in distinct $\mu_m$-cosets, then $\{x^k:x\in D_b\}$ and
$\{x^k:x\in D_{b'}\}$ are not both single cosets of $\mu_{8k'}$.
\end{proposition}
\begin{proof}
(i) Two cosets of $H$ that share the nonempty set $A_y$ are equal, so $F_b=D_b\setminus A_y=D_{b'}\setminus A_y=F_{b'}$;
but distinct fibres are disjoint and nonempty. (ii) $A_y^{(k)}:=\{a^k:a\in A_y\}$ is a union of $7$ cosets of
$\mu_{k'}$, pairwise distinct because $(e/e')^k\in\mu_{k'}$ forces $e/e'\in\mu_n\subseteq\mu_m$; and $F_b^{(k)}$ is one
coset of $\mu_{k'}$, disjoint from $A_y^{(k)}$ because $x^k=a^k$ would put $x\in D$ and $a\in A$ in one $\mu_m$-coset.
If both images are cosets of $\mu_{8k'}$, they share $A_y^{(k)}\ne\emptyset$ and so are equal, and
$F_b^{(k)}=F_{b'}^{(k)}$ is the single remaining $\mu_{k'}$-coset. Then $x^k=x'^k$ for some $x\in F_b$, $x'\in F_{b'}$,
so $x/x'\in\mu_k\subseteq\mu_n$ and $b=x^n=x'^n=b'$, a contradiction.
\end{proof}

Consequently, refinements of the incidence bound that require subgroup-aligned grids, and the approach through
\cite[Thm.~4.11]{BV26} on aligned grids, apply at the root only. This is one reason the locator lemma had to be proved
for arbitrary coset designs; Theorem~\ref{kl:thm:KL} does so without any alignment.

\subsection{Which hypotheses are necessary}\label{neg:table}

Categories: \textbf{P} = proved necessary (dropping it alone admits views of unbounded ratio charge$/|V|$);
\textbf{I} = per-instance certificates only; \textbf{N} = proved \emph{unnecessary}, by the locator lemma;
\textbf{M} = used by the proof, no negative known.

\smallskip
{\small\noindent
\begin{tabularx}{\textwidth}{@{}>{\raggedright\arraybackslash}p{0.27\textwidth}c>{\raggedright\arraybackslash}X@{}}
\toprule
hypothesis (of the locator lemma, or one that might be added) & cat. & evidence \\
\midrule
two masked lifts, not one & P & a single lift leaks at ratio $\ge u/(D_X+D_Y)=\Omega(\sqrt s)$ under point charging on the
standard layout (Theorem~\ref{neg:thm:T}, Proposition~\ref{neg:prop:VC}) \\
lift columns off the domain and not shared ($A\cap D=\emptyset$, $A_y\cap A_r=\emptyset$) & P & a shared column off or in
the domain forces charge $u$ at ratio about $n/2$--$n/4$ (Proposition~\ref{neg:prop:SC},
Corollary~\ref{neg:cor:T2}); for equal radices, two reads at one position of a shared column off the domain force
charge $n$, ratio $n/2$ (Proposition~\ref{neg:prop:SM}) \\
$u=m$ (radix divides the degree) & P & for degrees $u\equiv1$ modulo the radix, one full column per side forces charge
$u$, ratio about $n/2$ (Corollary~\ref{neg:cor:TB1}(b)); transversal analogue Corollary~\ref{neg:cor:TB1p} \\
column set a union of boundedly many cosets & I\,/\,P & I for equal radices, P for coprime radices. Equal radices on
free column sets: certificates of ratio about $u/(2(n+n'))$ at three sizes (Example~\ref{neg:prop:free}); coprime
radices: unbounded or free column sets fail for every constant (Theorem~\ref{neg:thm:F}, Corollary~\ref{neg:cor:Ffree},
Theorem~\ref{neg:thm:R}) \\
two lifts of coprime radices & N & not needed: the lemma uses equal radices (Theorem~\ref{kl:thm:KL}) \\
a design law (incidence restrictions) & N & not needed (Theorem~\ref{kl:thm:KL}); conjecturally, any ``no global
leak'' form would need laws of linear degree (Proposition~\ref{neg:prop:SATSR}; Example~\ref{neg:prop:deg}, computed for
$k\le3$ only) \\
a lower bound on the node size $m_1$ & N & not needed: the locator lemma has no size threshold; small nodes are absorbed
by the coset branch (Example~\ref{neg:prop:scale}) \\
$m<p$; at most $L_2$ row cosets & M & used by the incidence input and the constant; no negative known \\
\bottomrule
\end{tabularx}
}

\subsection{Summary and open questions}\label{neg:place}

Every departure from the hypotheses of the locator lemma that is proved to leak expensively lands at a
$\Theta(\sqrt s)$ ratio on balanced splits: a single lift on the standard layout (ratio at
least $u/(D_X+D_Y)$), a shared column (about $n_1/4$ in
the domain, $n_1/2$ off it), degrees not divisible by the radix (about $n/2$ for equal radices, about $n_{\min}/3$ for
transversal ones), and unbounded or free column sets for coprime radices (about $n_1/6$ down to $n_2/60$). This
matches the $\sqrt n$-locality of the sheaf simulation of \cite[Thm.~7.18]{ZOGS26} and the remark of \cite{GOS25}
that quasilinear zero-knowledge PCPs need new ideas. It also locates where one had to act: the joint privacy of two
masked lifts of the same radix on a coset-union grid, against views that do have global leaks
(Proposition~\ref{neg:prop:SATSR}). The locator lemma (Theorem~\ref{kl:thm:KL}) acts there with one polynomial
through the whole view, a domain-column refinement of the divisibility argument of \cite[Lemma~8, Thm.~9]{PS94} (in
the version of \cite[\S4.4]{BV26} for multiplicative groups of coprime orders, with the routine extension to unions of cosets), and the incidence bound \cite[Thm.~2]{CZ13}. All results of this appendix
concern one simulation paradigm; none shows that a composed PCP is not zero knowledge.

\paragraph{Open questions.} Section~\ref{sec:intro-limits} lists open questions about the main theorem; those
below concern the simulation paradigm of this appendix. (i) Whether the tangent leak of
Theorem~\ref{neg:thm:T} extends to the complementary
single-lift layout, that is, whether the proof there determines $n_1$ values of a row off its fibre. (ii) A general
lower bound for free column sets in the same-radix configuration, beyond the instances of
Example~\ref{neg:prop:free}. (iii) Whether the variant of the paradigm in which samplers may also read Hermite
data (derivatives) at domain points has a square-root barrier.

\endgroup

\section*{AI methodology}
\addcontentsline{toc}{section}{AI methodology}
We used Anthropic's Claude extensively, under the direction of the authors, to assist with developing the constructions
and the proofs, with checking the proofs, with the computations reported in Appendix~\ref{sec:negatives}, and with
drafting this paper. Its use
materially affected every section of the paper, including both appendices. Prior work is credited at each use
(Section~\ref{sec:related}).
The authors have also checked all proofs themselves. 

\begingroup
\providecommand{\arxivid}[1]{arXiv:~\texttt{\detokenize{#1}}}

\clearpage
\phantomsection

\endgroup

\end{document}